\documentclass[11pt]{amsbook}

\usepackage{amsmath}
\usepackage{amssymb}
\usepackage{amsthm}
\usepackage{enumerate}
\usepackage{epsf}
\usepackage{psfrag}
\usepackage{graphics}
\usepackage{mathrsfs}
\usepackage{hyperref}

\newcommand{\arctanh}{\operatorname{arctanh}}

\newcommand{\bbR}{\mathbb{R}}

\newcommand{\bbN}{\mathbb{N}}

\newcommand{\cL}{\mathcal{L}}

\newcommand{\contr}{\lrcorner \,}

\newcommand{\diag}{\operatorname{diag}}
\newcommand{\grad}{\operatorname{grad}}
\newcommand{\curl}{\operatorname{curl}}
\newcommand{\dive}{\operatorname{div}}

\newcommand{\tr}{\operatorname{tr}}

\newcommand{\cof}{\operatorname{cof}}

\newcommand{\inte}{\operatorname{int}}

\newtheorem{Thm}{Theorem}[section]
\newtheorem{Prop}[Thm]{Proposition}
\newtheorem{Lemma}[Thm]{Lemma}

\newtheorem{Cor}[Thm]{Corollary}

\newtheorem{Def}[Thm]{Definition}

\theoremstyle{definition}

\newtheorem{Example}[Thm]{Example}

\newtheorem{Remark}[Thm]{Remark}

\makeindex

\begin{document}
\begin{titlepage}

\vspace*{5cm}

\begin{center}
{\Huge \bf Mathematical Relativity} 
\end{center}

\vspace{2.5cm}

\begin{center}

{\normalsize \Large \bf Jos\'e Nat\'ario}\\

\vspace{4.5cm}
Lisbon, 2020

\end{center}

\end{titlepage}

\tableofcontents


\chapter*{Preface}

These lecture notes were written for a one-semester course in mathematical relativity aimed at mathematics and physics students, which has been taught at Instituto Superior T\'ecnico (Universidade de Lisboa) since 2010. They are not meant as an introduction to general relativity, but rather as a complementary, more advanced text, much like Part II of Wald's textbook \cite{W84}, on which they are loosely based. It is assumed that the reader is familiar at least with special relativity, and has taken a course either in Riemannian geometry (typically the mathematics students) or in general relativity (typically the physics students). In other words, the reader is expected to be proficient in (some version of) differential geometry and to be acquainted with the basic principles of relativity. 

I thank the many colleagues and students who read this text, or parts of it, for their valuable comments and suggestions. Special thanks are due to my colleague and friend Pedro Gir\~ao.

\chapter{Preliminaries} \label{chapter1}

In this initial chapter we give a very short introduction to special and general relativity for mathematicians. In particular, we relate the index-free differential geometry notation used in Mathematics (e.g.~\cite{ONeill83, Carmo93, Boothby03, GN14}) to the index notation used in Physics (e.g.~\cite{MTW73, W84, HE95}). As an exercise in index gymnastics, we derive the contracted Bianchi identities.

\section{Special relativity} \label{sec1.1}

Consider an inertial frame $S'$ moving with velocity $v$ with respect to another inertial frame $S$ along their common $x$-axis (Figure~\ref{inercial}). According to classical mechanics, coordinate $x'$ of a point $P$ on the frame $S'$ is related to its $x$ coordinate on the frame $S$ by
\[
x' = x - vt.
\]
Moreover, a clock in $S'$ initially synchronized with a clock in $S$ is assumed to keep the same time:
\[
t' = t.
\]
Thus the spacetime coordinates of events are related by a so-called {\bf Galileo transformation}
\[
\begin{cases}
x' = x - vt \\
t' = t
\end{cases}.
\]

\begin{figure}[h!]
\begin{center}
\psfrag{S}{$S$}
\psfrag{S'}{$S'$}
\psfrag{y}{$y$}
\psfrag{y'}{$y'$}
\psfrag{x'}{$x'$}
\psfrag{vt}{$vt$}
\psfrag{P}{$P$}
\epsfxsize=.6\textwidth
\leavevmode
\epsfbox{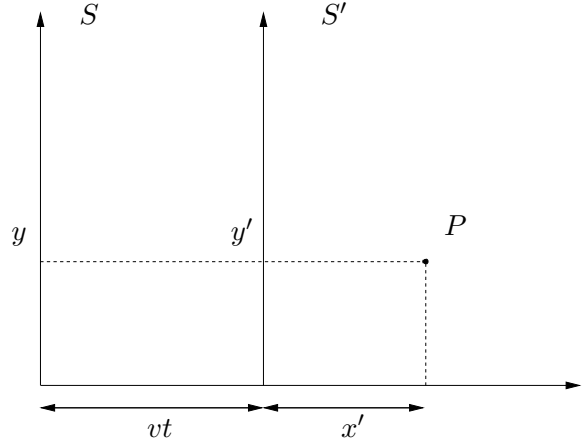}
\end{center}
\caption{Galileo transformation.} \label{inercial}
\end{figure}

If the point $P$ is moving, its velocity in $S'$ is related to its velocity in $S$ by
\[
\frac{dx'}{dt'} = \frac{dx - vdt}{dt} = \frac{dx}{dt} - v.
\]
This is in conflict with the experimental fact that the speed of light is the same in every inertial frame, indicating that classical mechanics is not correct. Einstein solved this problem in 1905 by replacing the Galileo transformation by the so-called {\bf Lorentz transformation}:
\[
\begin{cases}
x' = \gamma(x - vt) \\
t' = \gamma(t - vx)
\end{cases}.
\]
Here
\[
\gamma = \frac{1}{\sqrt{1-v^2}},
\]
and we are using units such that the speed of light is $c=1$ (for example measuring time in years and distance in light-years). Note that if $|v|$ is much smaller than the speed of light, $|v| \ll 1$, then $\gamma \simeq 1$, and we retrieve the Galileo transformation (assuming $\left|v\frac{x}{t}\right| \ll 1$).

Under the Lorentz transformation velocities transform as
\[
\frac{dx'}{dt'} = \frac{\gamma(dx - vdt)}{\gamma(dt - vdx)} = \frac{\frac{dx}{dt} - v}{1 - v \frac{dx}{dt}}.
\]
In particular,
\[
\frac{dx}{dt} = 1 \Rightarrow \frac{dx'}{dt'} = \frac{1 - v}{1 - v } = 1,
\]
that is, the speed of light is the same in the two inertial frames.

In 1908, Minkowski noticed that
\[
- (dt')^2 + (dx')^2 = - \gamma^2 (dt - vdx)^2 + \gamma^2 (dx - vdt)^2 = - dt^2 + dx^2,
\]
that is, the Lorentz transformations could be seen as isometries of $\bbR^4$ with the indefinite metric
\[
ds^2 = - dt^2 + dx^2 + dy^2 + dz^2 = - dt \otimes dt + dx \otimes dx + dy \otimes dy + dz \otimes dz.
\]

\begin{Def}
The pseudo-Riemannian manifold $(\bbR^4, ds^2) \equiv (\bbR^4, \langle \cdot, \cdot \rangle)$ is called the {\bf Minkowski spacetime}.
\end{Def}

Note that the set of vectors with zero square form a cone (the so-called {\bf light cone}):
\[
\langle v,v \rangle = 0 \Leftrightarrow -(v^0)^2 + (v^1)^2 + (v^2)^2 + (v^3)^2 = 0.
\]

\begin{Def}
A vector $v\in\bbR^4$ is said to be:
\begin{enumerate}
\item
{\bf timelike} if $\langle v, v \rangle < 0$;
\item
{\bf spacelike} if $\langle v, v \rangle > 0$;
\item
{\bf lightlike}, or {\bf null}, if $\langle v, v \rangle = 0$.
\item
{\bf causal} if it is timelike or null;
\item
{\bf future-pointing} if it is causal and $\left\langle v, \frac{\partial}{\partial t} \right\rangle < 0$.
\end{enumerate}
The same classification applies to (smooth) curves $c:[a,b] \to \bbR^4$ according to its tangent vector.
\end{Def}

\begin{figure}[h!]
\begin{center}
\psfrag{p}{$p$}
\psfrag{null}{null vector}
\psfrag{timelike future-pointing}{timelike future-pointing vector}
\psfrag{spacelike}{spacelike vector}
\psfrag{dt}{$\frac{\partial}{\partial t}$}
\psfrag{dx}{$\frac{\partial}{\partial x}$}
\psfrag{dy}{$\frac{\partial}{\partial y}$}
\epsfxsize=.8\textwidth
\leavevmode
\epsfbox{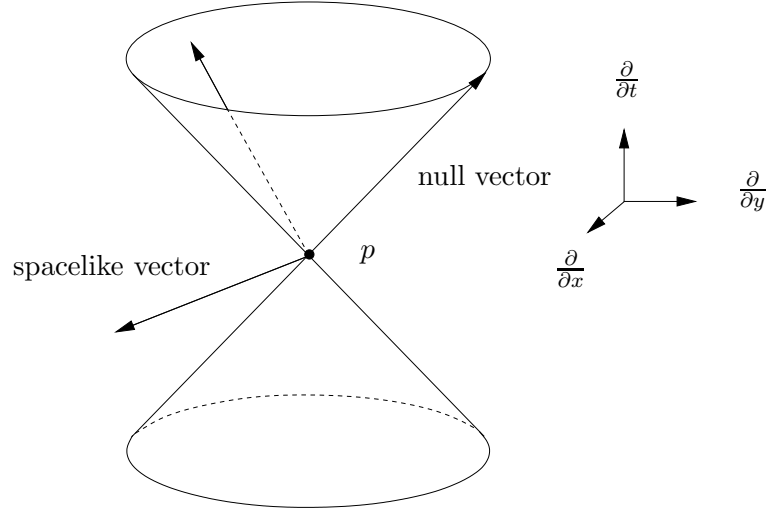}
\end{center}
\caption{Minkowski geometry (traditionally represented with the $t$-axis pointing upwards).}
\end{figure}

The length $|\langle v, v \rangle|^\frac12$ of a timelike (resp.~spacelike) vector $v \in \bbR^4$ represents the time (resp.~distance) measured between two events $p$ and $p + v$ in the inertial frame where these events happen in the same location (resp.~are simultaneous). If $c:[a,b] \to \bbR^4$ is a timelike curve then its length
\[
\tau(c) = \int_a^b |\langle \dot{c}(s), \dot{c}(s) \rangle|^\frac12 ds
\]
represents the {\bf proper time} measured by the particle between events $c(a)$ and $c(b)$. We have:

\begin{Prop} {\em (Twin paradox)}
Of all timelike curves connecting two events $p, q \in \bbR^4$, the curve with {\bf maximal} length is the line segment (representing inertial motion).
\end{Prop}

\begin{proof}
We may assume $p=(0,0,0,0)$ and $q=(T,0,0,0)$ on some inertial frame, and parameterize any timelike curve connecting $p$ to $q$ by the time coordinate:
\[
c(t)=(t,x(t),y(t),z(t)).
\]
Therefore
\[
\tau(c) = \int_0^T \left|-1+\dot{x}^2+\dot{y}^2+\dot{z}^2\right|^\frac12 dt = \int_0^T \left(1-\dot{x}^2-\dot{y}^2-\dot{z}^2\right)^\frac12 dt \leq \int_0^T 1 dt = T.
\]
\end{proof}

Most problems in special relativity can be recast as questions about the geometry of the Minkowski spacetime.

\begin{Prop} {\em (Doppler effect)}
An observer moving with velocity $v$ away from a source of light of period $T$ measures the period to be
\[
T' = T \sqrt{\frac{1+v}{1-v}} .
\]
\end{Prop}

\begin{proof}
Figure~\ref{Doppler} represents two light signals emitted by an observer at rest at $x=0$ with a time difference $T$. These signals are detected by an observer moving with velocity $v$, who measures a time difference $T'$ between them. Now, if the first signal is emitted at $t=t_0$, its history is the line $t = t_0 + x$. Consequently, the moving observer detects the signal at the event with coordinates
\[
\begin{cases}
t = t_0 + x \\ \\
\displaystyle x = vt
\end{cases}
\Leftrightarrow
\begin{cases}
\displaystyle t = \frac{t_0}{1-v} \\
\\
\displaystyle x = \frac{vt_0}{1-v}
\end{cases}.
\]

\begin{figure}[h!]
\begin{center}
\psfrag{t}{$t$}
\psfrag{x}{$x$}
\psfrag{x=vt}{$x=vt$}
\psfrag{T}{$T$}
\psfrag{T'}{$T'$}
\epsfxsize=.5\textwidth
\leavevmode
\epsfbox{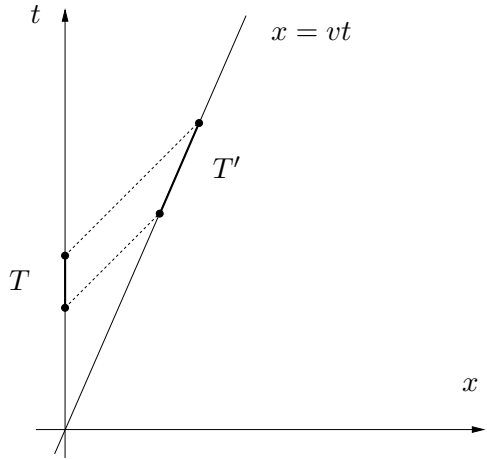}
\end{center}
\caption{Doppler effect.} \label{Doppler}
\end{figure}

Similarly, the second light signal is emitted at $t=t_0 + T$, its history is the line $t = t_0 + T + x$, and it is detected by the moving observer at the event with coordinates
\[
\begin{cases}
\displaystyle t = \frac{t_0 + T}{1-v} \\
\\
\displaystyle x = \frac{v (t_0 + T)}{1-v}
\end{cases}.
\]
Therefore the time difference between the signals as measured by the moving observer is
\begin{align*}
\hspace{2cm} T' & = \sqrt{\left(\frac{t_0 + T}{1-v} - \frac{t_0}{1-v} \right)^2 - \left(\frac{v (t_0 + T)}{1-v} - \frac{v t_0}{1-v} \right)^2} \\
& = \sqrt{\frac{T^2}{(1-v)^2} - \frac{v^2 T^2}{(1-v)^2}} = T \sqrt{\frac{1-v^2}{(1-v)^2}} = T \sqrt{\frac{1+v}{1-v}}.
\end{align*}
\end{proof}

In particular, two observers at rest in an inertial frame measure the same frequency for a light signal (Figure~\ref{desvio}). However, because the gravitational field couples to all forms of energy (as $E=mc^2$), one expects that a photon climbing in a gravitational field to lose energy, hence frequency. In 1912, Einstein realized that this could be modelled by considering curved spacetime geometries, so that equal line segments in a (flat) spacetime diagram do not necessarily correspond to the same length.

\begin{figure}[h!]
\begin{center}
\psfrag{t}{$t$}
\psfrag{x}{$x$}
\psfrag{T}{$T$}
\psfrag{T'}{$T'=T$}
\epsfxsize=.4\textwidth
\leavevmode
\epsfbox{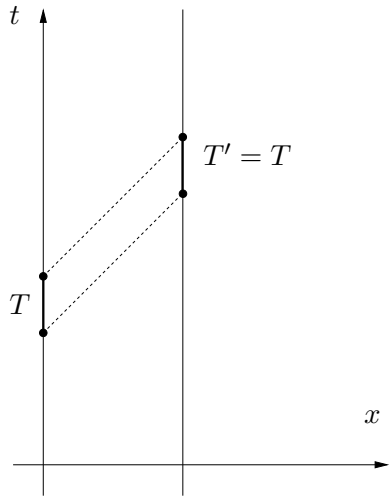}
\end{center}
\caption{Minkowski geometry is incompatible with the gravitational redshift.} \label{desvio}
\end{figure}

\section{Differential geometry: Mathematicians vs physicists} \label{sec1.2}

Einstein's idea to incorporate gravitation into relativity was to replace the Minkowski spacetime $(\bbR^4, \langle \cdot, \cdot \rangle)$ by a curved four-dimensional {\bf Lorentzian manifold} $(M,g)\equiv(M, \langle \cdot, \cdot \rangle)$. Here $g$ is a {\bf Lorentzian metric}, that is, a symmetric $2$-tensor field such that at each tangent space $g = \diag(-1,1,1,1)$ in an appropriate basis. Just like in Riemannian geometry, $g$ determines a {\bf Levi-Civita connection}, the unique connection $\nabla$ which is symmetric and compatible with $g$:
\begin{align*}
& \nabla_X Y - \nabla_Y X = [X,Y]; \\
& X\cdot\langle Y, Z \rangle = \langle \nabla_X Y, Z \rangle + \langle Y, \nabla_X Z \rangle,
\end{align*}
for all vector fields $X,Y,Z$. The {\bf curvature} of this connection is then given by the operator
\[
R(X,Y)Z = \nabla_X \nabla_Y Z - \nabla_Y \nabla_X Z - \nabla_{[X,Y]} Z.
\]

The formulas above were written using the abstract notation usually employed by mathematicians. It is often very useful (especially when dealing with contractions) to use the more explicit notation usually adopted by physicists, which emphasizes the indices of the various objects when written in local coordinates:

\begin{table}[h]
\begin{tabular}{ccccc}
{\bf Object} & & {\bf Mathematicians} & & {\bf Physicists} \\
Vector field & & $X$ & & $X^\mu$ \\
Tensor product & & $X \otimes Y$ & & $X^\mu Y^\nu$ \\
Metric & & $g \equiv \langle \cdot, \cdot \rangle$ & & $g_{\mu\nu}$ \\
Inner product & & $g(X,Y) \equiv \langle X, Y \rangle$ & & $g_{\mu\nu} X^\mu Y^\nu$ \\
Associated covector & & $X^\sharp \equiv g(X, \cdot)$ & & $X_\nu \equiv g_{\mu\nu} X^\mu$ \\
Covariant derivative & & $\nabla_X Y$ & & $X^\mu \nabla_\mu Y^\nu$ \\
Covariant derivative tensor & & $\nabla X$ & & $\nabla_\mu X^\nu \equiv \partial_\mu X^\nu + \Gamma^\nu_{\mu\alpha} X^\alpha$
\end{tabular}
\end{table}

Here $\Gamma^\alpha_{\mu\nu}$ are the {\bf Christoffel symbols} of the Levi-Civita connection; they can be computed from the components $g_{\mu\nu}$ of the metric tensor by the formula
\[
\Gamma^\alpha_{\mu\nu} = \frac12 g^{\alpha\beta} \left( \partial_\mu g_{\nu\beta} + \partial_\nu g_{\mu\beta} - \partial_\beta g_{\mu\nu} \right),
\]
and in turn be used to compute the components of the Riemann curvature tensor:
\[
R_{\alpha\beta\,\,\,\,\nu}^{\,\,\,\,\,\,\,\,\mu} = dx^\mu \left( R(\partial_\alpha,\partial_\beta) \partial_\nu \right) = \partial_\alpha \Gamma^\mu_{\beta\nu} - \partial_\beta \Gamma^\mu_{\alpha\nu} + \Gamma^\mu_{\alpha\gamma} \Gamma^\gamma_{\beta\nu} - \Gamma^\mu_{\beta\gamma} \Gamma^\gamma_{\alpha\nu} .
\]

The covariant derivative tensor of a vector field $X$, not always emphasized in differential geometry courses for mathematicians, is simply the $(1,1)$-tensor field defined by
\[
\nabla X (Y) = \nabla_Y X.
\]
Also not always emphasized in differential geometry courses for mathematicians is the fact that any connection can be naturally extended to act on tensor fields (via the Leibnitz rule). For instance, if $\omega$ is a covector field ($1$-form) then one defines
\[
(\nabla_X \omega) (Y) = X \cdot [\omega(Y)] - \omega(\nabla_X Y).
\]
In local coordinates, this is
\begin{align*}
(X^\mu \nabla_\mu \omega_\nu) Y^\nu & = X^\mu \partial_\mu (\omega_\nu Y^\nu) - \omega_\nu (X^\mu \nabla_\mu Y^\nu) \\
& = X^\mu (\partial_\mu \omega_\nu) Y^\nu + X^\mu \omega_\nu \partial_\mu Y^\nu - \omega_\nu (X^\mu \partial_\mu Y^\nu + X^\mu \Gamma_{\mu\alpha}^\nu Y^\alpha) \\
& = (\partial_\mu \omega_\nu - \Gamma_{\mu\nu}^\alpha \omega_\alpha) X^\mu Y^\nu,
\end{align*}
that is,
\[
\nabla_\mu \omega_\nu = \partial_\mu \omega_\nu - \Gamma_{\mu\nu}^\alpha \omega_\alpha.
\]
The generalization for higher rank tensors is obvious: for instance, if $T$ is a $(2,1)$-tensor then
\[
\nabla_\alpha T^\beta_{\mu\nu} = \partial_\alpha T^\beta_{\mu\nu} + \Gamma_{\alpha\gamma}^\beta T^\gamma_{\mu\nu} - \Gamma_{\alpha\mu}^\gamma T^\beta_{\gamma\nu} - \Gamma_{\alpha\nu}^\gamma T^\beta_{\mu\gamma}.
\]
Note that the condition of compatibility of the Levi-Civita connection with the metric is simply 
\[
\nabla g = 0.
\]
In particular, the operations of raising and lowering indices commute with covariant differentiation.

As an exercise in index gymnastics, we will now derive a series of identities involving the Riemann curvature tensor. We start by rewriting its definition in the notation of the physicists:
\begin{align*}
& R_{\alpha\beta\,\,\,\,\nu}^{\,\,\,\,\,\,\,\,\mu} X^\alpha Y^\beta Z^\nu = \\
& = X^\alpha \nabla_\alpha (Y^\beta \nabla_\beta Z^\mu) - Y^\alpha \nabla_\alpha (X^\beta \nabla_\beta Z^\mu) - (X^\alpha \nabla_\alpha Y^\beta - Y^\alpha \nabla_\alpha X^\beta) \nabla_\beta Z^\mu \\
& = (X^\alpha \nabla_\alpha Y^\beta) (\nabla_\beta Z^\mu) + X^\alpha Y^\beta \nabla_\alpha \nabla_\beta Z^\mu - (Y^\alpha \nabla_\alpha X^\beta) (\nabla_\beta Z^\mu) \\
& \,\,\,\, - Y^\alpha X^\beta \nabla_\alpha \nabla_\beta Z^\mu - (X^\alpha \nabla_\alpha Y^\beta) \nabla_\beta Z^\mu + (Y^\alpha \nabla_\alpha X^\beta) \nabla_\beta Z^\mu \\
& = X^\alpha Y^\beta (\nabla_\alpha \nabla_\beta - \nabla_\beta \nabla_\alpha)  Z^\mu.
\end{align*}
In other words,
\[
R_{\alpha\beta\,\,\,\,\nu}^{\,\,\,\,\,\,\,\,\mu} Z^\nu = (\nabla_\alpha \nabla_\beta - \nabla_\beta \nabla_\alpha)  Z^\mu,
\]
or, equivalently, 
\begin{equation} \label{commutator}
2\nabla_{[\alpha} \nabla_{\beta]} Z_\mu = R_{\alpha\beta\mu\nu} Z^\nu,
\end{equation}
where the square brackets indicate anti-symmetrization\footnote{Thus $T_{[\alpha\beta]}=\frac12\left(T_{\alpha\beta}-T_{\beta\alpha}\right)$, $T_{[\alpha\beta\gamma]}=\frac16\left(T_{\alpha\beta\gamma} + T_{\beta\gamma\alpha} + T_{\gamma\alpha\beta} - T_{\beta\alpha\gamma} - T_{\alpha\gamma\beta} - T_{\gamma\beta\alpha}\right)$, etc.}. This is readily generalized for arbitrary tensors: from
\begin{align*}
2\nabla_{[\alpha} \nabla_{\beta]} (Z_\mu W_\nu) & = (2\nabla_{[\alpha} \nabla_{\beta]} Z_\mu) W_\nu + (2\nabla_{[\alpha} \nabla_{\beta]} W_\nu) Z_\mu \\
& = R_{\alpha\beta\mu\sigma} Z^\sigma W_\nu + R_{\alpha\beta\nu\sigma} W^\sigma Z_\mu
\end{align*}
one readily concludes that
\[
2\nabla_{[\alpha} \nabla_{\beta]} T_{\mu\nu} = R_{\alpha\beta\mu\sigma} T^\sigma_{\,\,\,\,\nu} + R_{\alpha\beta\nu\sigma} T_{\mu}^{\,\,\,\,\sigma}.
\]

Let us choose
\[
Z_\mu = \nabla_\mu f \equiv \partial_\mu f
\]
in equation~\eqref{commutator}. We obtain 
\[
R_{[\alpha\beta\mu]\nu} Z^\nu = 2 \nabla_{[[\alpha} \nabla_{\beta]} Z_{\mu]} = 2 \nabla_{[\alpha} \nabla_{[\beta} Z_{\mu]]} = 0,
\]
because
\[
\nabla_{[\mu} Z_{\nu]} = \partial_{[\mu} Z_{\nu]} - \Gamma^\alpha_{[\mu\nu]} Z_\alpha = \partial_{[\mu} \partial_{\nu]} f = 0.
\]
Since we can choose $Z$ arbitrarily at a given point, it follows that
\[
R_{[\alpha\beta\mu]\nu} = 0 \Leftrightarrow R_{\alpha\beta\mu\nu} + R_{\beta\mu\alpha\nu} + R_{\mu\alpha\beta\nu} = 0.
\]
This is the so-called {\bf first Bianchi identity}, and is key for obtaining the full set of symmetries of the Riemann curvature tensor:
\[
R_{\alpha\beta\mu\nu} = - R_{\beta\alpha\mu\nu} = - R_{\alpha\beta\nu\mu} = R_{\mu\nu\alpha\beta}.
\]
In the notation of the mathematicians, it is written as
\[
R(X,Y)Z + R(Y,Z) X + R(Z,X) Y = 0
\]
for all vector fields $X,Y,Z$.

Let us now take the covariant derivative of equation~\eqref{commutator}:
\[
\nabla_\gamma R_{\alpha\beta\mu\nu} Z^\nu + R_{\alpha\beta\mu\nu} \nabla_\gamma Z^\nu = 2 \nabla_\gamma \nabla_{[\alpha} \nabla_{\beta]} Z_\mu.
\]
At any given point we can choose $Z$ such that
\[
\nabla_\gamma Z^\nu \equiv \partial_\gamma Z^\nu + \Gamma^\nu_{\gamma\delta} Z^\delta = 0.
\]
Assuming this, we then obtain\footnote{In the formula below the indices between vertical bars are not anti-symmetrized.}
\begin{align*}
\nabla_{[\gamma} R_{\alpha\beta]\mu\nu} Z^\nu & = 2 \nabla_{[\gamma} \nabla_{[\alpha} \nabla_{\beta]]} Z_\mu = 2 \nabla_{[[\gamma} \nabla_{\alpha]} \nabla_{\beta]} Z_\mu \\
& = R_{[\gamma\alpha\beta]\delta} \nabla^\delta Z_\mu + R_{[\gamma\alpha|\mu\delta|} \nabla_{\beta]} Z^\delta = 0.
\end{align*}
Since we can choose $Z$ arbitrarily at a given point, it follows that
\begin{equation} \label{second}
\nabla_{[\alpha}R_{\beta\gamma]\mu\nu} = 0 \Leftrightarrow \nabla_{\alpha}R_{\beta\gamma\mu\nu} + \nabla_{\beta}R_{\gamma\alpha\mu\nu} + \nabla_{\gamma}R_{\alpha\beta\mu\nu} = 0
\end{equation}
This is the so-called {\bf second Bianchi identity}. In the notation of the mathematicians, it is written as
\[
\nabla R (X,Y,Z,\cdot,\cdot) + \nabla R (Y,Z,X,\cdot,\cdot) + \nabla R (Z,X,Y,\cdot,\cdot) = 0
\]
for all vector fields $X,Y,Z$.

Recall that the Riemann curvature tensor has only one independent contraction, called the {\bf Ricci tensor}:
\[
R_{\mu\nu} = R_{\alpha\mu\,\,\,\,\nu}^{\,\,\,\,\,\,\,\,\alpha}.
\]
The trace of the Ricci tensor, in turn, is known as the {\bf scalar curvature}:
\[
R = g^{\mu\nu} R_{\mu\nu}.
\]
These quantities satisfy the so-called {\bf contracted Bianchi identity}, which is obtained from \eqref{second} by contracting the pairs of indices $(\beta,\mu)$ and $(\gamma,\nu)$:
\[
\nabla_\alpha R - \nabla^\beta R_{\alpha\beta}  - \nabla^\gamma R_{\alpha\gamma} = 0 \Leftrightarrow \nabla^\beta R_{\alpha\beta} - \frac12 \nabla_\alpha R = 0 \Leftrightarrow \nabla^\beta \left(R_{\alpha\beta} - \frac12 R g_{\alpha\beta}\right) = 0.
\]
The contracted Bianchi identity is equivalent to the statement that the {\bf Einstein tensor}
\[
G_{\mu\nu} = R_{\mu\nu} - \frac12 R g_{\mu\nu} 
\] 
is divergenceless:
\[
\nabla^\mu G_{\mu\nu} = 0.
\]

\section{General relativity} \label{sec1.3}

Newtonian gravity is described by a scalar function $\phi$, called the {\bf gravitational potential}. The equation of motion for a free-falling particle of mass $m$ in Cartesian coordinates is
\[
m \frac{d^2x^i}{dt^2} = - m \partial_i \phi \Leftrightarrow \frac{d^2x^i}{dt^2} = - \partial_i \phi.
\]
Note that all free-falling particles describe the same trajectories (an observation dating back to Galileo). The gravitational potential is determined from the matter mass density $\rho$ by the {\bf Poisson equation}
\[
\Delta \phi = 4 \pi \rho
\]
(using units such that Newton's gravitational constant is $G=1$; this choice, together with $c=1$, defines the so-called {\bf geometrized units}, where lengths, time intervals and masses all have the same dimensions).

To implement his idea of describing gravity via a curved four-dimensional Lorentzian manifold $(M,g)$, Einstein had to specify $(i)$ how free-falling particles would move on this manifold, and $(ii)$ how to determine the curved metric $g$. Since free particles move along straight lines in the Minkowski spacetime, Einstein proposed that free falling particles should move along timelike geodesics. In other words, he suggested replacing the Newtonian equation of motion by the geodesic equation
\[
\ddot{x}^\mu + \Gamma^\mu_{\alpha\beta}\dot{x}^\alpha\dot{x}^\beta = 0.
\]
Moreover, Einstein knew that it is possible to define the {\bf energy-momentum tensor} $T_{\mu\nu}$ of the matter content of the Minkowski spacetime, so that the conservation of energy and momentum is equivalent to the vanishing of its divergence:
\[
\nabla^\mu T_{\mu\nu} = 0.
\]
This inspired Einstein to propose that $g$ should satisfy the so-called {\bf Einstein field equations}:
\[
G_{\mu\nu} + \Lambda g_{\mu\nu} = 8 \pi T_{\mu\nu}.
\]
Here $\Lambda$ is a constant, known as the {\bf cosmological constant}. Note that the Einstein field equations imply, via the contracted Bianchi identity, that the energy-momentum tensor is divergenceless.

As a simple example, we consider a pressureless perfect fluid, known as {\bf dust}. Its energy-momentum tensor is
\[
T_{\mu\nu} = \rho U_\mu U_\nu,
\]
where $\rho$ is the dust rest density and $U$ is a unit timelike vector field tangent to the histories of the dust particles. The equations of motion for the dust can be found from
\begin{align*}
\nabla^\mu T_{\mu\nu} = 0 & \Leftrightarrow \left[ \nabla^\mu (\rho U_\mu) \right] U_ \nu + \rho U_\mu \nabla^\mu U_ \nu = 0 \\
& \Leftrightarrow \dive (\rho U) U + \rho \nabla_U U = 0.
\end{align*}
Since $U$ and $\nabla_UU$ are orthogonal (because $\langle U, U \rangle = -1$), we find
\[
\begin{cases}
\dive (\rho U) = 0 \\
\nabla_U U = 0
\end{cases}
\]
in the support of $\rho$. These are, respectively, the equation of conservation of mass and the geodesic equation. Thus the fact that free-falling particles move along geodesics can be seen as a consequence of the Einstein field equations (at least in this model).

\section{Exercises} \label{sec1.4}

\begin{enumerate}

\item
{\bf Twin paradox:} Two twins, Alice and Bob, are separated on their $20^\text{th}$ birthday. While Alice remains on Earth (which is an inertial frame to a very good approximation), Bob departs at $80\%$ of the speed of light towards Planet X, $8$ light-years away from Earth. Therefore Bob reaches his destination $10$ years later (as measured on the Earth's frame). After a short stay, he returns to Earth, again at $80\%$ of the speed of light. Consequently Alice is $40$ years old when she sees Bob again.
\begin{enumerate}
\item
How old is Bob when they meet again?
\item
How can the asymmetry in the twins' ages be explained? Notice that from Bob's point of view he is at rest in his spaceship and it is the Earth which moves away and then back again.
\item
Imagine that each twin watches the other trough a very powerful telescope. What do they see? In particular, how much time do they experience as they see one year elapse for their twin?
\end{enumerate}

\item
A particularly simple matter model is that of a smooth {\bf massless scalar field} $\phi:M\to \bbR$, whose energy-momentum tensor is
\[
T_{\mu\nu} = \partial_\mu \phi \, \partial_\nu \phi - \frac12 (\partial_\alpha \phi \, \partial^\alpha \phi) g_{\mu\nu}.
\]
Show that if the Lorentzian manifold $(M,g)$ satisfies the Einstein equations with this matter model then $\phi$ satisfies the {\bf wave equation}
\[
\Box \, \phi = 0 \Leftrightarrow \nabla^\mu \partial_\mu \phi = 0.
\]

\item
The energy-momentum tensor for {\bf perfect fluid} is
\[
T_{\mu\nu} = (\rho+p) U_\mu U_\nu + p g_{\mu\nu},
\]
where $\rho$ is the fluid's rest density, $p$ is the fluid's rest pressure, and $U$ is a unit timelike vector field tangent to the histories of the fluid particles. Show that:
\begin{enumerate}
\item
$\left( T_{\mu\nu} \right) = \diag(\rho,p,p,p)$ in any orthonormal frame including $U$;
\item
the motion equations for the perfect fluid are
\[
\begin{cases}
\dive(\rho U) + p \dive U = 0 \\
(\rho+p) \nabla_U U = - (\grad p)^\perp
\end{cases},
\]
where $^\perp$ represents the orthogonal projection on the spacelike hyperplane orthogonal to $U$.
\end{enumerate}

\end{enumerate}


\chapter{Exact solutions} \label{chapter2}

In this chapter we present a number of exact solutions of the Einstein field equations, as well as their Penrose diagrams. These solutions will be used as examples or counter-examples to the theorems in the subsequent chapters. We also discuss the matching of two different solutions across a timelike hypersurface. A different perspective on Penrose diagrams can be found in \cite{HE95}.

\section{Minkowski spacetime} \label{sec2.1}

The simplest solution of the Einstein field equations with zero cosmological constant in vacuum (i.e.~with vanishing energy-momentum tensor) is the Minkowski spacetime, that is, $\bbR^4$ with the metric
\[
ds^2 = - dt^2 + dx^2 + dy^2 + dz^2.
\]
Since this metric is flat, its curvature vanishes, and so do its Ricci and Einstein tensors. It represents a universe where there is no gravity whatsoever. Transforming the Cartesian coordinates $(x,y,z)$ to spherical coordinates $(r,\theta,\varphi)$ yields
\[
ds^2 = - dt^2 + dr^2 + r^2 \left(d\theta^2 + \sin^2\theta d\varphi^2\right).
\]
Performing the additional change of coordinates
\[
\begin{cases}
u = t - r \qquad \text{(retarded time)} \\
v = t + r \qquad \text{(advanced time)} 
\end{cases}
\]
we obtain
\[
ds^2 = - du \, dv + r^2 \left(d\theta^2 + \sin^2\theta d\varphi^2\right),
\]
where
\[
r(u,v) = \frac12 (v-u).
\]
The coordinates $(u,v)$ are called {\bf null coordinates}: their level sets are null cones formed by outgoing/ingoing null geodesics emanating from the center. Note that they are subject to the constraint
\[
r \geq 0 \Leftrightarrow v \geq u.
\]
Finally, the coordinate change
\begin{equation} \label{tanh}
\begin{cases}
\tilde{u} = \tanh u \\
\tilde{v} = \tanh v 
\end{cases}
\Leftrightarrow
\begin{cases}
u = \arctanh \tilde{u} \\
v = \arctanh \tilde{v}
\end{cases}
\end{equation}
brings the metric into the form
\[
ds^2 = - \frac{1}{\left(1-\tilde{u}^2\right)\left(1-\tilde{v}^2\right)} d\tilde{u} \, d\tilde{v} + r^2 \left(d\theta^2 + \sin^2\theta d\varphi^2\right),
\]
where now
\[
r\left(\tilde{u},\tilde{v}\right) = \frac12 (\arctanh \tilde{v}-\arctanh \tilde{u})
\]
and
\begin{equation} \label{range}
-1 < \tilde{u} \leq \tilde{v} < 1.
\end{equation}
Because $(\tilde{u},\tilde{v})$ are also null coordinates, it is common to represent their axes tilted by $45^\circ$. The plane region defined by \eqref{range} is then represented in Figure~\ref{uv}.

\begin{figure}[h!]
\begin{center}
\psfrag{u}{$\tilde{u}$}
\psfrag{v}{$\tilde{v}$}
\epsfxsize=.4\textwidth
\leavevmode
\epsfbox{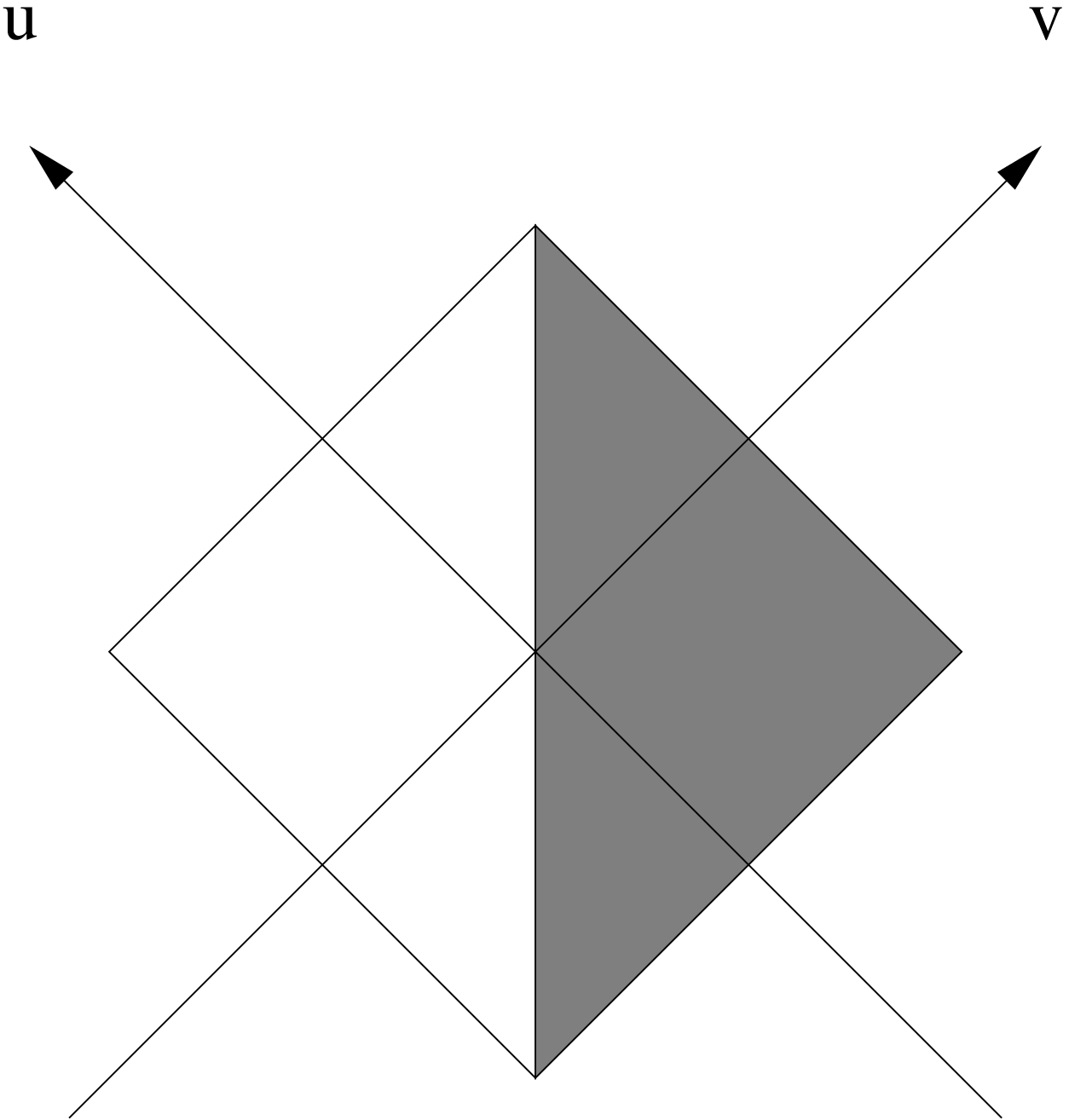}
\end{center}
\caption{Range of the coordinates $(\tilde{u},\tilde{v})$.} \label{uv}
\end{figure}

This region is usually called the {\bf Penrose diagram} for the Minkowski spacetime. If we take each point in the diagram to represent a sphere $S^2$ of radius $r\left(\tilde{u},\tilde{v}\right)$, the diagram itself represents the full spacetime manifold, in a way that makes causality relations apparent: any causal curve is represented in the diagram by a curve with tangent at most $45^\circ$ from the vertical. In Figure~\ref{Pen_Mink} we represent some level hypersurfaces of $t$ and $r$ in the Penrose diagram. The former approach the point $i^0$ in the boundary of the diagram, called the {\bf spacelike infinity}, whereas the later go from the boundary point $i^-$ ({\bf past timelike infinity}) to the boundary point  $i^+$ ({\bf future timelike infinity}). Finally, null geodesics start at the null boundary line $\mathscr{I^-}$ ({\bf past null infinity}) and end at the null boundary line $\mathscr{I^+}$ ({\bf future null infinity}). These boundary points and lines represent ideal points at infinity, and do not correspond to actual points in the Minkowski spacetime.

\begin{figure}[h!]
\begin{center}
\psfrag{i+}{$i^+$}
\psfrag{i0}{$i^0$}
\psfrag{i-}{$i^-$}
\psfrag{r=0}{$r=0$}
\psfrag{I+}{$\mathscr{I^+}$}
\psfrag{I-}{$\mathscr{I^-}$}
\epsfxsize=.4\textwidth
\leavevmode
\epsfbox{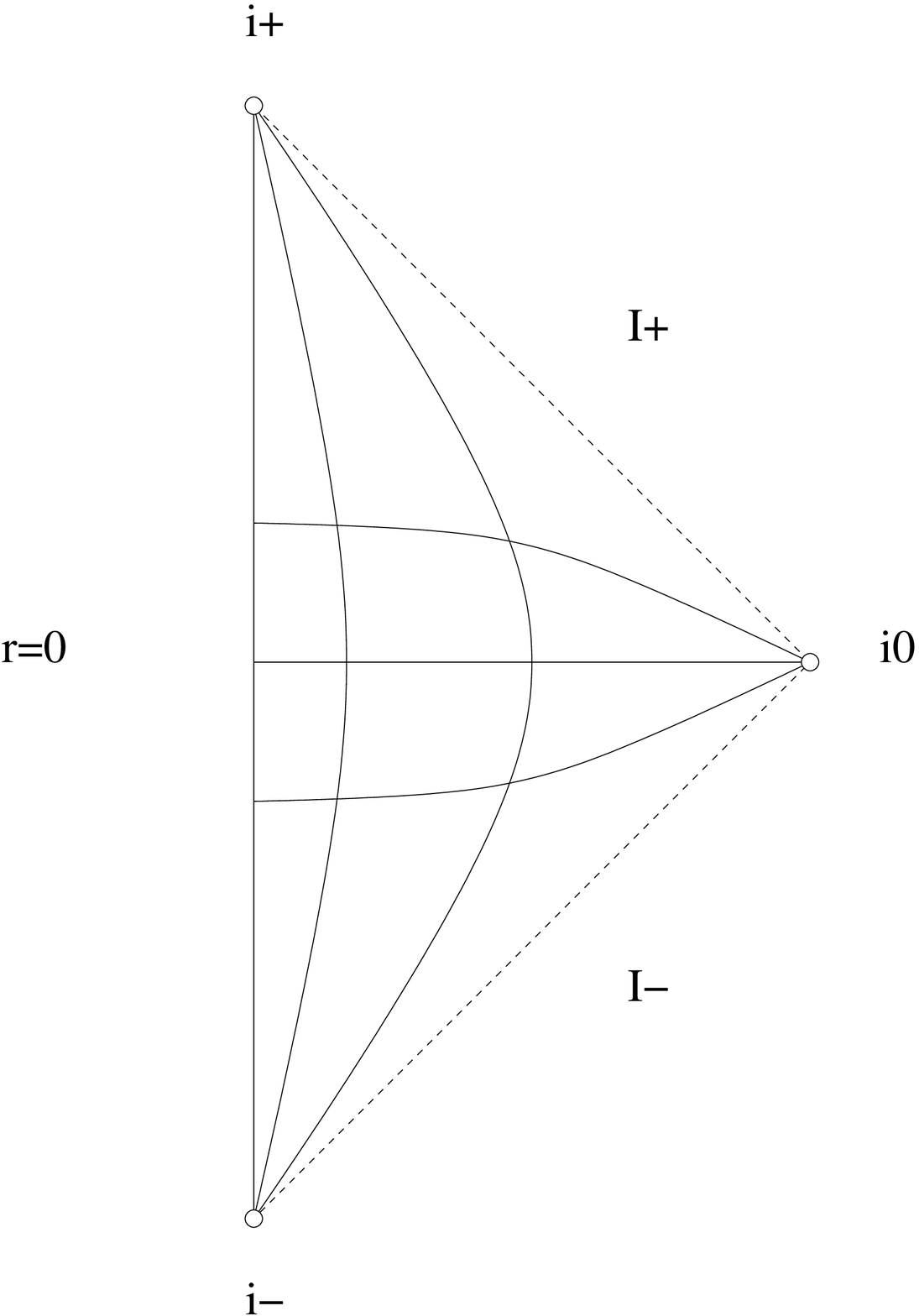}
\end{center}
\caption{Penrose diagram for the Minkowski spacetime, with some level hypersurfaces of $t$ and $r$ represented.} \label{Pen_Mink}
\end{figure}

\section{Penrose diagrams} \label{sec2.2}

The concept of Penrose diagram can be easily generalized for any spherically symmetric space-time. Such spacetimes have metric
\[
ds^2 = g_{AB} dx^A dx^B + r^2 \left(d\theta^2 + \sin^2\theta d\varphi^2\right), \qquad r=r(x^0,x^1),
\]
where $g_{AB} dx^A dx^B$ is a Lorentzian metric on a $2$-dimensional quotient manifold with boundary (which we assume to be diffeomorphic to a region of the plane). It turns out that any such metric is conformal to the Minkowski metric:
\begin{equation} \label{confmet}
g_{AB} dx^A dx^B = - \Omega^2 du \, dv, \qquad \Omega=\Omega(u,v).
\end{equation}
This can be seen locally as follows: choose a spacelike line $S$, a coordinate $u$ along it, and a coordinate $w$ along a family of null geodesics emanating from $S$, so that $S$ corresponds to $w=0$ (Figure~\ref{conformal}). Then near $S$ we have
\[
g_{uu} = \left\langle \frac{\partial}{\partial u}, \frac{\partial}{\partial u} \right\rangle > 0
\]
and
\[
g_{ww} = \left\langle \frac{\partial}{\partial w}, \frac{\partial}{\partial w} \right\rangle = 0.
\] 

\begin{figure}[h!]
\begin{center}
\psfrag{u}{$u$}
\psfrag{w}{$w$}
\epsfxsize=.6\textwidth
\leavevmode
\epsfbox{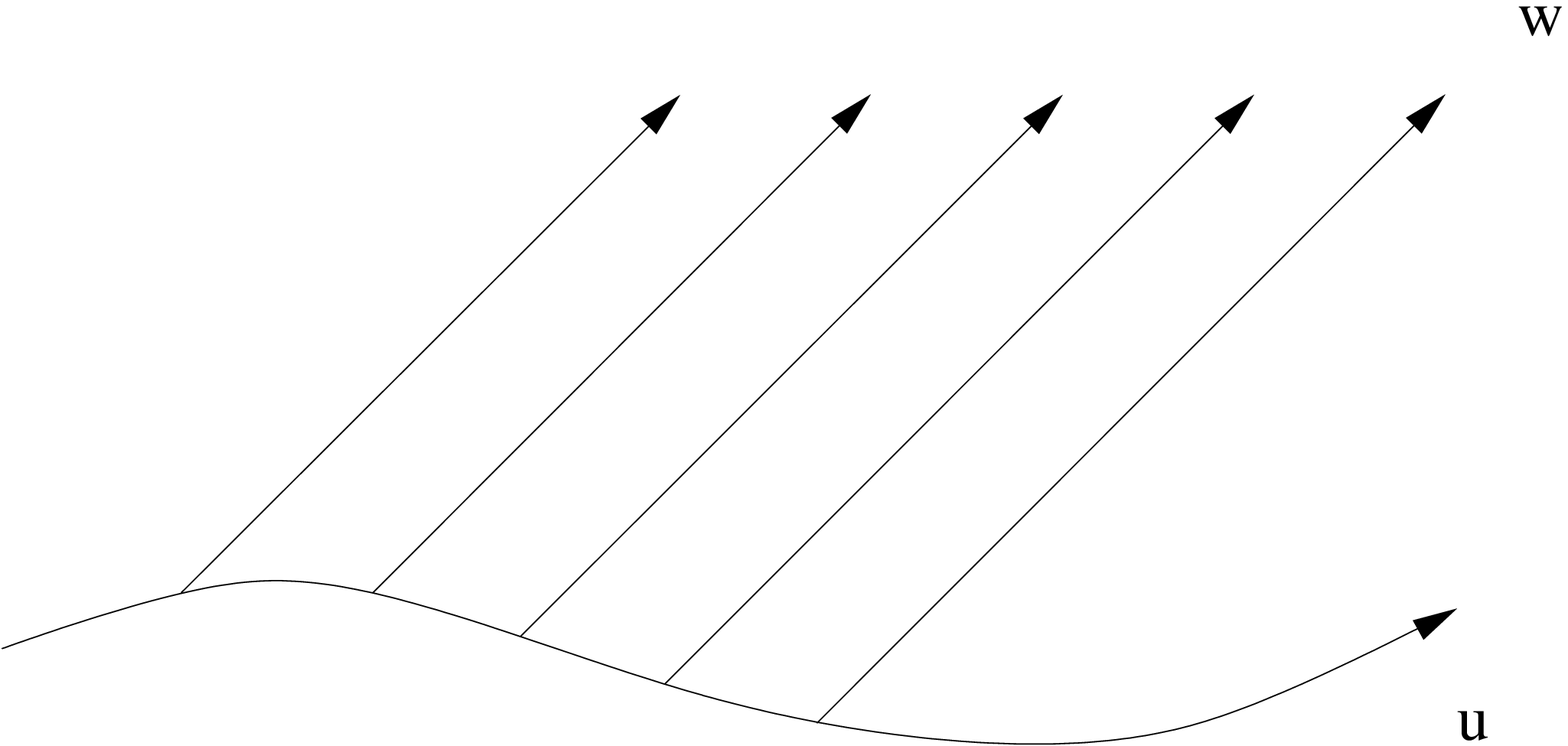}
\end{center}
\caption{Choice of the coordinates $(u,w)$.} \label{conformal}
\end{figure}

Therefore the $2$-dimensional metric is written in these coordinates
\[
g_{AB} dx^A dx^B = g_{uu} du^2 + 2g_{uw} du dw = g_{uu} du \left( du + \frac{2g_{uw}}{g_{uu}} dw \right).
\]
As for any $1$-form in a $2$-dimensional manifold, we have
\[
du + \frac{2g_{uw}}{g_{uu}} dw = f dv
\]
for suitable functions $f$ and $v$. Note that $f$ cannot vanish, because $(u,w)$ are local coordinates. Moreover, we can assume $f<0$ by replacing $v$ with $-v$ if necessary. Choosing $\Omega^2 = -fg_{uu}$ then yields \eqref{confmet}.

We then see that any spherically symmetric metric can be written as
\begin{equation} \label{Penrose_met}
ds^2 = - \Omega^2 du \, dv + r^2 \left(d\theta^2 + \sin^2\theta d\varphi^2\right)
\end{equation}
with $\Omega=\Omega(u,v)$ and $r=r(u,v)$. By rescaling $u$ and $v$ if necessary, we can assume that the range of $(u,v)$ is bounded, and hence obtain a Penrose diagram depicting the causal geometry. As we will see, this is extremely helpful in more complicated spherical symmetric solutions of the Einstein field equations.

\begin{Remark}
From
\[
\left( g_{AB} \right) = 
\left(
\begin{matrix}
0 & -\frac{\Omega^2}2 \\
-\frac{\Omega^2}2 & 0
\end{matrix}
\right)
\Rightarrow
\left( g^{AB} \right) = 
\left(
\begin{matrix}
0 & -\frac2{\Omega^2} \\
-\frac2{\Omega^2} & 0
\end{matrix}
\right)
\]
it is easily seen that
\[
\partial_A \left( \sqrt{ - \det\left(g_{CD}\right)} \, g^{AB} \partial_B u \right) = 0 \Leftrightarrow \nabla_A \nabla^A u = 0.
\]
and similarly for $v$. In other words, the null coordinates $u$ and $v$ are solutions of the wave equation in the $2$-dimensional Lorentzian manifold:
\[
\Box u = \Box v = 0.
\]
This is the Lorentzian analogue of the so-called {\bf isothermal coordinates} for Riemannian surfaces. The proof that the later exist locally is however slightly more complicated: given a point $p$ on the surface, one chooses a local harmonic function with nonvanishing derivative,
\[
\Delta u = 0, \qquad (du)_p \neq 0,
\]
and considers the equation
\begin{equation} \label{dv}
dv = \star du.
\end{equation}
Here $\star$ is the Hodge star, which for generic orientable $n$-dimensional pseudo-Riemannian manifolds is defined as follows: if $\{\omega^1, \ldots, \omega^n \}$ is any positively oriented orthonormal coframe then
\[
\star (\omega^1 \wedge \cdots \wedge \omega^k) = \langle \omega^1, \omega^1 \rangle \cdots \langle \omega^k, \omega^k \rangle \, \omega^{k+1} \wedge \cdots \wedge \omega^n.
\]
By the Poincar\'e Lemma, equation~\eqref{dv} can be locally solved, since
\[
d \star du =  \star \star d \star du = \star (\Delta u) = 0.
\]
Moreover, $v$ is itself harmonic, because
\[
\Delta v = \star d \star dv = \star d \star \star du =  \star d (-du) = 0.
\]
Finally,
\[
\| du \| = \| dv \| = \frac1{\Omega}
\]
for some local function $\Omega > 0$, and so the metric is written is these coordinates as
\[
ds^2 = \Omega^2 \left( du^2 + dv^2 \right).
\]
\end{Remark}

\section{The Schwarzschild solution} \label{sec2.3}

If we try to solve the vacuum Einstein field equations with zero cosmological constant for a spherically symmetric Lorentzian metric, we obtain, after suitably rescaling the time coordinate, the {\bf Schwarzschild metric}
\[
ds^2 = -\left(1 - \frac{2M}{r}\right) dt^2 + \left(1 - \frac{2M}{r}\right)^{-1} dr^2 + r^2 \left( d\theta^2 + \sin^2 \theta d\varphi^2 \right)
\]
(where $M \in \bbR$ is a constant). Note that for $M=0$ we retrieve the Minkowski metric in spherical coordinates. Note also that if $M>0$ then the metric is defined in two disconnected domains of coordinates, corresponding to $r \in (0, 2M)$ and $r \in (2M, + \infty)$.

The physical interpretation of the Schwarzschild solution can be found by considering the proper time of a timelike curve parameterized by the time coordinate:
\[
\tau = \int_{t_0}^{t_1} \left[ \left(1 - \frac{2M}{r}\right) - \left(1 - \frac{2M}{r}\right)^{-1} \dot{r}^2 - r^2 \dot{\theta}^2 - r^2 \sin^2 \theta \dot{\varphi}^2 \right]^{\frac12} dt,
\]
where $\dot{r}=\frac{dr}{dt}$, etc. The integrand $L_S$ is the Lagrangian for geodesic motion in the Schwarzschild spacetime when parameterized by the time coordinate. Now for motions with speeds much smaller than the speed of light we have $\dot{r}^2 \ll 1$, etc. Assuming $\frac{M}{r} \ll 1$ as well we have
\begin{align*}
L_S & = \left[ 1 - \frac{2M}{r} - \left(1 - \frac{2M}{r}\right)^{-1} \dot{r}^2 - r^2 \dot{\theta}^2 - r^2 \sin^2 \theta \dot{\varphi}^2 \right]^{\frac12} \\
& \simeq 1 - \frac{M}{r} - \frac12 \left( \dot{r}^2 + r^2 \dot{\theta}^2 + \sin^2 \theta \dot{\varphi}^2 \right) = 1 - L_N, 
\end{align*}
where
\[
L_N = \frac12 \left( \dot{r}^2 + r^2 \dot{\theta}^2 + r^2 \sin^2 \theta \dot{\varphi}^2 \right) + \frac{M}{r}
\]
is precisely the Newtonian Lagrangian for the motion of a particle in the gravitational field of a point mass $M$. The Schwarzschild solution should therefore be considered the relativistic analogue of this field.

To write the Schwarzschild metric in the form \eqref{Penrose_met} we note that the quotient metric is
\begin{align*}
ds^2 & = -\left(1 - \frac{2M}{r}\right) dt^2 + \left(1 - \frac{2M}{r}\right)^{-1} dr^2 \\
& = -\left(1 - \frac{2M}{r}\right) \left[ dt^2 - \left(1 - \frac{2M}{r}\right)^{-2} dr^2 \right] \\
& = -\left(1 - \frac{2M}{r}\right) \left[ dt - \left(1 - \frac{2M}{r}\right)^{-1} dr \right] \left[ dt + \left(1 - \frac{2M}{r}\right)^{-1} dr \right] \\
& = -\left(1 - \frac{2M}{r}\right) du \, dv,
\end{align*}
where we define
\[
u = t - \int \left(1 - \frac{2M}{r}\right)^{-1} dr = t - r - 2M \log |r - 2M|
\]
and
\[
v = t + \int \left(1 - \frac{2M}{r}\right)^{-1} dr = t + r + 2M \log |r - 2M|.
\]
In the domain of coordinates $r>2M$ we have $1 - \frac{2M}{r} > 0$, and so the quotient metric is already in the required form. Note however that, unlike what happened in the Minkowski spacetime, we now have
\[
v - u = 2r + 4M \log |r - 2M| \in (-\infty, +\infty).
\]
Consequently, by applying the coordinate rescaling \eqref{tanh} we obtain the full square, instead of a triangle (Figure~\ref{Pen_Schw_out}). Besides the infinity points and null boundaries also present in the Penrose diagram for the Minkowski spacetime, there are two new null boundaries, $\mathscr{H^-}$ ({\bf past event horizon}) and $\mathscr{H^+}$ ({\bf future event horizon}), where $r=2M$. 

\begin{figure}[h!]
\begin{center}
\psfrag{i+}{$i^+$}
\psfrag{i0}{$i^0$}
\psfrag{i-}{$i^-$}
\psfrag{H+}{$\mathscr{H^+}$}
\psfrag{H-}{$\mathscr{H^-}$}
\psfrag{I+}{$\mathscr{I^+}$}
\psfrag{I-}{$\mathscr{I^-}$}
\epsfxsize=.5\textwidth
\leavevmode
\epsfbox{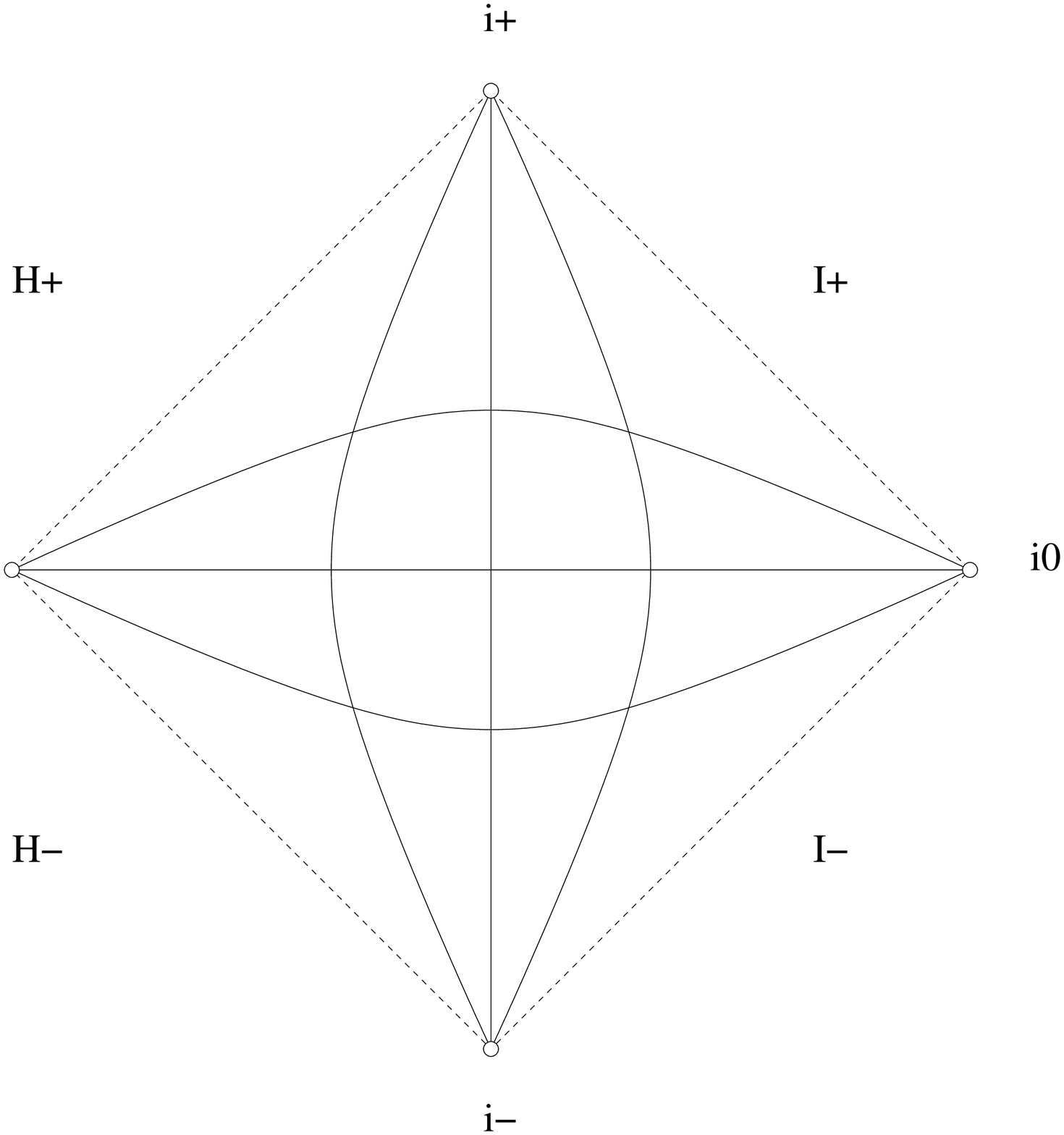}
\end{center}
\caption{Penrose diagram for the region $r>2M$ of the Schwarzschild spacetime, with some level hypersurfaces of $t$ and $r$ represented.} \label{Pen_Schw_out}
\end{figure}

It seems reasonable to expect that the metric can be extended across the horizons, since $r$ does not tend to zero nor to infinity there; this expectation is confirmed by calculating the so-called {\bf Kretschmann scalar}:
\[
R_{\alpha\beta\mu\nu} R^{\alpha\beta\mu\nu} = \frac{48M^2}{r^6}.
\]
This is perfectly well behaved as $r \to 2M$, and seems to indicate that the horizons are mere singularities of the coordinate system $(t,r)$. To show that this is indeed the case, note that in the $(u,r)$ coordinate system the quotient metric is written
\[
ds^2 = - \left(1 - \frac{2M}{r}\right) du^2 - 2 du \, dr.
\]
Since
\[
\det
\left(
\begin{matrix}
 - 1 + \frac{2M}{r} & -1 \\
-1 & 0
\end{matrix}
\right)
= -1,
\]
we see that the metric is well defined across $r=2M$ in this coordinate system. Moreover, we know that it solves the Einstein equations in the coordinate domains $r<2M$ and $r>2M$; by continuity, it must solve it in the whole domain $r \in (0, +\infty)$. Note that the coordinate domains $r<2M$ and $r>2M$ are glued along $r=2M$ so that the outgoing null geodesics $u=\text{constant}$ go from $r=0$ to $r=+\infty$; in other words, the gluing is along the past event horizon $\mathscr{H^-}$. 

To obtain the Penrose diagram for the coordinate domain $r<2M$ we note that the quotient metric can be written as
\begin{align*}
ds^2 & = -\left(1 - \frac{2M}{r}\right) du \, dv = - \left(\frac{2M}{r} - 1\right) du \, (-dv) = - \left(\frac{2M}{r} - 1\right) du \, dv',
\end{align*}
where $v'=-v$. Since in this coordinate domain $\frac{2M}{r} - 1 > 0$, the quotient metric is in the required form. Note however that we now have
\[
u + v' = - 2r - 4M \log |r - 2M| \in (- 4M \log (2M), +\infty),
\]
and by setting
\[
v'' = v' + 4M \log (2M)
\]
we obtain
\[
u + v'' > 0.
\]
Consequently, by applying the coordinate rescaling \eqref{tanh} we obtain a triangle (Figure~\ref{Pen_Schw_down}). There is now a spacelike boundary, where $r=0$, and two null boundaries $\mathscr{H^-}$, where $r=2M$. The Penrose diagram for the domain of the coordinates $(u,r)$ can be obtained gluing the Penrose diagrams in Figures~\ref{Pen_Schw_out} and \ref{Pen_Schw_down} along $\mathscr{H^-}$, so that the null geodesics $u=\text{constant}$ match (Figure~\ref{Pen_Schw_right}).

\begin{figure}[h!]
\begin{center}
\psfrag{r=0}{$r=0$}
\psfrag{i-}{$i^-$}
\psfrag{H-}{$\mathscr{H^-}$}
\epsfxsize=.6\textwidth
\leavevmode
\epsfbox{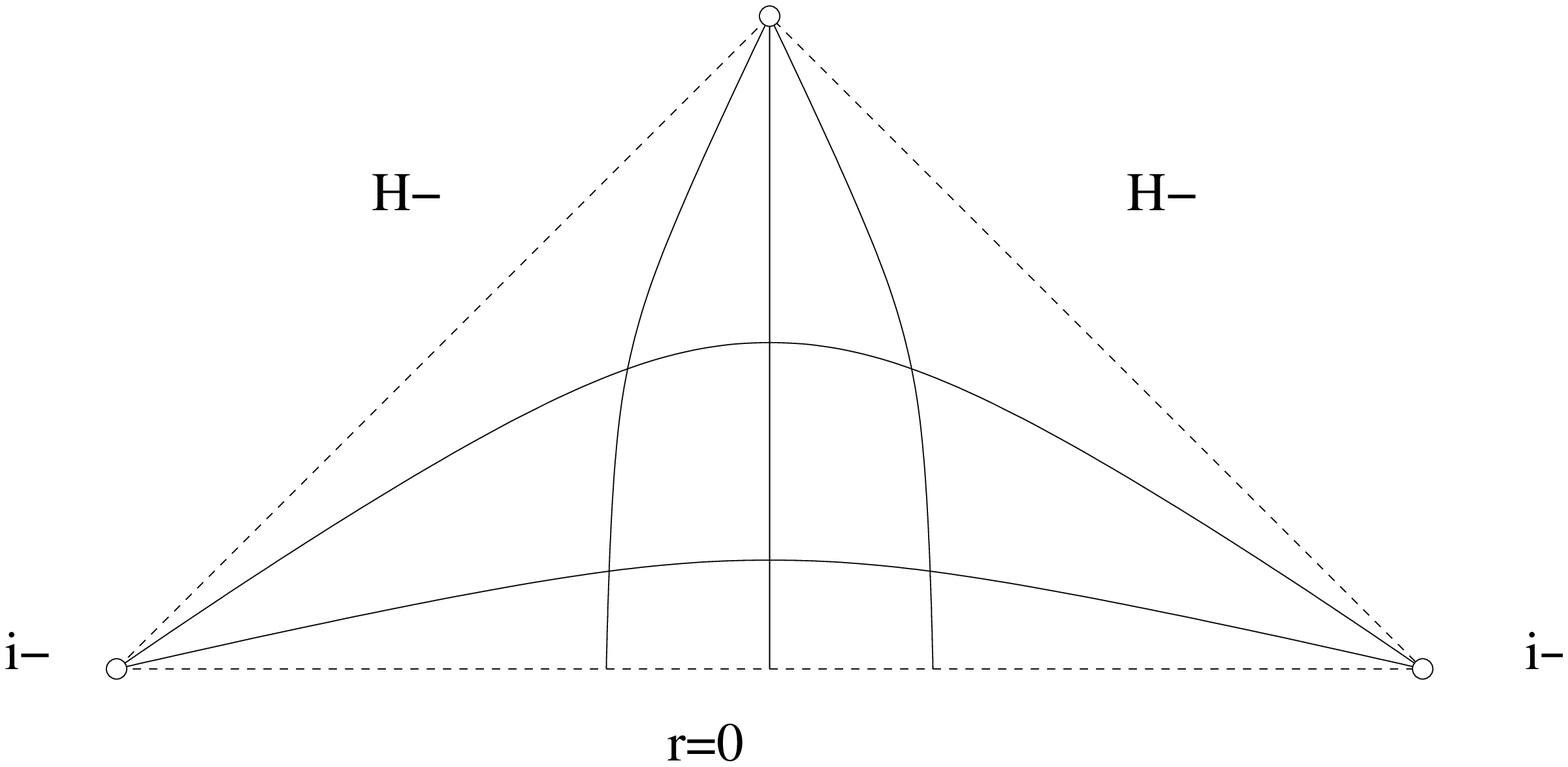}
\end{center}
\caption{Penrose diagram for a region $r<2M$ of the Schwarzschild spacetime, with some level hypersurfaces of $t$ and $r$ represented.} \label{Pen_Schw_down}
\end{figure}

\begin{figure}[h!]
\begin{center}
\psfrag{i+}{$i^+$}
\psfrag{i0}{$i^0$}
\psfrag{i-}{$i^-$}
\psfrag{r=0}{$r=0$}
\psfrag{H+}{$\mathscr{H^+}$}
\psfrag{H-}{$\mathscr{H^-}$}
\psfrag{I+}{$\mathscr{I^+}$}
\psfrag{I-}{$\mathscr{I^-}$}
\epsfxsize=.7\textwidth
\leavevmode
\epsfbox{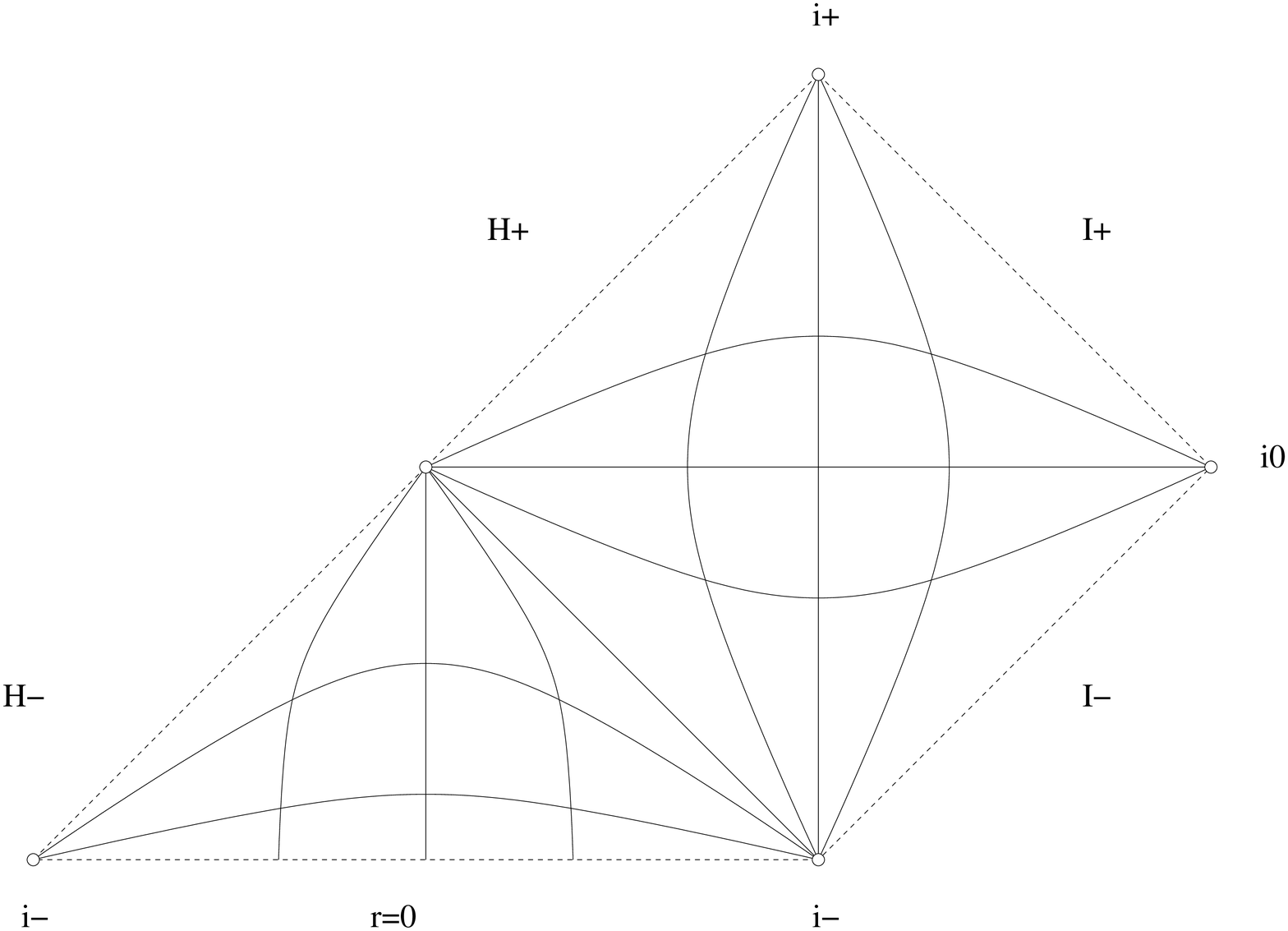}
\end{center}
\caption{Penrose diagram for the domain of the coordinates $(u,r)$ in the Schwarzschild spacetime, with some level hypersurfaces of $t$ and $r$ represented.} \label{Pen_Schw_right}
\end{figure}

If instead we use the $(v,r)$ coordinate system, the quotient metric is written
\[
ds^2 = - \left(1 - \frac{2M}{r}\right) dv^2 + 2 dv \, dr.
\]
Again the metric is well defined across $r=2M$ in this coordinate system, since
\[
\det
\left(
\begin{matrix}
 - 1 + \frac{2M}{r} & 1 \\
1 & 0
\end{matrix}
\right)
= - 1,
\]
and solves the Einstein equations in the whole coordinate domain $r \in (0, +\infty)$. The coordinate domains $r<2M$ and $r>2M$ are now glued along $r=2M$ so that the ingoing null geodesics $v=\text{constant}$ go from $r=+\infty$ to $r=0$; in other words, the gluing is along the future event horizon $\mathscr{H^+}$. 

To obtain the Penrose diagram for the coordinate domain $r<2M$ we note that the quotient metric can be written as
\begin{align*}
ds^2 & = -\left(1 - \frac{2M}{r}\right) du \, dv = - \left(\frac{2M}{r} - 1\right) (- du) \, dv = - \left(\frac{2M}{r} - 1\right) du' \, dv,
\end{align*}
where $u'=-u$. Since in this coordinate domain $\frac{2M}{r} - 1 > 0$, the quotient metric is in the required form. We have
\[
u' + v = 2r + 4M \log |r - 2M| \in (- \infty, 4M \log (2M) ),
\]
and by setting
\[
u'' = u' - 4M \log (2M)
\]
we obtain
\[
u'' + v < 0.
\]
Consequently, by applying the coordinate rescaling \eqref{tanh} we obtain a triangle (Figure~\ref{Pen_Schw_up}). Again there is a spacelike boundary, where $r=0$, and two null boundaries $\mathscr{H^+}$, where $r=2M$. The Penrose diagram for the domain of the coordinates $(v,r)$ can be obtained gluing the Penrose diagrams in Figures~\ref{Pen_Schw_out} and \ref{Pen_Schw_up} along $\mathscr{H^+}$, so that the null geodesics $v=\text{constant}$ match (Figure~\ref{Pen_Schw_right_up}).

\begin{figure}[h!]
\begin{center}
\psfrag{r=0}{$r=0$}
\psfrag{i+}{$i^+$}
\psfrag{H+}{$\mathscr{H^+}$}
\epsfxsize=.6\textwidth
\leavevmode
\epsfbox{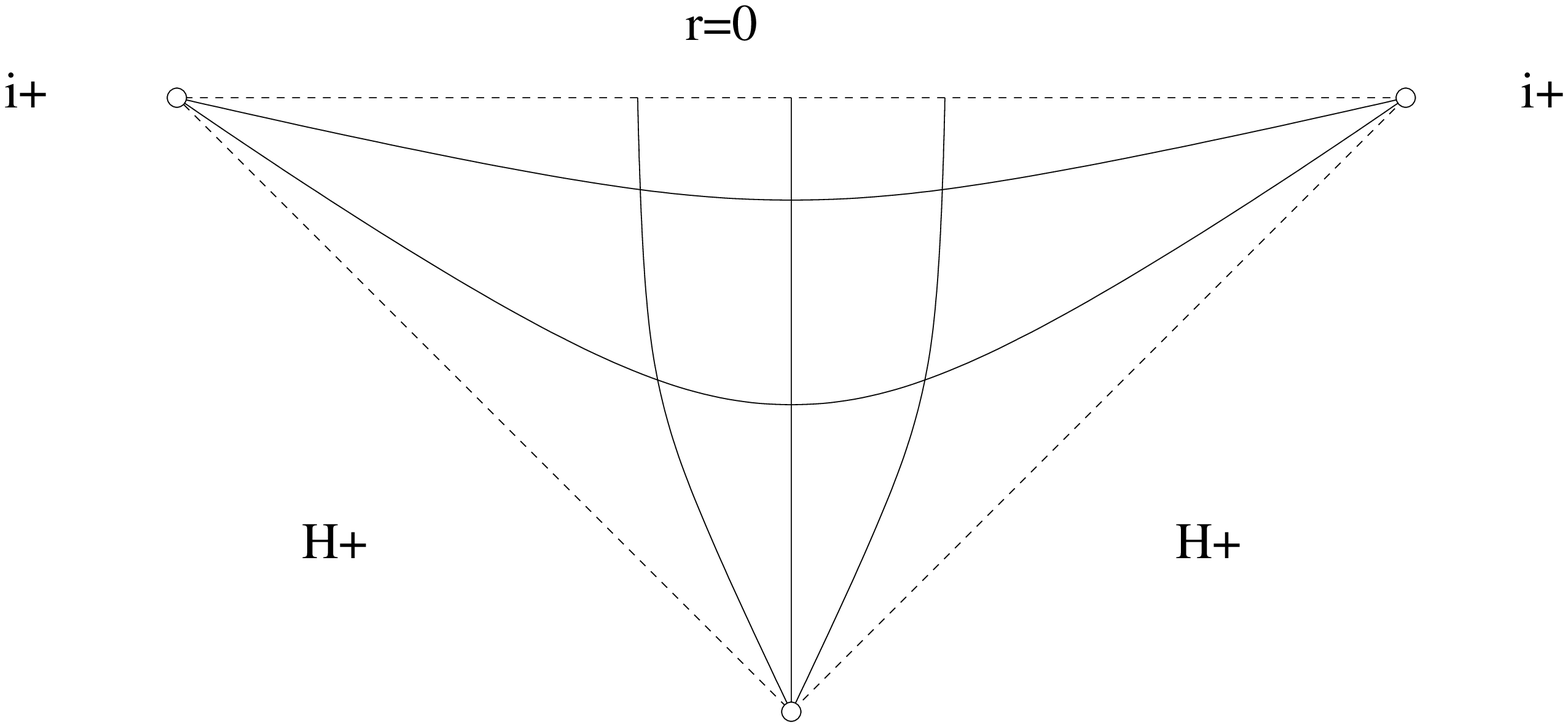}
\end{center}
\caption{Penrose diagram for a region $r<2M$ of the Schwarzschild spacetime, with some level hypersurfaces of $t$ and $r$ represented.} \label{Pen_Schw_up}
\end{figure}

\begin{figure}[h!]
\begin{center}
\psfrag{i+}{$i^+$}
\psfrag{i0}{$i^0$}
\psfrag{i-}{$i^-$}
\psfrag{r=0}{$r=0$}
\psfrag{H+}{$\mathscr{H^+}$}
\psfrag{H-}{$\mathscr{H^-}$}
\psfrag{I+}{$\mathscr{I^+}$}
\psfrag{I-}{$\mathscr{I^-}$}
\epsfxsize=.7\textwidth
\leavevmode
\epsfbox{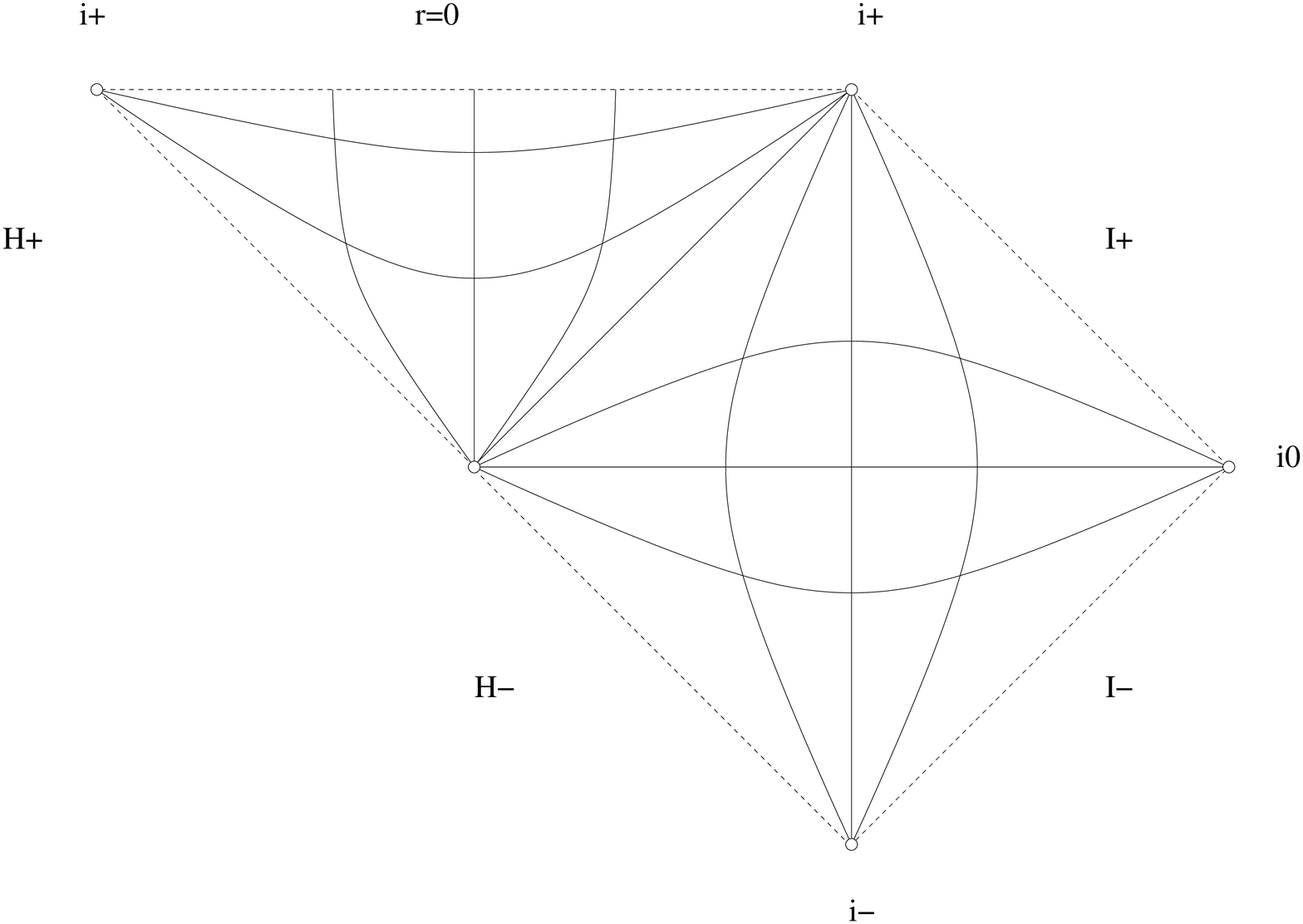}
\end{center}
\caption{Penrose diagram for the domain of the coordinates $(v,r)$ in the Schwarzschild spacetime, with some level hypersurfaces of $t$ and $r$ represented.} \label{Pen_Schw_right_up}
\end{figure}

Both regions $r<2M$ can of course be glued to the region $r>2M$ simultaneously. Since they are invariant under reflections with respect to $t=0$ (the vertical line through their common vertex), it is then clear that a mirror-reversed copy of the region $r>2M$ can be glued to the surviving null boundaries $\mathscr{H^-}$ and $\mathscr{H^+}$ (Figure~\ref{Pen_Schw}). The resulting spacetime, known as the {\bf maximal analytical extension} of the Schwarzschild solution, is a solution of the Einstein equations which cannot be extended any further, since $r \to 0$ or $r \to + \infty$ on the boundary of its Penrose diagram. Note that by continuity the Einstein equations hold at the point where the four Penrose diagrams intersect (known as the {\bf bifurcate sphere}).

\begin{figure}[h!]
\begin{center}
\psfrag{i+}{$i^+$}
\psfrag{i0}{$i^0$}
\psfrag{i-}{$i^-$}
\psfrag{r=0}{$r=0$}
\psfrag{I+}{$\mathscr{I^+}$}
\psfrag{I-}{$\mathscr{I^-}$}
\epsfxsize=.7\textwidth
\leavevmode
\epsfbox{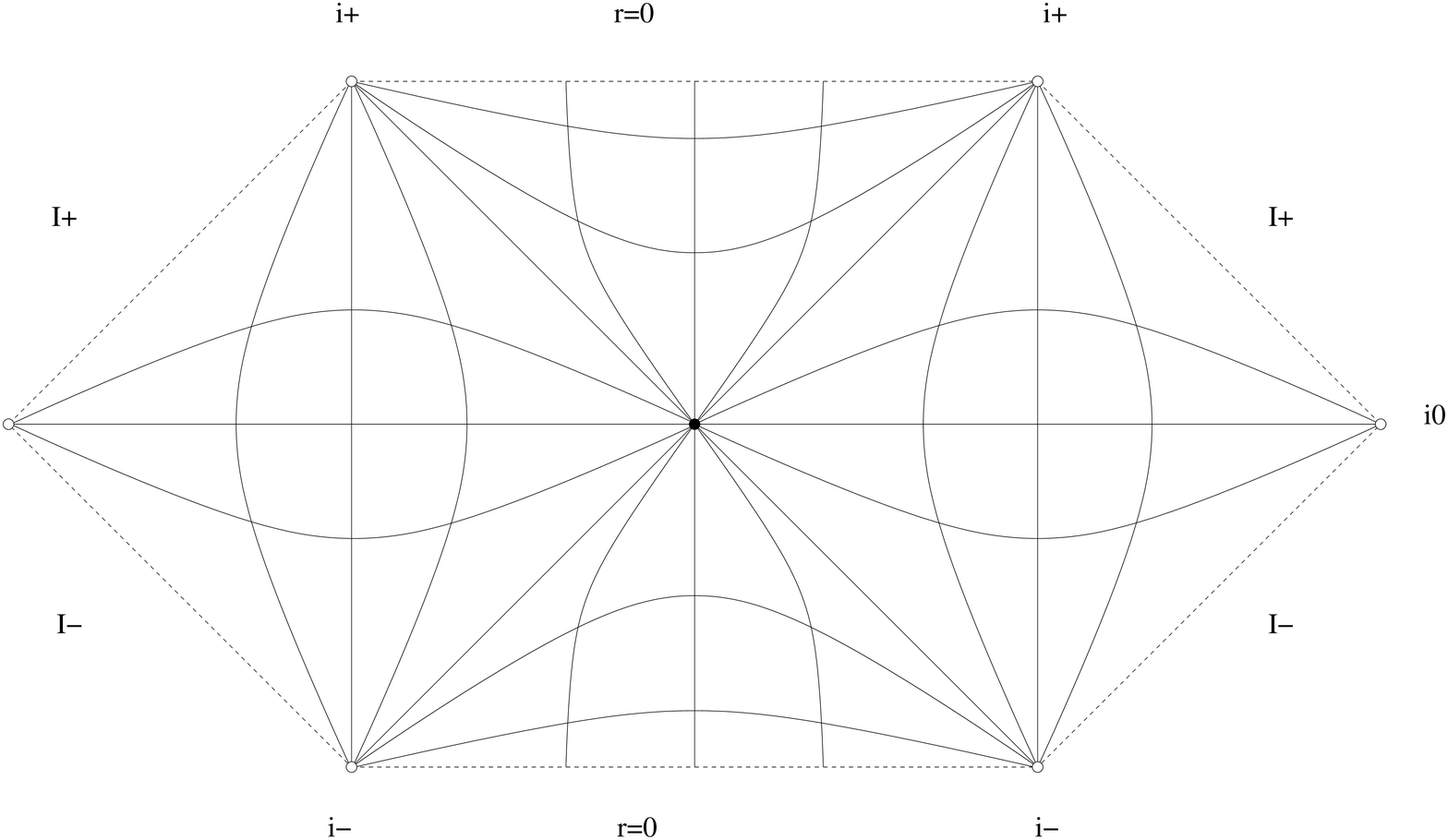}
\end{center}
\caption{Penrose diagram for the maximal analytic extension of the Schwarzschild spacetime, with some level hypersurfaces of $t$ and $r$ represented.} \label{Pen_Schw}
\end{figure}

Let us now analyze in detail the Penrose diagram for the maximal analytic extension of the Schwarzschild spacetime. There are two {\bf asymptotically flat} regions $r > 2M$, corresponding to two causally disconnected universes, joined by a {\bf wormhole}. There are also two regions where $r < 2M$: a {\bf black hole region}, bounded by the future event horizons $\mathscr{H^+}$, from which no causal curve can escape; and a {\bf white hole region}, bounded by the past event horizons $\mathscr{H^-}$, from which every causal curve must escape. Note that the horizons themselves correspond to spheres which are propagating at the speed of light, but whose radius remains constant, $r = 2M$. 

The black hole in the maximal analytic extension of the Schwarzschild spacetime is an {\bf eternal black hole}, that is, a black hole which has always existed (as opposed to having formed by some physical process). We will see shortly how to use the Schwarzschild solution to model physically realistic black holes.

\section{Friedmann-Lema\^\i tre-Robertson-Walker models} \label{sec2.4}

The simplest models of {\bf cosmology}, the study of the Universe as a whole, are obtained from the assumption that space is homogeneous and isotropic (which is true on average at very large scales). It is well known that the only isotropic $3$-dimensional Riemannian metrics are, up to scale, given by
\[
dl^2 = \frac{dr^2}{1-kr^2} + r^2\left( d\theta^2 + \sin^2 \theta d\varphi^2 \right),
\]
where
\[
k =
\begin{cases}
\,\,\,\, 1 \qquad \text{for the standard metric on } S^3 \\
\,\,\,\, 0 \qquad \text{for the standard metric on } \bbR^3 \\
-1 \qquad \text{for the standard metric on } H^3 \\
\end{cases}.
\]
Allowing for a time-dependent scale factor $a(t)$ (also known as the ``radius of the Universe''), we arrive at the {\bf Friedmann-Lema\^\i tre-Robertson-Walker} (FLRW) family of Lorentzian metrics:
\begin{equation} \label{FLRW}
ds^2 = -dt^2 + a^2(t) \left[\frac{dr^2}{1-kr^2} + r^2\left( d\theta^2 + \sin^2 \theta d\varphi^2 \right)\right].
\end{equation}

To interpret these metrics, we consider a general Lorentzian metric of the form
\[
ds^2 = - e^{2 \phi} dt^2 + h_{ij} dx^i dx^j = - e^{2 \phi} dt^2 + dl^2.
\]
The Riemannian metric $dl^2$ is readily interpreted as giving the distances measured between nearby observers with fixed space coordinates $x^i$ in radar experiments: indeed, such observers measure proper time $\tau$ given by
\[
d\tau^2 = e^{2\phi} dt^2.
\]
The null geodesics representing a radar signal bounced by a given observer from a nearby observer (Figure~\ref{radar}) satisfy
\[
ds^2 = 0 \Leftrightarrow  e^{2 \phi} dt^2 = dl^2 \Leftrightarrow d \tau^2 = dl^2 \Leftrightarrow d\tau = \pm dl.
\]

\begin{figure}[h!]
\begin{center}
\psfrag{2dl}{$2dl$}
\epsfxsize=.15\textwidth
\leavevmode
\epsfbox{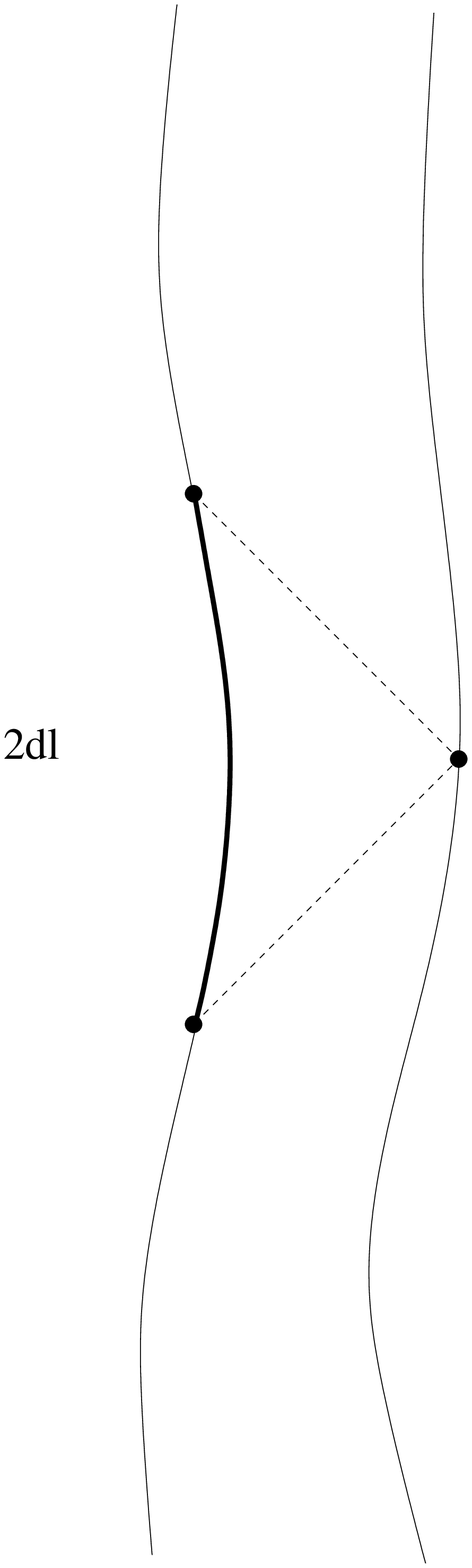}
\end{center}
\caption{Distance between nearby observers.} \label{radar}
\end{figure}

Since the speed of light is $c=1$, the distance traveled between the observers will be half the time between the emission and the reception of the signal:
\[
\frac{(\tau + dl) - (\tau - dl)}{2} = dl.
\]
Moreover, the unit timelike covector field tangent to the trajectories of the the observers with fixed space coordinates $x^i$ is 
\[
U = - e^\phi dt \Leftrightarrow U_\mu = - e^\phi \nabla_ \mu t.
\]
Therefore
\begin{align*}
\nabla_U U^\mu & = - U^\nu \nabla_\nu ( e^\phi \nabla^\mu t) = (U \cdot \phi) U^\mu - e^\phi U^\nu \nabla_\nu \nabla^\mu t \\
& = (U \cdot \phi) U^\mu - e^\phi U^\nu \nabla^\mu \nabla_\nu t  = (U \cdot \phi) U^\mu + e^\phi U^\nu \nabla^\mu (e^{-\phi} U_\nu) \\
& = (U \cdot \phi) U^\mu - U^\nu U_\nu \nabla^\mu \phi + U^\nu \nabla^\mu U_\nu \\
& = \nabla^\mu \phi + (U^\nu \nabla_\nu \phi) U^\mu + \frac12 \nabla^\mu (U^\nu U_\nu) = \nabla^\mu \phi + (U^\nu \nabla_\nu \phi) U^\mu,
\end{align*}
since $U^\nu U_\nu = -1$. In other words,
\[
\nabla_U U = (\grad \phi)^\perp,
\]
where $^\perp$ represents the orthogonal projection on the spacelike hyperplane orthogonal to $U$.

Therefore the observers with fixed space coordinates in the FLRW models have zero acceleration, that is, they are free-falling (by opposition to the corresponding observers in the Schwarzschild spacetime, who must accelerate to remain at fixed $r>2M$). Moreover, the distance between two such observers varies as
\[
d(t) = a(t) \frac{d_0}{a_0} \Rightarrow \dot{d} = \dot{a} \, \frac{d_0}{a_0} = \frac{\dot{a}}{a} \, d.
\]
This relation, known as the {\bf Hubble law}, is often written as
\[
v = H d,
\]
where $v$ is the relative velocity and
\[
H = \frac{\dot{a}}{a}
\]
is the so-called {\bf Hubble constant} (for historical reasons, since it actually varies in time).

We will model the matter content of the universe as an uniform dust of galaxies placed at fixed space coordinates (hence free-falling):
\begin{equation} \label{TFLRW}
T = \rho(t) dt \otimes dt.
\end{equation}
Plugging the metric~\eqref{FLRW} and the energy-momentum tensor~\eqref{TFLRW} into the Einstein equations, and integrating once, results in the so-called {\bf Friedmann equations} 
\[
\begin{cases}
\displaystyle \frac12 \dot{a}^2 - \frac{\alpha}{a} - \frac{\Lambda}6 a^2 = - \frac{k}2 \\
\\
\displaystyle \frac{4\pi}3 \rho a^3 = \alpha
\end{cases}
\]
(where $\alpha$ is an integration constant). The first Friedmann equation is a first order ODE for $a(t)$; it can be seen as the equation of conservation of energy for a particle moving in the $1$-dimensional effective potential
\[
V(a) = - \frac{\alpha}{a} - \frac{\Lambda}6 a^2
\]
with energy $- \frac{k}2$. Once this equation has been solved, the second Friedmann equation yields $\rho(t)$ from $a(t)$. We now examine in detail the FLRW models arising from the solutions of these equations.

\subsection{Milne universe}

If we set $\alpha=\Lambda=0$ then the first Friedmann equation becomes
\[
\dot{a}^2 = - k.
\]
Therefore either $k=0$ and $\dot{a} = 0$, which corresponds to the Minkowski spacetime, or $k = -1$ and $\dot{a}^2 = 1$, that is
\[
ds^2 = - dt^2 + t^2 dl^2_{H^3},
\]
where $dl^2_{H^3}$ represents the metric of the unit hyperbolic $3$-space; this is the so-called {\bf Milne universe}. It turns out that the Milne universe is isometric to an open region of the Minkowski spacetime, namely the region limited by the future (or past) light cone of the origin. This region is foliated by hyperboloids $S_t$ of the form
\[
T^2 - X^2 - Y^2 - Z^2 = t^2,
\]
whose induced metric is that of a hyperbolic space of radius $t$ (Figure~\ref{Milne}). Note that the light cone corresponds to $a(t)=t=0$, that is, the {\bf Big Bang} of the Milne universe.

\begin{figure}[h!]
\begin{center}
\psfrag{u}{$T$}
\psfrag{x}{$X$}
\psfrag{y}{$Y$}
\psfrag{St}{$S_t$}
\epsfxsize=.6\textwidth
\leavevmode
\epsfbox{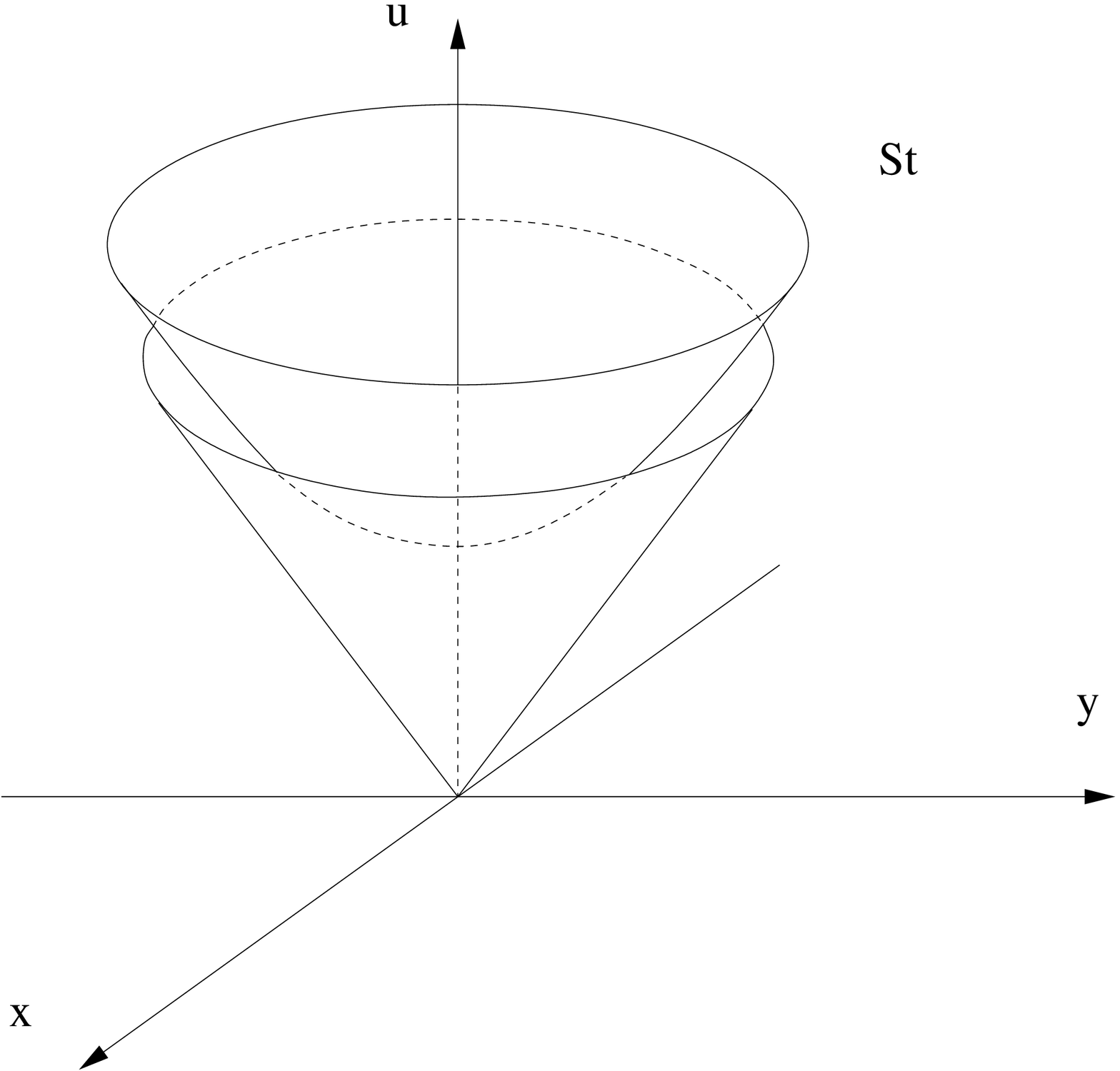}
\end{center}
\caption{Milne universe.} \label{Milne}
\end{figure}

\subsection{de Sitter universe}

If $\alpha=0$ and $\Lambda>0$ we can choose units such that $\Lambda=3$. The first Friedmann equation then becomes
\[
\dot{a}^2 - a^2 = - k.
\]
In this case all three values $k=1$, $k=0$ and $k=-1$ are possible; the corresponding metrics are, respectively,
\begin{align*}
& ds^2 = - dt^2 + \cosh^2t \, dl^2_{S^3}; \\
& ds^2 = - dt^2 + e^{2t} dl^2_{\bbR^3}; \\
& ds^2 = - dt^2 + \sinh^2t \, dl^2_{H^3},
\end{align*}
where $dl^2_{S^3}$, $dl^2_{\bbR^3}$ and $dl^2_{H^3}$ represent the metric of the unit $3$-sphere, the Euclidean $3$-space and the the unit hyperbolic $3$-space. 

It turns out that the last two models correspond to open regions of the first, which is then called the {\bf de Sitter universe}. It represents a spherical universe which contracts to a minimum radius ($1$ in our units) and then re-expands. It is easily seen to be isometric to the unit hyperboloid
\[
-T^2 + X^2 + Y^2 + Z^2 + W^2 = 1
\]
in the Minkowski $5$-dimensional spacetime (Figure~\ref{dS}).

\begin{figure}[h!]
\begin{center}
\psfrag{u}{$T$}
\psfrag{x}{$X$}
\psfrag{y}{$Y$}
\epsfxsize=.6\textwidth
\leavevmode
\epsfbox{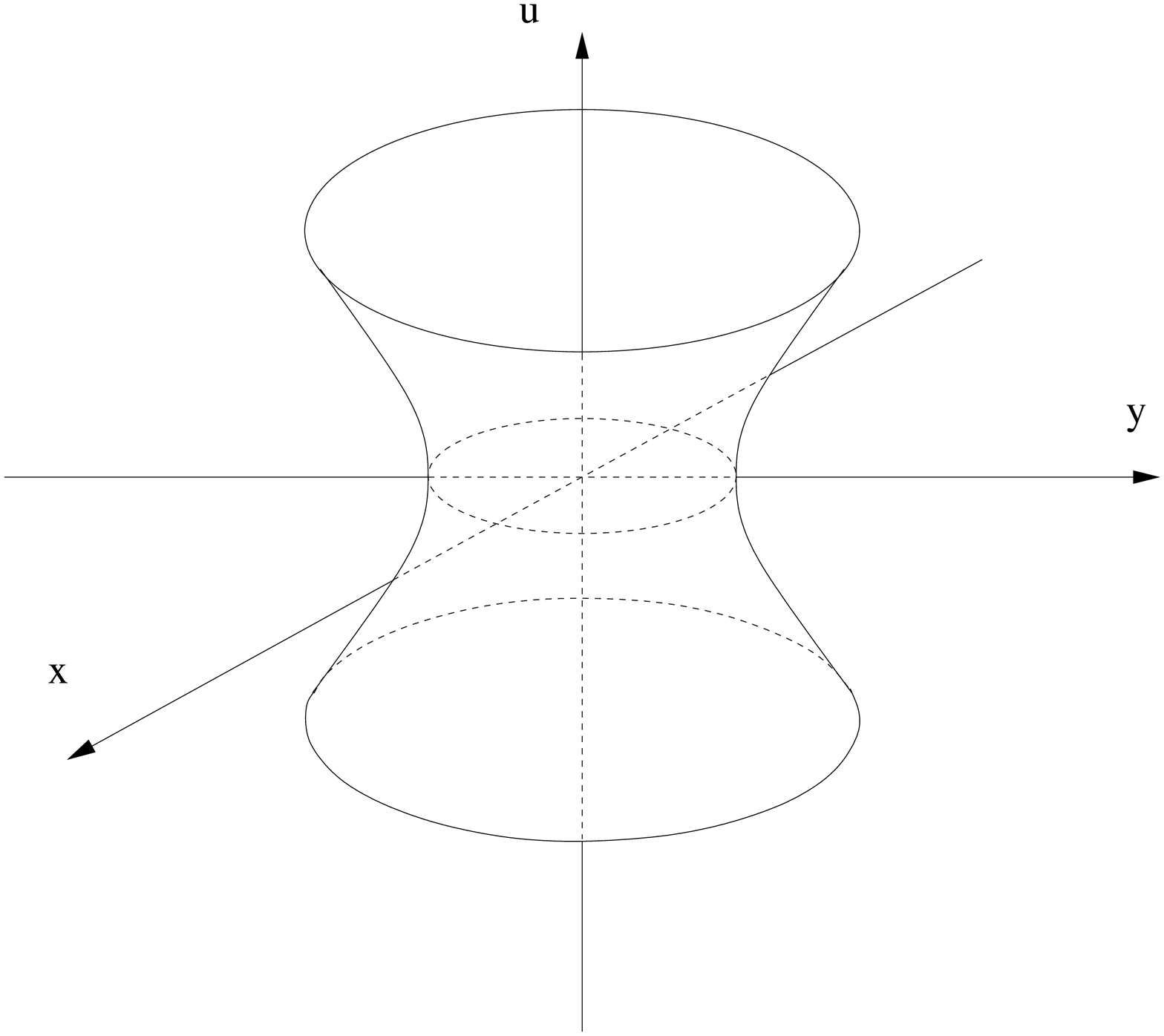}
\end{center}
\caption{de Sitter universe.} \label{dS}
\end{figure}

To obtain the Penrose diagram for the de Sitter universe we write its metric as
\begin{align*}
ds^2 & = - dt^2 + \cosh^2t \left[ d \psi^2 + \sin^2 \psi  \left(d\theta^2 + \sin^2\theta d\varphi^2\right) \right] \\
& = \cosh^2t \left[ - d\tau^2 + d \psi^2 + \sin^2 \psi  \left(d\theta^2 + \sin^2\theta d\varphi^2\right) \right] \\
& = \cosh^2t \left[ - d\tau^2 + d \psi^2 \right] + r^2 \left(d\theta^2 + \sin^2\theta d\varphi^2\right),
\end{align*}
where $\psi \in [0,\pi]$,
\[
\tau = \int_{- \infty}^t \frac{dt}{\cosh t}
\]
and
\[
r = \cosh t \sin \psi.
\]
Since
\[
\int_{- \infty}^{+ \infty} \frac{dt}{\cosh t} = \pi,
\]
we see that the quotient metric is conformal to the square $(0,\pi) \times [0, \pi]$ of the Minkowski $2$-dimensional spacetime, and so the Penrose diagram is as depicted in Figure~\ref{dS}. Note that there are two lines where $r=0$, corresponding to two antipodal points of the $3$-sphere. A light ray emitted from one of these points at $t=-\infty$ has just enough time to reach the other point at $t=+\infty$ (dashed line in the diagram). Note also that in this case $\mathscr{I^-}$ and $\mathscr{I^+}$ (defined as the past and future boundary points approached by null geodesics along which $r \to + \infty$) are spacelike boundaries.

\begin{figure}[h!]
\begin{center}
\psfrag{r=0}{$r=0$}
\psfrag{I+}{$\mathscr{I^+}$}
\psfrag{I-}{$\mathscr{I^-}$}
\epsfxsize=.5\textwidth
\leavevmode
\epsfbox{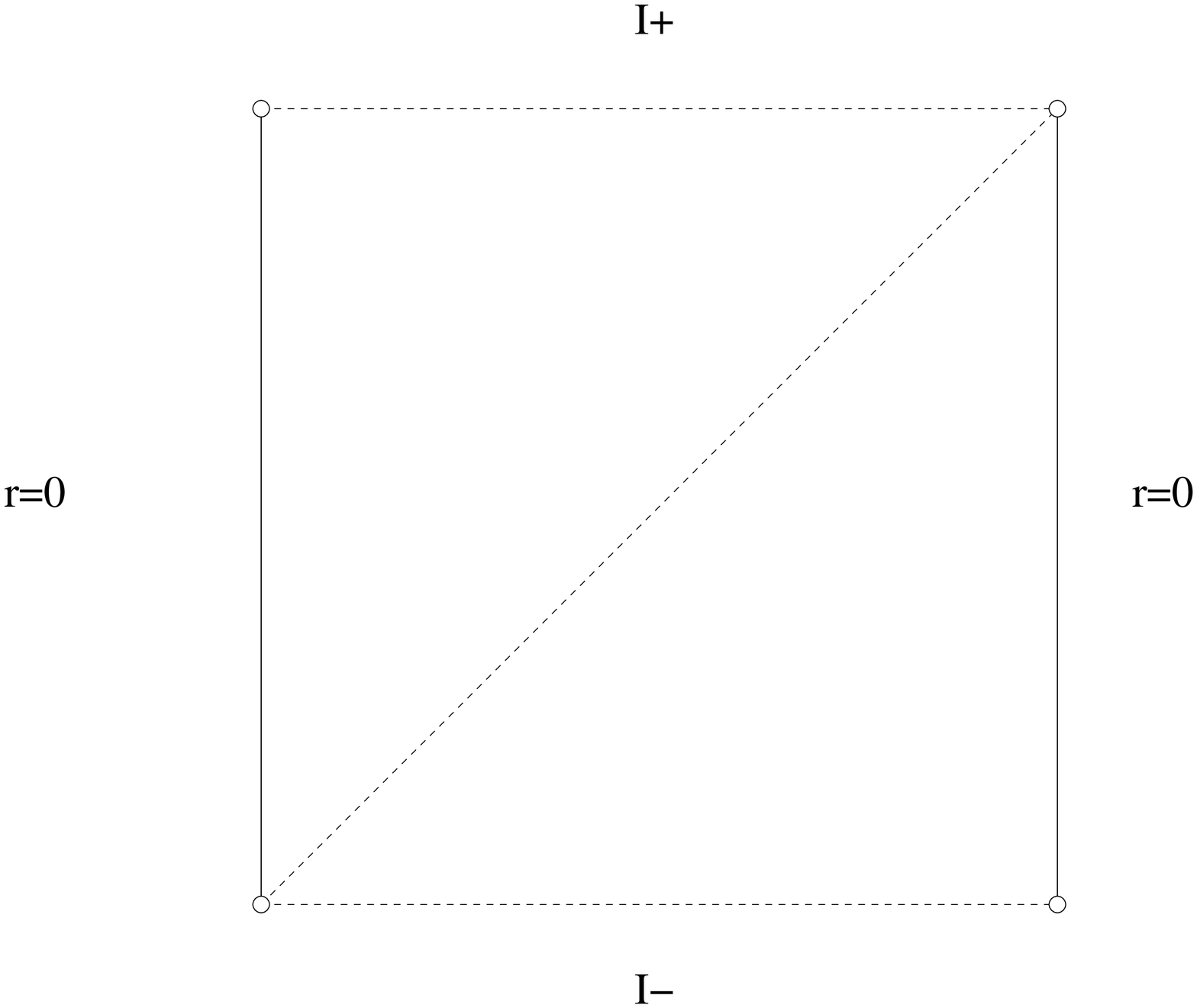}
\end{center}
\caption{Penrose diagram for the de Sitter universe.} \label{Pen_dS}
\end{figure}

\subsection{Anti-de Sitter universe}

If $\alpha=0$ and $\Lambda<0$ we can choose units such that $\Lambda=-3$. The the first Friedmann equation then becomes
\[
\dot{a}^2 + a^2 = - k.
\]
In this case only $k=-1$ is possible; the corresponding metric is
\[
ds^2 = - dt^2 + \cos^2t \, dl^2_{H^3}.
\]
It turns out (see Exercise~\ref{development} in Chapter~\ref{chapter5}) that this model is an open region of the spacetime with metric
\[
ds^2 = - \cosh^2\psi dt^2 + d \psi^2 + \sinh^2 \psi  \left(d\theta^2 + \sin^2\theta d\varphi^2\right)
\]
(where $\psi \in [0,+\infty)$), called the {\bf anti-de Sitter universe}. It represents a static hyperbolic universe (with radius $1$ in our units).

To obtain the Penrose diagram for the anti-de Sitter universe we write its metric as
\begin{align*}
ds^2 & = \cosh^2\psi \left[ - dt^2 + \frac{d \psi^2}{\cosh^2 \psi} \right] + \sinh^2 \psi  \left(d\theta^2 + \sin^2\theta d\varphi^2\right)  \\
& = \cosh^2\psi \left[ - dt^2 + dx^2 \right] + r^2  \left(d\theta^2 + \sin^2\theta d\varphi^2\right)
\end{align*}
where 
\[
x = \int_0^\psi \frac{d\psi}{\cosh \psi}
\]
and
\[
r = \sinh \psi.
\]
Since
\[
\int_0^{+ \infty} \frac{d\psi}{\cosh \psi} = \frac{\pi}2,
\]
we see that the quotient metric is conformal to the strip $\bbR \times [0, \frac{\pi}2)$ of the Minkowski $2$-dimensional spacetime, and so the Penrose diagram is as depicted in Figure~\ref{Pen_AdS}. The FLRW model above corresponds to the triangular region in the diagram. Note also that in this case $\mathscr{I^-} \equiv \mathscr{I^+} \equiv \mathscr{I}$ is a timelike boundary.

\begin{figure}[h!]
\begin{center}
\psfrag{r=0}{$r=0$}
\psfrag{I}{$\mathscr{I}$}
\epsfxsize=.4\textwidth
\leavevmode
\epsfbox{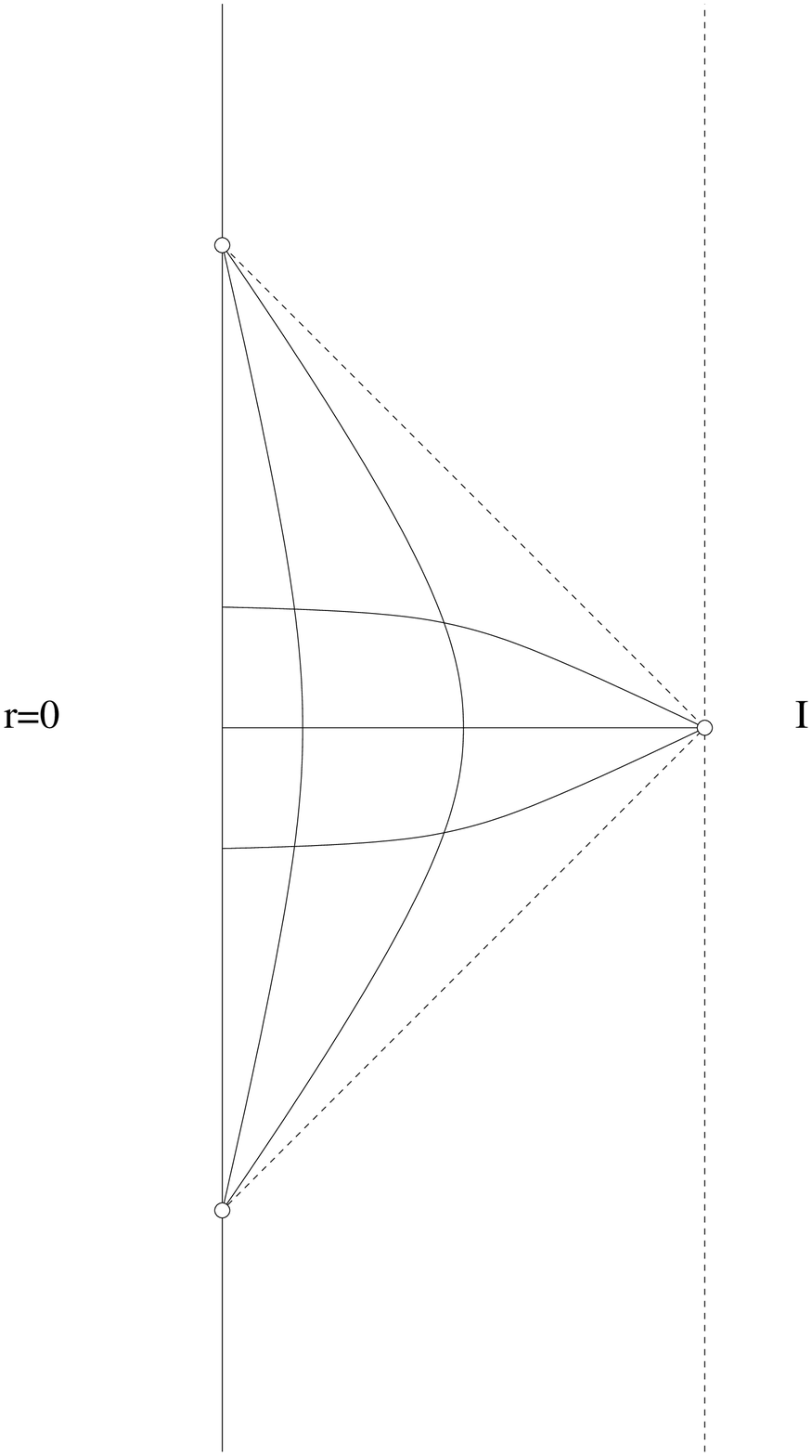}
\end{center}
\caption{Penrose diagram for the anti-de Sitter universe.} \label{Pen_AdS}
\end{figure}

\subsection{Universes with matter and $\Lambda=0$}

If $\alpha>0$ and $\Lambda=0$, the first Friedmann equation becomes
\[
\dot{a}^2 - \frac{2\alpha}a = - k.
\]
In this case all three values $k=1$, $k=0$ and $k=-1$ are possible. Although it is possible to obtain explicit formulas for the solutions of these equations, it is simpler to analyze the graph of the effective potential $V(a)$ (Figure~\ref{Lzero_graph}). Possibly by reversing and translating $t$, we can assume that all solutions are defined for $t>0$, with $\lim_{t\to 0} a(t) = 0$, implying $\lim_{t\to 0} \rho(t) = + \infty$. Therefore all three models have a true singularity at $t=0$, known as the {\bf Big Bang}, where the scalar curvature $R=8\pi \rho$ also blows up; this is not true for the Milne universe or the open region in the anti-de Sitter universe, which can be extended across the Big Bang. The spherical universe ($k=1$) reaches a maximum radius $2 \alpha$ and re-collapses, forming a second singularity (the {\bf Big Crunch}); the radius of the flat ($k=0$) and hyperbolic ($k=-1$) universes increases monotonically.

\begin{figure}[h!]
\begin{center}
\psfrag{a}{$a$}
\psfrag{V(a)}{$V(a)$}
\psfrag{k=1}{$k=1$}
\psfrag{k=0}{$k=0$}
\psfrag{k=-1}{$k=-1$}
\epsfxsize=.6\textwidth
\leavevmode
\epsfbox{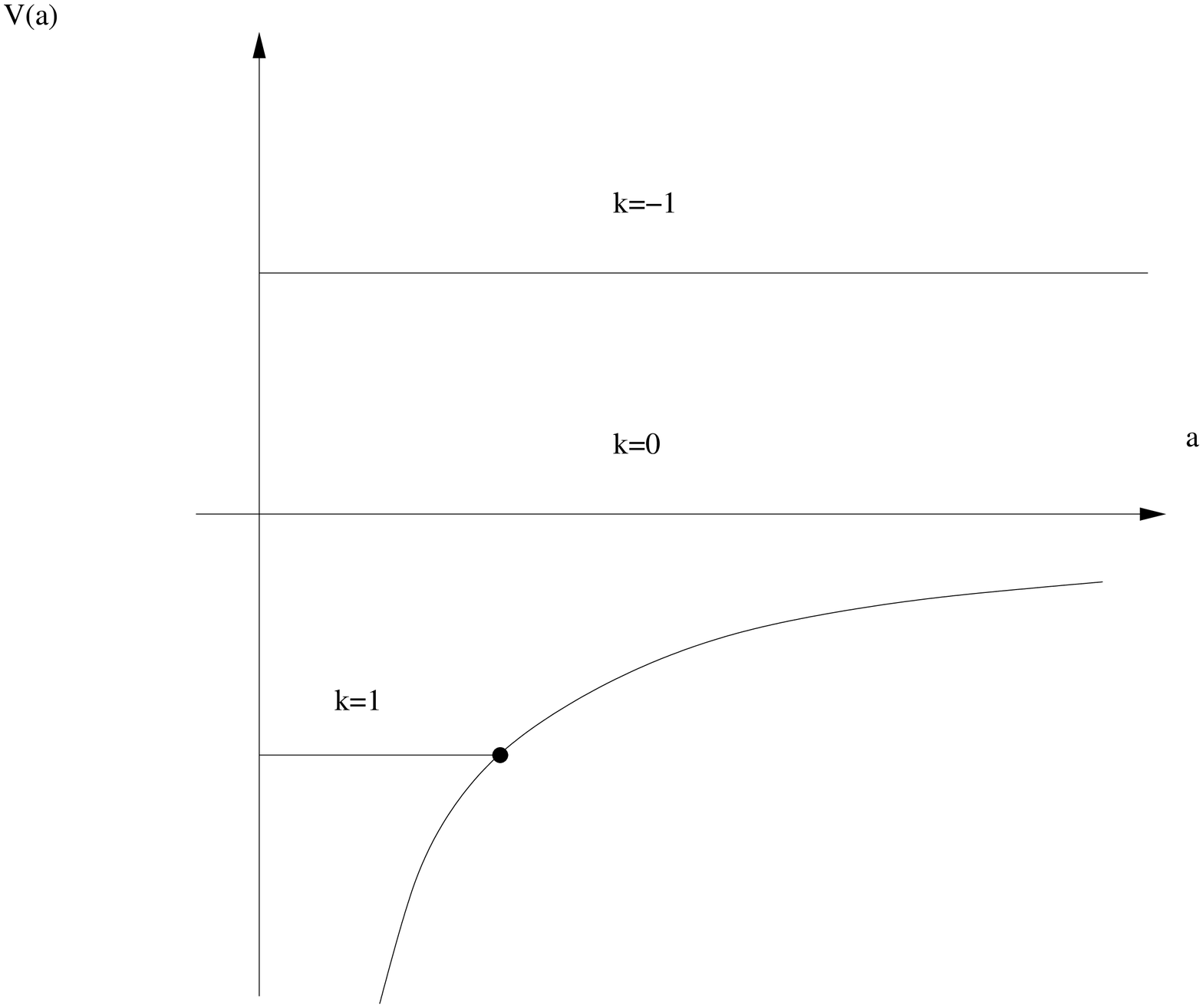}
\end{center}
\caption{Effective potential for FLRW models with $\Lambda=0$.} \label{Lzero_graph}
\end{figure}

To obtain the Penrose diagram for the spherical universe we write its metric as
\begin{align*}
ds^2 & = - dt^2 + a^2(t) \left[ d \psi^2 + \sin^2 \psi  \left(d\theta^2 + \sin^2\theta d\varphi^2\right) \right] \\
& = a^2(t) \left[ - d\tau^2 + d \psi^2 + \sin^2 \psi  \left(d\theta^2 + \sin^2\theta d\varphi^2\right) \right] \\
& = a^2(t) \left[ - d\tau^2 + d \psi^2 \right] + r^2 \left(d\theta^2 + \sin^2\theta d\varphi^2\right),
\end{align*}
where $\psi \in [0,\pi]$,
\[
\tau = \int_{0}^t \frac{dt}{a(t)}
\]
and
\[
r = a(t) \sin \psi.
\]
Since
\[
\int_{0}^{t_\text{max}} \frac{dt}{a(t)} = 2 \int_0^{a_\text{max}} \frac{da}{a\dot{a}} = 2 \int_0^{2 \alpha} \frac{da}{a\sqrt{\frac{2 \alpha}{a} - 1}} = 2\pi,
\]
we see that the quotient metric is conformal to the rectangle $(0,2\pi) \times [0, \pi]$ of the Minkowski $2$-dimensional spacetime, and so the Penrose diagram is as depicted in Figure~\ref{Pen_spherical}. Note that there are two lines where $r=0$, corresponding to two antipodal points of the $3$-sphere. A light ray emitted from one of these points at $t=0$ has just enough time to circle once around the universe an return at $t=t_\text{max}$ (dashed line in the diagram). Note also that the Big Bang and the Big Crunch are spacelike boundaries.

\begin{figure}[h!]
\begin{center}
\psfrag{r=0}{$r=0$}
\psfrag{Big Bang}{Big Bang}
\psfrag{Big Crunch}{Big Crunch}
\epsfxsize=.5\textwidth
\leavevmode
\epsfbox{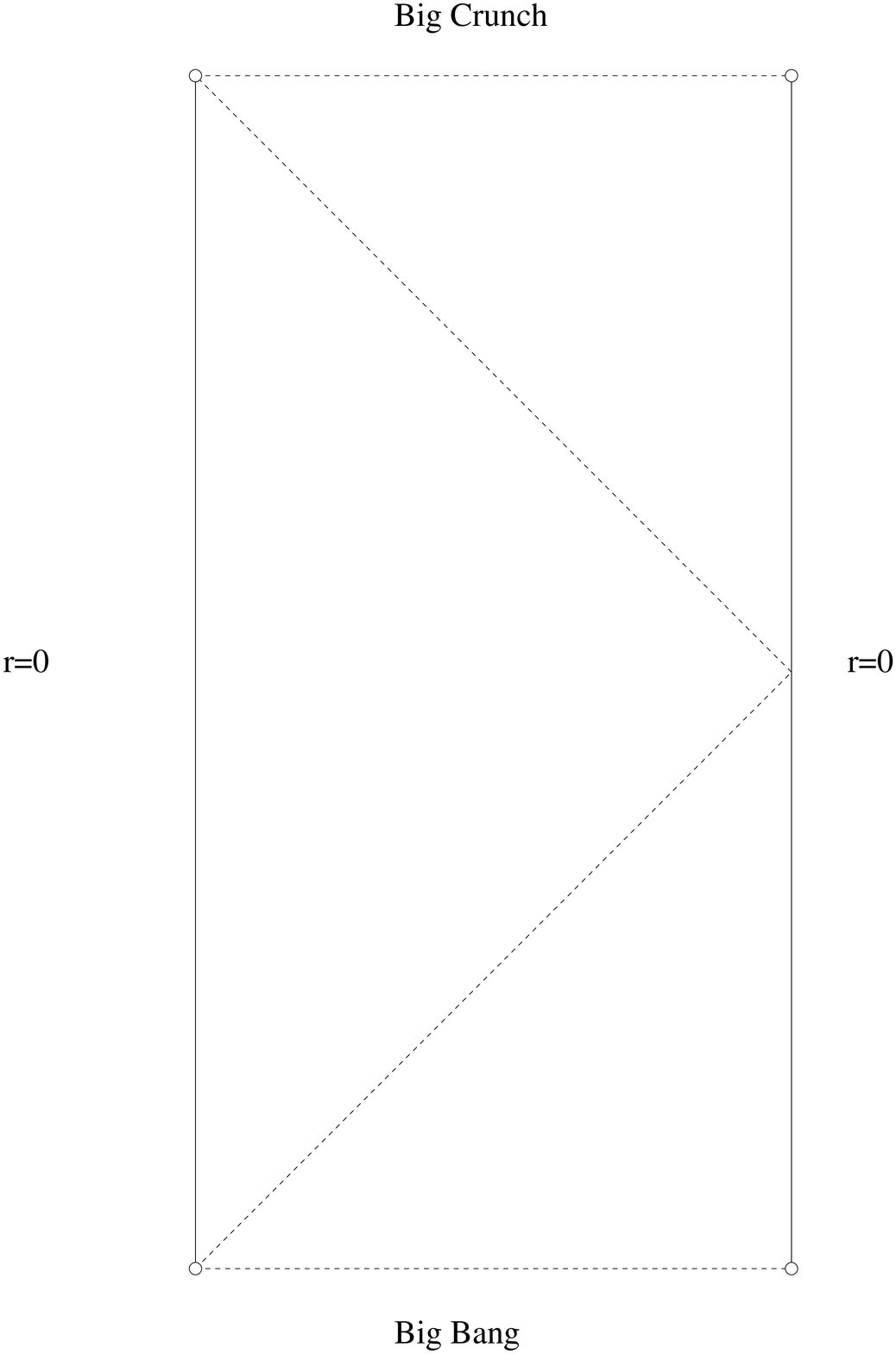}
\end{center}
\caption{Penrose diagram for the spherical universe.} \label{Pen_spherical}
\end{figure}

To obtain the Penrose diagram for the flat universe we write its metric as
\begin{align*}
ds^2 & = - dt^2 + a^2(t) \left[ d \rho^2 + \rho^2 \left(d\theta^2 + \sin^2\theta d\varphi^2\right) \right] \\
& = a^2(t) \left[ - d\tau^2 + d \rho^2 + \rho^2 \left(d\theta^2 + \sin^2\theta d\varphi^2\right) \right] \\
& = a^2(t) \left[ - d\tau^2 + d \rho^2 \right] + r^2 \left(d\theta^2 + \sin^2\theta d\varphi^2\right),
\end{align*}
where $\rho \in [0,+\infty)$ and $r = a(t) \rho$. Since
\[
\int_{0}^{+\infty} \frac{dt}{a(t)} = \int_0^{+\infty} \frac{da}{a\dot{a}} = \int_0^{+ \infty} \frac{da}{a\sqrt{\frac{2 \alpha}{a}}} = + \infty,
\]
we see that the quotient metric is conformal to the region $(0,+\infty) \times [0,+\infty)$ of the Minkowski $2$-dimensional spacetime, and so the Penrose diagram is as depicted in Figure~\ref{Pen_flat}. Note that the Big Bang is a spacelike boundary.

\begin{figure}[h!]
\begin{center}
\psfrag{r=0}{$r=0$}
\psfrag{Big Bang}{Big Bang}
\psfrag{i+}{$i^+$}
\psfrag{i0}{$i^0$}
\psfrag{I+}{$\mathscr{I^+}$}
\epsfxsize=.6\textwidth
\leavevmode
\epsfbox{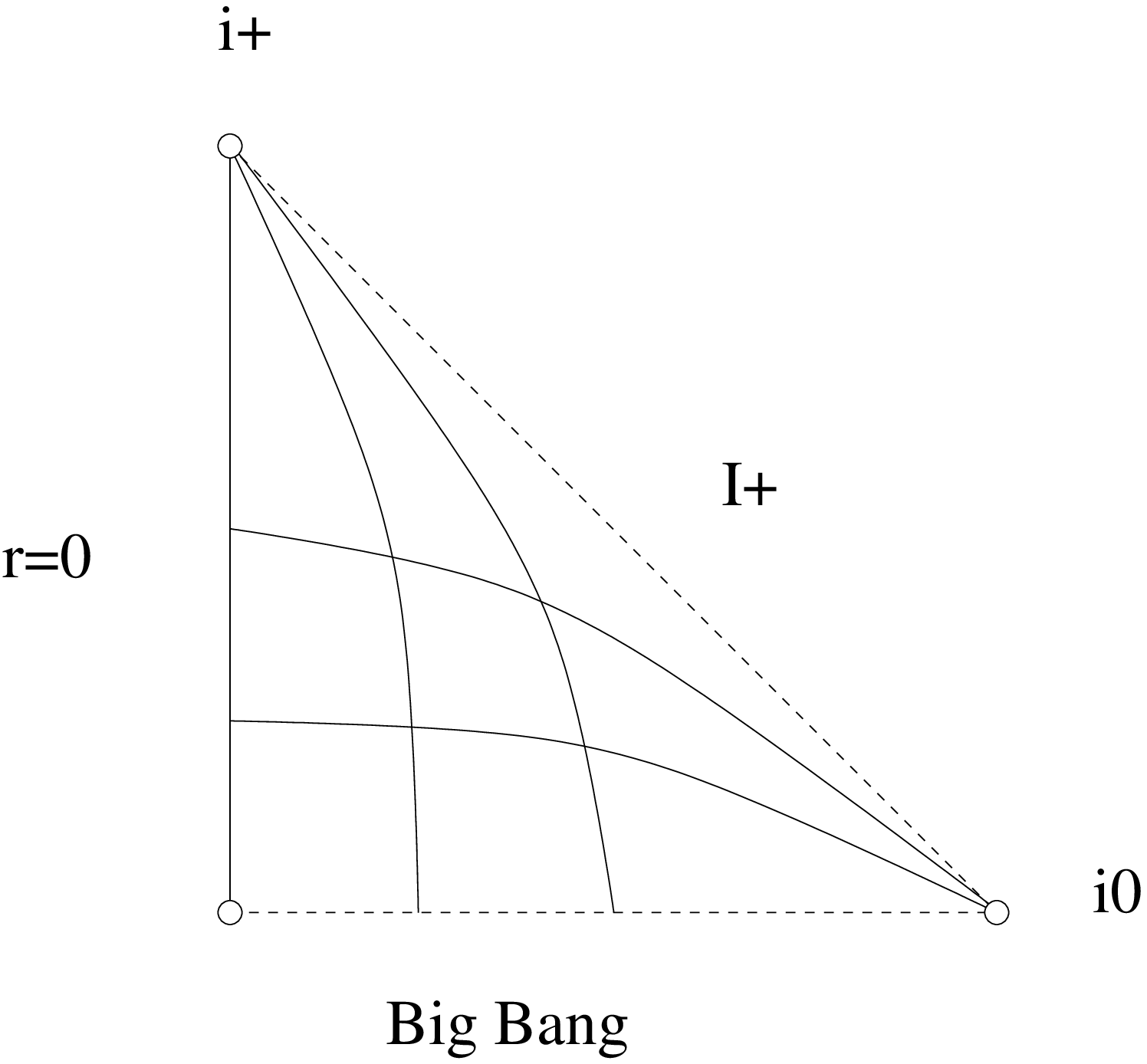}
\end{center}
\caption{Penrose diagram for the flat and hyperbolic universes.} \label{Pen_flat}
\end{figure}

The Penrose diagram for the hyperbolic universe turns out to be the same as for the flat universe. To see this we write its metric as
\begin{align*}
ds^2 & = - dt^2 + a^2(t) \left[ d \psi^2 + \sinh^2\psi \left(d\theta^2 + \sin^2\theta d\varphi^2\right) \right] \\
& = a^2(t) \left[ - d\tau^2 + d \psi^2 + \sinh^2\psi \left(d\theta^2 + \sin^2\theta d\varphi^2\right) \right] \\
& = a^2(t) \left[ - d\tau^2 + d \psi^2 \right] + r^2 \left(d\theta^2 + \sin^2\theta d\varphi^2\right),
\end{align*}
where $\psi \in [0,+\infty)$ and $r = a(t) \sinh \psi$, and note that
\[
\int_{0}^{+\infty} \frac{dt}{a(t)} = \int_0^{+\infty} \frac{da}{a\dot{a}} = \int_0^{+ \infty} \frac{da}{a\sqrt{\frac{2 \alpha}{a} + 1}} = + \infty.
\]

\subsection{Universes with matter and $\Lambda>0$}

If $\alpha>0$ and $\Lambda>0$ we can choose units such that $\Lambda=3$. The first Friedmann equation then becomes
\[
\dot{a}^2 - \frac{2\alpha}a - a^2 = - k.
\]
In this case all three values $k=1$, $k=0$ and $k=-1$ are possible. As before, we analyze the graph of the effective potential $V(a)$ (Figure~\ref{Lpositive_graph}). The hyperbolic and flat universes behave qualitatively like when $\Lambda=0$, although $\dot{a}(t)$ is now unbounded as $t \to + \infty$, instead of approaching some constant. The spherical universe has a richer spectrum of possible behaviors, depending on $\alpha$, represented in Figure~\ref{Lpositive_graph} by drawing the line of constant energy $-k=-1$ at three different heights. The higher line (corresponding to $\alpha > \frac{\sqrt{3}}{9}$) yields a behaviour similar to that of the hyperbolic and flat universes. The intermediate line (corresponding to $\alpha=\frac{\sqrt{3}}{9}$) gives rise to an unstable equilibrium point $a = \frac{\sqrt{3}}{3}$, where the attraction force of the matter is balanced by the repulsion force of the cosmological constant; it corresponds to the so-called {\bf Einstein universe}, the first cosmological model ever proposed. The intermediate line also yields two solutions asymptotic to the Einstein universe, one containing a Big Bang and the other endless expansion. Finally, the lower line (corresponding to $\alpha < \frac{\sqrt{3}}{9}$) yields two different types of behaviour (depending on the initial conditions): either similar to the spherical model with $\Lambda=0$, or to the de Sitter universe. 

\begin{figure}[h!]
\begin{center}
\psfrag{a}{$a$}
\psfrag{V(a)}{$V(a)$}
\psfrag{k=1}{$k=1$}
\psfrag{k=0}{$k=0$}
\psfrag{k=-1}{$k=-1$}
\epsfxsize=.6\textwidth
\leavevmode
\epsfbox{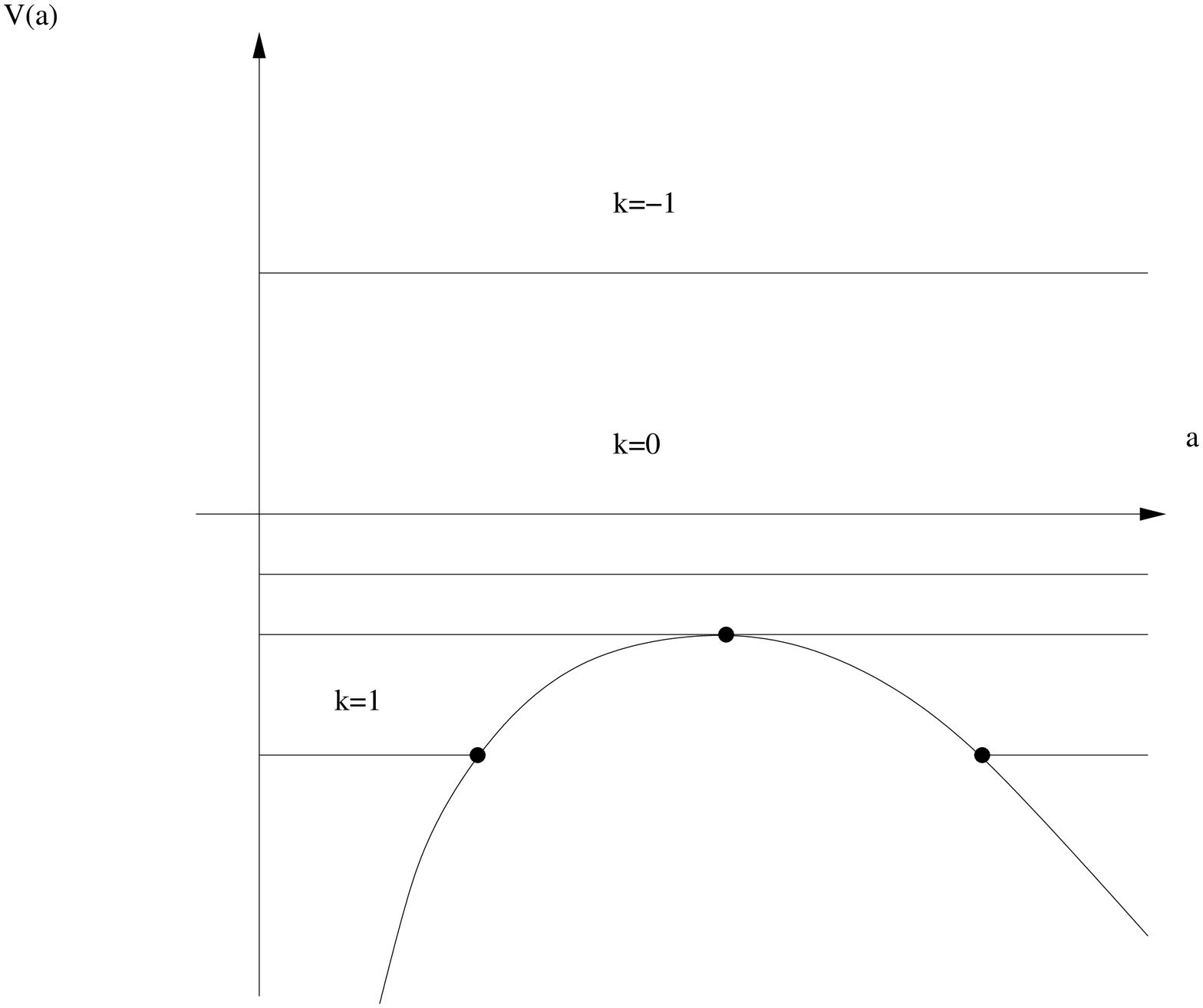}
\end{center}
\caption{Effective potential for FLRW models with $\Lambda>0$.} \label{Lpositive_graph}
\end{figure}

It is currently believed that the best model for our physical Universe is the the flat universe with $\Lambda > 0$. If we write the first Friedmann equation as
\[
H^2 \equiv \frac{\dot{a}^2}{a^2} = \frac{8 \pi}{3} \rho + \frac{\Lambda}{3}
\]
then the terms on the right-hand side are in the proportion $2 : 5$ at the present time.

\subsection{Universes with matter and $\Lambda<0$}

If $\alpha>0$ and $\Lambda<0$ we can choose units such that $\Lambda=-3$. The first Friedmann equation then becomes
\[
\dot{a}^2 - \frac{2\alpha}a + a^2 = - k.
\]
In this case all three values $k=1$, $k=0$ and $k=-1$ are possible. As before, we analyze the graph of the effective potential $V(a)$ (Figure~\ref{Lnegative_graph}). The qualitative behaviour of the hyperbolic, flat and spherical universes is the same as the spherical universe with $\Lambda=0$, namely starting at a Big Bang and ending at a Big Crunch.

\begin{figure}[h!]
\begin{center}
\psfrag{a}{$a$}
\psfrag{V(a)}{$V(a)$}
\psfrag{k=1}{$k=1$}
\psfrag{k=0}{$k=0$}
\psfrag{k=-1}{$k=-1$}
\epsfxsize=.6\textwidth
\leavevmode
\epsfbox{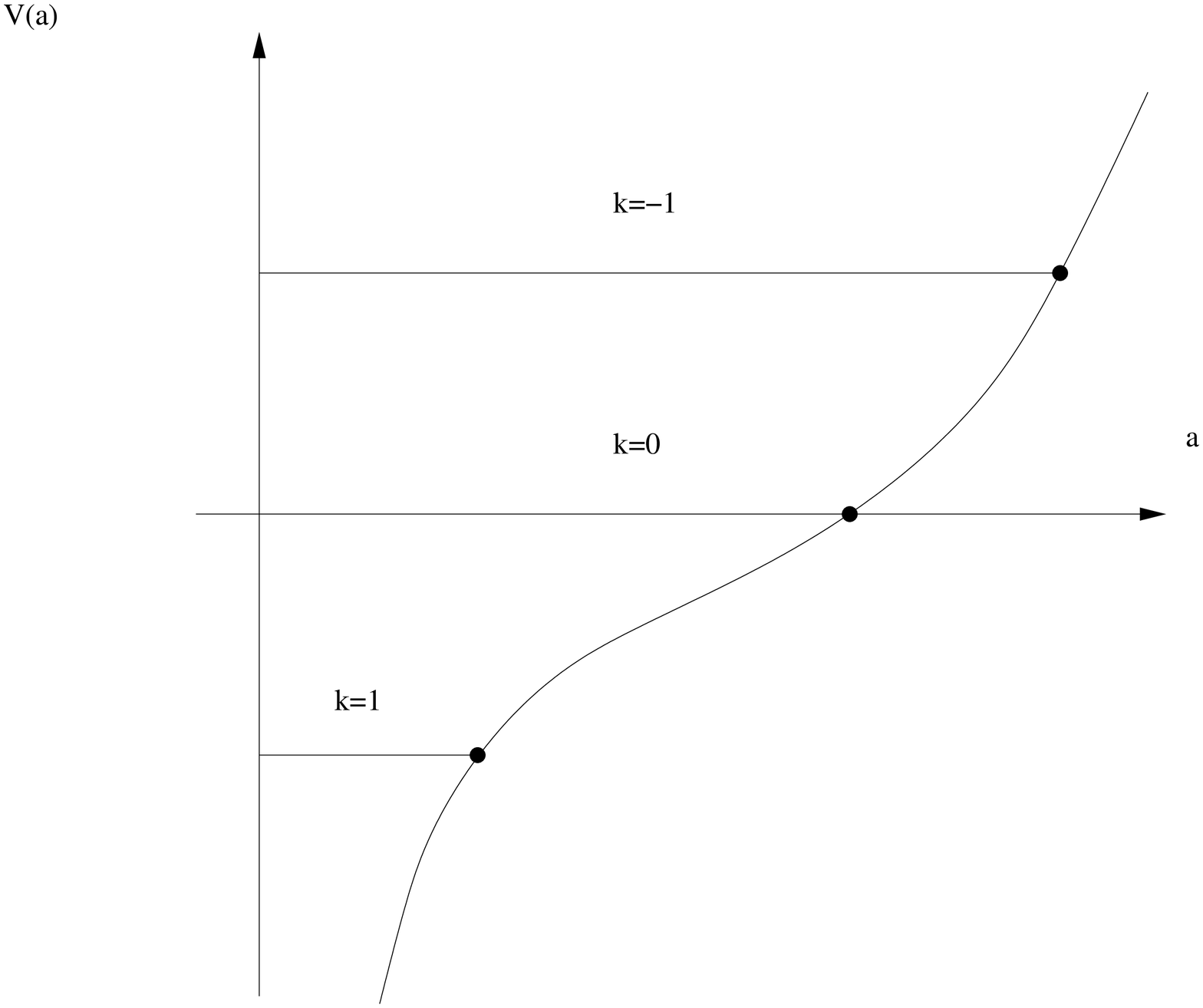}
\end{center}
\caption{Effective potential for FLRW models with $\Lambda>0$.} \label{Lnegative_graph}
\end{figure}

\section{Matching} \label{sec2.5}

Let $(M_1,g_1)$ and $(M_2,g_2)$ be solutions of the Einstein field equations containing open sets $U_1$ and $U_2$ whose boundaries $S_1$ and $S_2$ are {\bf timelike hypersurfaces}, that is, hypersurfaces whose induced metric is Lorentzian (or, equivalently, whose normal vector is spacelike). If $S_1$ is diffeomorphic to $S_2$ then we can identify them to obtain a new manifold $M$ gluing $U_1$ to $U_2$ along $S_1 \cong S_2$ (Figure~\ref{matching}).

\begin{figure}[h!]
\begin{center}
\psfrag{g1}{$g_1$}
\psfrag{g2}{$g_2$}
\psfrag{n}{$n$}
\psfrag{S1=S2}{$S_1 \cong S_2$}
\epsfxsize=.6\textwidth
\leavevmode
\epsfbox{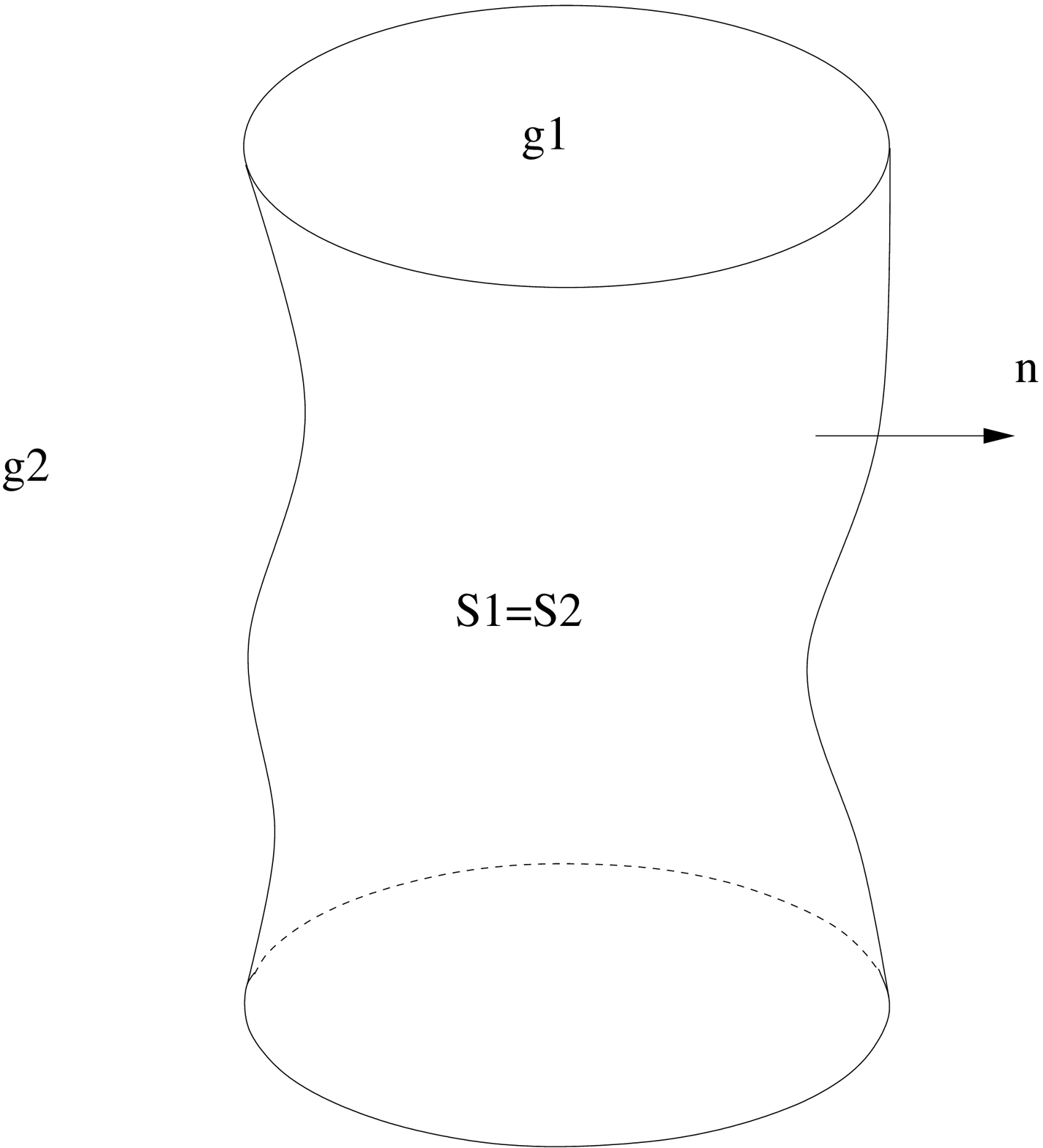}
\end{center}
\caption{Matching two spacetimes.} \label{matching}
\end{figure}

Let $n$ be the unit normal vector to $S_1$ pointing out of $U_1$, which we identify with the unit normal vector to $S_2$ pointing into $U_2$. If $(x^1, x^2, x^3)$ are local coordinates on $S \equiv S_1 \cong S_2$, we can construct a system of local coordinates $(t,x^1, x^2, x^3)$ in a neighbourhood of $S$ by moving a distance $t$ along the geodesics with initial condition $n$. Note that $U_1$, $S$ and $U_2$ correspond to $t<0$, $t=0$ and $t>0$ in these coordinates. Since $\frac{\partial}{\partial t}$ is the unit tangent vector to the geodesics, we have
\begin{align*}
& \frac{\partial}{\partial t} \left\langle \frac{\partial}{\partial t}, \frac{\partial}{\partial x^i} \right\rangle = \left\langle \nabla_{\frac{\partial}{\partial t}} \frac{\partial}{\partial t}, \frac{\partial}{\partial x^i} \right\rangle + \left\langle \frac{\partial}{\partial t}, \nabla_{\frac{\partial}{\partial t}} \frac{\partial}{\partial x^i} \right\rangle \\
& = \left\langle \frac{\partial}{\partial t}, \nabla_{\frac{\partial}{\partial x^i}} \frac{\partial}{\partial t} \right\rangle  = \frac{\partial}{\partial x^i} \left( \frac12 \left\langle \frac{\partial}{\partial t},  \frac{\partial}{\partial t} \right\rangle \right) = 0
\end{align*}
($i=1,2,3$), where we used
\[
\nabla_{\frac{\partial}{\partial t}} \frac{\partial}{\partial x^i} - \nabla_{\frac{\partial}{\partial x^i}} \frac{\partial}{\partial t} = \left[ \frac{\partial}{\partial t}, \frac{\partial}{\partial x^i} \right] = 0.
\]
Since for $t=0$ we have
\[
\left\langle \frac{\partial}{\partial t}, \frac{\partial}{\partial x^i} \right\rangle = \left\langle n, \frac{\partial}{\partial x^i} \right\rangle = 0,
\]
we see that  $\frac{\partial}{\partial t}$ remains orthogonal to the surfaces of constant $t$. This result will be used repeatedly.

\begin{Lemma} ({\bf Gauss Lemma I})
Let $(M,g)$ be a Riemannian or a Lorentzian manifold, and $S \subset M$ a hypersurface whose normal vector field $n$ satisfies $g(n,n) \neq 0$. The hypersurfaces $S_t$ obtained from $S$ by moving a distance $t$ along the geodesics orthogonal to $S$ remain orthogonal to the geodesics.
\end{Lemma}

The same ideas can be used to prove a closely related result.

\begin{Lemma} ({\bf Gauss Lemma II})
Let $(M,g)$ be a Riemannian or a Lorentzian manifold and $p \in M$. The hypersurfaces $S_t$ obtained from $p$ by moving a distance $t$ along the geodesics through $p$ remain orthogonal to the geodesics.
\end{Lemma}

In this coordinate system the metrics $g_1$ and $g_2$ are given on $t \leq 0$ and $t \geq 0$, respectively, by
\[
g_A = dt^2 + h^A_{ij}(t,x) dx^i dx^j
\] 
($A=1,2$). Therefore we can define a continuous metric $g$ on $M$ if
\[
h^1_{ij}(0,x) dx^i dx^j = h^2_{ij}(0,x) dx^i dx^j,
\]
that is, if
\[
{g_1}_{|_{TS}} = {g_2}_{|_{TS}}.
\]
This also guarantees continuity of all tangential derivatives of the metric, but not of the normal derivatives. In order to have a $C^1$ metric we must have
\[
\frac{\partial h^1_{ij}}{\partial t}(0,x) dx^i dx^j = \frac{\partial h^2_{ij}}{\partial t}(0,x) dx^i dx^j,
\]
that is,
\[
{\cL_n g_1}_{|_{TS}} = {\cL_n g_2}_{|_{TS}}.
\]
Note that in this case the curvature tensor (hence the energy-momentum tensor) is at most discontinuous across $S$. More importantly, as shown in Exercise~\ref{continuity}, the components $T_{tt}$ and $T_{ti}$ of the energy-momentum tensor are continuous across $S$ (that is, the flow $T_{\mu\nu} n^\mu$ of energy and momentum across $S$ is equal on both sides), implying that the energy-momentum tensor satisfies the integral version of the conservation equation $\nabla^\mu T_{\mu\nu} = 0$. Therefore we can consider $(M,g)$ a solution of the Einstein equations.

Recall that
\[
K_A = \frac12 {\cL_n g_A}_{|_{TS_A}}
\]
is known as the {\bf extrinsic curvature}, or {\bf second fundamental form}, of $S_A$. We can summarize the discussion above in the following statement. 

\begin{Prop}
Two solutions $(M_1,g_1)$ and $(M_2,g_2)$ of the Einstein field equations can be matched along diffeomorphic timelike boundaries $S_1$ and $S_2$ if and only if the induced metrics and second fundamental forms coincide:
\[
g_1 = g_2 \qquad \text{ and } \qquad K_1 = K_2.
\]
\end{Prop}

\section{Oppenheimer-Snyder collapse} \label{sec2.6}

We can use the matching technique to construct a solution of the Einstein field equations which describes a spherical cloud of dust collapsing to a black hole. This is a physically plausible model for a black hole, as opposed to the eternal black hole.

Let us take $(M_1,g_1)$ to be a flat collapsing FLRW universe:
\[
g_1 = - d\tau^2 + a^2(\tau) \left[ d \sigma^2 +  \sigma^2 \left(d\theta^2 + \sin^2\theta d\varphi^2\right) \right].
\]
We choose $S_1$ to be the hypersurface $\sigma=\sigma_0$, with normal vector
\[
n = \frac1a \frac{\partial}{\partial \sigma}.
\]
The induced metric then is
\[
{g_1}_{|_{TS_1}} = - d\tau^2 + a^2(\tau) {\sigma_0}^2 \left(d\theta^2 + \sin^2\theta d\varphi^2\right),
\]
and the second fundamental form
\[
K_1 = a(\tau) \sigma_0 \left(d\theta^2 + \sin^2\theta d\varphi^2\right).
\]
Here we used {\bf Cartan's magic formula}: if $\omega$ is a differential form and $X$ is a vector field then
\[
\cL_X \omega = X \contr d\omega + d (X \lrcorner \, \omega).
\]
Thus, for example,
\[
\cL_n d\tau = n \contr d^2 \tau + d (n \lrcorner \, d\tau) = 0.
\]
Note that the function $a(\tau)$ is constrained by the first Friedmann equation:
\begin{equation} \label{Friedmann}
\dot{a}^2 = \frac{2\alpha}{a}
\end{equation}
(we assume $\Lambda=0$).

We now take $(M_2,g_2)$ to be the Schwarzschild solution,
\[
g_2 = - V dt^2 + V^{-1} dr^2 + r^2 \left(d\theta^2 + \sin^2\theta d\varphi^2\right),
\]
where
\[
V = 1 - \frac{2M}r,
\]
and choose $S_2$ to be a spherically symmetric timelike hypersurface given by the parameterization
\[
\begin{cases}
t = t(\tau) \\
r = r(\tau)
\end{cases}.
\]
The exact form of the functions $t(\tau)$ and $r(\tau)$ will be fixed by the matching conditions; for the time being they are constrained only by the condition that $\tau$ is the proper time along $S_2$:
\[
- V \dot{t}^2 + V^{-1} \dot{r}^2 = -1.
\]
The induced metric is then
\[
{g_2}_{|_{TS_2}} = - d\tau^2 + r^2(\tau) \left(d\theta^2 + \sin^2\theta d\varphi^2\right),
\]
and so the two induced metrics coincide if and only if
\begin{equation} \label{matchingcond1}
r(\tau) = a(\tau) \sigma_0.
\end{equation}
To simplify the calculation of the second fundamental form, we note that since $S_1$ is ruled by timelike geodesics, so is $S_2$, because the induced metrics and extrinsic curvatures are the same (therefore so are the Christoffel symbols). Therefore $t(\tau)$ and $r(\tau)$ must be a solution of the radial geodesic equations, which are equivalent to
\begin{equation} \label{matchingcond2}
\begin{cases} 
V \dot{t} = E \\
- V \dot{t}^2 + V^{-1} \dot{r}^2 = -1
\end{cases}
\Leftrightarrow
\begin{cases}
V \dot{t} = E \\
\displaystyle \dot{r}^2 = E^2 - 1 + \frac{2M}{r}
\end{cases}
\end{equation}
(where $E>0$ is a constant). Equations~\eqref{Friedmann}, \eqref{matchingcond1} and \eqref{matchingcond2} are compatible if and only if
\[
\begin{cases}
E = 1 \\
M = \alpha {\sigma_0}^3
\end{cases},
\]
that is, if and only if $S_2$ represents a spherical shell dropped from infinity with zero velocity and the mass parameter of the Schwarzschild spacetime is related to the density $\rho(\tau)$ of the collapsing dust by
\[
M = \frac{4 \pi}{3} r^3(\tau) \rho(\tau).
\]
To compute the second fundamental form of $S_2$ we can then consider a family of free-falling spherical shells which includes $S_2$. If $s$ is the parameter indexing the shells (with $S_2$ corresponding to, say, $s=s_0$) then by the Gauss Lemma we can write the Schwarzschild metric in the form
\[
g_2 = - d \tau^2 + A^2(\tau,s) ds^2 + r^2(\tau,s) \left(d\theta^2 + \sin^2\theta d\varphi^2\right)
\]
(consider, for instance, the change of coordinates determined by the solution $t=t(\tau,s)$, $r=r(\tau,s)$ of the radial geodesic equations with initial conditions determined by $t(0,s)=0$, $r(0,s)=s$, $\dot{r}(0,s)=0$). The unit normal vector field $n$ to the hypersurfaces of constant $s$ is then
\[
n = \frac1A \frac{\partial}{\partial s},
\]
and so we have
\[
K_2 = \frac12 {\cL_n g_2}_{|_{TS_2}} = r (n \cdot r) \left(d\theta^2 + \sin^2\theta d\varphi^2\right).
\]
On the other hand, in Schwarzschild coordinates
\[
n = V^{-1} \dot{r} \frac{\partial}{\partial t} + V \dot{t} \frac{\partial}{\partial r},
\]
since $n$ must be unit and orthogonal to
\[
\dot{t} \frac{\partial}{\partial t} + \dot{r} \frac{\partial}{\partial r}.
\]
Therefore we have
\[
K_2 = r V \dot{t} \left(d\theta^2 + \sin^2\theta d\varphi^2\right) = r E \left(d\theta^2 + \sin^2\theta d\varphi^2\right),
\]
or, using $E=1$,
\[
K_2 = r(\tau) \left(d\theta^2 + \sin^2\theta d\varphi^2\right).
\]
In other words, $K_1 = K_2$ follows from the previous conditions, and so we indeed have a solution of the Einstein equations.

\begin{figure}[h!]
\begin{center}
\psfrag{r=0}{$r=0$}
\psfrag{Big Crunch}{Big Crunch}
\psfrag{i+}{$i^+$}
\psfrag{i0}{$i^0$}
\psfrag{i-}{$i^-$}
\psfrag{S1}{$S_1$}
\psfrag{S2}{$S_2$}
\psfrag{I+}{$\mathscr{I^+}$}
\psfrag{I-}{$\mathscr{I^-}$}
\epsfxsize=.9\textwidth
\leavevmode
\epsfbox{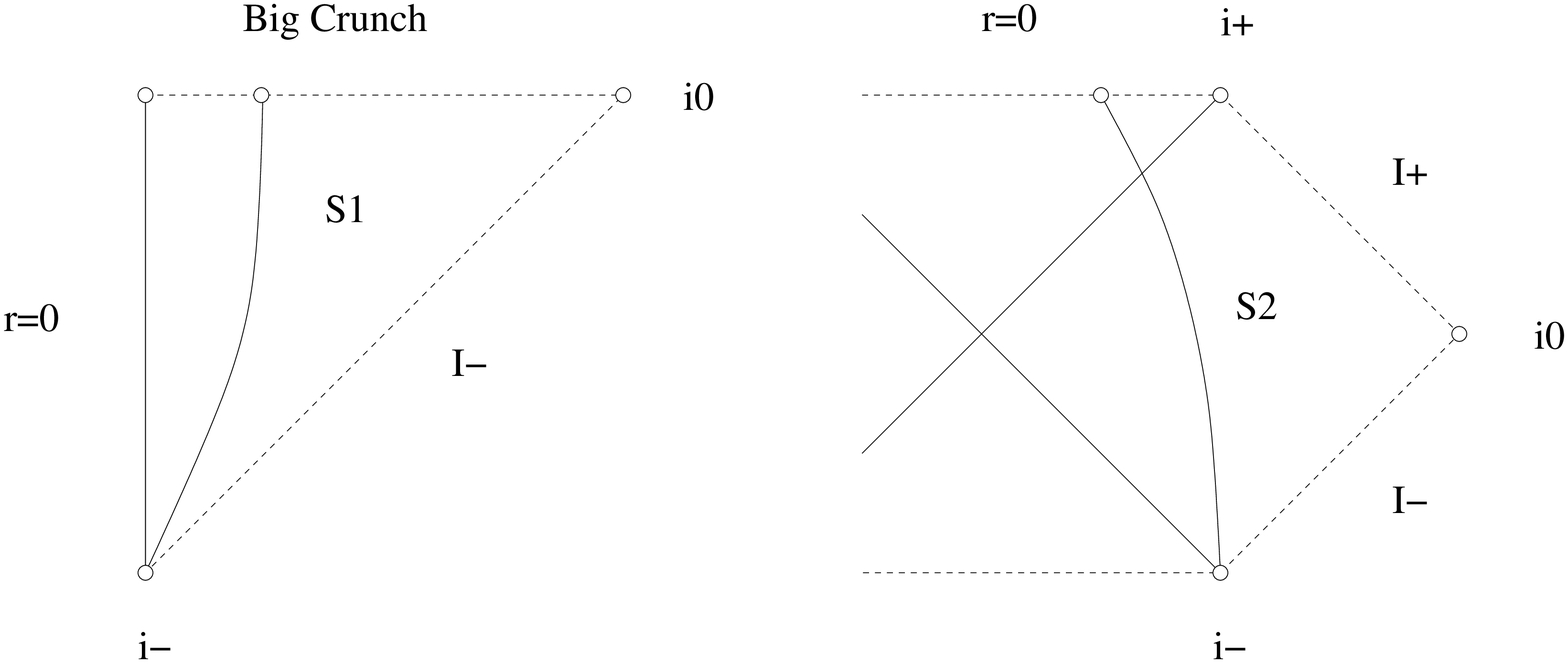}
\end{center}
\caption{Matching hypersurfaces in the collapsing flat FLRW universe and the Schwarzschild solution.} \label{Pen_OS1}
\end{figure}

To construct the Penrose diagram for this solution we represent $S_1$ and $S_2$ in the Penrose diagrams of the collapsing flat FLRW universe (obtained by reversing the time direction in the expanding flat FLRW universe) and the Schwarzschild solution (Figure~\ref{Pen_OS1}). Identifying these hypersurfaces results in the Penrose diagram depicted in Figure~\ref{Pen_OS2}.

\begin{figure}[h!]
\begin{center}
\psfrag{r=0}{$r=0$}
\psfrag{i+}{$i^+$}
\psfrag{i0}{$i^0$}
\psfrag{i-}{$i^-$}
\psfrag{S}{$S$}
\psfrag{I+}{$\mathscr{I^+}$}
\psfrag{I-}{$\mathscr{I^-}$}
\epsfxsize=.6\textwidth
\leavevmode
\epsfbox{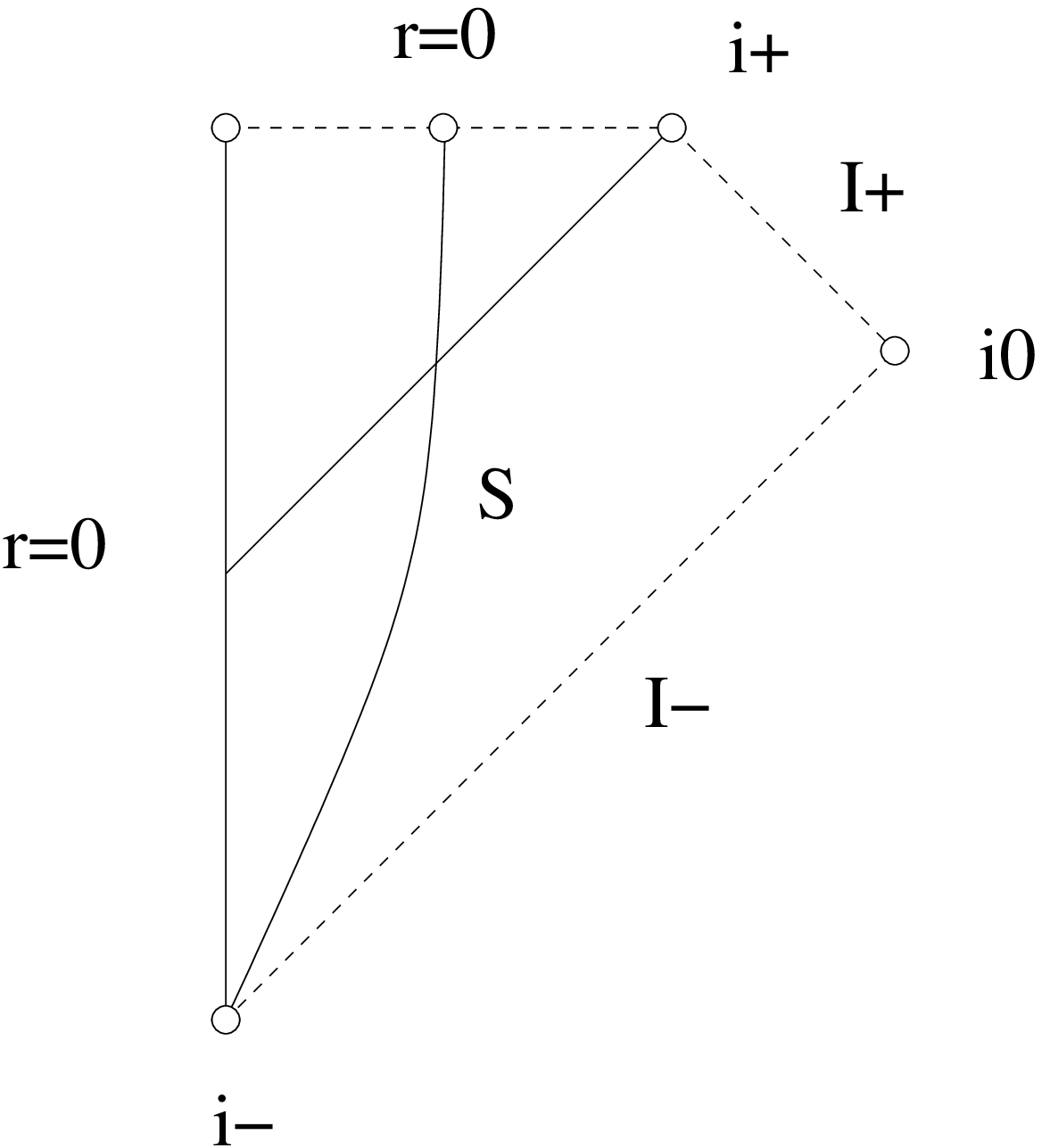}
\end{center}
\caption{Penrose diagram for the Oppenheimer-Snyder collapse.} \label{Pen_OS2}
\end{figure}

\section{Exercises} \label{sec2.7}

\begin{enumerate}

\item
In this exercise we will solve the vacuum Einstein equations (without cosmological constant) for the spherically symmetric Lorentzian metric given by
\[
\hspace{2cm} ds^2 = -(A(t,r))^2 dt^2 + (B(t,r))^2 dr^2 + r^2 \left( d\theta^2 + r^2 \sin^2 \theta d\varphi^2 \right),
\]
where $A$ and $B$ are positive smooth functions.
\begin{enumerate}
\item
Use Cartan's first structure equations,
\[
\begin{cases}
\omega_{\mu\nu}=-\omega_{\nu\mu}\\
d\omega^\mu + \omega^\mu_{\,\,\,\,\nu} \wedge \omega^\nu = 0
\end{cases},
\]
to show that the nonvanishing connection forms for the orthonormal frame dual to
\[
\hspace{2cm} \omega^0 = A dt, \quad \quad \omega^r = B dr, \quad \quad \omega^\theta = r d\theta, \quad \quad \omega^\varphi = r \sin \theta d\varphi
\]
are (using the notation $\dot{}=\frac{\partial}{\partial t}$ and ${}'=\frac{\partial}{\partial r}$)
\begin{align*}
& \omega^0_{\,\,\,\,r} = \omega^r_{\,\,\,\,0} = \frac{A'}{B} dt + \frac{\dot{B}}{A}dr \, ;\\
& \omega^\theta_{\,\,\,\,r} = - \omega^r_{\,\,\,\,\theta} = \frac1B d\theta \, ;\\
& \omega^\varphi_{\,\,\,\,r} = - \omega^r_{\,\,\,\,\varphi} = \frac{\sin\theta}{B} d\varphi \, ;\\
& \omega^\varphi_{\,\,\,\,\theta} = - \omega^\theta_{\,\,\,\,\varphi} = \cos\theta d\varphi \, .
\end{align*}
\item
Use Cartan's second structure equations,
\[
\Omega^\mu_{\,\,\,\,\nu} = d\omega^\mu_{\,\,\,\,\nu} + \omega^\mu_{\,\,\,\,\alpha} \wedge \omega^\alpha_{\,\,\,\,\nu} \, ,
\]
to show that the curvature forms on this frame are
\begin{align*}
\hspace{2cm} & \Omega^0_{\,\,\,\,r} = \Omega^r_{\,\,\,\,0} = \left(\frac{A''B-A'B'}{AB^3}+\frac{\dot{A}\dot{B}-A\ddot{B}}{A^3B} \right)\, \omega^r \wedge \omega^0 \, ; \\
& \Omega^0_{\,\,\,\,\theta} = \Omega^\theta_{\,\,\,\,0} = \frac{A'}{rAB^2} \, \omega^\theta \wedge \omega^0 + \frac{\dot{B}}{rAB^2} \, \omega^\theta \wedge \omega^r \, ; \\
& \Omega^0_{\,\,\,\,\varphi} = \Omega^\varphi_{\,\,\,\,0} = \frac{A'}{rAB^2} \, \omega^\varphi \wedge \omega^0 + \frac{\dot{B}}{rAB^2} \, \omega^\varphi \wedge \omega^r \, ; \\
& \Omega^\theta_{\,\,\,\,r} = -\Omega^r_{\,\,\,\,\theta} = \frac{B'}{rB^3} \, \omega^\theta \wedge \omega^r + \frac{\dot{B}}{rAB^2} \, \omega^\theta \wedge \omega^0 \, ; \\
& \Omega^\varphi_{\,\,\,\,r} = -\Omega^r_{\,\,\,\,\varphi} = \frac{B'}{rB^3} \, \omega^\varphi \wedge \omega^r + \frac{\dot{B}}{rAB^2} \, \omega^\varphi \wedge \omega^0 \, ; \\
& \Omega^\varphi_{\,\,\,\,\theta} = -\Omega^\theta_{\,\,\,\,\varphi} = \frac{B^2-1}{r^2B^2} \, \omega^\varphi \wedge \omega^\theta \, .
\end{align*}
\item
Using
\[
\Omega^\mu_{\,\,\,\,\nu} = \sum_{\alpha<\beta} R_{\alpha\beta\,\,\,\,\nu}^{\,\,\,\,\,\,\,\,\,\mu} \omega^\alpha \wedge \omega^\beta \, ,
\]
determine the components $R_{\alpha\beta\,\,\,\,\nu}^{\,\,\,\,\,\,\,\,\,\mu}$ of the curvature tensor in this orthonormal frame, and show that the nonvanishing components of the Ricci tensor in this frame are
 \begin{align*}
& R_{00} = \frac{A''B-A'B'}{AB^3}+\frac{\dot{A}\dot{B}-A\ddot{B}}{A^3B} + \frac{2A'}{rAB^2} \, ; \\
& R_{0r} = R_{r0}= \frac{2\dot{B}}{rAB^2} \, ; \\
& R_{rr} = \frac{A'B'-A''B}{AB^3}+\frac{A\ddot{B}-\dot{A}\dot{B}}{A^3B} + \frac{2B'}{rB^3} \, ; \\
& R_{\theta\theta} = R_{\varphi\varphi} = - \frac{A'}{rAB^2} + \frac{B'}{rB^3} + \frac{B^2-1}{r^2B^2} \, .
\end{align*}
Conclude that the nonvanishing components of the Einstein tensor in this frame are 
\begin{align*}
\hspace{2cm} & G_{00} = \frac{2B'}{rB^3} + \frac{B^2-1}{r^2B^2} \, ; \\
& G_{0r} = G_{r0}= \frac{2\dot{B}}{rAB^2} \, ; \\
& G_{rr} = \frac{2A'}{rAB^2} - \frac{B^2-1}{r^2B^2} \, ; \\
& G_{\theta\theta} = G_{\varphi\varphi} = \frac{A''B-A'B'}{AB^3}+\frac{\dot{A}\dot{B}-A\ddot{B}}{A^3B} + \frac{A'}{rAB^2} - \frac{B'}{rB^3} \, .
\end{align*}
\item
Show that if we write
\[
B(t,r) = \left(1 - \frac{2m(t,r)}r\right)^{-\frac12}
\]
for some smooth function $m$ then
\[
G_{00} = \frac{2m'}{r^2}.
\]
Conclude that the Einstein equations $G_{00}=G_{0r}=0$  are equivalent to
\[
B = \left(1 - \frac{2M}r\right)^{-\frac12},
\]
where $M \in \bbR$ is an integration constant.
\item
Show that the Einstein equation $G_{00}+G_{rr}=0$ is equivalent to $A = \frac{\alpha(t)}B$ for some positive smooth function $\alpha$. 
\item
Check that if $A$ and $B$ are as above then the remaining Einstein equations $G_{\theta\theta} = G_{\varphi\varphi} = 0$ are automatically satisfied.
\item
Argue that it is always possible to rescale the time coordinate $t$ so that the metric is written
\[
\hspace{2cm} ds^2 = -\left(1 - \frac{2M}r\right) dt^2 + \left(1 - \frac{2M}r\right)^{-1} dr^2 + r^2 \left( d\theta^2 + \sin^2 \theta d\varphi^2 \right)
\]
(the statement that any spherically symmetric solution of the vacuum Einstein equations without cosmological constant is of this form is known as {\bf Birkhoff's theorem}).
\end{enumerate}

\item
Show that the Riemannian manifold obtained by gluing the hypersurfaces $t=0$ of the two exterior regions in the maximally extended Schwarzschild solution along the horizon $r=2M$ is isometric to the {\bf Flamm paraboloid}, that is, the hypersurface in $\bbR^4$ with equation
\[
\sqrt{x^2+y^2+z^2} = 2M + \frac{w^2}{8M}
\]
(Figure~\ref{Flamm}).

\begin{figure}[h!]
\begin{center}
\epsfxsize=.5\textwidth
\leavevmode
\epsfbox{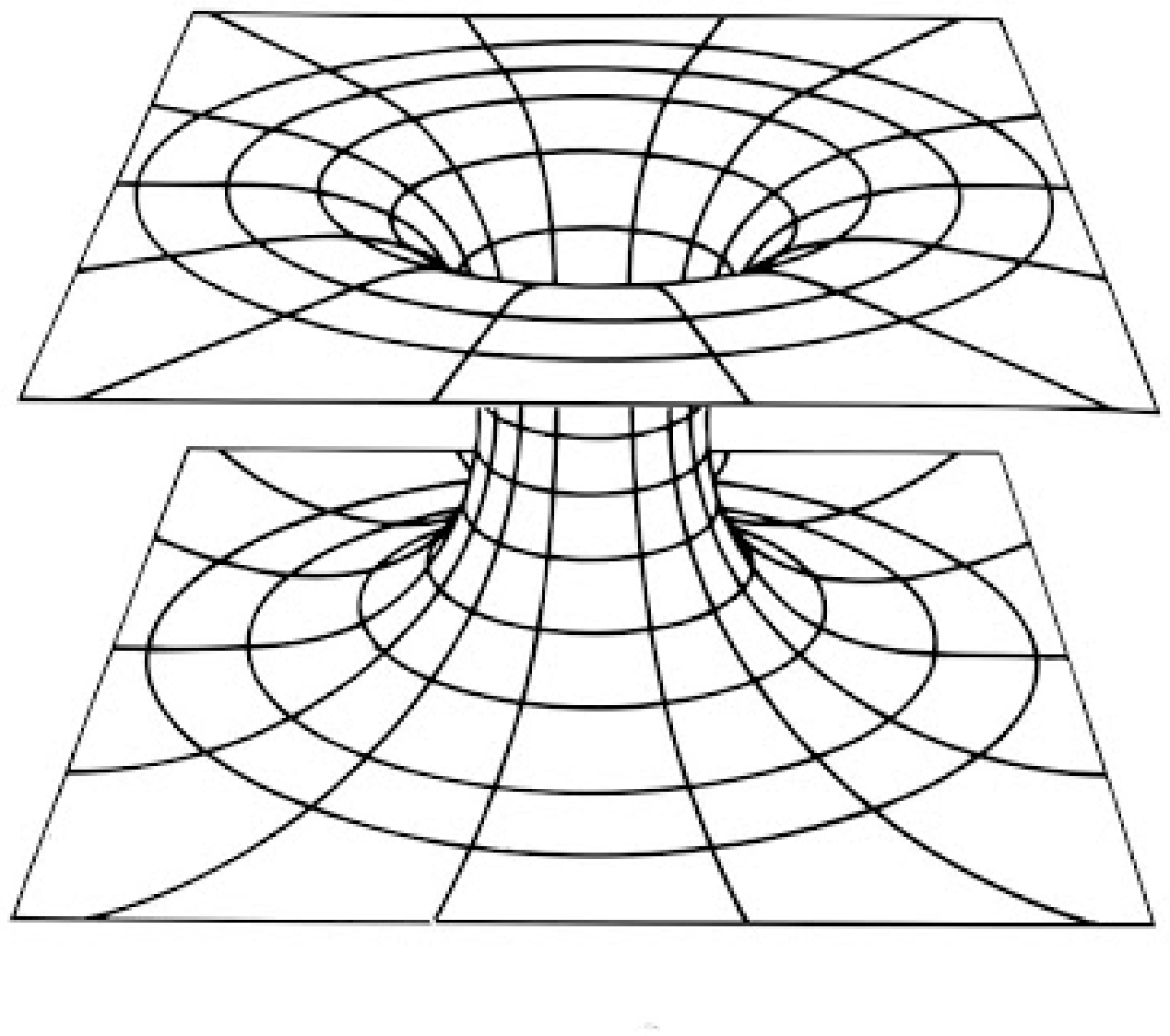}
\end{center}
\caption{Two-dimensional analogue of the Flamm paraboloid.} \label{Flamm}
\end{figure}

\item
Recall that the nonvanishing components of the Einstein tensor of the spherically symmetric Lorentzian metric
\[
\hspace{2cm} ds^2 = -(A(t,r))^2 dt^2 + (B(t,r))^2 dr^2 + r^2 \left( d\theta^2 + \sin^2 \theta d\varphi^2 \right)
\]
in the orthonormal frame dual to
\[
\hspace{2cm} \omega^0 = A dt, \quad \quad \omega^r = B dr, \quad \quad \omega^\theta = r d\theta, \quad \quad \omega^\varphi = r \sin \theta d\varphi,
\]
are given by (using the notation $\dot{}=\frac{\partial}{\partial t}$ and ${}'=\frac{\partial}{\partial r}$)
\begin{align*}
\hspace{2cm} & G_{00} =  \frac{2B'}{rB^3} + \frac{B^2-1}{r^2B^2} = \frac{2m'}{r^2}; \\
& G_{0r} = G_{r0}= \frac{2\dot{B}}{rAB^2}; \\
& G_{rr} = \frac{2A'}{rAB^2} - \frac{B^2-1}{r^2B^2}; \\
& G_{\theta\theta} = G_{\varphi\varphi} = \frac{A''B-A'B'}{AB^3}+\frac{\dot{A}\dot{B}-A\ddot{B}}{A^3B} + \frac{A'}{rAB^2} - \frac{B'}{rB^3},
\end{align*}
where
\[
B(t,r) = \left(1 - \frac{2m(t,r)}r\right)^{-\frac12}.
\]
\begin{enumerate}
\item
Assuming
\begin{itemize}
\item
$G_{0r}=0$ (so that $B$, and hence $m$, do not depend on $t$);
\item
$G_{00}+G_{rr}=0$ (so that $A = \frac{\alpha(t)}B$ for some positive smooth function $\alpha(t)$);
\item
$\alpha(t)=1$ (which can always be achieved by rescaling $t$),
\end{itemize}
show that 
\[
G_{\theta\theta} = G_{\varphi\varphi} = \frac12\left(A^2\right)''+\frac1r\left(A^2\right)'.
\]
\item
Prove that the general spherically symmetric solution of the vacuum Einstein field equations with a cosmological constant $\Lambda$ is the {\bf Kottler metric}
\begin{align*}
\hspace{2cm} g = & -\left(1 - \frac{2M}r - \frac{\Lambda}3 r^2 \right) dt^2 + \left(1 - \frac{2M}r - \frac{\Lambda}3 r^2 \right)^{-1} dr^2 \\
& + r^2 \left( d\theta^2 + \sin^2 \theta d\varphi^2 \right).
\end{align*}
\item
Obtain the Penrose diagram for the maximal extension of the Kottler solution with $\Lambda>0$ and $0<M<\frac{1}{3\sqrt{\Lambda}}$.
\item
Consider now the spherically symmetric electromagnetic field
\[
F = E(t,r) \, \omega^r \wedge \omega^0.
\]
Show that this field satisfies the vacuum Maxwell equations
\[
dF = d\star F=0
\]
(where $\star$ is the Hodge star) if and only if
\[
E(t,r) = \frac{e}{r^2}
\]
for some constant $e \in \bbR$ (the electric charge in units for which $4 \pi \varepsilon_0=1$).
\item
As we shall see in Chapter~\ref{chapter6}, this electromagnetic field corresponds to the energy-momentum tensor
\[
\hspace{2cm} T = \frac{E^2}{8\pi} \left( \omega^0 \otimes \omega^0 - \omega^r \otimes \omega^r + \omega^\theta \otimes \omega^\theta + \omega^\varphi \otimes \omega^\varphi \right).
\]
Prove that the general spherically symmetric solution of the Einstein field equations with an electromagnetic field of this kind is the {\bf Reissner-Nordstr\"om metric}
\begin{align*}
\hspace{2cm} g = & -\left(1 - \frac{2M}r + \frac{e^2}{r^2} \right) dt^2 + \left(1 - \frac{2M}r + \frac{e^2}{r^2} \right)^{-1} dr^2 \\
& + r^2 \left( d\theta^2 + \sin^2 \theta d\varphi^2 \right).
\end{align*}
\item
Obtain the Penrose diagram for the maximal extension of the Reissner-Nordstr\"om solution with $M>0$ and $0 < e^2 < M^2$.
\end{enumerate}

\item
Consider the spherically symmetric Lorentzian metric
\[
\hspace{2cm} ds^2 = - dt^2 + a^2(t) \left(\frac1{1-kr^2}dr^2 + r^2 \left( d\theta^2 + \sin^2 \theta d\varphi^2 \right) \right),
\]
where $a$ is a positive smooth function.
\begin{enumerate}
\item
Use Cartan's first structure equations,
\[
\begin{cases}
\omega_{\mu\nu}=-\omega_{\nu\mu}\\
d\omega^\mu + \omega^\mu_{\,\,\,\,\nu} \wedge \omega^\nu = 0
\end{cases},
\]
to show that the nonvanishing connection forms for the orthonormal frame dual to
\begin{align*}
& \omega^0 = dt, \quad \quad \quad \quad \, \omega^r = a(t)\left(1-kr^2\right)^{-\frac12} dr, \\
& \omega^\theta = a(t) r d\theta, \quad \quad \omega^\varphi = a(t) r \sin \theta d\varphi
\end{align*}
are
\begin{align*}
& \omega^0_{\,\,\,\,r} = \omega^r_{\,\,\,\,0} = \dot{a} \left(1-kr^2\right)^{-\frac12} dr\, ;\\
& \omega^0_{\,\,\,\,\theta} = \omega^\theta_{\,\,\,\,0} = \dot{a} r d\theta\, ;\\
& \omega^0_{\,\,\,\,\varphi} = \omega^\varphi_{\,\,\,\,0} = \dot{a} r \sin \theta d\varphi\, ;\\
& \omega^\theta_{\,\,\,\,r} = - \omega^r_{\,\,\,\,\theta} = \left(1-kr^2\right)^\frac12 d\theta\, ;\\
& \omega^\varphi_{\,\,\,\,r} = - \omega^r_{\,\,\,\,\varphi} = \left(1-kr^2\right)^\frac12 \sin\theta d\varphi\, ;\\
& \omega^\varphi_{\,\,\,\,\theta} = - \omega^\theta_{\,\,\,\,\varphi} = \cos\theta d\varphi\, .
\end{align*}
\item
Use Cartan's second structure equations,
\[
\Omega^\mu_{\,\,\,\,\nu} = d\omega^\mu_{\,\,\,\,\nu} + \omega^\mu_{\,\,\,\,\alpha} \wedge \omega^\alpha_{\,\,\,\,\nu}\, ,
\]
to show that the curvature forms on this frame are
\begin{align*}
& \Omega^0_{\,\,\,\,r} = \Omega^r_{\,\,\,\,0} = \frac{\ddot{a}}{a} \omega^0 \wedge \omega^r\, ; \\
& \Omega^0_{\,\,\,\,\theta} = \Omega^\theta_{\,\,\,\,0} = \frac{\ddot{a}}{a} \omega^0 \wedge \omega^\theta\, ; \\
& \Omega^0_{\,\,\,\,\varphi} = \Omega^\varphi_{\,\,\,\,0} = \frac{\ddot{a}}{a} \omega^0 \wedge \omega^\varphi\, ; \\
& \Omega^\theta_{\,\,\,\,r} = -\Omega^r_{\,\,\,\,\theta} = \left(\frac{k}{a^2}+\frac{\dot{a}^2}{a^2}\right) \omega^\theta \wedge \omega^r\, ; \\
& \Omega^\varphi_{\,\,\,\,r} = -\Omega^r_{\,\,\,\,\varphi} = \left(\frac{k}{a^2}+\frac{\dot{a}^2}{a^2}\right) \omega^\varphi \wedge \omega^r\, ; \\
& \Omega^\varphi_{\,\,\,\,\theta} = -\Omega^\theta_{\,\,\,\,\varphi} =  \left(\frac{k}{a^2}+\frac{\dot{a}^2}{a^2}\right) \omega^\varphi \wedge \omega^\theta\, .
\end{align*}
\item
Using
\[
\Omega^\mu_{\,\,\,\,\nu} = \sum_{\alpha<\beta} R_{\alpha\beta\,\,\,\,\nu}^{\,\,\,\,\,\,\,\,\,\mu} \omega^\alpha \wedge \omega^\beta\, ,
\]
determine the components $R_{\alpha\beta\,\,\,\,\nu}^{\,\,\,\,\,\,\,\,\,\mu}$ of the curvature tensor on this orthonormal frame, and show that the nonvanishing components of the Ricci tensor on this frame are
\begin{align*}
& R_{00} = -\frac{3\ddot{a}}{a}\, ; \\
& R_{rr} = R_{\theta\theta} = R_{\varphi\varphi} = \frac{\ddot{a}}{a} + \frac{2\dot{a}^2}{a^2} + \frac{2k}{a^2}\, .
\end{align*}
Conclude that the nonvanishing components of the Einstein tensor on this frame are
\begin{align*}
& G_{00} = \frac{3\dot{a}^2}{a^2} + \frac{3k}{a^2}\, ; \\
& G_{rr} = G_{\theta\theta} = G_{\varphi\varphi} = -\frac{2\ddot{a}}{a} - \frac{\dot{a}^2}{a^2} - \frac{k}{a^2}\, .
\end{align*}
\item
Show that the Einstein equations with a cosmological constant $\Lambda$ for a comoving pressureless perfect fluid of nonnegative density $\rho$, $G+\Lambda g=8\pi \rho \, dt^2$, are equivalent to the system
\[
\begin{cases}
\displaystyle \frac{\dot{a}^2}{a^2} + \frac{k}{a^2} = \frac{8\pi\rho}{3} + \frac{\Lambda}3 \\
\\
\displaystyle \frac{2\ddot{a}}{a} + \frac{\dot{a}^2}{a^2} + \frac{k}{a^2} = \Lambda
\end{cases}.
\]
Show that this system can be integrated to
\[
\begin{cases}
\displaystyle \frac{4\pi\rho}{3}a^3 = \alpha \\
\\
\displaystyle \frac12\dot{a}^2 - \frac{\alpha}{a} - \frac{\Lambda}6 a^2 = - \frac{k}{2}
\end{cases},
\]
where $\alpha$ is a nonnegative integration constant.
\item
Draw the Penrose diagram of the solutions with $\alpha > 0$, $\Lambda > 0$ and $k=0$ (currently believed to model our physical Universe).
\end{enumerate}

\item
Compute the metrics of the following manifolds in the local coordinates indicated, and sketch the corresponding Penrose diagrams:
\begin{enumerate}
\item
The region $T > \sqrt{X^2+Y^2+Z^2}$ of the $4$-dimensional Minkowski spacetime using the parameterization
\[
\begin{cases}
T = t \cosh \psi \\
X = t \sinh \psi \sin\theta \cos\varphi \\
Y = t \sinh \psi \sin\theta \sin\varphi \\
Z = t \sinh \psi \cos\theta
\end{cases}.
\]
\item
The hyperboloid $X^2+Y^2+Z^2+W^2=1+T^2$ in the $5$-dimensional Minkowski spacetime using the parameterization
\[
\begin{cases}
T = \sinh t \\
X = \cosh t \sin \psi \sin\theta \cos\varphi \\
Y = \cosh t \sin \psi \sin\theta \sin\varphi \\
Z = \cosh t \sin \psi \cos\theta \\
W = \cosh t \cos \psi
\end{cases}.
\]
\item
The region $W>1$ of the same hyperboloid using the parameterization
\[
\begin{cases}
T = \sinh t \cosh \psi \\
X = \sinh t \sinh \psi \sin\theta \cos\varphi \\
Y = \sinh t \sinh \psi \sin\theta \sin\varphi \\
Z = \sinh t \sinh \psi \cos\theta \\
W = \cosh t
\end{cases}.
\]
\item
The region $T>W$ of the same hyperboloid using the parameterization defined implicitly by
\[
\begin{cases}
T = W + e^t \\
X = e^t x \\
Y = e^t y \\
Z = e^t z
\end{cases}.
\]
\item
The hyperboloid $T^2+U^2=1+X^2+Y^2+Z^2$ in $\bbR^5$ with the pseudo-Riemannian metric
\[
ds^2=-dT^2-dU^2+dX^2+dY^2+dZ^2
\]
using the parameterization
\[
\begin{cases}
T = \cosh \psi \cos t \\
U = \cosh \psi \sin t \\
X = \sinh \psi \sin\theta \cos\varphi \\
Y = \sinh \psi \sin\theta \sin\varphi \\
Z = \sinh \psi \cos\theta
\end{cases}.
\]
\end{enumerate}

\item
Show that the anti-de Sitter metric
\[
ds^2 = - \cosh^2\psi dt^2 + d \psi^2 + \sinh^2 \psi  \left(d\theta^2 + \sin^2\theta d\varphi^2\right)
\]
is a solution of the vacuum Einstein field equations with cosmological constant $\Lambda=-3$.

\item
The $3$-dimensional anti-de Sitter space can be obtained by ``unwrapping'' the hyperboloid $T^2+U^2=1+X^2+Y^2$ in $\bbR^4$ with the pseudo-Riemannian metric
\[
ds^2=-dT^2-dU^2+dX^2+dY^2.
\]
In this exercise we identify $\bbR^4$ with the space of $2 \times 2$ matrices by the map
\[
(T,U,X,Y) \mapsto
\left(
\begin{matrix}
T + X & U + Y \\
-U + Y & T - X 
\end{matrix}
\right).
\]
\begin{enumerate}
\item
Show that the hyperboloid corresponds to the Lie group $SL(2,\bbR)$ of $2 \times 2$ matrices with unit determinant.
\item
Check that the squared norm of a vector $v\in\bbR^4$ in the metric above is $\langle v, v \rangle = - \det V$, where $V$ is the $2 \times 2$ matrix associated to $v$. Conclude that the metric induced on the hyperboloid is bi-invariant (that is, invariant under left and right multiplication).
\item
Use the Penrose diagram in Figure~\ref{SL2R} and the fact that the one-parameter subgroups of a Lie group with a bi-invariant metric are geodesics of that metric to conclude that the exponential map $\exp: \mathfrak{sl}(2,\bbR) \to SL(2,\bbR)$ is not surjective.
\item
Write explicitly the matrices in $SL(2,\bbR)$ which are not in the image of $\exp: \mathfrak{sl}(2,\bbR) \to SL(2,\bbR)$ by using the parameterization
\[
\begin{cases}
T = \cosh \psi \cos t \\
U = \cosh \psi \sin t \\
X = \sinh \psi \cos\varphi \\
Y = \sinh \psi \sin\varphi
\end{cases}.
\]
\end{enumerate}

\begin{figure}[h!]
\begin{center}
\psfrag{p=0}{$\psi=0$}
\psfrag{p=8}{$\psi=+\infty$}
\psfrag{t=p}{$t=\pi$}
\psfrag{t=0}{$t=0$}
\psfrag{t=-p}{$t=-\pi$}
\epsfxsize=.4\textwidth
\leavevmode
\epsfbox{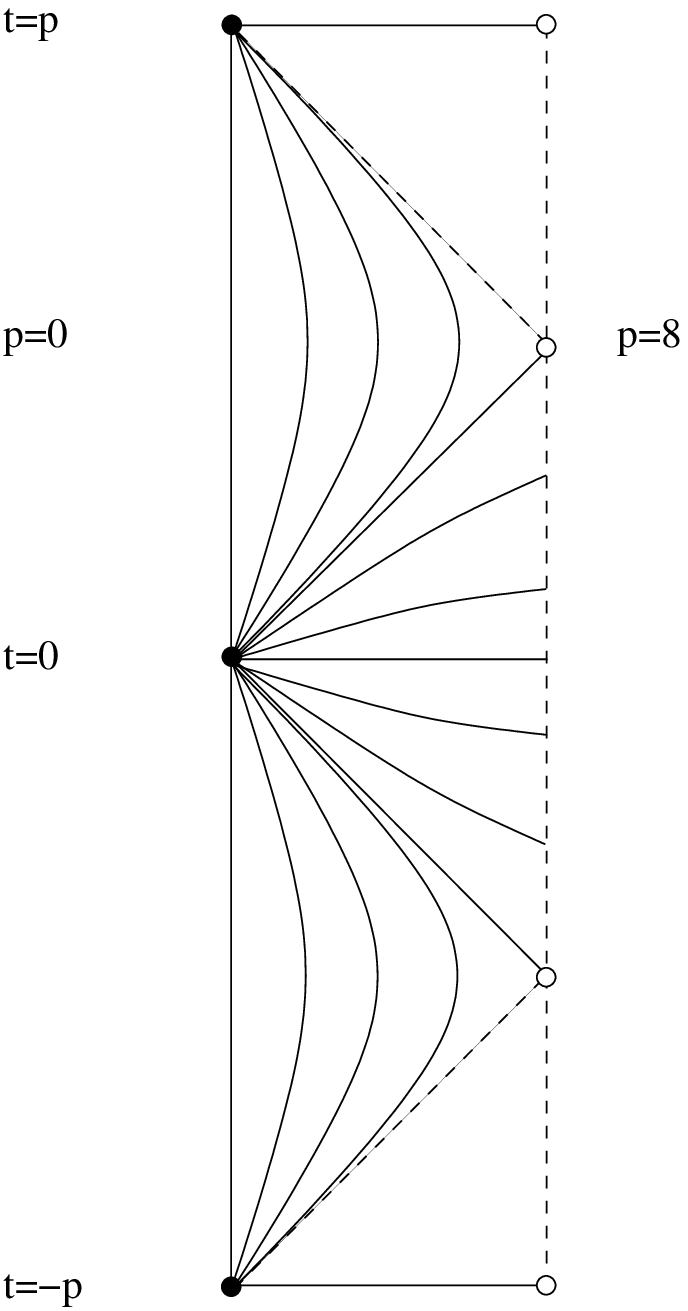}
\end{center}
\caption{Exponential map on $SL(2,\bbR)$.} \label{SL2R}
\end{figure}

\item\label{continuity}
Consider a Riemannian or Lorentzian metric given in the Gauss Lemma form
\[
g = dt^2 + h_{ij}(t,x) dx^i dx^j,
\]
so that the level sets of $t$ are Riemannian or Lorentzian manifolds with induced metric $h(t)=h_{ij}dx^i dx^j$ and second fundamental form 
\[
K(t)=\frac12\frac{\partial h_{ij}}{\partial t}dx^idx^j.
\]
Show that in these coordinates: 
\begin{enumerate}
\item
The Christoffel symbols are
\[
\Gamma^0_{ij} = - K_{ij}; \quad \Gamma^i_{jk} = \bar{\Gamma}^i_{jk}; \quad \Gamma^i_{0j} = K^i_{\,\,j},
\]
where $\bar{\Gamma}^i_{jk}$ are the Christoffel symbols of $h$.
\item
The components of the Riemann tensor are
\begin{align*}
\label{Riemann1} & R_{0i0}^{\,\,\,\,\,\,\,\, j} = - \frac{\partial}{\partial t} K^{j}_{\,\, i} - K_{il} K^{lj}; \\ 
\nonumber & R_{ij0}^{\,\,\,\,\,\,\,\, l} = - \bar{\nabla}_i K^l_{\,\, j} + \bar{\nabla}_j K^{l}_{\,\, i}; \\ 
\nonumber & R_{ijl}^{\,\,\,\,\,\,\,\, m} = \bar{R}_{ijl}^{\,\,\,\,\,\,\,\, m} - K_{il} K^{m}_{\,\,\,\, j} + K_{jl} K^{m}_{\,\,\,\, i},
\end{align*}
where $\bar{\nabla}$ is the Levi-Civita connection of $h$ and $\bar{R}_{ijl}^{\,\,\,\,\,\,\,\, m}$ are the components of the Riemann tensor of $h$.
\item
The time derivative of the inverse metric is given by the formula
\[
\frac{\partial h^{ij}}{\partial t} = -2K^{ij}.
\]
\item
The components of the Ricci tensor are
\begin{align*}
& R_{00} = - \frac{\partial}{\partial t} K^{i}_{\,\, i} - K_{ij} K^{ij}; \\
& R_{0i} = - \bar{\nabla}_i K^j_{\,\, j} + \bar{\nabla}_j K^{j}_{\,\, i}; \\
& R_{ij} = \bar{R}_{ij} - \frac{\partial}{\partial t} K_{ij} + 2 K_{il} K^{l}_{\,\, j} - K^{l}_{\,\, l} K_{ij},
\end{align*}
where $\bar{R}_{ij}$ are the components of the Ricci tensor of $h$.
\item
The scalar curvature is
\[
R = \bar{R} - 2 \frac{\partial}{\partial t} K^{i}_{\,\, i} - \left(K^{i}_{\,\, i}\right)^2 - K_{ij} K^{ij},
\]
where $\bar{R}$ is the scalar curvature of $h$.
\item
The component $G_{00}$ of the Einstein tensor is
\[
G_{00} = \frac12 \left( - \bar{R} + \left(K^{i}_{\,\, i}\right)^2 - K_{ij} K^{ij} \right).
\]
This shows that the matching conditions guarantee the continuity of $G_{00}$ and $G_{0i}=R_{0i}$.
\end{enumerate}

\item
Recall that the nonvanishing components of the Einstein tensor of the static, spherically symmetric Lorentzian metric
\[
\hspace{2cm} ds^2 = -(A(r))^2 dt^2 + (B(r))^2 dr^2 + r^2 \left( d\theta^2 + \sin^2 \theta d\varphi^2 \right)
\]
in the orthonormal frame dual to
\[
\hspace{2cm} \omega^0 = A dt, \quad \quad \omega^r = B dr, \quad \quad \omega^\theta = r d\theta, \quad \quad \omega^\varphi = r \sin \theta d\varphi,
\]
are given by
\begin{align*}
\hspace{2cm} & G_{00} =  \frac{2B'}{rB^3} + \frac{B^2-1}{r^2B^2} = \frac{2m'}{r^2}; \\
& G_{rr} = \frac{2A'}{rAB^2} - \frac{B^2-1}{r^2B^2}; \\
& G_{\theta\theta} = G_{\varphi\varphi} = \frac{A''B-A'B'}{AB^3} + \frac{A'}{rAB^2} - \frac{B'}{rB^3},
\end{align*}
where
\[
B(r) = \left(1 - \frac{2m(r)}r\right)^{-\frac12}.
\]
In this exercise we will solve the Einstein equations (without cosmological constant) for a static perfect fluid of constant rest density $\rho$ and rest pressure $p$, and match it to a Schwarzschild exterior.
\begin{enumerate}
\item
Show that
\[
B(r) = \left(1 - kr^2 \right)^{-\frac12},
\]
where $k=\frac{8 \pi}{3} \rho$. Conclude that the spatial metric is that of a sphere $S^3$ with radius $\frac1{\sqrt{k}}$.
\item
Solve the ordinary differential equation $G_{rr}=G_{\theta\theta}$ to obtain
\[
A(r) = C \left(1 - kr^2 \right)^{\frac12} + D,
\]
where $C,D \in \bbR$ are integration constants.
\item
Show that the matching conditions to a Schwarzschild exterior of mass $M > 0$ across a surface $r=R$ are
\[
\begin{cases}
\displaystyle A(R)=\left(1 - \frac{2M}R \right)^{\frac12} \\
\displaystyle A'(R)=\frac{M}{R^2}\left(1 - \frac{2M}R \right)^{-\frac12} \\
\displaystyle B(R)=\left(1 - \frac{2M}R \right)^{-\frac12}
\end{cases}.
\]
\item
Conclude that
\[
A(r) = \frac32 \left(1 - \frac{2M}R \right)^{\frac12} - \frac12 \left(1 - \frac{2Mr^2}{R^3} \right)^{\frac12}.
\]
\item
Show that
\[
p(r) = \frac{k\left(1 - kr^2 \right)^{\frac12}}{\frac32 \left(1 - kR^2 \right)^{\frac12} - \frac12 \left(1 - kr^2 \right)^{\frac12}} - k.
\]
What is the value of $p(R)$?
\item
Show that $M$ and $R$ must satisfy {\bf Buchdahl's limit}:
\[
\frac{2M}{R} < \frac89.
\]
What happens to $p(0)$ as $\frac{2M}{R} \to \frac89$?
\end{enumerate}

\end{enumerate}


\chapter{Causality} \label{chapter3}

In this chapter we briefly discuss the causality theory of a Lorentzian manifold, following \cite{GN14}. We take a minimal approach; more details can be found in \cite{ONeill83, W84, Penrose87, Naber88, HE95, Ringstrom09}.

\section{Past and future} \label{sec3.1}

A spacetime $(M,g)$ is said to be {\bf time-orientable} if there exists a timelike vector field, that is, a vector field $X$ satisfying $g(X,X)<0$. In this case, we can define a time orientation on each tangent space $T_pM$ by declaring causal vectors $v \in T_pM$ to be {\bf future-pointing} if $g(v, X_p) \leq 0$. It can be shown that any non-time-orientable spacetime admits a {\bf time-orientable double covering} (just like any non-orientable manifold admits an orientable double covering).

Assume that $(M,g)$ is {\bf time-oriented} (i.e.~time-orientable with a definite choice of time orientation). A timelike or causal curve $c:I \subset \bbR \to M$ is said to be {\bf future-directed} if $\dot{c}$ is future-pointing. The {\bf chronological future} of $p\in M$ is the set $I^+(p)$ of all points to which $p$ can be connected by a future-directed timelike curve. The {\bf causal future} of $p\in M$ is the set $J^+(p)$ of all points to which $p$ can be connected by a future-directed causal curve. Notice that $I^+(p)$ is simply the set of all events which are accessible to a particle with nonzero mass at $p$, whereas $J^+(p)$ is the set of events which can be causally influenced by $p$ (as this causal influence cannot propagate faster than the speed of light). Analogously, the {\bf chronological past} of $p\in M$ is the set $I^-(p)$ of all points which can be connected to $p$ by a future-directed timelike curve, and the {\bf causal past} of $p\in M$ is the set $J^-(p)$ of all points which can be connected to $p$ by a future-directed causal curve.

In general, the chronological and causal pasts and futures can be quite complicated sets, because of global features of the spacetime. Locally, however, causal properties are similar to those of Minkowski spacetime. More precisely, we have the following statement:

\begin{Prop} \label{local}
Let $(M,g)$ be a time-oriented spacetime. Then each point $p_0\in M$ has an open neighborhood $V \subset M$ such that the spacetime $(V,g)$ obtained by restricting $g$ to $V$ satisfies:
\begin{enumerate}
\item $V$  is {\bf geodesically convex}, that is, $V$ is a {\bf normal neighborhood} of each of its points such that given $p,q \in V$ there exists a unique geodesic (up to reparameterization) connecting $p$ to $q$;
\item $q\in I^+(p)$ if and only if there exists a future-directed timelike geodesic connecting $p$ to $q$;
\item $J^+(p) = \overline{I^+(p)}$;
\item $q \in J^+(p) \setminus I^+(p)$ if and only if there exists a future-directed null geodesic connecting $p$ to $q$.
\end{enumerate}
\end{Prop}

\begin{proof}
Recall that the {\bf exponential map} $\exp_p: U \subset T_pM \to M$ is the map given by
\[
\exp_p(v)=c_v(1),
\]
where $c_v$ is the geodesic with initial conditions $c_v(0)=p$, $\dot{c}_v(0)=v$; equivalently, $\exp_p(tv)=c_v(t)$ (since $c_{tv}(1)=c_v(t)$). Recall also that $V$ is a {\bf normal neighborhood} of $p$ if $\exp_p: U \to V$ is a diffeomorphism. The existence of geodesically convex neighborhoods is true for any affine connection and is proved for instance in \cite{KN96}.

To prove assertion (2), we start by noticing that if there exists a future-directed timelike geodesic connecting $p$ to $q$ then it is obvious that $q \in I^+(p)$. Suppose now that $q \in I^+(p)$; then there exists a future-directed timelike curve $c:[0,1] \to V$ such that $c(0)=p$ and $c(1)=q$. Choose {\bf normal coordinates} $(x^0,x^1,x^2,x^3)$, given by the parameterization
\[
\varphi(x^0,x^1,x^2,x^3)=\exp_p(x^0 E_0 + x^1 E_1 + x^2 E_2 + x^3 E_3),
\]
where $\{E_0,E_1,E_2,E_3\}$ is an orthonormal basis of $T_pM$ with $E_0$ timelike and future-pointing. These are global coordinates in $V$, since $\exp_p:U \to V$ is a diffeomorphism. Defining
\begin{align*}
W_p(q) & = - \left(x^0(q)\right)^2 + \left(x^1(q)\right)^2 + \left(x^2(q)\right)^2 + \left(x^3(q)\right)^2 \\
& = \sum_{\mu,\nu=0}^3 \eta_{\mu\nu}x^\mu(q)x^\nu(q),
\end{align*}
with $(\eta_{\mu\nu})=\diag(-1,1,1,1)$, we have to show that $W_p(q)<0$. Let $W_p(t)= W_p(c(t))$. Since $x^\mu(p)=0$, we have $W_p(0)=0$. Setting $x^\mu(t)=x^\mu(c(t))$, we obtain
\begin{align*}
& \dot{W}_p(t) = 2 \sum_{\mu,\nu=0}^3 \eta_{\mu\nu}x^\mu(t)\dot{x}^\nu(t);\\
& \ddot{W}_p(t) = 2 \sum_{\mu,\nu=0}^3 \eta_{\mu\nu}x^\mu(t)\ddot{x}^\nu(t) + 2\sum_{\mu,\nu=0}^3 \eta_{\mu\nu}\dot{x}^\mu(t)\dot{x}^\nu(t),
\end{align*}
and consequently
\begin{align*}
& \dot{W}_p(0) = 0;\\
& \ddot{W}_p(0) = 2\langle \dot{c}(0), \dot{c}(0) \rangle < 0.
\end{align*}
Therefore there exists $\varepsilon > 0$ such that $W_p(t) < 0$ for $t \in (0, \varepsilon)$.

By the Gauss Lemma, the level surfaces of $W_p$ are orthogonal to the geodesics through $p$. Therefore, if $c_v(t)=\exp_p(tv)$ is the geodesic with initial condition $v \in T_pM$, we have
\[
(\grad W_p)_{c_v(1)} = a(v) \dot{c}_v(1).
\]
Now
\begin{align*}
\left\langle (\grad W_p)_{c_v(t)}, \dot{c}_v(t) \right\rangle & = \frac{d}{dt} W_p(c_v(t)) = \frac{d}{dt} \langle tv, tv \rangle \\
& = \frac{d}{dt} \left( t^2 \langle v, v \rangle \right) = 2t \langle v, v \rangle,
\end{align*}
and hence
\[
\left\langle (\grad W_p)_{c_v(1)}, \dot{c}_v(1) \right\rangle = 2 \langle v, v \rangle.
\]
On the other hand,
\[
\left\langle (\grad W_p)_{c_v(1)}, \dot{c}_v(1) \right\rangle = \langle a(v) \dot{c}_v(1) , \dot{c}_v(1) \rangle = a(v) \langle v, v \rangle.
\]
We conclude that $a(v)=2$, and therefore
\[
(\grad W_p)_{c_v(1)} = 2 \dot{c}_v(1).
\]
Consequently, $\grad W_p$ is tangent to geodesics through $p$, being future-pointing on future-directed geodesics.

Suppose that $W_p(t) < 0$. Then $\left(\grad W_p\right)_{c(t)}$ is timelike future-pointing, and so
\[
\dot{W}(t) = \left\langle \left(\grad W_p\right)_{c(t)}, \dot{c}(t) \right\rangle < 0,
\]
as $\dot{c}(t)$ is also timelike future-pointing. We conclude that we must have $W_p(t)<0$ for all $t\in[0,1]$. In particular, $W_p(q)=W_p(1)<0$, and hence there exists a future-directed timelike geodesic connecting $p$ to $q$.

To prove assertion (3), let us see first that $\overline{I^+(p)}\subset J^+(p)$. If $q \in \overline{I^+(p)}$, then $q$ is the limit of a sequence of points $q_n\in I^+(p)$. By (2), $q_n = \exp_p(v_n)$ with $v_n \in T_pM$ timelike future-pointing. Since $\exp_p$ is a diffeomorphism, $v_n$ converges to a causal future-pointing vector $v \in T_pM$, and so $q=\exp_p(v)$ can be reached from $p$ by a future-directed causal geodesic. The converse inclusion $J^+(p) \subset \overline{I^+(p)}$ holds in general (cf.~Proposition~\ref{global}).

Finally, (4) is obvious from (3) and the fact that $\exp_p$ is a diffeomorphism onto $V$.
\end{proof}

This local behavior can be used to prove the following global result.

\begin{Prop} \label{global}
Let $(M,g)$ be a time oriented spacetime and $p \in M$. Then:
\begin{enumerate}
\item
$I^+(p)$ is open;
\item
$J^+(p) \subset \overline{I^+(p)}$;
\item
$I^+(p)=\inte J^+(p)$
\item
if $r \in J^+(p)$ and $q \in I^+(r)$ then $q \in I^+(p)$;
\item
if $r \in I^+(p)$ and $q \in J^+(r)$ then $q \in I^+(p)$.
\end{enumerate}
\end{Prop}

\begin{proof}
Exercise.
\end{proof}

The twin paradox also holds locally for general spacetimes. More precisely, we have the following statement:

\begin{Prop}
Let $(M,g)$ be a time-oriented spacetime, $p_0 \in M$ and $V \subset M$ a geodesically convex open neighborhood of $p_0$. The spacetime $(V,g)$ obtained by restricting $g$ to $V$ satisfies the following property: if $p,q \in V$ with $q \in I^+(p)$, $c$ is the timelike geodesic connecting $p$ to $q$ and $\gamma$ is any timelike curve connecting $p$ to $q$, then $\tau(\gamma) \leq \tau(c)$, with equality if and only if $\gamma$ is a reparameterization of $c$.
\end{Prop}

\begin{proof}
Any timelike curve $\gamma:[0,1] \to V$ satisfying $\gamma(0)=p$, $\gamma(1)=q$ can be written as
\[
\gamma(t)=\exp_p(r(t)n(t)),
\]
for $t \in [0,1]$, where $r(t) \geq 0$ and $\langle n(t), n(t) \rangle = -1$. We have
\[
\dot{\gamma}(t)=(\exp_p)_*\left(\dot{r}(t)n(t)+r(t)\dot{n}(t)\right).
\]
Since $\langle n(t), n(t) \rangle = -1$, we have $\langle \dot{n}(t), n(t) \rangle = 0$, and consequently $\dot{n}(t)$ is tangent to the level surfaces of the function $v \mapsto \langle v, v \rangle$. We conclude that
\[
\dot{\gamma}(t) = \dot{r}(t) X_{\gamma(t)} + Y(t),
\]
where $X$ is the unit tangent vector field to timelike geodesics through $p$ and $Y(t)=r(t)(\exp_p)_*\dot{n}(t)$ is tangent to the level surfaces of $W_p$ (hence orthogonal to $X_{\gamma(t)}$). Consequently,
\begin{align*}
\tau(\gamma) & = \int_0^1 \left|\left\langle \dot{r}(t) X_{\gamma(t)} + Y(t),\dot{r}(t) X_{\gamma(t)} + Y(t) \right\rangle\right|^\frac12 dt \\
& = \int_0^1 \left( \dot{r}(t)^2 - |Y(t)|^2 \right)^\frac12 dt \\
& \leq \int_0^1 \dot{r}(t) dt = r(1) = \tau(c),
\end{align*}
where we have used the facts that $\gamma$ is timelike, $\dot{r}(t)> 0$ for all $t \in [0,1]$ (as $\dot{\gamma}$ is future-pointing) and $\tau(c)=r(1)$ (as $q=\exp_p(r(1)n(1))$). It should be clear that $\tau(\gamma)=\tau(c)$ if and only if $|Y(t)|\equiv 0 \Leftrightarrow Y(t)\equiv 0$ ($Y(t)$ is spacelike or zero) for all $t \in [0,1]$, implying that $n$ is constant. In this case, $\gamma(t)=\exp_p({r(t)n})$ is, up to reparameterization, the geodesic through $p$ with initial condition $n \in T_p M$.
\end{proof}

There is also a local property characterizing null geodesics.

\begin{Prop}
Let $(M,g)$ be a time-oriented spacetime, $p_0 \in M$ and $V \subset M$ a geodesically convex open neighborhood of $p_0$. The spacetime $(V,g)$ obtained by restricting $g$ to $V$ satisfies the following property: if for $p,q \in V$ there exists a future-directed null geodesic $c$ connecting $p$ to $q$ and $\gamma$ is a causal curve connecting $p$ to $q$ then $\gamma$ is a reparameterization of $c$.
\end{Prop}

\begin{proof}
Since $p$ and $q$ are connected by a null geodesic, we conclude from Proposition~\ref{local} that $q \in J^+(p) \setminus I^+(p)$. Let $\gamma:[0,1] \to V$ be a causal curve connecting $p$ to $q$. Then we must have $\gamma(t) \in J^+(p) \setminus I^+(p)$ for all $t \in [0,1]$, since $\gamma(t_0) \in I^+(p)$ implies $\gamma(t) \in I^+(p)$ for all $t>t_0$ (see Proposition~\ref{global}). Consequently, we have
\[
W_p(\gamma(t))=0 \Rightarrow \left\langle \left(\grad W_p\right)_{\gamma(t)}, \dot{\gamma}(t) \right\rangle = 0,
\]
where $W_p$ was defined in the proof of Proposition~\ref{local}. The formula 
\[
(\grad W_p)_{c_v(1)} = 2 \dot{c}_v(1), 
\]
which was proved for timelike geodesics $c_v$ with initial condition $v \in T_pM$, must also hold for null geodesics (by continuity). Hence $\grad W_p$ is tangent to the null geodesics ruling $J^+(p) \setminus I^+(p)$ and future-pointing. Since  $\dot{\gamma}(t)$ is also future-pointing, we conclude that $\dot{\gamma}$ is proportional to $\grad W_p$, and therefore $\gamma$ is a reparameterization of a null geodesic, which must be $c$.
\end{proof}

\begin{Cor} \label{null_geodesic}
Let $(M,g)$ be a time-oriented spacetime and $p \in M$. If $q \in J^+(p)\setminus I^+(p)$ then any future-directed causal curve connecting $p$ to $q$ must be a reparameterized null geodesic.
\end{Cor}

\section{Causality conditions} \label{sec3.2}

For physical applications, it is important to require that the spacetime satisfies reasonable causality conditions. The simplest of these conditions excludes time travel, i.e.~the possibility of a particle returning to an event in its past history.

\begin{Def}
A spacetime $(M,g)$ is said to satisfy the {\bf chronology condition} if it does not contain closed timelike curves.
\end{Def}

This condition is violated by compact spacetimes:

\begin{Prop}
Any compact spacetime $(M,g)$ contains closed timelike curves.
\end{Prop}

\begin{proof}
Taking if necessary the time-orientable double covering, we can assume that $(M,g)$ is time-oriented. Since $I^+(p)$ is an open set for any $p \in M$, it is clear that $\{ I^+(p) \}_{p \in M}$ is an open cover of $M$. If $M$ is compact, we can obtain a finite subcover $\{ I^+(p_1), \ldots, I^+(p_N) \}$. Now if $p_1 \in I^+(p_i)$ for $i \neq 1$ then $I^+(p_1) \subset I^+(p_i)$, and we can exclude $I^+(p_1)$ from the subcover. Therefore, we can assume without loss of generality that $p_1 \in I^+(p_1)$, and hence there exists a closed timelike curve starting and ending at $p_1$.
\end{proof}

A stronger restriction on the causal behavior of the spacetime is the following:

\begin{Def}
A spacetime $(M,g)$ is said to be {\bf stably causal} if there exists a {\bf global time function}, i.e.~a smooth function $t:M \to \bbR$ such that $\grad(t)$ is timelike.
\end{Def}

In particular, a stably causal spacetime is time-orientable. We choose the time orientation defined by $-\grad(t)$, so that $t$ increases along future-directed timelike curves. Notice that this implies that no closed timelike curves can exist, i.e.~any stably causal spacetime satisfies the chronology condition. In fact, any small perturbation of a stably causal spacetime still satisfies the chronology condition (Exercise~\ref{stable}).

Let $(M,g)$ be a time-oriented spacetime. A smooth future-directed causal curve $c:(a,b) \to M$ (with possibly $a=-\infty$ or $b=+\infty$) is said to be {\bf future-inextendible} if $\lim_{t \to b} c(t)$ does not exist. The definition of a {\bf past-inextendible} causal curve is analogous. The {\bf future domain of dependence} of $S\subset M$ is the set $D^+(S)$ of all events $p \in M$ such that any past-inextendible causal curve starting at $p$ intersects $S$. Therefore any causal influence on an event $p \in D^+(S)$ had to register somewhere in $S$, and one can expect that what happens at $p$ can be predicted from data on $S$. Similarly, the {\bf past domain of dependence} of $S$ is the set $D^-(S)$ of all events $p \in M$ such that any future-inextendible causal curve starting at $p$ intersects $S$. Therefore any causal influence of an event $p \in D^-(S)$ will register somewhere in $S$, and one can expect that what happened at $p$ can be retrodicted from data on $S$. The {\bf domain of dependence} of $S$ is simply the set $D(S)=D^+(S)\cup D^-(S)$.

Let $(M,g)$ be a stably causal spacetime with time function $t:M \to \bbR$. The level sets $S_a = t^{-1}(a)$ are said to be {\bf Cauchy hypersurfaces} if $D(S_a)=M$. Spacetimes for which this happens have particularly good causal properties.

\begin{Def}
A stably causal spacetime possessing a time function whose level sets are Cauchy hypersurfaces is said to be {\bf globally hyperbolic}.
\end{Def}

Notice that the future and past domains of dependence of the Cauchy hypersurfaces $S_a$ are $D^+(S_a) = t^{-1}([a, +\infty))$ and $D^-(S_a) = t^{-1}((-\infty,a])$.

\section{Exercises} \label{sec3.?}

\begin{enumerate}
\item
Let $(M,g)$ be the quotient of the $2$-dimensional Minkowski spacetime by the discrete group of isometries generated by the map $f(t,x)=(-t,x+1)$. Show that $(M,g)$ is not time orientable.  
\item 
Let $(M,g)$ be a time oriented spacetime and $p \in M$. Show that:
\begin{enumerate}
\item
$I^+(p)$ is open;
\item
$J^+(p)$ is not necessarily closed;
\item
$J^+(p) \subset \overline{I^+(p)}$;
\item
$I^+(p)=\inte J^+(p)$
\item
if $r \in J^+(p)$ and $q \in I^+(r)$ then $q \in I^+(p)$;
\item
if $r \in I^+(p)$ and $q \in J^+(r)$ then $q \in I^+(p)$.
\end{enumerate}
\item 
Consider the $3$-dimensional Minkowski spacetime $(\bbR^3, g)$, where
\[
g = - dt^2 + dx^2 + dy^2.
\]
Let $c:\bbR \to \bbR^3$ be the curve $c(t)=(t,\cos t, \sin t)$. Show that although $\dot{c}(t)$ is null for all $t \in \bbR$ we have $c(t)\in I^+(c(0))$ for all $t>0$. What kind of motion does this curve represent?
\item \label{stable}
Let $(M,g)$ be a stably causal spacetime and $h$ an arbitrary symmetric $(2,0)$-tensor field with compact support. Show that for sufficiently small $|\varepsilon|$ the tensor field $g_\varepsilon = g + \varepsilon h$ is still a Lorentzian metric on $M$, and $(M,g_\varepsilon)$ satisfies the chronology condition.
\item
Let $(M,g)$ be the quotient of the $2$-dimensional Minkowski spacetime by the discrete group of isometries generated by the map $f(t,x)=(t+1,x+1)$. Show that $(M,g)$ satisfies the chronology condition, but there exist arbitrarily small perturbations of $(M,g)$ (in the sense of Exercise~\ref{stable}) which do not.
\item 
Let $(M,g)$ be a time oriented spacetime and $S \subset M$. Show that:
\begin{enumerate}
\item
$S \subset D^+(S)$;
\item
$D^+(S)$ is not necessarily open;
\item
$D^+(S)$ is not necessarily closed.
\end{enumerate}
\item
Let $(M,g)$ be the $2$-dimensional spacetime obtained by removing the positive $x$-semi-axis of Minkowski $2$-dimensional spacetime (cf.~Figure~\ref{stably_causal}). Show that:
\begin{enumerate}
\item
$(M,g)$ is stably causal but not globally hyperbolic;
\item
there exist points $p,q \in M$ such that $J^+(p) \cap J^-(q)$ is not compact;
\item
there exist points $p,q \in M$ with $q \in I^+(p)$ such that the supremum of the lengths of timelike curves connecting $p$ to $q$ is not attained by any timelike curve.
\end{enumerate}

\begin{figure}[h!]
\begin{center}
\psfrag{t}{$t$}
\psfrag{x}{$x$}
\psfrag{S}{$S$}
\psfrag{D}{$D(S)$}
\psfrag{p}{$p$}
\psfrag{J+}{$J^+(p)$}
\epsfxsize=1.0\textwidth
\leavevmode
\epsfbox{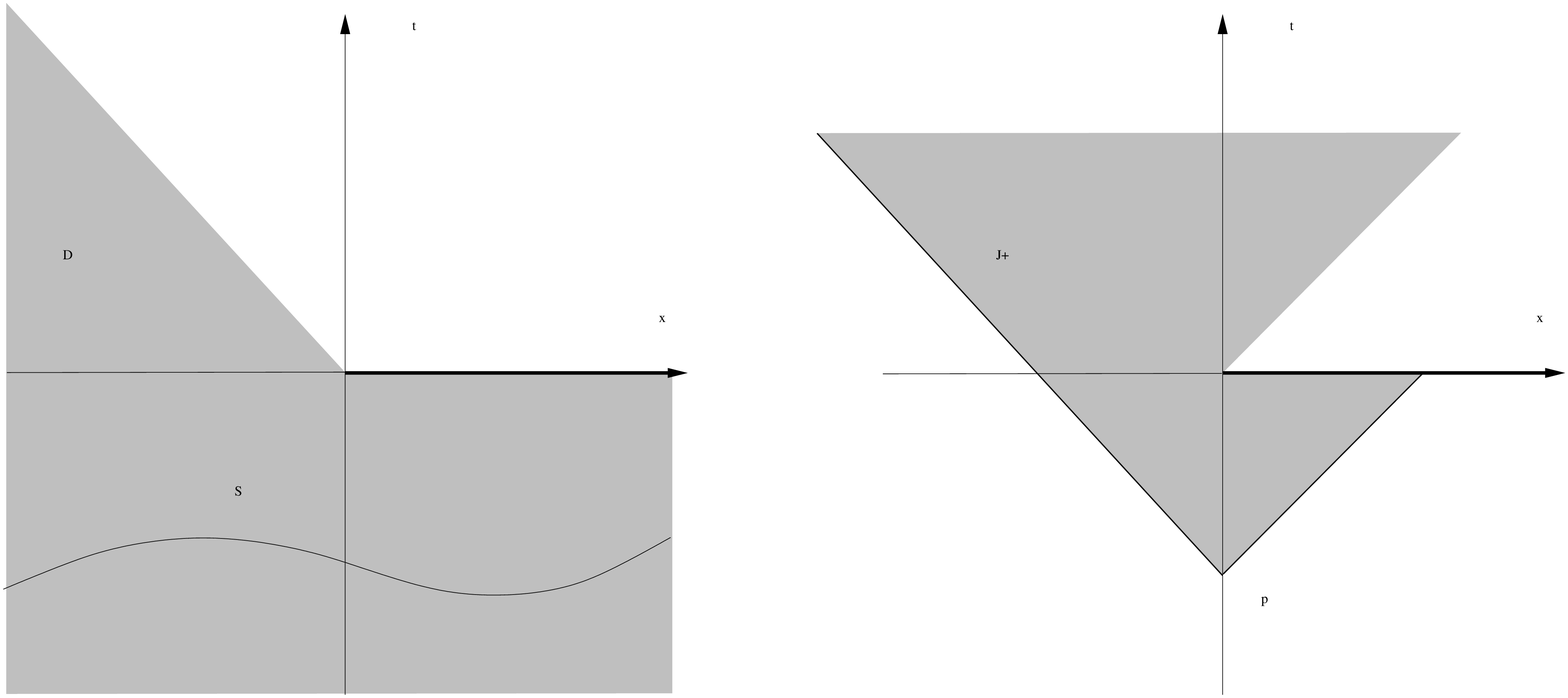}
\end{center}
\caption{Stably causal but not globally hyperbolic spacetime.} \label{stably_causal}
\end{figure}

\item
Let $(\Sigma, h)$ be a $3$-dimensional Riemannian manifold. Show that the spacetime $(M,g)=(\bbR \times \Sigma, -dt \otimes dt + h)$ is globally hyperbolic if and only if $(\Sigma,h)$ is complete.

\item
Show that the following spacetimes are globally hyperbolic:
\begin{enumerate}
\item
the Minkowski spacetime;
\item
the FLRW spacetimes;
\item
the region $\{r>2m\}$ of the Schwarzschild spacetime;
\item
the region $\{r<2m\}$ of the Schwarzschild spacetime;
\item
the maximal analytical extension of the Schwarzschild spacetime.
\end{enumerate}

\item
Let $(M,g)$ be a global hyperbolic spacetime with Cauchy hypersurface $S$. Show that $M$ is diffeomorphic to $\bbR \times S$.
\end{enumerate}


\chapter{Singularity theorems} \label{chapter4}

As we have seen in Chapter~\ref{chapter2}, both the Schwarzschild solution and the FLRW cosmological models display singularities, beyond which timelike and null geodesics cannot be continued. It was once thought that these solutions were singular due to their high degree of symmetry, and that more realistic spacetimes would be non-singular. In this chapter we show that this is not the case: any sufficiently small perturbation of these solutions will still be singular. We follow \cite{W84} when discussing conjugate points and \cite{GN14} for the details of the proofs. See also \cite{ONeill83, Penrose87, Naber88, HE95}.

\section{Geodesic congruences} \label{sec4.1}

Let $(M,g)$ be a Lorentzian manifold. A {\bf congruence} of curves in an open set $U \subset M$ is the family of integral curves of a nonvanishing vector field $X$ in $U$. We will assume that $X$ is unit timelike and geodesic, that is,
\[
\left\langle X, X \right\rangle = -1 \qquad \text{ and } \qquad \nabla_XX = 0.
\]
The properties of the congruence determined by $X$ are best analyzed by considering its {\bf second fundamental form}
\[
B_{\mu\nu} = \nabla_\nu X_\mu. 
\]
This tensor is {\bf purely spatial}, that is,
\[
B_{\mu\nu} X^\mu = B_{\mu\nu} X^\nu = 0.
\]
Indeed, since $X$ is unit,
\[
B_{\mu\nu} X^\mu = X^\mu \nabla_\nu X_\mu = \frac12 \nabla_\nu (X_\mu X^\mu) = 0.
\]
On the other hand, because $X$ is geodesic,
\[
B_{\mu\nu} X^\nu = X^\nu \nabla_\nu X_\mu = \nabla_X X_\mu = 0.
\]

\begin{Prop}
The second fundamental form $B$ satisfies
\[
\nabla_X B_{\mu\nu} = - B_{\mu\alpha} B^\alpha_{\,\,\,\,\nu} + R_{\alpha \nu \mu \beta} X^\alpha X^\beta.
\]
\end{Prop}

\begin{proof}
We have
\begin{align*}
X^\alpha \nabla_\alpha B_{\mu\nu} & = X^\alpha \nabla_\alpha \nabla_\nu X_{\mu} = X^\alpha \nabla_\nu \nabla_\alpha X_{\mu} + X^\alpha R_{\alpha\nu\mu\beta} X^\beta \\
& = \nabla_\nu (X^\alpha \nabla_\alpha X_{\mu}) - (\nabla_\nu X^\alpha) (\nabla_\alpha X_{\mu}) + R_{\alpha\nu\mu\beta}  X^\alpha X^\beta \\
& = - B_{\mu\alpha} B^\alpha_{\,\,\,\,\nu} + R_{\alpha\nu\mu\beta}  X^\alpha X^\beta.
\end{align*}
\end{proof}

Let $c(t,s)$ be a one-parameter family of geodesics of the congruence, parameterized such that
\[
\frac{\partial c}{\partial t} = X
\]
(Figure~\ref{deviation}). The {\bf geodesic deviation vector} associated to $c$ is
\[
Y = \frac{\partial c}{\partial s}.
\]

\begin{figure}[h!]
\begin{center}
\psfrag{X}{$X$}
\psfrag{Y}{$Y$}
\epsfxsize=.3\textwidth
\leavevmode
\epsfbox{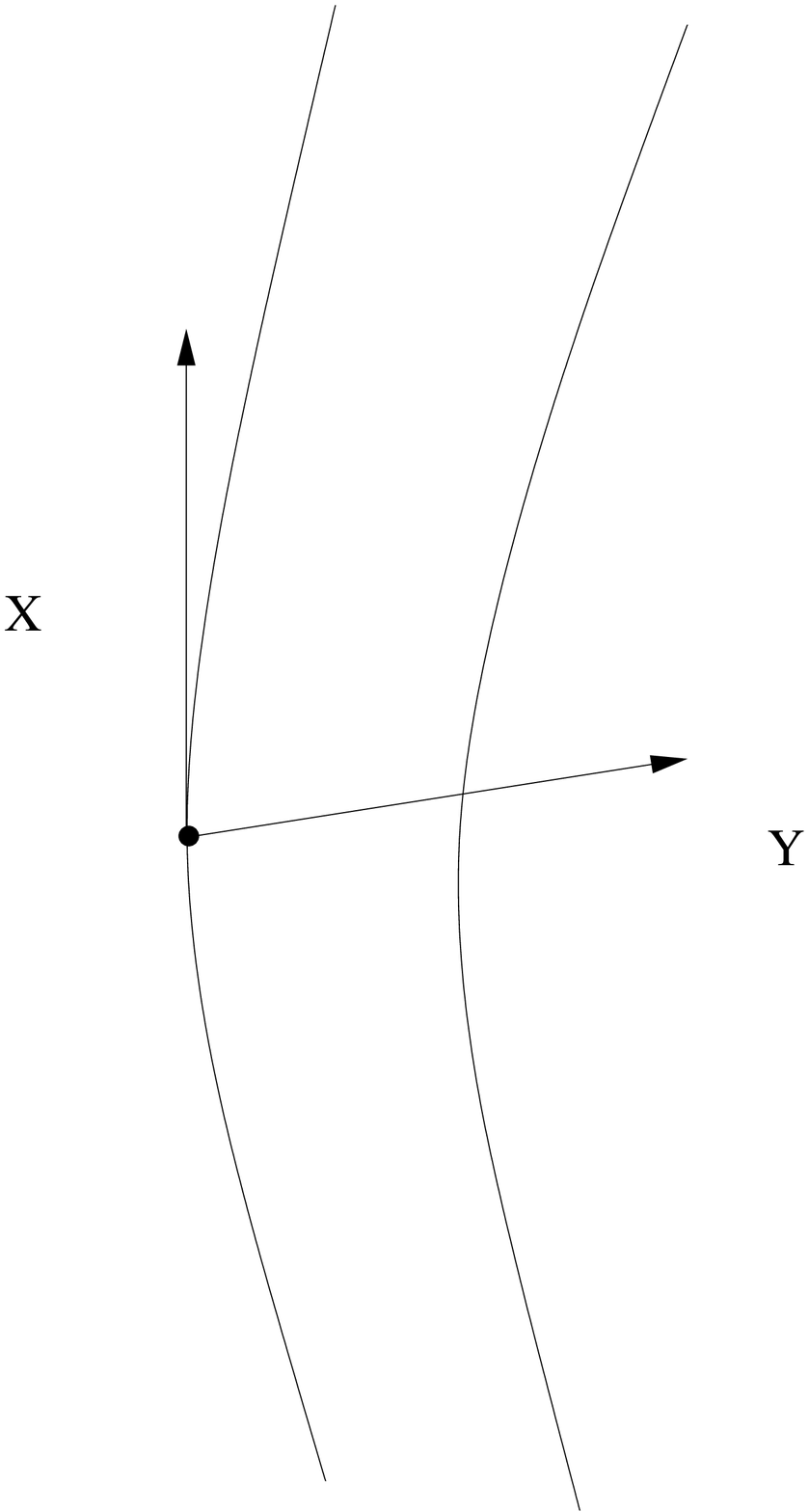}
\end{center}
\caption{Geodesic deviation.} \label{deviation}
\end{figure}

\begin{Prop}
The geodesic deviation vector satisfies
\[
\nabla_X Y^\mu = B^\mu_{\,\,\,\,\nu} Y^\nu.
\]
\end{Prop}

\begin{proof}
The definition of $Y$ implies that
\[
[X,Y]=0 \Leftrightarrow \nabla_XY - \nabla_Y X = 0.
\]
Consequently, we have
\[
\nabla_X Y^\mu = \nabla_Y X^\mu = Y^\nu \nabla_\nu X^\mu = B^\mu_{\,\,\,\,\nu} Y^\nu.
\]
\end{proof}

The equation for $\nabla_X B_{\mu\nu}$ then yields the following famous result.

\begin{Prop}
The geodesic deviation vector satisfies the {\bf Jacobi equation}
\[
\nabla_X \nabla_X Y = R(X,Y)X.
\]
\end{Prop}

\begin{proof}
We have
\begin{align*}
\nabla_X \nabla_X Y^\alpha & = \nabla_X (B^\alpha_{\,\,\,\,\beta} Y^\beta) = (\nabla_X B^\alpha_{\,\,\,\,\beta}) Y^\beta + B^\alpha_{\,\,\,\,\beta} \nabla_X Y^\beta \\
& = - B^{\alpha\mu} B_{\mu\beta} Y^\beta + R_{\mu\beta\,\,\,\,\nu}^{\,\,\,\,\,\,\,\,\alpha}  X^\mu X^\nu Y^\beta + B^\alpha_{\,\,\,\,\beta} B^\beta_{\,\,\,\,\mu} Y^\mu \\
& = R^\alpha_{\,\,\,\,\beta\mu\nu} X^\beta X^\mu Y^\nu.
\end{align*}
Alternatively, we can simply notice that
\[
\nabla_X \nabla_X Y = \nabla_X \nabla_Y X = \nabla_Y \nabla_X X + R(X,Y) X = R(X,Y) X.
\]
\end{proof}

We now define the kinematic quantities associated to the congruence.

\begin{Def}
The {\bf spatial metric} associated to the congruence is
\[
h_{\mu\nu} = g_{\mu\nu} + X_\mu X_\nu.
\]
The {\bf expansion}, {\bf shear} and {\bf vorticity} are defined as\footnote{Curved brackets indicate symmetrization: $B_{(\mu\nu)}=\frac12\left(B_{\mu\nu}+B_{\nu\mu}\right)$.}
\begin{align*}
& \theta = h^{\mu\nu} B_{\mu\nu} = g^{\mu\nu} B_{\mu\nu}, \\
& \sigma_{\mu\nu} = B_{(\mu\nu)} - \frac13 \theta h_{\mu\nu}, \\
& \omega_{\mu\nu} = B_{[\mu\nu]},
\end{align*}
so that we have the decomposition
\[
B_{\mu\nu} = \frac13 \theta h_{\mu\nu} + \sigma_{\mu\nu} + \omega_{\mu\nu}.
\]
\end{Def}

Note that all the tensors above are purely spatial:
\[
h_{\mu\nu} X^\nu = \sigma_{\mu\nu} X^\nu = \omega_{\mu\nu} X^\nu = 0,
\]
Moreover, the trace of $h$ is
\[
h^{\mu\nu} h_{\mu\nu} = g^{\mu\nu} h_{\mu\nu} = g^{\mu\nu} (g_{\mu\nu} + X_\mu X_\nu) = 4 - 1 = 3,
\]
and so the shear is traceless:
\[
h^{\mu\nu} \sigma_{\mu\nu} = g^{\mu\nu} \sigma_{\mu\nu} = 0.
\]

Fix a geodesic $c$, and let $Y$ be a geodesic deviation vector along $c$. If $Y$ is initially orthogonal to $c$ then it will remain orthogonal:
\[
X \cdot (X_\mu Y^\mu) = (\nabla_X X_\mu) Y^\mu + X_\mu \nabla_X Y^\mu = X_\mu B^{\mu\nu} Y_\nu = 0.
\]
In an orthonormal frame $\{X,E_1,E_2,E_3\}$ parallel along $c$ we then have
\[
\dot{Y}^i = B_{ij} Y^j = \left(\frac13 \theta \delta_{ij} + \sigma_{ij} + \omega_{ij}\right) Y^j =  \frac13 \theta Y^i + \sigma_{ij} Y^j + \omega_{ij}Y^j
\]
($i=1,2,3$). If we consider a small spacelike sphere in the hypersurface orthogonal to $c$ and let it be carried by the geodesics of the congruence, we see that $\theta$ measures the rate at which the sphere's volume grows, $\sigma$ describes the sphere's volume-preserving shape deformations, and $\omega$ gives the sphere's angular velocity.

\begin{Prop}
The expansion of the congruence satisfies the {\bf Raychaudhuri equation}
\[
X \cdot \theta = - \frac13 \theta^2 - \sigma_{\mu\nu} \sigma^{\mu\nu} + \omega_{\mu\nu} \omega^{\mu\nu} - R_{\mu\nu} X^\mu X^\nu.
\]
\end{Prop}

\begin{proof}
Taking the trace of the equation for $\nabla_X B_{\mu\nu}$ (that is, contracting with $g^{\mu\nu}$) yields
\begin{align*}
X \cdot \theta & = - B_{\mu\nu} B^{\nu\mu} - R_{\alpha \beta} X^\alpha X^\beta \\
& = - \left(\frac13 \theta h_{\mu\nu} + \sigma_{\mu\nu} + \omega_{\mu\nu} \right) \left(\frac13 \theta h^{\mu\nu} + \sigma^{\mu\nu} - \omega^{\mu\nu} \right) - R_{\mu \nu} X^\mu X^\nu \\
& = - \frac13 \theta^2 - \sigma_{\mu\nu} \sigma^{\mu\nu} + \omega_{\mu\nu} \omega^{\mu\nu} - R_{\mu\nu} X^\mu X^\nu.
\end{align*}
\end{proof}

\section{Energy conditions} \label{sec4.2}

\begin{Def}
A given energy-momentum tensor $T_{\mu\nu}$, with trace $T = g^{\mu\nu}T_{\mu\nu}$, is said to satisfy:
\begin{enumerate}
\item
the {\bf strong energy condition (SEC)} if $T_{\mu\nu}X^\mu X^\nu + \frac12 T \geq 0$ for all unit timelike vectors $X$;
\item
the {\bf weak energy condition (WEC)} if $T_{\mu\nu}X^\mu X^\nu \geq 0$ for all timelike vectors $X$;
\item
the {\bf null energy condition (NEC)} if $T_{\mu\nu}X^\mu X^\nu \geq 0$ for all null vectors $X$;
\item
the {\bf dominant energy condition (DEC)} if $-T^{\mu\nu} X_\nu$ is causal and future-pointing for all causal future-pointing vectors $X$.
\end{enumerate}
\end{Def}

The weak energy condition is the reasonable requirement that any observer should measure a non-negative energy density, and the null energy condition can be thought of as the same requirement for observers moving at the speed of light. The dominant energy condition, on the other hand, demands that any observer should measure the flow of energy and momentum to be causal. To understand the strong energy condition, we write the Einstein equations as
\[
R_{\mu\nu} - \frac12 R g_{\mu\nu} = 8 \pi T_{\mu\nu}
\] 
(possibly including the cosmological constant in the energy-momentum tensor). Note that the trace of this equation yields
\[
- R = 8 \pi T,
\]
and so the Einstein equations can also be written as
\[
R_{\mu\nu} = 8 \pi \left( T_{\mu\nu} - \frac12 T g_{\mu\nu} \right).
\]
Therefore the strong energy condition simply requires that the Ricci tensor satisfies $R_{\mu\nu}X^\mu X^\nu \geq 0$ for all timelike vectors $X$ (given that the Einstein equations are written as above).

Generically, the energy-momentum tensor is diagonalizable, that is, there exists an orthonormal frame $\{E_0,E_1,E_2,E_3\}$ in which the energy-momentum tensor is diagonal,
\[
(T_{\mu\nu}) = \diag(\rho, p_1, p_2, p_3).
\]
The timelike eigenvalue $\rho$ and is called the {\bf rest energy density}, and the spacelike eigenvalues $p_1, p_2,p_3$ are known as the {\bf principal pressures}. In terms of these eigenvalues, we have:
\begin{enumerate}
\item
SEC $\Leftrightarrow \rho+\sum_{i=1}^3p_i\geq 0$ and $\rho+p_i\geq 0$ ($i=1,2,3$).
\item
WEC $\Leftrightarrow \rho\geq 0$ and $\rho+p_i\geq 0$ ($i=1,2,3$).
\item
NEC $\Leftrightarrow \rho+p_i\geq 0$ ($i=1,2,3$).
\item
DEC $\Leftrightarrow \rho\geq |p_i|$ ($i=1,2,3$).
\end{enumerate}
Using this characterization, it is easy to see that the NEC is the weakest energy condition, that is, it is implied by any of the other conditions. The remaining three energy conditions are largely independent, except that the DEC implies the WEC. Notice that in particular the SEC does not imply the WEC.

\section{Conjugate points} \label{sec4.3}

\begin{Def}
Let $(M,g)$ be a Lorentzian manifold. A point $q \in M$ is said to be {\bf conjugate} to $p \in M$ along a timelike geodesic $c$ if there exists a nonvanishing solution $Y$ of the Jacobi equation $\nabla_X \nabla_X Y = R(X,Y)X$ such that $Y_p=Y_q=0$.
\end{Def}

Informally, two points $p$ and $q$ are conjugate along $c$ if there exists a nearby timelike geodesic intersecting $c$ at both $p$ and $q$ (Figure~\ref{conjugate}).

\begin{figure}[h!]
\begin{center}
\psfrag{p}{$p$}
\psfrag{q}{$q$}
\psfrag{g}{$c$}
\psfrag{X}{$X$}
\psfrag{Y}{$Y$}
\epsfxsize=.3\textwidth
\leavevmode
\epsfbox{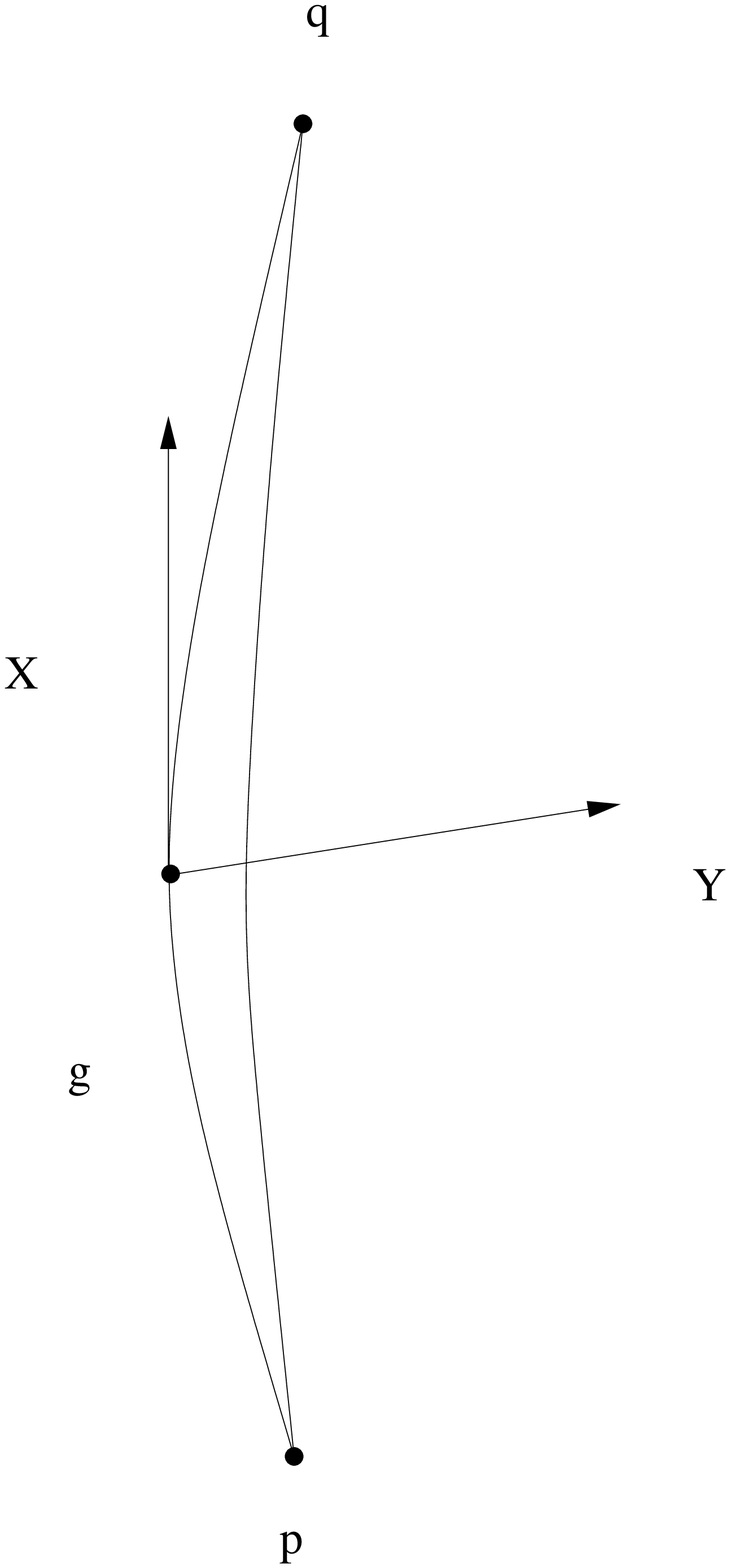}
\end{center}
\caption{Geodesic deviation.} \label{conjugate}
\end{figure}

Chose an orthonormal frame $\{X,E_1, E_2, E_3\}$ parallel along $c$, where $X$ is the unit tangent vector. Let $Y$ be a geodesic deviation vector $Y$ which vanishes at $p = c(0)$. If we write
\[
Y = Y^0 X + Y^i E_i
\]
then the Jacobi equation becomes
\[
\begin{cases}
\ddot{Y}^0 = 0 \\
\ddot{Y}^i = R^i_{\,\,\,\,00j} Y^j
\end{cases}.
\]
Since $Y^0$ is an affine function of the proper time $\tau$, it must vanish identically if $Y$ vanishes again along $c$. We will therefore assume that $Y^0=0$, that is, that $Y$ is orthogonal to $c$. Since the remaining components of $Y$ satisfy a linear ODE, we know that
\[
Y^i(\tau) = A_{ij}(\tau) \dot{Y}^j(0),
\]
where $A(\tau)$ is the fundamental matrix solution vanishing at $\tau=0$:
\[
\begin{cases}
A(0)=0 \\
\dot{A}(0)=I \\
\ddot{A}_{ij} = R^i_{\,\,\,\,00k} A_{kj}
\end{cases}.
\]
Although $A(0)=0$ is singular, $A(\tau)$ is not singular for $\tau > 0$ sufficiently small, because $\dot{A}(0)=I$. If $A(\tau)$ becomes singular for some $\tau_*>0$ then $q = c(\tau_*)$ is conjugate to $p$ (we just have to choose $\dot{Y}(0)$ to be a nonvanishing column vector in the kernel of $A(\tau_*)$).

Consider the congruence of timelike geodesics through $p$. Since on the one hand
\[
\dot{Y}^i = B_{ij} Y^j
\]
and on the other
\[
\dot{Y} = \dot{A} \dot{Y}(0) =  \dot{A} A^{-1} A \dot{Y}(0) =  \dot{A} A^{-1} Y,
\]
we conclude that
\[
B = \dot{A} A^{-1},
\]
and so
\[
\theta = \tr B = \tr (\dot{A} A^{-1}) = \frac{d}{d\tau} \log(\det A).
\]
Therefore the expansion of the congruence blows up if and only if the geodesic approaches the first conjugate point $q$.

\begin{Thm}
Let $(M,g)$ be a $4$-dimensional Lorentzian manifold satisfying the SEC, $c$ a timelike geodesic and $p = c(0)$. Suppose that the expansion $\theta$ of the congruence of timelike geodesics through $p$ takes a negative value $\theta_0 < 0$ at some point $r = c(\tau_0)$, with $\tau_0>0$. Then there exists a point $q$ conjugate to $p$ along $c$ at a distance at most $\frac{3}{|\theta_0|}$ from $r$.
\end{Thm}

\begin{proof}
By the Gauss Lemma, we can find local coordinates $(t,x^1,x^2,x^3)$ such that $X^\sharp = dt$, and so 
\[
dX^\sharp = 0 \Leftrightarrow \nabla_{[\mu}X_{\nu]} = 0.
\]
In other words, the congruence has no vorticity, and the Raychaudhuri equation becomes
\[
\frac{d\theta}{d\tau} = - \frac13 \theta^2 - \sigma_{\mu\nu} \sigma^{\mu\nu} - R_{\mu\nu} X^\mu X^\nu.
\]
Because $\sigma$ is purely spatial and $(M,g)$ satisfies the SEC, we have
\[
\frac{d\theta}{d\tau} \leq - \frac13 \theta^2 \Leftrightarrow - \frac1{\theta^2}\frac{d\theta}{d\tau} \geq \frac13 \Rightarrow \frac1{\theta} \geq \frac1{\theta_0} + \frac13(\tau-\tau_0).
\]
We conclude that $\frac1{\theta}$ vanishes, and so $\theta$ blows up, at proper time at most $\tau_0 + \frac{3}{|\theta_0|}$.
\end{proof}

Let $c(t,s)$ be a one-parameter family of timelike curves connecting two points $p$ and $q$:
\[
c(t_0,s) = p \qquad \text{ and } \qquad c(t_1,s) = q
\]
for all $s$. Then the connecting vector
\[
Y = \frac{\partial c}{\partial s},
\]
which in general is not a Jacobi field, satisfies
\[
Y_p = Y_q = 0.
\]
We assume that $c(t,s)$ has been parameterized in such a way that $Y$ does not vanish identically. 

The tangent vector
\[
X = \frac{\partial c}{\partial t}
\]
is timelike, and if we define
\[
f(t,s) = \left( - \left\langle X, X \right\rangle \right)^\frac12
\]
then the length of each curve is
\[
\tau(s) = \int_{t_0}^{t_1} f(t,s) dt.
\]
We have
\begin{align*}
\frac{d\tau}{ds} & =  \int_{t_0}^{t_1} \frac{\partial f}{\partial s} dt = - \int_{t_0}^{t_1} \frac1{f} \left\langle X, \nabla_Y X \right\rangle dt = - \int_{t_0}^{t_1} \frac1{f} \left\langle X, \nabla_X Y \right\rangle dt \\
& = - \int_{t_0}^{t_1} X \cdot \left( \frac1{f} \left\langle X, Y \right\rangle \right) dt + \int_{t_0}^{t_1} \left\langle \nabla_X \left(\frac{X}{f}\right), Y \right\rangle dt \\
& = \int_{t_0}^{t_1} \left\langle \nabla_X \left(\frac{X}{f}\right), Y \right\rangle dt,
\end{align*}
where we used the Fundamental Theorem of Calculus and the fact that $Y_p=Y_q=0$. This shows that the timelike curve $c$ defined as $c(t)=c(t,0)$ has extremal length among all timelike curves in such one-parameter families if and only if
\[
\nabla_X \left(\frac{X}{f}\right) = 0,
\]
that is, if and only if it is a timelike geodesic. Assume this to be the case. Then
\[
\frac{d^2\tau}{ds^2}(0) =  \int_{t_0}^{t_1} Y \cdot \left\langle \nabla_X \left(\frac{X}{f}\right), Y \right\rangle dt = \int_{t_0}^{t_1} \left\langle \nabla_Y\nabla_X \left(\frac{X}{f}\right), Y \right\rangle dt.
\]
Assuming that $f(t,0)=1$ (that is, $c$ is parameterized by its proper time) and $\langle X, Y \rangle = 0$ for $s=0$ (which is always possible by reparameterizing $c(t,s)$) leads to
\[
\frac{d^2\tau}{ds^2}(0) =  \int_{t_0}^{t_1} \left\langle \nabla_Y\nabla_X X, Y \right\rangle dt.
\]
Finally, using
\[
R(X,Y)Z = \nabla_X\nabla_Y Z - \nabla_Y\nabla_X Z - \nabla_{[X,Y]} Z = \nabla_X\nabla_Y Z - \nabla_Y\nabla_X Z
\]
we obtain
\begin{align*}
\frac{d^2\tau}{ds^2}(0) & =  \int_{t_0}^{t_1} \left\langle \nabla_X\nabla_Y X - R(X,Y)X, Y \right\rangle dt \\
& = \int_{t_0}^{t_1} \left\langle \nabla_X\nabla_X Y - R(X,Y)X, Y \right\rangle dt.
\end{align*}

\begin{Thm} \label{Thmconjugate}
A timelike curve $c$ connecting the points $p, q \in M$ locally maximizes the proper time (along any one-parameter family of timelike curves connecting the same points) if and only if it is a timelike geodesic without conjugate points to $p$ between $p$ and $q$. 
\end{Thm}

\begin{proof}
It is clear from what was done above that $c$ being a geodesic is a necessary condition.

Let us assume that $c$ has no conjugate points (to $p$, say) between $p$ and $q$. In an orthonormal frame parallel along $c$ we have
\[
\frac{d^2\tau}{ds^2}(0) =  \int_{t_0}^{t_1} Y^i \left(\ddot{Y}^i - R^i_{\,\,\,\,00j} Y^j \right) dt.
\]
Because there are no conjugate points, the fundamental matrix solution $A(t)$ is nondegenerate for $t \in (t_0,t_1)$, and we can set
\[
Y^i = A_{ij} Z^j.
\]
We have
\[
\ddot{Y}^i = A_{ij} \ddot{Z}^j + 2\dot{A}_{ij} \dot{Z}^j + \ddot{A}_{ij} Z^j = A_{ij} \ddot{Z}^j + 2\dot{A}_{ij} \dot{Z}^j + R^i_{\,\,\,\,00k} {A}_{kj} Z^j,
\]
and so
\begin{align*}
\frac{d^2\tau}{ds^2}(0) & =  \int_{t_0}^{t_1} A_{ij} Z^j \left( A_{ik} \ddot{Z}^k + 2\dot{A}_{ik} \dot{Z}^k \right) dt \\
& = \int_{t_0}^{t_1} Z^t A^t \left( A \ddot{Z} + 2\dot{A} \dot{Z} \right) dt \\
& = \int_{t_0}^{t_1} \left[ \frac{d}{dt} \left(Z^tA^tA \dot{Z}\right) - \dot{Z}^tA^tA \dot{Z} - Z^t\dot{A}^tA \dot{Z} + Z^tA^t\dot{A} \dot{Z} \right] dt \\
& = - \int_{t_0}^{t_1} (A\dot{Z})^t A \dot{Z} dt + \int_{t_0}^{t_1} Z^t \left(A^t\dot{A} - \dot{A}^tA\right) \dot{Z} dt.
\end{align*}
Above we used the Fundamental Theorem of Calculus and the fact that $(AZ)^t A\dot{Z} = Y^t (\dot{Y} - \dot{A}Z) = Y^t (\dot{Y} - BY)$ vanishes at $t_0$ and $t_1$ (although $B$ blows up as $(t-t_0)^{-1}$, $Y$ vanishes as $t-t_0$ or faster by Taylor's formula). From $\dot{A}=BA$ we have
\[
A^t\dot{A} - \dot{A}^tA = A^tBA - A^tB^tA = A^t \left( B - B^t \right) A = 2 A^t \omega A = 0,
\]
because the vorticity matrix $\omega$ vanishes for the congruence of timelike geodesics through $p$. Therefore
\[
\frac{d^2\tau}{ds^2}(0) = - \int_{t_0}^{t_1} (A\dot{Z})^t A \dot{Z} dt \leq 0,
\]
with equality if and only if
\[
\dot{Z} \equiv 0 \Rightarrow Z \equiv 0 \Rightarrow Y \equiv 0
\]
(note that if $Z$ is constant then it must be zero because $0 = Y_q = A(t_1) Z$ and $A(t_1)$ is nonsingular). We conclude that $c$ is indeed a maximum of the proper time along any one-parameter family of timelike curves connecting $p$ and $q$.

On the other hand, if there exists a conjugate point along $c$ between $p$ and $q$, say $r=c(t^*)$, then let $\hat{Y}$ be a nonvanishing Jacobi field such that $\hat{Y}(t_0) = \hat{Y}(t^*) = 0$ (in particular $\hat{Y}$ is orthogonal to $c$), and let $Y$ be the vector field along $c$ that coincides with $\hat{Y}$ between $p$ and $r$ and is zero between $r$ and $q$. Similarly, let $\hat{Z}$ be the (necessarily spacelike) vector field parallel along $c$ such that $\hat{Z}(t^*) = -\nabla_X\hat{Y}(t^*)$, and let $Z(t)=\theta(t)\hat{Z}(t)$, where $\theta$ is a smooth function satisfying $\theta(t_0)=\theta(t_1)=0$ and $\theta(t^*)=1$. Finally, let $Y_\varepsilon$ be the vector field along $c$ defined by $Y_\varepsilon = Y + \varepsilon Z$, and consider a one-parameter family of curves $c_\varepsilon(t,s)$ such that $c_\varepsilon(t,0)=c(t)$ and $Y_\varepsilon = \frac{\partial c_\varepsilon}{\partial s}$. Since $Y_\varepsilon$ is not $C^1$, we must write the formula for the second derivative of the length as
\[
\frac{d^2\tau}{ds^2}(0) = - \int_{t_0}^{t_1} \biggl( \left\langle \nabla_X Y_\varepsilon, \nabla_X Y_\varepsilon\right\rangle + \left\langle R(X,Y_\varepsilon)X, Y_\varepsilon \right\rangle \biggr) dt = I(Y_\varepsilon, Y_\varepsilon),
\]
where the bilinear form $I$ is clearly symmetric. Therefore
\[
\frac{d^2\tau}{ds^2}(0) = I(Y,Y) + 2\varepsilon I(Y,Z) + \varepsilon^2 I(Z,Z).
\]
Since $Y$ is a Jacobi field between $p$ and $r$, and zero between $r$ and $q$, we have $I(Y,Y)=0$. On the other hand,
\begin{align*}
I(Y,Z) & = - \int_{t_0}^{t^*} \biggl( \left\langle \nabla_X Y, \nabla_X Z \right\rangle + \left\langle R(X,Y)X, Z \right\rangle \biggr) dt \\
& = - \biggl[ \left\langle \nabla_X Y, Z\right\rangle\biggr]_{t_0}^{t^*} + \int_{t_0}^{t^*} \biggl( \left\langle \nabla_X\nabla_X Y, Z \right\rangle - \left\langle R(X,Y)X, Z \right\rangle \biggr) dt \\
& = \left\langle \nabla_X \hat{Y}(t^*), \nabla_X \hat{Y}(t^*)\right\rangle > 0.
\end{align*}
Therefore for $\varepsilon > 0$ sufficiently small the one-parameter family $c_\varepsilon(t,s)$ contains curves whose length is greater than the length of $c$.

Figure~\ref{conjugate2} illustrates the geometric idea behind the proof above: $c_s$ represents a generic curve of a one-parameter family corresponding to $Y$, and has the same length as $c$; adding $\varepsilon Z$ changes $c_s$ between points $u$ and $v$, say, making it longer by the twin paradox.


\end{proof}

\begin{figure}[h!]
\begin{center}
\psfrag{p}{$p$}
\psfrag{q}{$q$}
\psfrag{r}{$r$}
\psfrag{s}{$u$}
\psfrag{t}{$v$}
\psfrag{c}{$c$}
\psfrag{ct}{$c_s$}
\psfrag{Y}{$Y$}
\psfrag{DY}{$\nabla_X\hat{Y}(t^*)$}
\epsfxsize=.5\textwidth
\leavevmode
\epsfbox{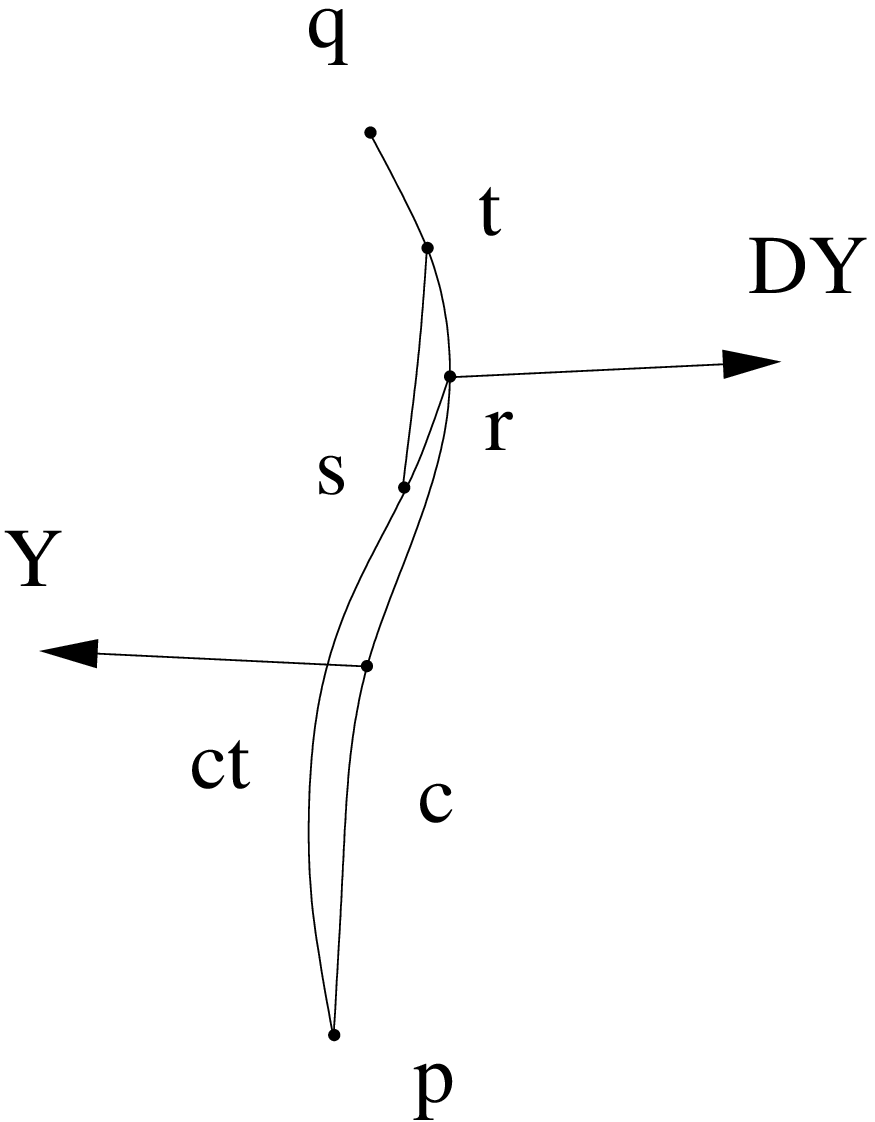}
\end{center}
\caption{Proof of Theorem~\ref{Thmconjugate}.} \label{conjugate2}
\end{figure}

The results above can be generalized for timelike geodesics orthogonal to a spacelike hypersurface $S$. If one considers the congruence of such geodesics then at $S$
\begin{align*}
\nabla_X Y_\mu & = \nabla_Y X_\mu = Y^\nu (\nabla_\nu X_\mu) = Y^\nu (\nabla_{(\nu} X_{\mu)} + \nabla_{[\nu} X_{\mu]}) \\
& = \nabla_{(\mu} X_{\nu)}  Y^\nu = \frac12 \cL_X g_{\mu\nu} Y^\nu = K_{\mu\nu} Y^\nu.
\end{align*}

\begin{Def}
Let $(M,g)$ be a Lorentzian manifold and let $S \subset M$ be a spacelike hypersurface with second fundamental form $K$. A point $q \in M$ is said to be {\bf conjugate} to $S$ along a timelike geodesic $c$ orthogonal to $S$ at some point $p \in S$ if there exists a nonvanishing solution $Y$ of the Jacobi equation $\nabla_X \nabla_X Y = R(X,Y)X$ such that $Y_p \in T_pS$, $(\nabla_X Y_\mu )_p = (K_{\mu\nu} Y^\nu)_p$ and $Y_q = 0$.
\end{Def}

In an orthonormal frame $\{X,E_1, E_2, E_3\}$ parallel along $c$, again we can assume that $Y^0=0$, and have for the remaining components
\[
Y^i(t) = A_{ij}(t) Y^j(0),
\]
where $A(t)$ is the fundamental matrix solution:
\[
\begin{cases}
A(0)=I \\
\dot{A}(0) = K \\
\ddot{A}_{ij} = R^i_{\,\,\,\,00k} A_{kj}
\end{cases}.
\]
Arguing as above, we have the following result.

\begin{Thm} \label{theta_0}
Let $(M,g)$ be a $4$-dimensional Lorentzian manifold satisfying the SEC, $S \subset M$ a spacelike hypersurface and $c$ a timelike geodesic orthogonal to $S$ at some point $p \in S$. Suppose that the expansion $\theta$ of the congruence of timelike geodesics orthogonal to $S$ takes a negative value $\theta_0 < 0$ at $p$. Then there exists a point $q$ conjugate to $S$ along $c$ at a distance at most $\frac{3}{|\theta_0|}$ from $S$.
\end{Thm}

\begin{Thm} \label{maximizing}
A timelike curve $c$ connecting the spacelike hypersurface $S \subset M$ to the point $q \in M$ locally maximizes the proper time (along any one-parameter family of timelike curves connecting $S$ to $q$) if and only if it is a timelike geodesic orthogonal to $S$ without conjugate points to $S$ between $S$ and $q$. 
\end{Thm}

\begin{proof}
The proof is basically the same as for curves connecting two points. The main differences are the in formula
\begin{align*}
\frac{d\tau}{ds} & = - \int_{t_0}^{t_1} X \cdot \left( \frac1{f} \left\langle X, Y \right\rangle \right) dt + \int_{t_0}^{t_1} \left\langle \nabla_X \left(\frac{X}{f}\right), Y \right\rangle dt \\
& = \frac1{f(t_0,s)} \left\langle X_p, Y_p \right\rangle + \int_{t_0}^{t_1} \left\langle \nabla_X \left(\frac{X}{f}\right), Y \right\rangle dt,
\end{align*}
which requires $c$ to be orthogonal to $S$ at $p = c(t_0)$; the fact that $A(t_0)$ does not vanish, but instead
\begin{align*}
(AZ)^t A\dot{Z} & = Y^t (\dot{Y} - \dot{A}Z) = Y^t (\dot{Y} - BAZ) \\
& = Y^t (\dot{Y} - B Y) = Y^t (\dot{Y} - K Y) = 0
\end{align*}
at $t_0$; and the integrated formula
\begin{align*}
\frac{d^2\tau}{ds^2}(0) & = (K_{\mu\nu} Y^\mu Y^\nu)(p) - \int_{t_0}^{t_1} \biggl( \left\langle \nabla_X Y , \nabla_X Y \right\rangle + \left\langle R(X,Y)X, Y \right\rangle \biggr) dt \\
& = (K_{\mu\nu} Y^\mu Y^\nu)(p) + I(Y,Y),
\end{align*}
which vanishes when $Y$ is a Jacobi field.
\end{proof}

\section{Existence of maximizing geodesics} \label{sec4.4}

\begin{Prop} \label{compact}
Let $(M,g)$ be a globally hyperbolic spacetime, $S$ a Cauchy hypersurface and $p \in D^+(S)$. Then $D^+(S)\cap J^-(p)$ is compact.
\end{Prop}

\begin{proof}
Let us define a {\bf simple neighborhood} $U \subset M$ to be a geodesically convex open set diffeomorphic to an open ball whose boundary is a compact submanifold of a larger geodesically convex open set (therefore $\partial U$ is diffeomorphic to $S^3$ and $\overline{U}$ is compact). It is clear that simple neighborhoods form a basis for the topology of $M$. Also, it is easy to show that any open cover $\{ V_\alpha \}_{\alpha \in A}$ has a countable, locally finite refinement $\{ U_n \}_{n \in \bbN}$ by simple neighborhoods.

If $A=D^+(S)\cap J^-(p)$ were not compact, there would exist a countable, locally finite open cover $\{ U_n \}_{n \in \bbN}$ of $A$ by simple neighborhoods not admitting any finite subcover. Take $q_n \in A \cap U_n$ such that $q_m \neq q_n$ for $m \neq n$. The sequence $\{ q_n \}_{n \in \bbN}$ cannot have accumulation points, since any point in $M$ has a neighborhood intersecting only finite simple neighborhoods $U_n$. In particular, each simple neighborhood $U_n$ contains only a finite number of points in the sequence (as $\overline{U}_n$ is compact).

Set $p_1=p$. Since $p_1 \in A$, we have $p_1 \in U_{n_1}$ for some $n_1 \in \bbN$. Let $q_n \not\in U_{n_1}$. Since $q_n \in J^-(p_1)$, there exists a future-directed causal curve $c_n$ connecting $q_n$ to $p_1$. This curve will necessarily intersect $\partial U_{n_1}$. Let $r_{1,n}$ be an intersection point. Since $U_{n_1}$ contains only a finite number of points in the sequence $\{ q_n \}_{n \in \bbN}$, there will exist infinite intersection points $r_{1,n}$. As $\partial U_{n_1}$ is compact, these will accumulate to some point $p_2 \in \partial U_{n_1}$ (cf.~Figure~\ref{proof}).

Because $\overline{U}_{n_1}$ is contained in a geodesically convex open set $V$, which can be chosen so that $v \mapsto (\pi(v),\exp(v))$ is a diffeomorphism onto $V \times V$, we have $p_2 \in J^-(p_1)$: if $\gamma_{1,n}$ is the unique causal geodesic connecting $p_1$ to $r_{1,n}$, parameterized by the global time function $t:M \to \bbR$, then the subsequence of $\{\gamma_{1,n}\}$ corresponding to a convergent subsequence of $\{r_{1,n}\}$ will converge to a causal geodesic $\gamma_1$ connecting $p_1$ to $p_2$. If $S=t^{-1}(0)$ then we have $t(r_{1,n}) \geq 0$, implying that $t(p_2) \geq 0$ and hence $p_2 \in A$. Since $p_2 \not\in U_{n_1}$, there must exist $n_2 \in \bbN$ such that $p_2 \in U_{n_2}$. 

Since $U_{n_2}$ contains only a finite number of points in the sequence $\{ q_n \}_{n \in \bbN}$, an infinite number of curves $c_n$ must intersect $\partial U_{n_2}$ to the past of $r_{1,n}$. Let $r_{2,n}$ be the intersection points. As $\partial U_{n_2}$ is compact, $\{r_{2,n}\}$ must accumulate to some point $p_3 \in \partial U_{n_2}$. Because $\overline{U}_{n_2}$ is contained in a geodesically convex open set, $p_3 \in J^-(p_2)$: if $\gamma_{2,n}$ is the unique causal geodesic connecting $r_{1,n}$ to $r_{2,n}$, parameterized by the global time function, then the subsequence of $\{\gamma_{2,n}\}$ corresponding to convergent subsequences of both $\{r_{1,n}\}$ and $\{r_{2,n}\}$ will converge to a causal geodesic connecting $p_2$ to $p_3$. Since $J^-(p_2) \subset J^-(p_1)$ and $t(r_{2,n}) \geq 0 \Rightarrow t(p_3) \geq 0$, we have $p_3 \in A$.

Iterating the procedure above, we can construct a sequence $\{p_i\}_{i\in\bbN}$ of points in $A$ satisfying $p_i \in U_{n_i}$ with $n_i \neq n_j$ if $i \neq j$, such that $p_{i}$ is connected to $p_{i+1}$ by a causal geodesic $\gamma_i$. It is clear that $\gamma_i$ cannot intersect $S$, for $t(p_{i+1}) > t(p_{i+2}) \geq 0$. On the other hand, the piecewise smooth causal curve obtained by joining the curves $\gamma_i$ can easily be smoothed into a past-directed causal curve starting at $p_1$ which does not intersect $S$. Finally, such curve is inextendible: it cannot converge to any point, as $\{p_i\}_{i\in\bbN}$ cannot accumulate. But since $p_1 \in D^+(S)$, this curve would have to intersect $S$. Therefore $A$ must be compact.
\end{proof}

\begin{figure}[h!]
\begin{center}
\psfrag{p=p1}{$p=p_1$}
\psfrag{p2}{$p_2$}
\psfrag{p3}{$p_3$}
\psfrag{Un1}{$U_{n_1}$}
\psfrag{Un2}{$U_{n_2}$}
\epsfxsize=.7\textwidth
\leavevmode
\epsfbox{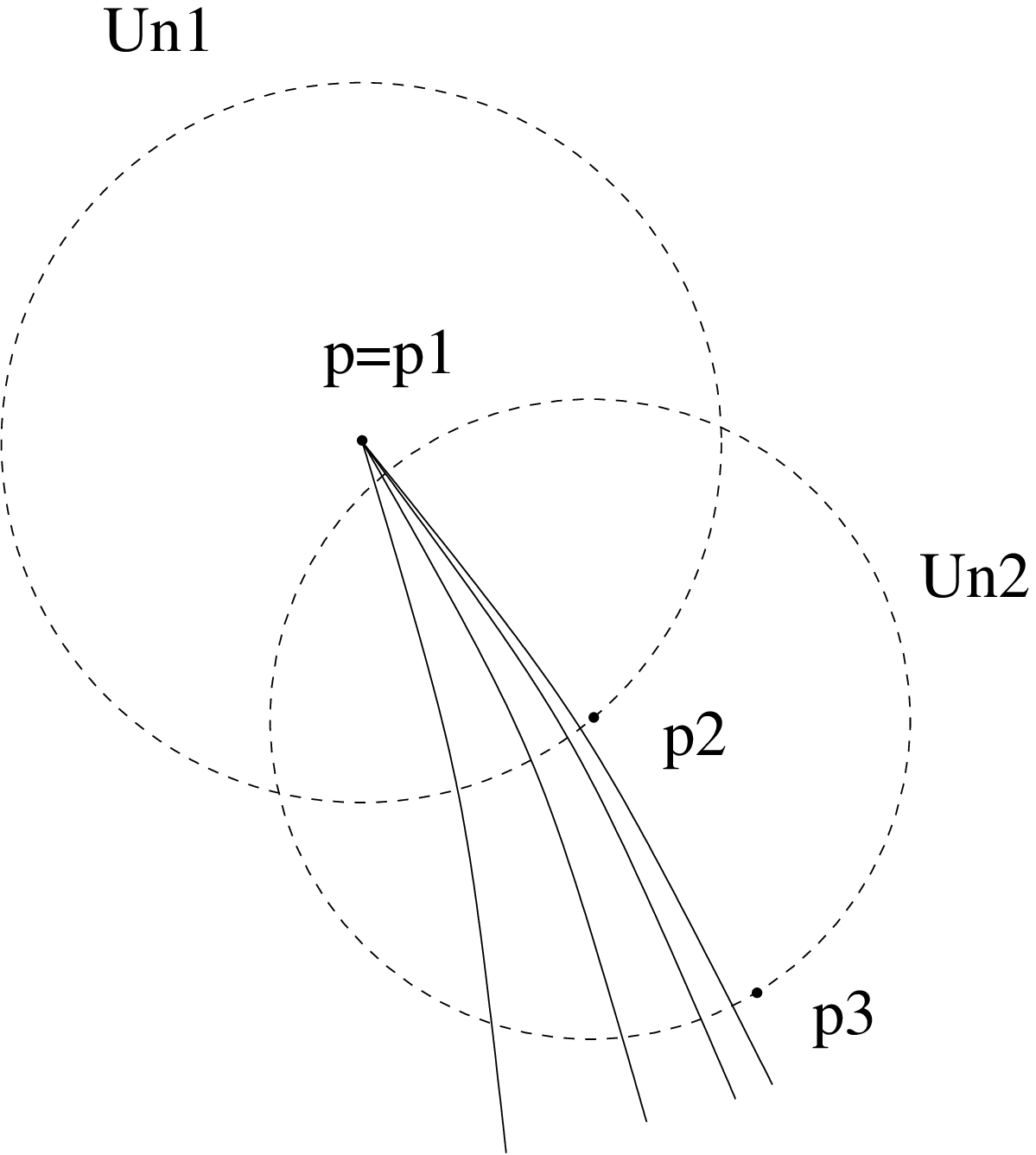}
\end{center}
\caption{Proof of Proposition~\ref{compact}.} \label{proof}
\end{figure}

\begin{Cor} \label{closed_compact}
Let $(M,g)$ be a globally hyperbolic spacetime and $p,q \in M$. Then
\begin{enumerate}[(i)]
\item
$J^+(p)$ is closed;
\item
$J^+(p) \cap J^-(q)$ is compact.
\end{enumerate}
\end{Cor}

\begin{proof}
Exercise.
\end{proof}

Proposition~\ref{compact} is a key ingredient in establishing the following fundamental result.

\begin{Thm} \label{maximal}
Let $(M,g)$ be a globally hyperbolic spacetime with Cauchy hypersurface $S$, and $p \in D^+(S)$. Then, among all timelike curves connecting $p$ to $S$, there exists a timelike curve with maximal length. This curve is a timelike geodesic, orthogonal to $S$.
\end{Thm}

\begin{proof}
Consider the set $T(S,p)$ of all timelike curves connecting $S$ to $p$. Since we can always use the global time function $t:M\to \bbR$ as a parameter, these curves are determined by their images, which are compact subsets of the compact set $A=D^+(S) \cap J^-(p)$. As it is well known (see for instance \cite{Munkres00}), the set $C(A)$ of all compact subsets of $A$ is a compact metric space for the {\bf Hausdorff metric} $d_H$, defined as follows: if $d:M\times M \to \bbR$ is a metric yielding the topology of $M$,
\[
d_H(K,L) = \inf \{ \varepsilon > 0 \mid K \subset U_\varepsilon(L) \text{ and } L \subset U_\varepsilon(K) \},
\]
where $U_{\varepsilon}(K)$ is a $\varepsilon$-neighborhood of $K$ for the metric $d$. Therefore, the closure $C(S,p) = \overline{T(S,p)}$ is a compact subset of $C(A)$. It is not difficult to show that $C(S,p)$ can be identified with the set of {\bf continuous causal curves} connecting $S$ to $p$ (a continuous curve $c:[0,t(p)]\to M$ is said to be {\bf causal} if $c(t_2) \in J^+(c(t_1))$ whenever $t_2 > t_1$).

The length function $\tau:T(S,p) \to \bbR$ is defined by
\[
\tau(c)=\int_0^{t(p)} |\dot{c}(t)| dt.
\]
This function is {\bf upper semicontinuous}, i.e.~continuous for the topology 
\[
\mathcal{O}=\{(-\infty,a) \mid -\infty \leq a \leq +\infty \}
\]
in $\bbR$. Indeed, let $c \in T(S,p)$ be parameterized by its arclength $u$. For a sufficiently small $\varepsilon > 0$, the function $u$ can be extended to the $\varepsilon$-neighborhood $U_{\varepsilon}(c)$ in such a way that its level hypersurfaces are spacelike and orthogonal to $c$, that is, $-\grad u$ is timelike and coincides with $\dot{c}$ on $c$ (cf.~Figure~\ref{proof2}). If $\gamma \in T(S,p)$ is in the open ball $B_\varepsilon(c) \subset C(A)$ for the Hausdorff metric $d_H$ then we can use $u$ as a parameter, thus obtaining
\[
du(\dot{\gamma}) = 1 \Leftrightarrow \langle \dot{\gamma}, \grad u \rangle = 1.
\]
Therefore $\dot{\gamma}$ can be decomposed as
\[
\dot{\gamma} = \frac1{\langle \grad u, \grad u \rangle} \grad u + X,
\]  
where $X$ is spacelike and orthogonal to $\grad u$, and so
\[
|\dot{\gamma}| = \left| \frac1{\langle \grad u, \grad u \rangle} + \langle X,X \rangle \right|^\frac12.
\]
Given $\delta > 0$, we can choose $\varepsilon>0$ sufficiently small so that
\[
-\frac1{\langle \grad u, \grad u \rangle} < \left(1 + \frac{\delta}{2\tau(c)}\right)^2
\]
on the $\varepsilon$-neighborhood $U_{\varepsilon}(c)$ (as $\langle \grad u, \grad u \rangle=-1$ along $c$). We have
\[
\tau(\gamma) =  \int_0^{t(p)} \left|\frac{d \gamma}{dt} \right| \, dt = \int_0^{t(p)} |\dot{\gamma}| \frac{du}{dt} \, dt = \int_{u(\gamma\cap S)}^{\tau(c)} |\dot{\gamma}| \, d u,
\]
where we have to allow for the fact that $c$ is not necessarily orthogonal to $S$, and so the initial point of $\gamma$ is not necessarily at $u=0$ (cf.~Figure~\ref{proof2}). Consequently, 
\begin{align*}
\tau(\gamma) & = \int_{u(\gamma\cap S)}^{\tau(c)} \left| -\frac1{\langle \grad u, \grad u \rangle} - \langle X,X \rangle \right|^\frac12 \, d u \\
& < \int_{u(\gamma\cap S)}^{\tau(c)} \left(1 + \frac{\delta}{2\tau(c)}\right) \, d u  = \left(1 + \frac{\delta}{2\tau(c)}\right) \left(\tau(c) - u(\gamma\cap S)\right).
\end{align*}
Choosing $\varepsilon$ sufficiently small so that 
\[
|u|< \left( \frac1{\tau(c)} + \frac{2}{\delta} \right)^{-1}
\]
on $S \cap U_{\varepsilon}(c)$, we obtain $\tau(\gamma) < \tau(c) + \delta$, proving upper semicontinuity in $T(S,p)$. As a consequence, the length function can be extended to $C(S,p)$ through
\[
\tau(c)=\lim_{\varepsilon \to 0} \sup\{ \tau(\gamma) \mid \gamma \in B_\varepsilon(c) \cap T(S,p) \}
\]
(as for $\varepsilon>0$ sufficiently small the supremum will be finite). Also, it is clear that if $c \in T(S,p)$ then the upper semicontinuity of the length forces the two definitions of $\tau(c)$ to coincide. The extension of the length function to $C(S,p)$ is trivially upper semicontinuous: given $c \in C(S,p)$ and $\delta > 0$, let $\varepsilon>0$ be such that $\tau(\gamma) < \tau(c) + \frac{\delta}2$ for any $\gamma \in B_{2\varepsilon}(c) \cap T(S,p)$. Then it is clear that $\tau(c') \leq \tau(c) + \frac{\delta}2 < \tau(c) + \delta$ for any $c' \in B_{\varepsilon}(c) \cap C(S,p)$.

Finally, we notice that the compact sets of $\bbR$ for the topology $\mathcal{O}$ are the sets with a maximum. Therefore, the length function attains a maximum at some point $c \in C(S,p)$. All that remains to be seen is that the maximum is also attained at a smooth timelike curve $\gamma$. To do so, cover $c$ with finitely many geodesically convex neighborhoods and choose points $p_1, \ldots, p_k$ in $c$ such that $p_1 \in S$, $p_k=p$ and the portion of $c$ between $p_{i-1}$ and $p_{i}$ is contained in a geodesically convex neighborhood for all $i=2, \ldots, k$. It is clear that there exists a sequence $c_n \in T(S,p)$ such that $c_n \to c$ and $\tau(c_n) \to \tau(c)$. Let $t_i=t(p_i)$ and $p_{i,n}$ be the intersection of $c_n$ with $t^{-1}(t_i)$. Replace $c_n$ by the sectionally geodesic curve $\gamma_n$ obtained by joining $p_{i-1,n}$ to $p_{i,n}$ in the corresponding geodesically convex neighborhood. Then $\tau(\gamma_n) \geq \tau(c_n)$, and therefore $\tau(\gamma_n) \to \tau(c)$. Since each sequence $p_{i,n}$ converges to $p_i$, $\gamma_n$ converges to the sectionally geodesic curve $\gamma$ obtained by joining $p_{i-1}$ to $p_{i}$ ($i=2, \ldots, k$), and it is clear that $\tau(\gamma_n) \to \tau(\gamma)=\tau(c)$. Therefore $\gamma$ is a point of maximum for the length. Finally, we notice that $\gamma$ must be smooth at the points $p_i$, for otherwise we could increase its length by using the twin paradox. Therefore $\gamma$ must be a timelike geodesic. It is also clear that $\gamma$ must be orthogonal to $S$, for otherwise it would be possible to increase its length by small deformations.
\end{proof}

\begin{figure}[h!]
\begin{center}
\psfrag{p}{$p$}
\psfrag{c}{$c$}
\psfrag{g}{$\gamma$}
\psfrag{S}{$S$}
\psfrag{u=0}{$u=0$}
\psfrag{u=t(c)}{$u=\tau(c)$}
\psfrag{Ue(c)}{$U_{\varepsilon}(c)$}
\epsfxsize=.6\textwidth
\leavevmode
\epsfbox{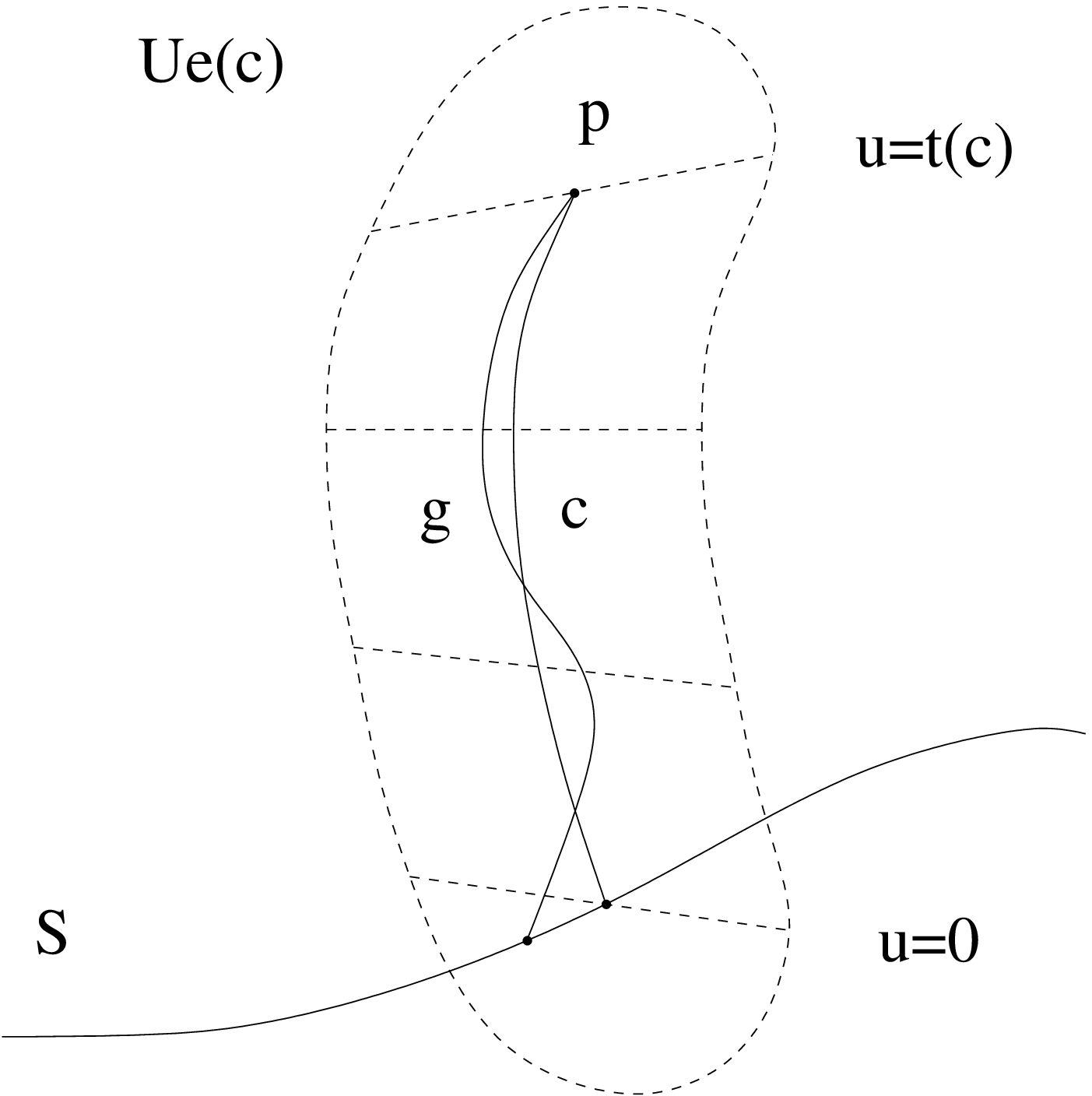}
\end{center}
\caption{Proof of Theorem~\ref{maximal}.} \label{proof2}
\end{figure}

\section{Hawking's singularity theorem} \label{sec4.5}

We have now all the necessary ingredients to prove the Hawking singularity theorem.

\begin{Def}
A spacetime $(M,g)$ is said to be {\bf singular} if it is not geodesically complete.
\end{Def}

\begin{Thm} ({\bf Hawking \cite{H66}}) \label{Hawking}
Let $(M,g)$ be a globally hyperbolic spacetime satisfying the strong energy condition, and suppose that the expansion of the congruence of future-pointing timelike geodesics orthogonal to $S$ satisfies $\theta\leq\theta_0 < 0$ on a Cauchy hypersurface $S$. Then $(M,g)$ is singular.
\end{Thm}

\begin{proof}
We will show that no future-directed timelike geodesic orthogonal to $S$ can be extended to proper time greater than $\tau_0=-\frac{3}{\theta_0}$ to the future of $S$. Suppose that this was not so. Then there would exist a future-directed timelike geodesic $c$ orthogonal to $S$, parameterized by proper time, defined in an interval $[0,\tau_0+\varepsilon]$ for some $\varepsilon>0$. Let $p=c(\tau_0+\varepsilon)$. According to Theorem~\ref{maximal}, there would exist a timelike geodesic $\gamma$ with maximal length connecting $S$ to $p$, orthogonal to $S$. Because $\tau(c)=\tau_0+\varepsilon$, we would necessarily have $\tau(\gamma)\geq\tau_0+\varepsilon$. Theorem~\ref{theta_0} guarantees that $\gamma$ would develop a conjugate point at a distance of at most $\tau_0$ to the future of $S$, and Theorem~\ref{maximizing} states that $\gamma$ would cease to be maximizing beyond this point. Therefore we arrive at a contradiction.
\end{proof}

\begin{Remark}
It should be clear that $(M,g)$ is singular if the condition $\theta\leq\theta_0 < 0$ on a Cauchy hypersurface $S$ is replaced by the condition $\theta\geq\theta_0 > 0$ on $S$. In this case, no {\bf past-directed} timelike geodesic orthogonal to $S$ can be extended to proper time greater than $\tau_0=\frac{3}{\theta_0}$ to the {\bf past} of $S$.
\end{Remark}

\begin{Example} \hspace{1cm}
\begin{enumerate}
\item
The FLRW models with $\alpha>0$ and $\Lambda = 0$ are globally hyperbolic, and satisfy the strong energy condition (as $\rho > 0$). Moreover, the expansion of the congruence tangent to $\frac{\partial}{\partial t}$ is $\theta = \frac{3\dot{a}}{a}$. Assume that the model is expanding at time $t_0$. Then $\theta = \theta_0 = \frac{3\dot{a}(t_0)}{a(t_0)}>0$ on the Cauchy hypersurface $S = \{ t=t_0 \}$, and hence Theorem~\ref{Hawking} guarantees that this model is singular to the past of $S$ (i.e.~there exists a big bang). Moreover, Theorem~\ref{Hawking} implies that this singularity is generic: any sufficiently small perturbation of an expanding FLRW model satisfying the strong energy condition will also be singular. Loosely speaking, any expanding universe must have begun at a big bang.
\item
The region $\{ r < 2m \}$ of the Schwarzschild solution is globally hyperbolic, and satisfies the strong energy condition (as $R_{\mu\nu}=0$). The metric can be written in this region as
\[
\hspace{2cm} g = - d\tau^2 + \left( \frac{2m}r - 1 \right) dt^2 + r^2 \left(d\theta^2 + \sin^2 \theta d\varphi^2 \right), 
\]
where
\[
\tau = \int_r^{2m} \left( \frac{2m}u - 1 \right)^{-\frac12} du.
\]
Therefore the inside of the black hole can be pictured as a cylinder $\bbR \times S^2$ whose shape is evolving in time. As $r \to 0$, the $S^2$ contracts to a singularity, with the $t$-direction expanding. Since
\[
K = \frac{dr}{d\tau} \left( -\frac{m}{r^2} dt^2 + r d\theta^2 + r \sin^2 \theta d\varphi^2\right),
\]
we have
\[
\theta = \left( \frac{2m}r - 1 \right)^{-\frac12}\left( \frac{2}r - \frac{3m}{r^2} \right).
\]
Therefore we have $\theta = \theta_0 < 0$ on any Cauchy hypersurface $S = \{ r=r_0 \}$ with $r_0 < \frac{3m}2$, and hence Theorem~\ref{Hawking} guarantees that the Schwarzschild solution is singular to the future of $S$. Moreover, Theorem~\ref{Hawking} implies that this singularity is generic: any sufficiently small perturbation of the Schwarzschild solution satisfying the strong energy condition will also be singular. Loosely speaking, once the collapse has advanced long enough, nothing can prevent the formation of a singularity.
\item
It should be noted that Theorem~\ref{Hawking} proves geodesic incompleteness, not the existence of curvature singularities. For instance, it applies to the Milne universe, or a globally hyperbolic region of the anti-de Sitter universe, whose curvature is bounded (they are simply globally hyperbolic regions in larger, inextendible Lorentzian manifolds).
\end{enumerate}
\end{Example}

\section{Penrose's singularity theorem} \label{sec4.6}

Let $(M,g)$ be a globally hyperbolic spacetime, $S$ a Cauchy hypersurface with future-pointing unit normal vector field $N$, and $\Sigma \subset S$ a compact $2$-dimensional submanifold with unit normal vector field $n$ in $S$. Let $c_p$ be the null geodesic with initial condition $N_p+n_p$ for each point $p \in \Sigma$. We define a smooth map $\exp:(-\varepsilon,\varepsilon) \times \Sigma \to M$ for some $\varepsilon > 0$ as $\exp(r,p)=c_p(r)$. 

\begin{Def}
The critical values of $\exp$ are said to be {\bf conjugate points} to $\Sigma$.
\end{Def}

Loosely speaking, conjugate points are points where geodesics starting orthogonally at nearby points of $\Sigma$ intersect.

Let $q=\exp(r_0,p)$ be a point not conjugate to $\Sigma$. If $\varphi$ is a local parameterization of $\Sigma$ around $p$, then we can construct a system of local coordinates $(u,r,x^2,x^3)$ on some open set $V \ni q$ by using the map
\[
(u,r,x^2,x^3) \mapsto \exp(r, \psi_u(\varphi(x^2,x^3))),
\]
where $\psi_u$ is the flow along the timelike geodesics orthogonal to $S$ and the map $\exp:(-\varepsilon,\varepsilon) \times \psi_u(\Sigma) \to M $ is defined as above.

Since $\frac{\partial}{\partial r}$ is tangent to null geodesics, we have $g_{rr}=\left\langle\frac{\partial}{\partial r},\frac{\partial}{\partial r}\right\rangle = 0$. On the other hand, we have
\begin{align*}
\frac{\partial g_{r\mu}}{\partial r} & = \frac{\partial}{\partial r} \left\langle\frac{\partial}{\partial r},\frac{\partial}{\partial x^\mu}\right\rangle
= \left\langle\frac{\partial}{\partial r},\nabla_{\frac{\partial}{\partial r}}\frac{\partial}{\partial x^\mu}\right\rangle \\
& = \left\langle\frac{\partial}{\partial r},\nabla_{\frac{\partial}{\partial x^\mu}}\frac{\partial}{\partial r}\right\rangle
= \frac12 \frac{\partial}{\partial x^\mu} \left\langle\frac{\partial}{\partial r},\frac{\partial}{\partial r}\right\rangle = 0,
\end{align*}
for $\mu = 0,1,2,3$. Since $g_{ru}=-1$ and $g_{r2}=g_{r3}=0$ on $\psi_u(\Sigma)$, we have $g_{ru}=-1$ and $g_{r2}=g_{r3}=0$ on $V$. Therefore the metric is written in this coordinate system as
\[
g = \alpha du^2 - 2 du dr + 2 \beta_A du dx^A + \gamma_{AB} dx^A dx^B.
\]
Since
\[
\det \left(
\begin{matrix}
\alpha & -1 & \beta_2 & \beta_3 \\
-1 & 0 & 0 & 0 \\
\beta_2 & 0 & \gamma_{22} & \gamma_{23} \\
\beta_3 & 0 & \gamma_{32} & \gamma_{33}
\end{matrix}
\right)
=
-\det \left(
\begin{matrix}
\gamma_{22} & \gamma_{23} \\
\gamma_{32} & \gamma_{33}
\end{matrix}
\right),
\]
we see that the functions 
\[
\gamma_{AB} = \left\langle \frac{\partial}{\partial x^A}, \frac{\partial}{\partial x^B} \right\rangle
\]
form a positive definite matrix, and so $g$ induces a Riemannian metric on the $2$-dimensional surfaces $\exp(r,\psi_u(\Sigma))$, which are then spacelike. Since the vector fields $\frac{\partial}{\partial x^A}$ can always be defined along $c_{p}$, the matrix $(\gamma_{AB})$ is also well defined along $c_p$, even at points where the coordinate system breaks down, i.e.~at points which are conjugate to $\Sigma$. These are the points for which $\gamma=\det\left(\gamma_{AB}\right)$ vanishes, since only then will $\left\{\frac{\partial}{\partial u},\frac{\partial}{\partial r},\frac{\partial}{\partial x^2},\frac{\partial}{\partial x^3}\right\}$ fail to be linearly independent. In fact the vector fields $\frac{\partial}{\partial x^A}$ are Jacobi fields along $c_p$.

It is easy to see that
\[
\Gamma^u_{ur} = \Gamma^u_{rr} = \Gamma^u_{rA} = \Gamma^r_{rr} = \Gamma^A_{rr} = 0 \quad \text{ and } \quad \Gamma^A_{rB} = \gamma^{AC} \beta_{CB},
\]
where $(\gamma^{AB})=(\gamma_{AB})^{-1}$ and $\beta_{AB} = \frac12 \frac{\partial \gamma_{AB}}{\partial r}$. Consequently,
\begin{align*}
R_{rr} & = R_{urr}^{\;\;\;\;\;u} + R_{Arr}^{\;\;\;\;\;A} = \left( - \frac{\partial \Gamma^A_{Ar}}{\partial r} - \Gamma^B_{Ar} \Gamma^A_{rB}\right)\\
& = - \frac{\partial}{\partial r} \left(\gamma^{AB} \beta_{AB}\right) - \gamma^{BC} \gamma^{AD} \beta_{CA} \beta_{DB}.
\end{align*}
The quantity
\[
\theta = \gamma^{AB} \beta_{AB}
\]
appearing in this expression is called the {\bf expansion} of the null geodesics, and has an important geometric meaning: 
\[
\theta = \frac12 \tr \left((\gamma_{AB})^{-1} \frac{\partial}{\partial r} (\gamma_{AB})\right) = \frac12 \frac{\partial}{\partial r} \log \left(\det\left(\gamma_{AB}\right)\right) = \frac{\partial}{\partial r} \log \left(\det\left(\gamma_{AB}\right)\right)^\frac12.
\]
Therefore the expansion yields the variation of the area element of the spacelike  $2$-dimensional surfaces $\exp(r,\psi_u(\Sigma))$. More importantly for our purposes, we see that a singularity of the expansion indicates a zero of $\det\left(\gamma_{AB}\right)$, i.e.~a conjugate point to $\psi_u(\Sigma)$.

\begin{Prop}\label{conjugate_null}
Let $(M,g)$ be a globally hyperbolic spacetime satisfying the null energy condition, $S \subset M$ a Cauchy hypersurface with future-pointing unit normal vector field $N$, $\Sigma \subset S$ a compact $2$-dimensional submanifold with unit normal vector field $n$ in $S$ and $p \in \Sigma$ a point where $\theta=\theta_0 < 0$. Then the null geodesic $c_p$ with initial condition $N_p+n_p$ contains at least a point conjugate to $\Sigma$, at an affine parameter distance of at most $-\frac{2}{\theta_0}$ to the future of $\Sigma$ (assuming that it can be extended that far).
\end{Prop}

\begin{proof}
Since $(M,g)$ satisfies the null energy condition, we have $R_{rr}=R_{\mu\nu}\left(\frac{\partial}{\partial r}\right)^\mu\left(\frac{\partial}{\partial r}\right)^\nu \geq 0$. Consequently,
\[
\frac{\partial \theta}{\partial r} + \gamma^{BC} \gamma^{AD} \beta_{CA} \beta_{DB} \leq 0.
\]
Choosing an orthonormal basis (where $\gamma^{AB}=\delta_{AB}$), and using the inequality
\[
(\tr A)^2 \leq n \tr(A^tA)
\]
for square $n \times n$ matrices, it is easy to show that
\[
\gamma^{BC} \gamma^{AD} \beta_{CA} \beta_{DB} = \beta_{BA} \beta_{AB} = \tr\left((\beta_{AB}) (\beta_{AB})^t\right) \geq \frac12 \theta^2.
\]
Consequently $\theta$ must satisfy
\[
\frac{\partial \theta}{\partial r} + \frac12 \theta^2 \leq 0.
\] 
Integrating this inequality yields
\[
\frac1{\theta} \geq \frac1{\theta_0} + \frac{r}2,
\]
and hence $\theta$ must blow up at a value of $r$ no greater than $-\frac{2}{\theta_0}$.
\end{proof}

We define the {\bf chronological future} and the {\bf causal future} of the compact surface $\Sigma$ as
\[
I^+(\Sigma)=\bigcup_{p \in \Sigma} I^+(p) \quad \text{ and } \quad J^+(\Sigma)=\bigcup_{p \in \Sigma} J^+(p)
\]
(with similar definitions for the {\bf chronological past} and the {\bf causal past} of $\Sigma$). It is clear that $I^+(\Sigma)$, being the union of open sets, is itself open, and also that $J^+(\Sigma) \subset \overline{I^+(\Sigma)}$ and $I^+(\Sigma) = \inte J^+(\Sigma)$. On the other hand, it is easy to generalize Proposition~\ref{compact} (and consequently Corollary~\ref{closed_compact}) to the corresponding statements with compact surfaces replacing points. In particular, $J^+(\Sigma)$ is closed. Therefore
\[
\partial J^+(\Sigma) = \partial I^+(\Sigma) = J^+(\Sigma) \setminus I^+(\Sigma),
\]
and so, by a straightforward generalization of Corollary~\ref{null_geodesic} in Chapter~\ref{chapter3}, every point in this boundary can be reached from a point in $\Sigma$ by a future-directed null geodesic. Moreover, this geodesic must be orthogonal to $\Sigma$. Indeed, at $\Sigma$ we have
\[
\frac{\partial}{\partial u} = N \quad \text{ and } \quad \frac{\partial}{\partial r} = N + n,
\]
and so the metric takes the form
\[
g = - du^2 - 2 du dr + \gamma_{AB} dx^A dx^B.
\]
If $c:I \subset \bbR \to M$ is a future-directed null geodesic with $c(0)\in\Sigma$, its initial tangent vector
\[
\dot{c}(0) = \dot{u} \frac{\partial}{\partial u} + \dot{r} \frac{\partial}{\partial r} + \dot{x}^A \frac{\partial}{\partial x^A} = (\dot{u} + \dot{r}) N + \dot{r} n + \dot{x}^A \frac{\partial}{\partial x^A}
\]
must satisfy
\[
\dot{u} (\dot{u} + 2 \dot{r}) = \gamma_{AB} \dot{x}^A \dot{x}^B.
\]
Since $c$ is future-directed we must have $\dot{u} + \dot{r}>0$. On the other hand, by choosing the unit normal to $\Sigma$ on $S$ to be either $n$ or $-n$, we can assume $\dot{r}\geq 0$. If $c$ is not orthogonal to $\Sigma$ we then have
\[
\gamma_{AB} \dot{x}^A \dot{x}^B > 0 \Rightarrow \dot{u} (\dot{u} + 2 \dot{r}) > 0 \Rightarrow \dot{u} > 0.
\]
Now the region where $u>0$ and $r \geq 0$ is clearly a subset of $I^+(\Sigma)$, since its points can be reached from $\Sigma$ by a sectionally smooth curve composed of an arc of timelike geodesic and an arc of null geodesic. Therefore, we see that if $c$ is not orthogonal to $\Sigma$ then $c(t)\in I^+(\Sigma)$ for all $t>0$.

Even future-directed null geodesics orthogonal to $\Sigma$ may eventually enter $I^+(\Sigma)$. A sufficient condition for this to happen is given in the following result.

\begin{Prop} \label{conjugate_null_I+}
Let $(M,g)$ be a globally hyperbolic spacetime, $S$ a Cauchy hypersurface with future-pointing unit normal vector field $N$, $\Sigma \subset S$ a compact $2$-dimensional submanifold with unit normal vector field $n$ in $S$, $p \in \Sigma$, $c_p$ the null geodesic through $p$ with initial condition $N_p+n_p$ and $q = c_p(r)$ for some $r>0$. If $c_p$ has a conjugate point between $p$ and $q$ then $q \in I^+(\Sigma)$.
\end{Prop}

\begin{proof}
We will offer only a sketch of the proof. Let $s$ be the first conjugate point along $c_p$ between $p$ and $q$. Since $q$ is conjugate to $p$, there exists another null geodesic $\gamma$ starting at $\Sigma$ which (approximately) intersects $c_p$ at $s$. The piecewise smooth null curve obtained by following $\gamma$ between $\Sigma$ and $s$, and $c_p$ between $s$ and $q$ is a causal curve but not a null geodesic. This curve can be easily smoothed while remaining causal and nongeodesic, and so by the generalization of Corollary~\ref{null_geodesic} in Chapter~\ref{chapter3} we have $q \in I^+(\Sigma)$.
\end{proof}

\begin{figure}[h!]
\begin{center}
\psfrag{p}{$p$}
\psfrag{q}{$q$}
\psfrag{s}{$s$}
\psfrag{S}{$S$}
\psfrag{Sg}{$\Sigma$}
\psfrag{g}{$\gamma$}
\psfrag{cp}{$c_p$}
\epsfxsize=.8\textwidth
\leavevmode
\epsfbox{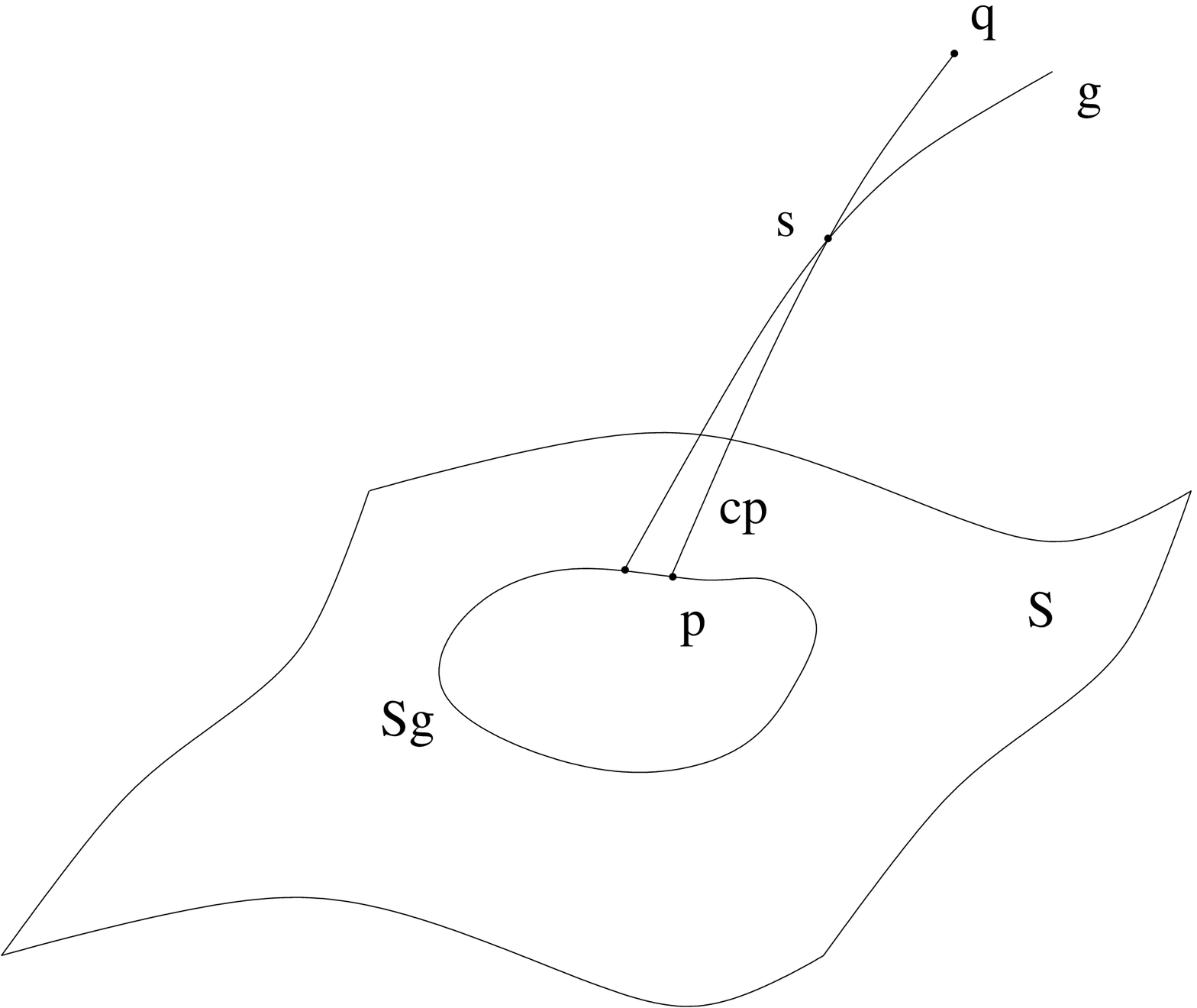}
\end{center}
\caption{Proof of Proposition~\ref{conjugate_null_I+}.} \label{proof3}
\end{figure}

\begin{Def}
Let $(M,g)$ be a globally hyperbolic spacetime and $S$ a Cauchy hypersurface with future-pointing unit normal vector field $N$. A compact $2$-dimensional submanifold $\Sigma \subset S$ with unit normal vector field $n$ in $S$ is said to be {\bf trapped} if the expansions $\theta^+$ and $\theta^-$ of the null geodesics with initial conditions $N + n$ and $N - n$ are both negative everywhere on $\Sigma$.
\end{Def}

We have now all the necessary ingredients to prove the Penrose singularity theorem.

\begin{Thm} ({\bf Penrose \cite{P65}}) \label{Penrose}
Let $(M,g)$ be a connected globally hyperbolic spacetime with a noncompact Cauchy hypersurface $S$, satisfying the null energy condition. If $S$ contains a trapped surface $\Sigma$ then $(M,g)$ is singular.
\end{Thm}

\begin{proof}
Let $t:M \to \bbR$ be a global time function such that $S=t^{-1}(0)$. The integral curves of $\grad t$, being timelike, intersect $S$ exactly once, and $\partial I^+(\Sigma)$ at most once. This defines a continuous injective map $\pi: \partial I^+(\Sigma) \to S$, whose image is open. Indeed, if $q = \pi(p)$, then all points is some neighborhood of $q$ are images of points in $\partial I^+(\Sigma)$, as otherwise there would be a sequence $q_k \in S$ with $q_k \to q$ such that the integral curves of $\grad t$ through $q_k$ would not intersect $\partial I^+(\Sigma)$. Letting $r_k$ be the intersections of these curves with the Cauchy hypersurface $t^{-1}(t(r))$, for some point $r$ to the future of $p$ along the integral line of $\grad t$, we would have $r_k \to r$, and so $r_k \in I^+(\Sigma)$ for sufficiently large $k$ (as $I^+(\Sigma)$ is open), leading to a contradiction.

Since $\Sigma$ is trapped (and compact), there exists $\theta_0 < 0$ such that the expansions $\theta^+$ and $\theta^-$ of the null geodesics orthogonal to $\Sigma$ both satisfy $\theta^+, \theta^- \leq \theta_0$. We will show that there exists a future-directed null geodesic orthogonal to $\Sigma$ which cannot be extended to an affine parameter greater than $r_0=-\frac{2}{\theta_0}$ to the future of $\Sigma$. Suppose that this was not so. Then, according to Proposition~\ref{conjugate_null}, any null geodesic orthogonal to $\Sigma$ would have a conjugate point at an affine parameter distance of at most $r_0$ to the future of $\Sigma$, after which it would be in $I^+(\Sigma)$, by Proposition~\ref{conjugate_null_I+}. Consequently, $\partial I^+(\Sigma)$ would be a (closed) subset of the compact set 
\[
\exp^+([0,r_0] \times \Sigma) \cup \exp^-([0,r_0] \times \Sigma)
\]
(where $\exp^+$ and $\exp^-$ refer to the exponential map constructed using the unit normals $n$ and $-n$), hence compact. Therefore the image of $\pi$ would also be compact, hence closed as well as open. Since $M$, and therefore $S$, are connected, the image of $\pi$ would be $S$, which would then be homeomorphic to $\partial I^+(\Sigma)$. But $S$ is noncompact by hypothesis, and we reach a contradiction.
\end{proof}

\begin{Remark}
It should be clear that $(M,g)$ is singular if the condition of existence of a trapped surface is replaced by the condition of existence of an {\bf anti-trapped surface}, that is, a compact surface $\Sigma \subset S$ such that the expansions of null geodesics orthogonal to $\Sigma$ are both positive. In this case, there exists a {\bf past-directed} null geodesic orthogonal to $\Sigma$ which cannot be extended to an affine parameter time greater than $r_0=\frac{2}{\theta_0}$ to the {\bf past} of $\Sigma$.
\end{Remark}

\begin{Example} \hspace{1cm}
\begin{enumerate}
\item
The region $\{ r < 2m \}$ of the Schwarzschild solution is globally hyperbolic, and satisfies the null energy condition (as $R_{\mu\nu}=0$). Since $r$ (or $-r$) is clearly a time function (depending on the choice of time orientation), it must increase (or decrease) along any future-pointing null geodesic, and therefore any sphere $\Sigma$ of constant $(t,r)$ is anti-trapped (or trapped). Since any Cauchy hypersurface is diffeomorphic to $\bbR \times S^2$, hence noncompact, we conclude from Theorem~\ref{Penrose} that the Schwarzschild solution is singular to past (or future) of $\Sigma$. Moreover, Theorem~\ref{Penrose} implies that this singularity is generic: any sufficiently small perturbation of the Schwarzschild solution satisfying the null energy condition will also be singular. Loosely speaking, once the collapse has advanced long enough, nothing can prevent the formation of a singularity.
\item
The FLRW models with $\alpha>0$ and $\Lambda = 0$ are globally hyperbolic, and satisfy the null energy condition. Moreover, radial null geodesics satisfy
\[
\frac{dr}{dt} = \pm \frac1a \sqrt{1-kr^2}.
\]
Therefore, if we start with a sphere $\Sigma$ of constant $(t,r)$ and follow the orthogonal null geodesics along the direction of increasing or decreasing $r$, we obtain spheres whose radii $ar$ satisfy
\[
\frac{d}{dt} (ar) = \dot{a}r+a\dot{r} = \dot{a}r \pm \sqrt{1-kr^2}.
\]
Assume that the model is expanding, with the big bang at $t=0$, and spatially noncompact (in particular $k \neq 1$). Then, for sufficiently small $t>0$, the sphere $\Sigma$ is anti-trapped, and hence Theorem~\ref{Penrose} guarantees that this model is singular to the past of $\Sigma$ (i.e.~there exists a big bang). Moreover, Theorem~\ref{Penrose} implies that this singularity is generic: any sufficiently small perturbation of an expanding, spatially noncompact FLRW model satisfying the null energy condition will also be singular. Loosely speaking, any expanding universe must have begun at a big bang.
\end{enumerate}
\end{Example}

\section{Exercises} \label{sec4.7}

\begin{enumerate}

\item
Let $g$ be a Lorentzian metric given in the Gauss Lemma form,
\[
g = - dt^2 + h_{ij} dx^i dx^j,
\]
and consider the geodesic congruence tangent to $\frac{\partial}{\partial t}$.
\begin{enumerate}
\item
Show that
\[
B_{ij} = \Gamma^0_{ij} = \frac12 \frac{\partial h_{ij}}{\partial t},
\]
that is,
\[
B = \frac12 \cL_{\frac{\partial}{\partial t}} g = K,
\]
where $K$ is the second fundamental form of the hypersurfaces of constant $t$.
\item
Conclude that the expansion of the congruence of the galaxies in a FLRW model is
\[
\theta = \frac{3\dot{a}}{a} = 3H.
\]
\end{enumerate}

\item
Let $(M,g)$ be a Lorentzian manifold.
\begin{enumerate}
\item
Use the formula for the Lie derivative of a tensor,
\[
\hspace{2cm} (\cL_X g) (Y,Z) = X \cdot g(Y,Z) - g(\cL_X Y, Z) - g(Y, \cL_X Z)
\]
to show that
\[
(\cL_X g) (Y,Z) = g(\nabla_Y X, Z) + g(Y, \nabla_Z X).
\]
\item
Show that this formula can be written as
\[
(\cL_X g)_{\mu\nu} = \nabla_\mu X_\nu + \nabla_\nu X_\mu.
\]
\item
Suppose that $X$ is a Killing vector field, i.e.~$\cL_X g=0$. Use the Killing equation
\[
\nabla_\mu X_\nu + \nabla_\nu X_\mu = 0
\]
to show that $X$ is a solution of the Jacobi equation. Give a geometric interpretation of this fact.
\end{enumerate}

\item
Let $T$ be a diagonalizable energy-momentum tensor, that is, $\left(T_{\mu\nu}\right) = \diag(\rho, p_1,p_2, p_3)$ on some orthonormal frame $\{E_0,E_1,E_2,E_3\}$. Show that:
\begin{enumerate}
\item
$T$ satisfies the SEC if and only if $\rho+\sum_{i=1}^3p_i\geq 0$ and $\rho+p_i\geq 0$ ($i=1,2,3$).
\item
$T$ satisfies the WEC if and only if $\rho\geq 0$ and $\rho+p_i\geq 0$ ($i=1,2,3$).
\item
$T$ satisfies the DEC if and only if $\rho\geq |p_i|$ ($i=1,2,3$).
\item
$T$ satisfies the NEC if and only if $\rho+p_i\geq 0$ ($i=1,2,3$).
\item
The first three conditions are independent except that the DEC implies the WEC.
\item
The first three conditions imply the NEC.
\end{enumerate}

\item
Let $(M,g)$ be the globally hyperbolic Lorentzian manifold corresponding to the exterior region of the Schwarzschild solution, that is, $M=\bbR \times \left(\bbR^3 \setminus \overline{B_{2m}(0)}\right)$ and
\[
\hspace{2cm} g = - \left( 1 - \frac{2m}r \right) dt^2 + \left( 1 - \frac{2m}r \right)^{-1} dr^2 + r^2 \left(d\theta^2 + \sin^2 \theta d \varphi^2 \right)
\]
(with $m>0$).
\begin{enumerate}
\item
Show that for any $r_0>2m$ the curve
\[
c(t)=\left(t,r_0,\frac\pi2,\sqrt{\frac{m}{{r_0}^3}} t\right)
\]
is a timelike, null or spacelike geodesic, according to whether $r_0>3m$, $r_0=3m$ or $r_0<3m$.
\item
Argue that the point $q=\left(\pi\sqrt{\frac{{r_0}^3}{m}},r_0,\frac\pi2,\pi\right)$ is conjugate to the point $p=\left(0,r_0,\frac\pi2,0\right)$ along $c$ (note that this can be done without solving the Jacobi equation).
\item
Show explicitly that if $r_0>3m$ then $c$ stops being maximizing for $t>\pi\sqrt{\frac{{r_0}^3}{m}}$.
\end{enumerate}

\item
Let $(M,g)$ be a globally hyperbolic spacetime and $p,q \in M$ with $q \in I^+(p)$. Show that among all timelike curves connecting $p$ to $q$ there exists a timelike curve with maximal length, which is a timelike geodesic.

\item
Use ideas similar to those leading to the proof Hawking's singularity theorem to prove {\bf Myers's theorem}: if $(M,\langle\cdot,\cdot\rangle)$ is a complete Riemannian manifold whose Ricci curvature satisfies $R_{\mu\nu} X^\mu X^\nu \geq \varepsilon g_{\mu\nu} X^\mu X^\nu$ for some $\varepsilon > 0$ then $M$ is compact. Can these ideas be used to prove a singularity theorem in Riemannian geometry?

\item
Explain why Hawking's singularity theorem does not apply to each of the following geodesically complete Lorentzian manifolds:
\begin{enumerate}
\item
Minkowski's spacetime;
\item
Einstein's universe;
\item
de Sitter's universe;
\item
Anti-de Sitter spacetime.
\end{enumerate}

\item
Consider the metric
\[
\hspace{2cm} ds^2 = \alpha \, du^2 - 2 \, du \, dr + 2 \beta_A \, du \, dx^A + \gamma_{AB} \, dx^A dx^B.
\]
\begin{enumerate}
\item
Show that the Christoffel symbols satisfy
\[
\hspace{2cm} \Gamma^u_{ur} = \Gamma^u_{rr} = \Gamma^u_{rA} = \Gamma^r_{rr} = \Gamma^A_{rr} = 0 \quad \text{ and } \quad \Gamma^A_{rB} = \gamma^{AC} \beta_{CB},
\]
where $(\gamma^{AB})=(\gamma_{AB})^{-1}$ and $\beta_{AB} = \frac12 \frac{\partial \gamma_{AB}}{\partial r}$. 
\item
Conclude that
\[
\hspace{2cm} R_{rr} = - \frac{\partial}{\partial r} \left( \gamma^{AB} \beta_{AB}\right) - \gamma^{AB} \gamma^{CD} \beta_{AC} \beta_{BD}.
\]
\end{enumerate}

\item
Let $(M,g)$ be a globally hyperbolic spacetime with Cauchy hypersurfaces $S_0$ and $S_1$ satisfying $S_1 \subset D^+(S_0)$, and $\Sigma \subset S_1$ a compact surface. Show that:
\begin{enumerate}
\item
$D^+(S_0)\cap J^-(\Sigma)$ is compact;
\item
$J^-(\Sigma)$ is closed.
\end{enumerate}

\item
Explain why Penrose's singularity theorem does not apply to each of the following geodesically complete Lorentzian manifolds:
\begin{enumerate}
\item
Minkowski's spacetime;
\item
Einstein's universe;
\item
de Sitter's universe;
\item
Anti-de Sitter spacetime.
\end{enumerate}

\end{enumerate}


\chapter{Cauchy problem} \label{chapter5}

In this chapter we discuss the Cauchy problem for the Einstein field equations, following \cite{W84}. We start by studying the Klein-Gordon equation, as a prototypical wave equation, and the the Maxwell equations, where the issues of constraints on the inital data and gauge freedom also arise. We then sketch the proof of the Choquet-Bruhat theorem and discuss the Lichnerowicz method for solving the constraint equations. See \cite{Ringstrom09} for a more complete discussion.

\section{Divergence theorem} \label{sec5.1}

It is possible to define the divergence of a vector field on any orientable manifold where a volume form has been chosen.

\begin{Def}
Let $M$ be an orientable $n$-dimensional manifold with volume form $\epsilon$, and let $X$ be a vector field on $M$. The {\bf divergence} of $X$ is the function $\dive X$ such that
\[
d(X \contr \epsilon) = (\dive X) \epsilon.
\]
\end{Def}

The following result in a straightforward application of the Stokes theorem.

\begin{Thm} ({\bf Divergence theorem})
If $M$ is a compact orientable $n$-dimensional manifold with boundary $\partial M$ then
\[
\int_M (\dive X) \epsilon = \int_{\partial M} X \contr \epsilon.
\]
\end{Thm}

\begin{Prop}
If $(M,g)$ is a pseudo-Riemannian manifold with Levi-Civita connection $\nabla$ then
\[
\dive X = \nabla_\mu X^\mu.
\]
\end{Prop}

\begin{proof}
Take local coordinates such that $X=\partial_1$ around each point where $X$ does not vanish. Since
\[
\epsilon = \sqrt{|\det(g_{\mu\nu})|} \, dx^1 \wedge \ldots \wedge dx^n,
\]
we have
\[
d(X \contr \epsilon) = \cL_X \epsilon - X \contr d\epsilon = \partial_1 \log\left|\det(g_{\mu\nu})\right|^\frac12 \, \epsilon,
\]
where we used $d \epsilon = 0$. Since
\[
\partial_1 \log \left|\det A\right| = \tr (A^{-1} \partial_1 A)
\]
for any matrix-valued function $A$ in $M$, we have
\begin{align*}
d(X \contr \epsilon) & = \frac12 \left(g^{\mu\nu} \partial_1 g_{\mu\nu}\right) \epsilon = \frac12 \left(g^{\mu\nu} (\cL_X g)_{\mu\nu}\right) \epsilon \\
& = \frac12 g^{\mu\nu} (\nabla_\mu X_\nu + \nabla_\nu X_\mu ) \epsilon = (\nabla_\mu X^\mu) \epsilon.
\end{align*}
This formula can be easily extended to the set of zeros of $X$: by continuity on the boundary, and trivially in the interior.
\end{proof}

Assume now that $(M,g)$ is a compact orientable $n$-dimensional Lorentzian manifold with boundary $\partial M$, and let $\{E_1, \ldots, E_n \}$ be a positive orthonormal frame with $E_1 = N$ the outward unit normal vector on $\partial M$. The volume element is
\[
\epsilon = - E_1^\sharp \wedge \ldots \wedge E_n^\sharp,
\]
and the volume element of $\partial M$ with the induced orientation is
\[
\sigma = \pm E_2^\sharp \wedge \ldots \wedge E_n^\sharp
\]
(according to whether $N$ is timelike or spacelike). Therefore we have on $\partial M$
\[
X \contr \epsilon = \pm \langle X, N \rangle \sigma,
\]
implying that the divergence theorem can be written as
\[
\int_M (\dive X) \epsilon = \int_{\partial M} \langle X, n \rangle \sigma,
\]
where $n$ is the outward unit normal vector in the points where it is spacelike, and is the inward unit normal vector in the points where it is timelike (Figure~\ref{normal}).

\begin{figure}[h!]
\begin{center}
\psfrag{ntimelike}{$n$ timelike}
\psfrag{nnull}{$n$ null}
\psfrag{nspacelike}{$n$ spacelike}
\epsfxsize=.6\textwidth
\leavevmode
\epsfbox{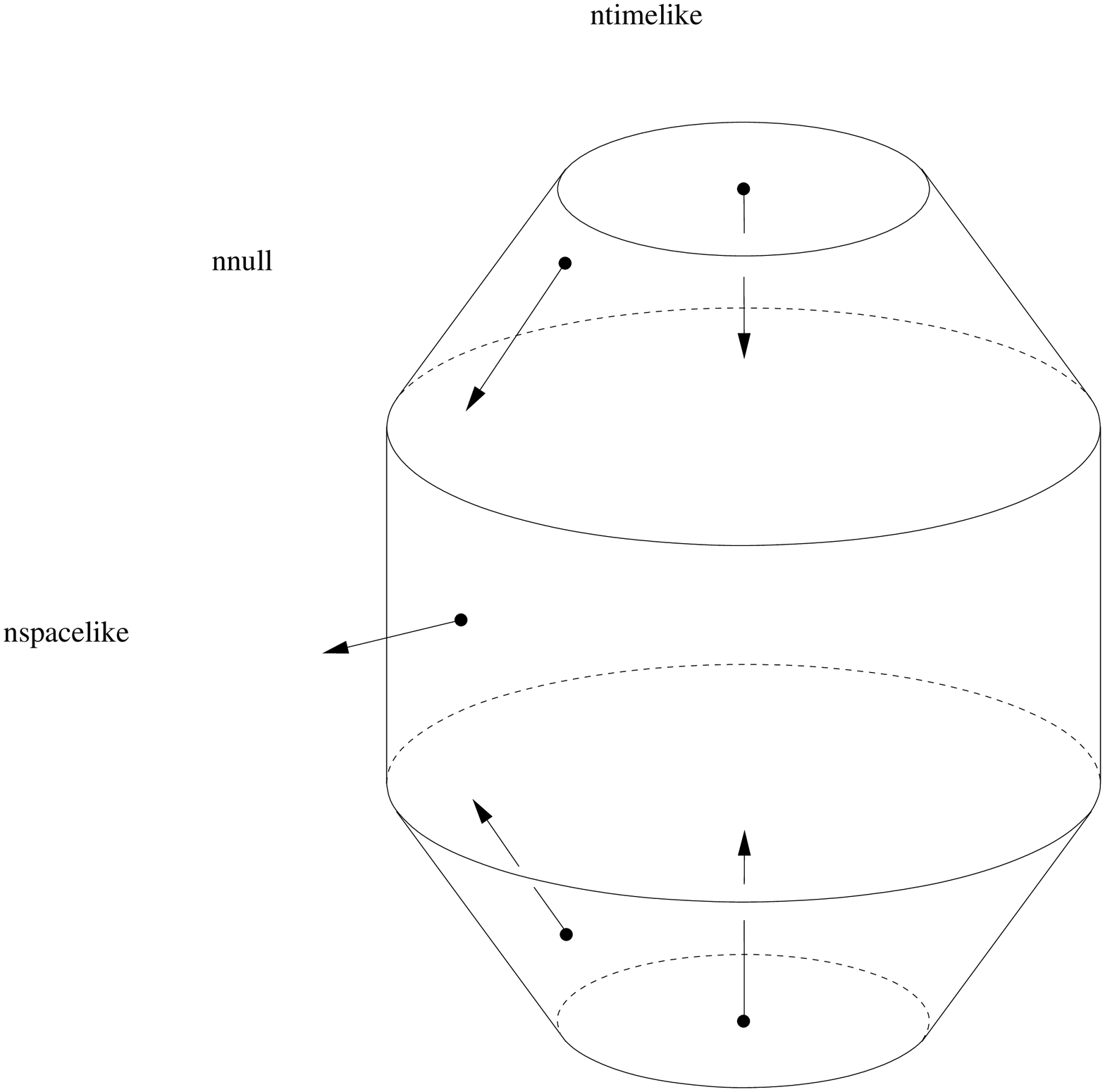}
\end{center}
\caption{Normal vector for the divergence theorem on a Lorentzian manifold.} \label{normal}
\end{figure}

It may happen that $\partial M$ has points where the normal is null, hence tangent to $\partial M$. In that case we choose the orthonormal frame such that
\[
N = E_1 + E_2
\]
is a null normal, with $E_1$ timelike and pointing outwards, and $E_2$ spacelike and (necessarily) pointing inwards. Then
\[
\epsilon = - N^\sharp \wedge E_2^\sharp \wedge \ldots \wedge E_n^\sharp
\]
and so
\[
X \contr \epsilon = - \langle X, N \rangle \sigma,
\]
where
\[
\sigma = E_2^\sharp \wedge \ldots \wedge E_n^\sharp
\]
is a volume element on $\partial M$ (compatible with the induced orientation, as $E_1$ points outwards). So we must choose in this case $n = -N$, that is, we must use the normal whose timelike component points inwards and whose spacelike component points outwards (Figure~\ref{normal}). Note that the magnitude (but not the sign) of the volume element $\sigma$ will depend on the choice of $n$.

\section{Klein-Gordon equation} \label{sec5.2}

Let $(M,g)$ be a Lorentzian manifold. A smooth function $\phi:M\to\bbR$ is a solution of the {\bf Klein-Gordon equation} if it satisfies
\[
\Box \phi - m^2 \phi = 0 \Leftrightarrow \nabla^\mu \partial_\mu \phi - m^2 \phi = 0.
\]
For reasons that will become clear in Chapter~\ref{chapter6}, we define the {\bf energy-momentum tensor} associated to this equation as
\[
T_{\mu\nu} = \partial_\mu \phi \, \partial_\nu \phi - \frac12 g_{\mu\nu} \left( \partial_\alpha \phi \, \partial^\alpha \phi + m^2 \phi^2 \right).
\]
If $\phi$ is a solution of the Klein-Gordon equation then
\begin{align*}
\nabla^\mu T_{\mu\nu} & = \Box \phi \, \partial_\nu \phi +  \partial_\mu \phi \, \nabla^\mu \partial_\nu \phi - \partial_\alpha \phi \, \nabla_\nu \partial^\alpha \phi - m^2 \phi \, \partial_\nu \phi \\
& = (\Box \phi - m^2 \phi) \partial_\nu \phi = 0
\end{align*}
Moreover, if $(M,g)$ is time-oriented then it is possible to prove that $T$ satisfies the dominant energy condition, that is, $T_{\mu\nu}X^\nu$ corresponds to a past-pointing causal vector whenever $X$ is a future-pointing causal vector. Assume that $X$ is a future-pointing timelike Killing vector field, and define $Y^\mu = T^{\mu\nu}X_\nu$. Then $Y$ is a past-pointing causal vector field satisfying
\[
\nabla_\mu Y^\mu = T^{\mu\nu}\nabla_\mu X_\nu = 0
\]
(where we used the Killing equation $\nabla_{(\mu}X_{\nu)}=0$ and the symmetry of $T$).

Let us now focus on the case when $(M,g)$ is flat Minkowski spacetime and $X = \frac{\partial}{\partial t}$. Consider the Cauchy hypersurface $S_0 = \{ t=t_0 \}$, and let $B_0$ be the ball of radius $R$ in that hypersurface:
\[
B_0 = \{ t=t_0, x^2 + y^2 + z^2 \leq R^2 \}.
\]
Let $S_1 = \{ t = t_1 \}$ be another Cauchy hypersurface, and consider the ball $B_1 = D(B_0) \cap S_1$ in that hypersurface (Figure~\ref{domain}). By the divergence theorem we have
\[
\int_{B_0} \langle Y, X \rangle + \int_C \langle Y, n \rangle + \int_{B_1} \langle Y, -X \rangle = 0,
\]
where $C$ is the null portion of $\partial D(B_0)$ between $S_0$ and $S_1$ and $n$ is a past-pointing normal. Since $Y$ is a past-pointing causal vector we have $\langle Y, n \rangle \leq 0$, and so
\begin{equation} \label{ineq}
\int_{B_1} \langle Y, X \rangle \leq \int_{B_0} \langle Y, X \rangle.
\end{equation}

\begin{figure}[h!]
\begin{center}
\psfrag{n}{$n$}
\psfrag{S0}{$S_0$}
\psfrag{B0}{$B_0$}
\psfrag{B1}{$B_1$}
\psfrag{C}{$C$}
\psfrag{nspacelike}{$n$ spacelike}
\epsfxsize=1.0\textwidth
\leavevmode
\epsfbox{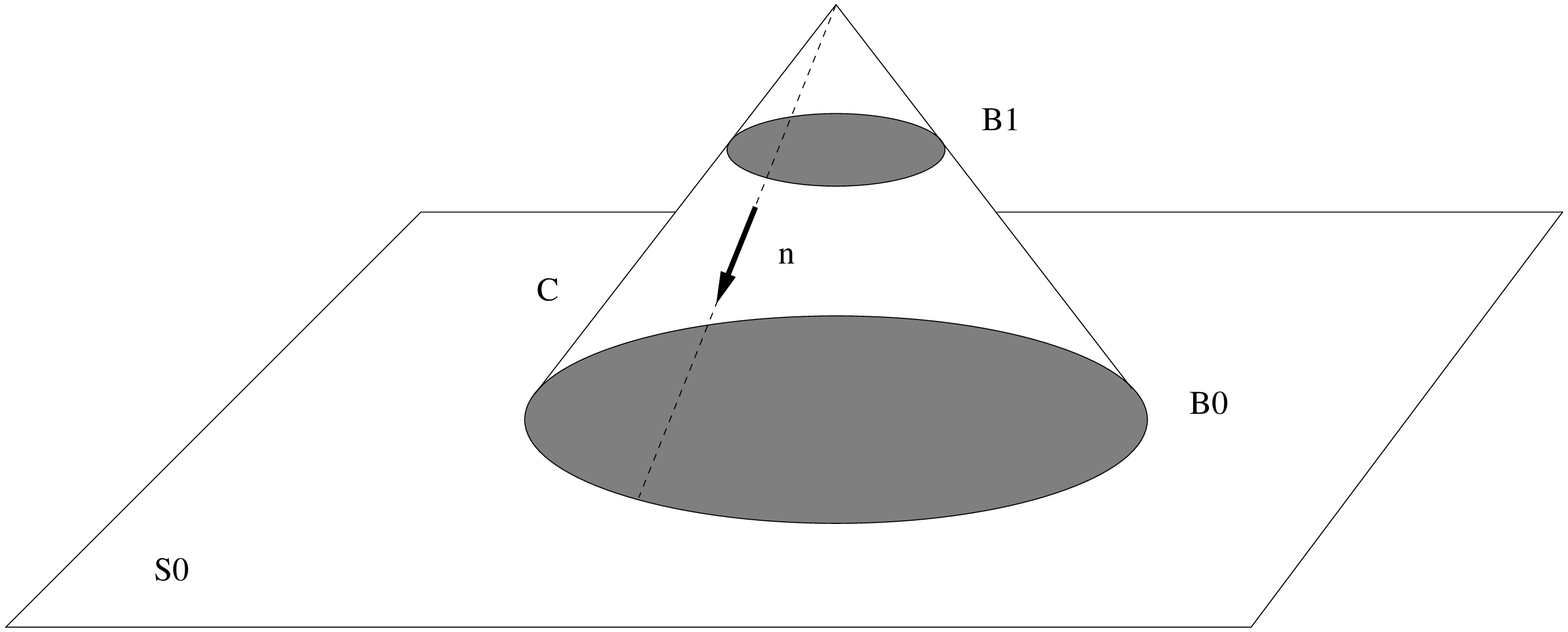}
\end{center}
\caption{Proof of Proposition~\ref{unicity}.} \label{domain}
\end{figure}

Note that
\begin{align*}
\langle Y, X \rangle & = T_{\mu\nu}X^\mu X^\nu = (X \cdot \phi)^2 + \frac12 \left( \partial_\alpha \phi \, \partial^\alpha \phi + m^2 \phi^2 \right) \\
& = \frac12 \left[ (\partial_0 \phi)^2 + (\partial_x \phi)^2 + (\partial_y \phi)^2 + (\partial_z \phi)^2 + m^2 \phi^2 \right].
\end{align*}
We conclude immediately that if $\phi$ is a solution of the Klein-Gordon equation and $\phi=\partial_0\phi=0$ in $B_0$ then $\phi=0$ in $B_1$, and indeed in $D(S_0)$ (since $t_1$ is arbitrary). Because the Klein-Gordon equation is linear we can then deduce the following result.

\begin{Prop}\label{unicity}
Given two smooth functions $\phi_0,\psi_0:S_0 \to \bbR$, there exists at most a solution $\phi$ of the Klein-Gordon equation satisfying $\phi=\phi_0$ and $\partial_0\phi = \psi_0$ on $S_0$.
\end{Prop}

Given a smooth function $\phi:M \to \bbR$ and some set $A \subset \bbR^4$ we define the {\bf Sobolev norm}
\[
\| \phi \|^2_{H^1(A)} = \int_A \left[ \phi^2 + (\partial_0 \phi)^2 +  (\partial_x \phi)^2 + (\partial_y \phi)^2 + (\partial_z \phi)^2 \right].
\]
In this definition $A$ can either be an open set or a submanifold (in which case we use the induced volume form in the integral).

More generally, for each $k \in \bbN_0$ we define the Sobolev norms
\[
\| \phi \|^2_{H^k(A)} = \int_A \sum_{|\alpha| \leq k} (\partial^\alpha \phi)^2,
\]
where $\alpha = (\alpha_0, \alpha_1, \alpha_2, \alpha_3) \in {\bbN_0}^4$, $|\alpha| = \alpha_0 + \alpha_1 + \alpha_2 + \alpha_3$, $\partial = (\partial_0,\partial_x,\partial_y,\partial_z)$ and $\partial^\alpha = \partial_0^{\alpha_0} \partial_x^{\alpha_1} \partial_y^{\alpha_2} \partial_z^{\alpha_3}$. Similarly, given a smooth function $\phi_0:S_0 \to \bbR$ we define
\[
\| \phi_0 \|^2_{H^k(B_0)} = \int_{B_0} \sum_{|\alpha| \leq k} (D^\alpha \phi_0)^2,
\]
where now $D = (\partial_x,\partial_y,\partial_z)$. Inequality \eqref{ineq} can then be written as
\[
\| \phi \|^2_{H^1(B_1)} \leq C \| \phi \|^2_{H^1(B_0)} \leq C \| \phi_0 \|^2_{H^1(B_0)} + C \| \psi_0 \|^2_{H^0(B_0)},
\]
where $C>0$ is a generic positive constant (not always the same). Integrating this inequality in $t$ from $t_0 - R$ to $t_0 + R$ we obtain
\[
\| \phi \|^2_{H^1(D(B_0))} \leq C \| \phi_0 \|^2_{H^1(B_0)} + C \| \psi_0 \|^2_{H^0(B_0)}.
\]
On $S_0$, all spatial partial derivatives of $\phi$ and $\partial_0\phi$ are given by partial derivatives of $\phi_0$ and $\psi_0$. On the other hand, from the Klein-Gordon equation we have
\[
\partial_0^2 \phi_{|_{S_0}} = \partial_x^2 \phi_0 + \partial_y^2 \phi_0 + \partial_z^2 \phi_0 - m^2 \phi_0,
\]
and so, by partial differentiation, we can obtain all spatial partial derivatives of $\partial^2_0\phi$ on $S_0$ from partial derivatives of $\phi_0$. Differentiating the Klein-Gordon equation with respect to $t$ yields
\[
\partial_0^3 \phi_{|_{S_0}} = \partial_x^2 \psi_0 + \partial_y^2 \psi_0 + \partial_z^2 \psi_0 - m^2 \psi_0,
\]
and it should be clear that all partial derivatives of $\phi$ on $S_0$ are given by partial derivatives of $\phi_0$ and $\psi_0$. Note from the general partial derivative of the Klein-Gordon equation,
\[
\partial_0^2 \partial^\alpha\phi = \partial_x^2 \partial^\alpha\phi + \partial_y^2 \partial^\alpha\phi + \partial_z^2\partial^\alpha \phi - m^2 \partial^\alpha\phi,
\]
that $\partial^\alpha\phi$ also satisfies the Klein-Gordon equation, and so
\begin{align*}
& \| \partial^\alpha\phi \|^2_{H^1(B_1)} \leq C \| \partial^\alpha\phi \|^2_{H^1(B_0)} \Rightarrow \\
& \| \phi \|^2_{H^k(B_1)} \leq C \| \phi_0 \|^2_{H^k(B_0)} + C \| \psi_0 \|^2_{H^{k-1}(B_0)} \Rightarrow \\
& \| \phi \|^2_{H^k(D(B_0))} \leq C \| \phi_0 \|^2_{H^k(B_0)} + C \| \psi_0 \|^2_{H^{k-1}(B_0)} ,
\end{align*}
whence
\[
\| \phi \|_{H^k(D(B_0))} \leq C \| \phi_0 \|_{H^k(B_0)} + C \| \psi_0 \|_{H^{k-1}(B_0)} .
\]

Another important norm is the {\bf supremum norm},
\[
\| \phi \|_{C^0(A)} = \sup_{p \in A} |\phi(p)|.
\]
More generally, we define the norms
\[
\| \phi \|_{C^k(A)} = \sum_{|\alpha| \leq k} \sup_{p \in A} |\partial^\alpha\phi(p)|.
\]

\begin{Def}
A set $A \subset \bbR^n$ is said to satisfy the {\bf interior cone condition} if there exists a (closed) cone of height $h>0$ and solid angle $\Omega > 0$ at the vertex such that for each point $p \in A$ it is possible to map the cone isometrically into $A$ in such a way that the vertex is mapped to $p$.
\end{Def}

\begin{Thm}({\bf Sobolev inequality})
If $A \subset \bbR^n$ satisfies the interior cone condition and $k > \frac{n}2$ then for any smooth function $f: A \to \bbR$
\[
\| f \|_{C^0(A)} \leq C \| f \|_{H^k(A)}.
\]
\end{Thm}

\begin{proof}
Exercise.
\end{proof}

For a solution of the Klein-Gordon equation we then have
\[
\| \phi \|_{C^0(D(B_0))} \leq \| \phi \|_{H^3(D(B_0))} \leq C \| \phi_0 \|_{H^3(B_0)} + C \| \psi_0 \|_{H^{2}(B_0)}.
\]
More generally,
\[
\| \partial^\alpha \phi \|_{C^0(D(B_0))} \leq \| \partial^\alpha \phi \|_{H^3(D(B_0))} \leq C \| \phi_0 \|_{H^{m+3}(B_0)} + C \| \psi_0 \|_{H^{m+2}(B_0)},
\]
where $m=|\alpha|$. We conclude that
\[
\| \phi \|_{C^m(D(B_0))} \leq C \| \phi_0 \|_{H^{m+3}(B_0)} + C \| \psi_0 \|_{H^{m+2}(B_0)},
\]
whence
\[
\| \phi \|_{C^m(D(B_0))} \leq C \| \phi_0 \|_{C^{m+3}(B_0)} + C \| \psi_0 \|_{C^{m+2}(B_0)}.
\]

\begin{Thm}
Given initial data $\phi_0, \psi_0: S_0 \to \bbR$ for the Klein-Gordon equation, there exists a unique smooth solution $\phi$ satisfying $\phi=\phi_0$ and $\partial_0\phi = \psi_0$ on $S_0$. Moreover, if $B_0 \subset S_0$ is a ball then the solution in $D(B_0)$ depends only on the initial data in $B_0$, and the map
\[
C^{m+3}(B_0) \times C^{m+2}(B_0) \ni (\phi_0, \psi_0) \mapsto \phi \in C^m(D(B_0))
\]
is continuous.
\end{Thm}

\begin{proof}
We just have to prove existence of solution. To do that, we note that if $\phi$ is a solution then we know all its partial derivatives on $S_0$. Therefore, we can construct a power series for $\phi$ around each point of $S_0$. If $\phi_0$ and $\psi_0$ are analytic then the {\bf Cauchy-Kowalewski theorem} guarantees that these series converge, and so there exists an analytic solution $\phi$. If $\phi_0$ and $\psi_0$ are smooth then there exist sequences $\phi_{0,n}$ and $\psi_{0,n}$ of analytic functions which converge to $\phi_0$ and $\psi_0$ in all spaces $C^{m}(B_0)$. The corresponding analytic solutions $\phi_n$ of the Klein-Gordon equation thus form a Cauchy sequence in all spaces $C^m(D(B_0))$, and hence must converge in all these spaces to some function $\phi$. This function is therefore smooth, and, passing to the limit, must satisfy the Klein-Gordon equation in all sets $(D(B_0))$.
\end{proof}

Using similar tecnhiques, it is possible to prove a much stronger result.

\begin{Thm}
Let $(M,g)$ be a globally hyperbolic spacetime with a Cauchy hypersurface $S$ and future unit normal $N$. Then the {\bf linear, diagonal second order hyperbolic system}
\[
g^{\mu\nu} \nabla_\mu \partial_\nu \phi_i + A^\mu_{ij} \partial_\mu \phi_j + B_{ij} \phi_j + C_i = 0 \qquad (i=1, \ldots, n),
\]
where $A^\mu_{ij}$, $B_{ij}$ and $C_i$ are smooth and $\nabla$ is any connection, yields a well-posed Cauchy problem with initial data in $S$. More precisely, given smooth initial data $(\phi_1, \ldots, \phi_n, N \cdot \phi_1, \ldots, N \cdot \phi_n)$ on $S$ there exists a unique smooth solution of the system, defined in $M$. Moreover, the solutions depend continuously on the initial data, and if two initial data sets coincide on some closed subset $B \subset S$ then the corresponding solutions coincide in $D(B)$.
\end{Thm}

To solve the Cauchy problem for the Einstein equations, we will need to solve more complicated systems of hyperbolic equations.

\begin{Thm} \label{hyperbolic2}
Consider the {\bf quasi-linear, diagonal second order hyperbolic system}
\[
g^{\mu\nu}(x,\phi,\partial\phi) \nabla_\mu \partial_\nu \phi_i = F_i(x,\phi,\partial\phi) \qquad (i=1, \ldots, n),
\]
where $g^{\mu\nu}$ and $F_i$ are smooth and $\nabla$ is any connection on some manifold $M$. Let $(\phi_0)_1, \ldots, (\phi_0)_n$ be a solution of this system, and define $(g_0)^{\mu\nu}=g^{\mu\nu}(x,\phi_0,\partial\phi_0)$. Assume that $(M,g)$ is globally hyperbolic, and let $S$ be a Cauchy hypersurface. Then the system above yields a well-posed Cauchy problem with initial data in $S$, in the following sense: given initial data in $S$ sufficiently close to the initial data for  $(\phi_0)_1, \ldots, (\phi_0)_n$ there exists an open neighborhood $V$ of $S$ such that the system has a unique solution in $V$, and $(V,g(x,\phi,\partial\phi))$ is globally hyperbolic. Moreover, the solutions depend continuously on the initial data, and if two initial data sets coincide on some closed subset $B \subset S$ then the corresponding solutions coincide in $D(B)$.
\end{Thm}

\begin{proof}
The idea of the proof is to start with the linear hyperbolic system
\[
g^{\mu\nu}(x,\phi_0,\partial\phi_0) \nabla_\mu \partial_\nu \phi_i = F_i(x,\phi_0,\partial\phi_0) \qquad (i=1, \ldots, n),
\]
which by the previous theorem has a unique solution $\phi_1$ close to $\phi_0$. Because of this, there exists a neighborhood $V_1$ of $S$ such that $(V_1,g(x,\phi_1,\partial\phi_1))$ is globally hyperbolic with Cauchy hypersurface $S$, and so the system
\[
g^{\mu\nu}(x,\phi_1,\partial\phi_1) \nabla_\mu \partial_\nu \phi_i = F_i(x,\phi_1,\partial\phi_1) \qquad (i=1, \ldots, n),
\]
again has a unique solution $\phi_2$ close to $\phi_1$. Iterating this procedure we obtain a sequence $\phi_n$ which can then be shown to converge to the unique solution of the quasi-linear hyperbolic system.
\end{proof}

\section{Maxwell's equations: constraints and gauge} \label{sec5.3}

As a warm-up problem to solving the Einstein field equations we consider the considerably easier problem of solving the Maxwell equations without sources in flat Minkowski spacetime. These equations can be split into what we shall call {\bf constraint equations},
\[
\begin{cases}
\dive {\bf E} = 0 \\
\dive {\bf B} = 0 
\end{cases},
\]
and {\bf evolution equations}:
\[
\begin{cases}
\displaystyle \frac{\partial {\bf E}}{\partial t} = \curl {\bf B} \\ \\
\displaystyle \frac{\partial {\bf B}}{\partial t} = - \curl {\bf E}
\end{cases}.
\]
One expects that the evolution equations completely determine ${\bf E}(t, {\bf x})$ and ${\bf B}(t, {\bf x})$ from initial data
\[
\begin{cases}
{\bf E}(0, {\bf x}) = {\bf E}_0({\bf x}) \\
{\bf B}(0, {\bf x}) = {\bf B}_0({\bf x})
\end{cases}.
\]
This initial data, however, is not completely free: it must satisfy the constraint equations
\[
\dive {\bf E}_0 = \dive {\bf B}_0 = 0.
\]
This suffices to guarantee that the constraint equations are satisfied by ${\bf E}(t, {\bf x})$ and ${\bf B}(t, {\bf x})$, since they are preserved by the evolution: for instance,
\[
\frac{\partial}{\partial t} (\dive {\bf E}) = \dive \left( \frac{\partial {\bf E}}{\partial t} \right) = \dive (\curl {\bf B}) = 0.
\]
As we will see, solving the Einstein field equations will also require splitting them into constraint equations and evolution equations. Another issue that will have to be dealt with, gauge freedom, also occurs when solving the Maxwell equations by using the electromagnetic {\bf gauge potentials}. To do so, we note that two of the Maxwell equations (those which do not admit sources) are equivalent to the existence of a vector potential ${\bf A}$ and a scalar potential $\phi$ in terms of which ${\bf B}$ and ${\bf E}$ can be written:
\[
\begin{cases}
\dive {\bf B} = 0 \\
\displaystyle \curl {\bf E} = - \frac{\partial {\bf B}}{\partial t} 
\end{cases}
\Leftrightarrow
\begin{cases}
{\bf B} = \curl {\bf A} \\
\displaystyle {\bf E} = - \grad \phi - \frac{\partial {\bf A}}{\partial t} 
\end{cases}.
\]
These potentials, however, are nonunique: given any smooth function $\chi$, the potentials
\[
\begin{cases}
{\bf A}' = {\bf A} + \grad \chi \\
\displaystyle \phi' = \phi - \frac{\partial \chi}{\partial t} 
\end{cases}
\]
yield the same fields ${\bf B}$ and ${\bf E}$ (${\bf A}'$ and $\phi'$ are said to be related to ${\bf A}$ and $\phi$ by a {\bf gauge transformation}). The remaining Maxwell equations can now be written as
\[
\begin{cases}
\dive {\bf E} = 0 \\
\displaystyle \curl {\bf B} = \frac{\partial {\bf E}}{\partial t}
\end{cases}
\Leftrightarrow
\begin{cases}
\displaystyle \Delta \phi + \frac{\partial}{\partial t} (\dive {\bf A}) = 0 \\
\displaystyle \grad (\dive {\bf A}) - \Delta {\bf A} = - \grad \frac{\partial \phi}{\partial t} - \frac{\partial^2 {\bf A}}{\partial t^2}
\end{cases}.
\]
Therefore, if there exist gauge potentials satisfying
\[
\dive {\bf A} = - \frac{\partial \phi}{\partial t}
\]
(the so-called {\bf Lorentz gauge}) then these equations reduce to uncoupled wave equations:
\[
\begin{cases}
\Box{\phi} = 0 \\
\Box {\bf A} = 0
\end{cases}
\]
To solve the Maxwell equations using the gauge potentials we then solve these wave equations with initial data $\phi_0$, $\left(\frac{\partial \phi}{\partial t}\right)_0$ and ${\bf A}_0$, $\left(\frac{\partial {\bf A}}{\partial t}\right)_0$ satisfying:
\begin{enumerate}
\item
$\curl {\bf A}_0 = {\bf B}_0$ (possible because $\dive {\bf B}_0 = 0$);
\item
$\phi_0 = 0$ (by choice);
\item
$\left(\frac{\partial {\bf A}}{\partial t}\right)_0 = - {\bf E}_0$ (giving the correct initial electric field);
\item
$\left(\frac{\partial \phi}{\partial t}\right)_0 = - \dive {\bf A}_0$ (so that the Lorentz gauge condition holds).
\end{enumerate}
These potentials will determine a solution of the Maxwell equations with the correct initial data {\bf if the Lorentz gauge condition holds for all time}. Now
\[
\Box \left( \dive {\bf A} + \frac{\partial \phi}{\partial t} \right) = \dive (\Box {\bf A}) + \frac{\partial}{\partial t} (\Box \phi) = 0
\]
and
\[
\left( \dive {\bf A} + \frac{\partial \phi}{\partial t} \right)_0 = 0.
\]
Moreover,
\[
\frac{\partial}{\partial t} \left( \dive {\bf A} + \frac{\partial \phi}{\partial t} \right) = \dive\left(\frac{\partial {\bf A}}{\partial t}\right) + \frac{\partial^2 \phi}{\partial t^2} =  \dive\left(\frac{\partial {\bf A}}{\partial t}\right) + \Delta \phi,
\]
and so
\[
\left(\frac{\partial}{\partial t} \left( \dive {\bf A} + \frac{\partial \phi}{\partial t} \right) \right)_0 = - \dive {\bf E}_0 = 0.
\]
By uniqueness of solution of the wave equation, we conclude that the Lorentz gauge condition does hold for all time, and so the potentials obtained by solving the wave equation with the initial conditions above do determine the solution of the Maxwell equations with initial data ${\bf B}_0$, ${\bf E}_0$.

\begin{Remark}
The electromagnetic potentials can be seen as the components of the {\bf electromagnetic potential one-form}
\[
A = - \phi \, dt + A^1 dx + A^2 dy + A^3 dz.
\]
Note that a gauge transformation can be written as
\[
A' = A + d\chi.
\]
The electric and magnetic fields can in turn be seen as the components of the {\bf Faraday tensor}
\begin{align*}
F = dA & = E^1 dx \wedge dt + E^2 dy \wedge dt + E^3 dz \wedge dt \\
& + B^1 dy \wedge dz + B^2 dz \wedge dx + B^3 dx \wedge dy.
\end{align*}
It should be obvious that $F$ remains invariant under a gauge transformation. The Maxwell equations can be written as
\[
dF = 0 \Leftrightarrow F = dA
\]
and
\[
d \star F = 0,
\]
since
\begin{align*}
\star F & = E^1 dy \wedge dz + E^2 dz \wedge dx + E^3 dx \wedge dy \\
& - B^1 dx \wedge dt - B^2 dy \wedge dt - B^3 dz \wedge dt
\end{align*}
(that is, the Hodge star replaces ${\bf B}$ with ${\bf E}$ and ${\bf E}$ with $-{\bf B}$).
\end{Remark}

\section{Einstein's equations} \label{sec5.4}

Let $(M,g)$ be a globally hyperbolic spacetime and $S \subset M$ a Cauchy hypersurface. Let us write $g$ in the Gauss Lemma form near $S$,
\[
g = -dt^2 + h_{ij}(t,x) dx^i dx^j,
\]
so that the level sets of $t$ are Riemannian manifolds with induced metric $h(t)=h_{ij}dx^i dx^j$ and second fundamental form 
\[
\displaystyle K(t)=\frac12\frac{\partial h_{ij}}{\partial t}dx^idx^j.
\]
For this choice of coordinates (gauge), finding the metric is equivalent to finding a time-dependent Riemannian metric $h(t)$ on $S$. The vacuum Einstein field equations $G_{\mu\nu}=0$ can be split into {\bf constraint equations}
\[
\begin{cases}
G_{00} = 0 \\
G_{0i} = 0
\end{cases}
\Leftrightarrow
\begin{cases}
\bar{R} + \left(K^{i}_{\,\, i}\right)^2 - K_{ij} K^{ij} = 0 \\
\bar{\nabla}_i K^j_{\,\, j} - \bar{\nabla}_j K^{j}_{\,\, i} = 0
\end{cases}
\]
and {\bf evolution equations}
\[
G_{ij} = 0 \Leftrightarrow \frac{\partial}{\partial t} K_{ij} = - \bar{R}_{ij} + 2 K_{il} K^{l}_{\,\, j} - K^{l}_{\,\, l} K_{ij} \; ,
\]
where $\bar{\nabla}$, $\bar{R}$ and $\bar{R}_{ij}$ are the Levi-Civita connection, the scalar curvature and the Ricci tensor of $h$. Note that the evolution equations allow us to evolve $h(t)$ and $K(t)$, whereas the constraint equations restrict their initial values $h(0)$ and $K(0)$. If the initial data satisfy the constraint equations, so does the solution of the evolution equations. Indeed, the contracted Bianchi identities give us for free the equations
\begin{align*}
\nabla^\alpha G_{\alpha \beta} = 0 & \Leftrightarrow \nabla^0 G_{0 \beta} +  \nabla^i G_{i \beta} = 0 \Leftrightarrow - \nabla_0 G_{0 \beta} + h^{ij} \nabla_j G_{i \beta} = 0 \\
& \Leftrightarrow \partial_0 G_{0 \beta} - \Gamma_{00}^\alpha G_{\alpha\beta} - \Gamma_{0\beta}^\alpha G_{0\alpha} = h^{ij} (\partial_j G_{i \beta} - \Gamma_{ji}^\alpha G_{\alpha\beta} - \Gamma_{j\beta}^\alpha G_{i \alpha}).
\end{align*}
If the evolution equations hold, we have $G_{ij} = \partial_\alpha G_{ij} = 0$, and so the contracted Bianchi identities become
\begin{align*}
&\begin{cases}
\partial_0 G_{00} = \Gamma_{00}^\alpha G_{\alpha 0} + \Gamma_{00}^\alpha G_{0\alpha} + h^{ij} (\partial_j G_{i 0} - \Gamma_{ji}^\alpha G_{\alpha 0} - \Gamma_{j 0}^0 G_{i 0}) \\
\partial_0 G_{0k} = \Gamma_{00}^0 G_{0k} + \Gamma_{0k}^\alpha G_{0\alpha} - h^{ij} (\Gamma_{ji}^0 G_{0k} + \Gamma_{jk}^0 G_{i 0})
\end{cases} \\
&\Leftrightarrow\begin{cases}
\partial_0 G_{00} = h^{ij} (\partial_j G_{i 0} - K_{ji} G_{00} - \bar\Gamma_{ji}^k G_{k0} ) \\
\partial_0 G_{0k} = K^{i}_{\,\,k} G_{0i} - h^{ij} (K_{ji} G_{0k} + K_{jk} G_{i0})
\end{cases}
.
\end{align*}
This is a system of linear first order partial differential equations on $G_{0 \beta}$; integrating the last three equations, and then the first, it is easy to see that, since the initial data vanishes at $t=0$, the solution vanishes for all $t$.

\begin{Remark}
In general, any time function $t:M \to \bbR$ whose level sets $S_t$ are Cauchy hypersurfaces can be completed into a system of local coordinates $(t,x^1,x^2,x^3)$. If
\[
N    = - \frac{\grad t}{|\grad t|}
\]
is the future-pointing unit normal to $S_t$, we have the orthogonal decomposition
\[
\frac{\partial}{\partial t} = \alpha N + \beta,
\]
where the positive function $\alpha$ is known as the {\bf lapse function} and the vector field $\beta$, tangent to $S_t$, is known as the {\bf shift vector} (Figure~\ref{lapse}). In these coordinates, the metric is written
\begin{align*}
g & = (-\alpha^2 + \beta_i \beta^i) dt^2 + 2 \beta_i dt dx^i + h_{ij} dx^i dx^j \\
& = -\alpha^2 dt^2 + h_{ij} (dx^i - \beta^i dt)(dx^j - \beta^j dt).
\end{align*}

\begin{figure}[h!]
\begin{center}
\psfrag{S}{$S_t$}
\psfrag{n}{$N$}
\psfrag{an}{$\alpha N$}
\psfrag{b}{$\beta$}
\psfrag{t}{$\frac{\partial}{\partial t}$}
\epsfxsize=.8\textwidth
\leavevmode
\epsfbox{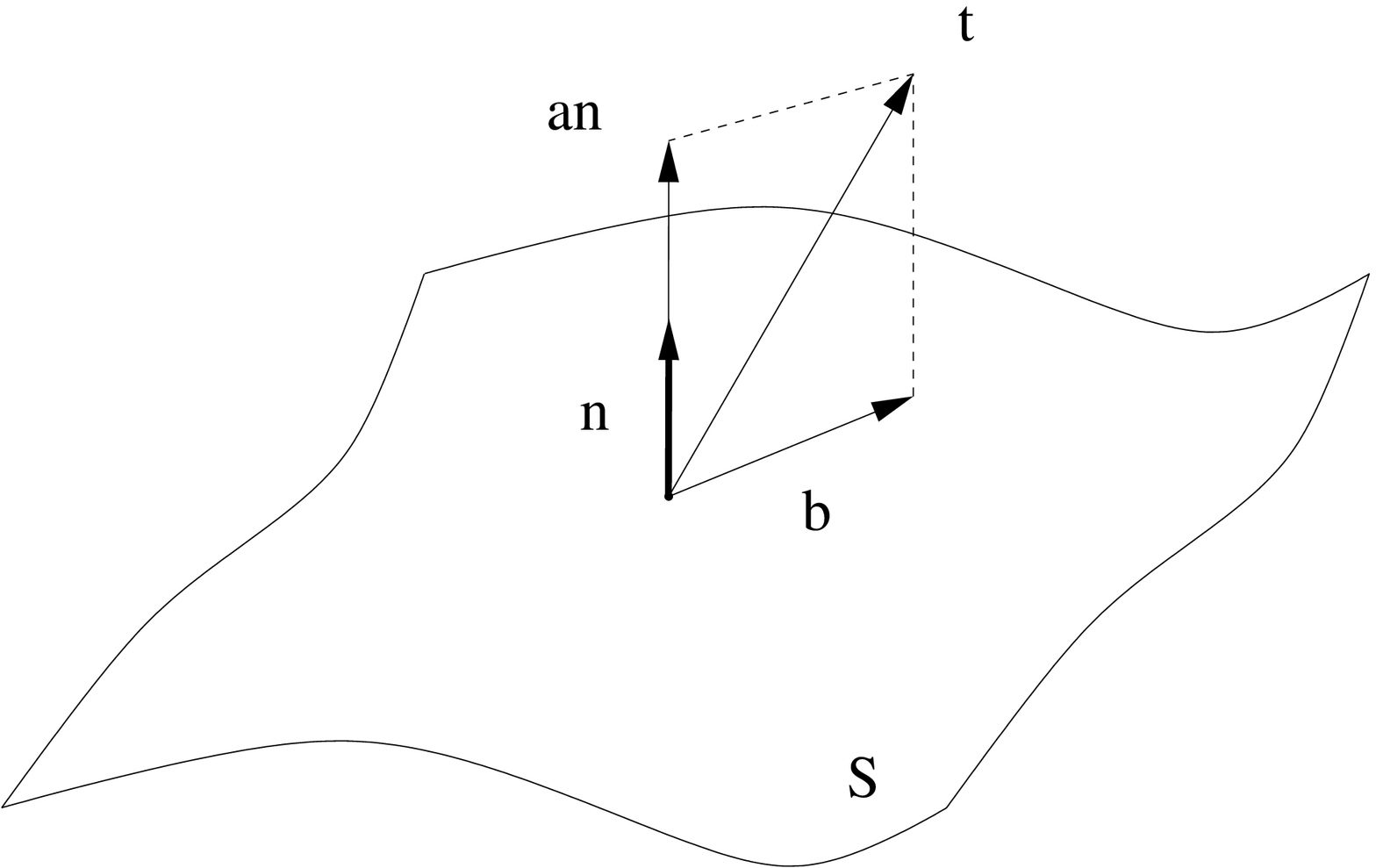}
\end{center}
\caption{Lapse function and shift vector.} \label{lapse}
\end{figure}

Note that the Riemannian metric of the Cauchy hypersurfaces $S_t$ is still $h(t)=h_{ij}dx^i dx^j$; the lapse function and the shift vector merely specify how points with the same coodinates $x^i$ in different Cauchy hypersurfaces $S_t$ are related, which is a matter of choice, that is, a gauge freedom. The Gauss Lemma form of the metric, for instance, corresponds $\alpha=1$ and $\beta=0$, but this is not necessarily the best choice. 
\end{Remark}

To prove the existence and uniqueness result for the vacuum Einstein field equations it is best to choose so-called {\bf harmonic coordinates}, that is, coordinates $x^\mu$ satisfying the wave equation, $\Box x^\mu=0$. This equation can be written as
\[
H^\mu \equiv \partial_\alpha g^{\alpha\mu} + \frac12 g^{\alpha\mu} g^{\rho\sigma}\partial_{\alpha}g_{\rho\sigma} = 0
\]
We define the {\bf reduced Ricci tensor} to be
\[
R^H_{\mu\nu} \equiv R_{\mu\nu} + g_{\alpha(\mu}\partial_{\nu)}H^\alpha = - \frac12 g^{\alpha\beta} \partial_\alpha \partial_\beta g_{\mu\nu} + F_{\mu\nu}(g,\partial g),
\]
and the {\bf reduced Einstein equations} to be
\[
R^H_{\mu\nu} = 0.
\]
Note that $R^H_{\mu\nu}=R_{\mu\nu}$ if the coordinates $x^\mu$ are harmonic; in this case, the reduced Einstein equations coincide with the Einstein equations. Moreover, the reduced Einstein equations are a quasi-linear, diagonal second order hyperbolic system for the components of the metric. Given initial data $(h_{ij},K_{ij})$ satisfying the constraint equations, consider the following initial data for the reduced Einstein equations:
\begin{enumerate}
\item
$g_{ij} = h_{ij}$ (forced);
\item
$g_{i0} = 0$ (by choice);
\item
$g_{00} = -1$ (by choice);
\item
$\frac{\partial g_{ij}}{\partial t} = 2K_{ij}$ (forced);
\item
$\frac{\partial g_{0\mu}}{\partial t}$ such that $H^\mu = 0$ in $S$.
\end{enumerate} 
If $(h_{ij},K_{ij})$ is close to the trivial data $(\delta_{ij},0)$ for the Minkowski spacetime, Theorem~\ref{hyperbolic2} guarantees that we can solve the reduced Einstein equations in some open neighborhood $V$ of $S$. From
\[
G_{\mu\nu} = R^H_{\mu\nu} - \frac12 R^H g_{\mu\nu} - g_{\alpha(\mu}\partial_{\nu)}H^\alpha + \frac12 g_{\mu\nu} \partial_\alpha H^\alpha
\]
it is easily seen that if $g_{\mu\nu}$ satisfies the reduced Einstein equations $R^H_{\mu\nu} = 0$ then the contracted Bianchi identities yield
\[
\nabla^\mu G_{\mu\nu} = 0 \Rightarrow g^{\mu\nu} \partial_\mu \partial_\nu H^\alpha + A^{\alpha\mu}_\beta \partial_\mu H^\beta = 0.
\]
Moreover, we have from the constraint equations
\[
G_{\mu 0} = 0 \Rightarrow \partial_0 H^\mu = 0
\]
on $S$. Therefore $H^\mu$ is the solution of a linear, diagonal second order hyperbolic system with vanishing initial conditions. We conclude that $H^\mu=0$ in $V$, and therefore $g_{\mu\nu}$ solves the Einstein equations in $V$.

We can always assume that our initial data is close to the trivial data by rescaling: if $x^\mu$ are local coordinates such that $g_{\mu\nu}=\eta_{\mu\nu}$ at some point $p \in S$, where $\eta_{\mu\nu}$ is the Minkowski metric, we define new coordinates $\bar{x}^\mu$ by the formula
\[
\bar{x}^\mu = \frac1{\lambda} x^\mu,
\]
where $\lambda>0$ is a constant. In these new coordinates the metric is
\[
\bar{g}_{\mu\nu} = \frac{\partial x^\alpha}{\partial \bar{x}^\mu} \frac{\partial x^\beta}{\partial \bar{x}^\nu} g_{\alpha\beta} = \lambda^2 g_{\mu\nu}.
\]
If this metric is a solution of the vacuum Einstein field equations, then so is
\[
\tilde{g}_{\mu\nu} = \frac1{\lambda^2}\bar{g}_{\mu\nu} = g_{\mu\nu}.
\]
Note that this metric satisfies
\[
\frac{\partial \tilde{g}_{\mu\nu}}{\partial \bar{t}} =  \frac{\partial g_{\mu\nu}}{\partial t} \frac{\partial t}{\partial \bar{t}} = \lambda \frac{\partial g_{\mu\nu}}{\partial t}.
\]
Therefore for $\lambda$ sufficiently small the initial data $(\tilde{g}_{\mu\nu},\frac{\partial \tilde{g}_{\mu\nu}}{\partial \bar{t}})$  will be close to $(\eta_{\mu\nu},0)$.

In this way we can obtain a local solution of the Einstein field equations in a neighborhood of each point $p \in S$. By uniqueness of solution of a quasi-linear, diagonal second order hyperbolic system we can glue these local solutions to obtain a global solution defined on an open neighborhood of $S$. In other words, given initial data $(h,K)$ satisfying the constraint equations, there exists a globally hyperbolic spacetime $(M,g)$ satisfying the Einstein field equations such that $S$ is a Cauchy surface with induced metric $h$ and second fundamental form $K$.

Finally, now that we have proved the existence of such solutions, it is possible to prove the existence of a {\bf maximal} solution. The proof is as follows: if $(M_1,g_1)$ and $(M_2,g_2)$ are two solutions of the Einstein field equations containing $S$ with the same initial data $(h,K)$, we say that $(M_1,g_1) \leq (M_2,g_2)$ if there is an isometric embedding $\psi:M_1 \to M_2$ preserving $(S,h,K)$. Note that it is possible that neither $(M_1,g_1) \leq (M_2,g_2)$ nor $(M_2,g_2) \leq (M_1,g_1)$. A set of solutions with the property that any two are related by $\leq$ is called a {\bf chain}; it is clear that every chain has an upper bound (the union up to isometric embeddings). Under these conditions, {\bf Zorn's Lemma} guarantees that there is a maximal element $(M,g)$ in the set of all solutions, that is, a solution which cannot be isometric embedded into any other solution. It is possible to prove that this element is unique (if there were two such maximal solutions it would be possible to patch them together to construct a larger solution). We then have the following fundamental result.

\begin{Thm} ({\bf Choquet-Bruhat \cite{CB55,CBG69}})
Let $(S,h)$ be a $3$-dimensional Riemannian manifold and $K$ a symmetric tensor field in $S$ satisfying the constraint equations
\[
\begin{cases}
\bar{R} + \left(K^{i}_{\,\, i}\right)^2 - K_{ij} K^{ij} = 0 \\
\bar{\nabla}_i K^j_{\,\, j} - \bar{\nabla}_j K^{j}_{\,\, i} = 0
\end{cases},
\]
where $\bar{\nabla}$ and $\bar{R}$ are the Levi-Civita connection and the scalar curvature of $h$. Then there exists a unique (up to isometry) $4$-dimensional Lorentzian manifold $(M,g)$, called the {\bf maximal Cauchy development} of $(S,h,K)$, satisfying:
\begin{enumerate}[(i)]
\item
$(M,g)$ is a solution of the vacuum Einstein equations;
\item
$(M,g)$ is globally hyperbolic with Cauchy surface $S$;
\item
The induced metric and second fundamental forms of $S$ are $h$ and $K$;
\item
Any $4$-dimensional Lorentzian manifold satisfying $(i)-(iii)$ can be isometrically embedded into $(M,g)$.
\end{enumerate}
Moreover, if $(S,h,K)$ and  $(\bar{S},\bar{h},\bar{K})$ coincide on some closed subset $B \cong \bar{B}$ then $D(B)$ and $D(\bar{B})$ are isometric. Finally, $g$ depends continuously on the initial data $(h,K)$ (for appropriate topologies).
\end{Thm}

\section{Constraint equations} \label{sec5.5}

To obtain initial data for the Einstein equations it is necessary to solve the nonlinear constraint equations
\[
\begin{cases}
\bar{R} + \left(K^{i}_{\,\, i}\right)^2 - K_{ij} K^{ij} = 0 \\
\bar{\nabla}_i K^j_{\,\, j} - \bar{\nabla}_j K^{j}_{\,\, i} = 0
\end{cases}.
\]
The {\bf Lichnerowicz method} for solving these equations is as follows: one starts by choosic an arbitrary Riemannian metric $h$ and an arbitrary symmetric tensor $K$ satisfying
\[
\begin{cases}
K^{i}_{\,\, i} = 0 \\
\bar{\nabla}_j K^{j}_{\,\, i} = 0
\end{cases}
\]
(that is, $K$ is traceless and divergenceless). These choices satisfy the second, but not the first, constraint equations. On then defines the conformally rescaled metric
\[
\tilde{h}=u^4 h
\]
and the rescaled symmetric tensor
\[
\tilde{K}=u^{-2} K
\]
Clearly $\tilde{K}$ is still traceless, and it is easily seen that it is also divergenceless:
\[
\tilde{\nabla}_j \tilde{K}^{j}_{\,\, i} = 0,
\]
where $\tilde{\nabla}$ is the Levi-Civita connection of $\tilde{h}$. The first constraint equation for the metric $\tilde{h}$ and the symmetric tensor $\tilde{K}$, on the other hand, becomes
\[
\tilde{R} - \tilde{K}_{ij} \tilde{K}^{ij} = 0 \Leftrightarrow \bar{\Delta} u - \frac18 \bar{R} u + \frac18 u^{-7} K_{ij} K^{ij} = 0.
\] 
This is a nonlinear elliptic equation on one variable, much simpler than the original system. If one chooses the so-called time symmetric case $K=0$ (the reason for this designation being that in the Gauss Lemma coordinates $t \to -t$ is clearly an isometry of the solution), this equation becomes linear:
\[
\bar{\Delta} u - \frac18 \bar{R} u = 0.
\]
As an example, choose $h$ to be the Euclidean metric, $h_{ij} = \delta_{ij}$; then $\bar{R}=0$ and the equation above is simply the Laplace equation
\[
\Delta u = 0.
\]
A simple solution, related to the gravitational field of a point mass, is
\[
u = 1 + \frac{M}{2r}.
\]
This solution leads exactly to the Schwarzschild solution, which in isotropic coordinates is written
\[
ds^2 = - \left( \frac{1 - \frac{M}{2r}}{1 + \frac{M}{2r}}\right)^2 dt^2 + \left(1 + \frac{M}{2r}\right)^4 (dr^2 + r^2(d\theta^2 + \sin^2 \theta d\varphi^2)).
\]
Since the Laplace equation is linear, one can superimpose solutions. Thus an initial data set for a set of $N$ black holes initially at rest can be obtained by choosing
\[
u = \sum_{i=1}^N \left(1 + \frac{M_i}{2r_i}\right),
\]
where $M_i > 0$ and $r_i$ is the Euclidean distance to a fixed point $p_i \in \bbR^3$ ($i=1, \ldots, N$).

\section{Einstein equations with matter} \label{sec5.6}

So far we have analyzed only the vacuum Einstein field equations. To include matter (and also a cosmological constant $\Lambda$) we must introduce matter fields, generically represented by $\psi$, and consider a system of the form
\[
\begin{cases}
G_{\mu\nu} + \Lambda g_{\mu\nu} = 8 \pi T_{\mu\nu}(g,\psi) \\
\text{Field equations for } \psi
\end{cases}.
\]
If the equations for $\psi$ form a hyperbolic system then Choquet-Bruhat's theorem still applies, with the constraint equations
\[
\begin{cases}
\bar{R} + \left(K^{i}_{\,\, i}\right)^2 - K_{ij} K^{ij} - 2\Lambda = 16 \pi \rho \\
\bar{\nabla}_i K^j_{\,\, j} - \bar{\nabla}_j K^{j}_{\,\, i} = 8 \pi J_i
\end{cases}.
\]
Here $\rho = T_{00}$ and $J_i = - T_{0i}$ are computed from $h$ and the initial data for $\psi$. As an example, the Einstein-Klein-Gordon system is given by
\[
\begin{cases}
G_{\mu\nu} + \Lambda g_{\mu\nu} = 8 \pi \left( \partial_\mu \phi \, \partial_\nu \phi - \frac12 g_{\mu\nu} \left( \partial_\alpha \phi \, \partial^\alpha \phi + m^2 \phi^2 \right) \right) \\
\nabla^\mu \partial_\mu \phi - m^2 \phi = 0
\end{cases},
\]
and we have
\[
\begin{cases}
\rho = \frac12 \left( {\psi_0}^2 + h^{ij} \partial_i \phi_0 \partial_j \phi_0 + m^2 {\phi_0}^2 \right) \\
J_i = - \psi_0 \partial_i \phi_0
\end{cases},
\]
where $\phi_0$ and $\psi_0$ are the initial data for $\phi$.

\section{Exercises} \label{sec5.7}

\begin{enumerate}

\item
Let $(M,g)$ be a time oriented Lorentzian manifold and $X$ a future-pointing timelike vector field. Given a smooth function $\phi:M \to \bbR$ let
\[
T_{\mu\nu} = \partial_\mu \phi \partial_\nu \phi - \frac12 g_{\mu\nu} \left(\partial_\alpha\phi \partial^\alpha \phi + m^2 \phi^2 \right)
\]
be the energy-momentum tensor associated to the Klein-Gordon equation for $\phi$, and let $Y$ be the vector field defined by
\[
Y_\mu = T_{\mu\nu}X^\nu.
\]
Show that:
\begin{enumerate}
\item
$Y = (X \cdot \phi) \grad \phi - \frac12\left(\left\langle\grad\phi,\grad\phi\right\rangle + m^2\phi^2\right) X$;
\item
$Y$ is causal.
\item
$Y$ is past-pointing.
\end{enumerate}

\item
({\bf Sobolev inequality}) Let $Q$ be a closed solid cone in $\bbR^n$ with height $H$, solid angle $\Omega$ and vertex at the origin. Let $\psi:\bbR \to \bbR$ be a smooth nonincreasing function with $\psi(r)=1$ for $r<\frac{H}3$ and $\psi(r)=0$ for $r>\frac{2H}3$. Show that:
\begin{enumerate}
\item
For any smooth function $f:Q \to \bbR$ and any $k \in \bbN$ we have
\[
f(0)=\frac{(-1)^k}{(k-1)!} \int_0^{R(\theta)} r^{k-1}\frac{\partial^k}{\partial r^k} \left(\psi(r)f(r,\theta)\right) dr,
\]
where $(r,\theta)$ are the usual spherical coordinates in $\bbR^n$ and $r=R(\theta)$ is the equation for the base of the cone (here $f(r,\theta)$ represents the function $f$ written in spherical coordinates).
\item
There exists a constant $C$, depending on $k$ and $\Omega$, such that
\[
f(0)=C\int_Q r^{k-n} \frac{\partial^k}{\partial r^k}\left(\psi f\right).
\]
\item
For $k>\frac{n}2$ we have $|f(0)| \leq C' \|f\|_{H^k(Q)}$, where the constant $C'$ depends on $k$, $H$ and $\Omega$ only (you will need to use the Cauchy-Schwarz inequality for multiple integrals).
\end{enumerate}

\item
Consider a Lorentzian metric given in the Gauss Lemma form
\[
g = -dt^2 + h_{ij}(t,x) dx^i dx^j,
\]
so that the level sets of $t$ are Riemannian manifolds with induced metric $h(t)=h_{ij}dx^i dx^j$ and second fundamental form 
\[
K(t)=\frac12\frac{\partial h_{ij}}{\partial t}dx^idx^j.
\]
Show that in these coordinates: 
\begin{enumerate}
\item
The Christoffel symbols are
\[
\Gamma^0_{ij} = K_{ij}; \quad \Gamma^i_{jk} = \bar{\Gamma}^i_{jk}; \quad \Gamma^i_{0j} = K^i_{\,\,j},
\]
where $\bar{\Gamma}^i_{jk}$ are the Christoffel symbols of $h$.
\item
The components of the Riemann tensor are
\begin{align*}
\label{Riemann1} & R_{0i0}^{\,\,\,\,\,\,\,\, j} = - \frac{\partial}{\partial t} K^{j}_{\,\, i} - K_{il} K^{lj}; \\ 
\nonumber & R_{ij0}^{\,\,\,\,\,\,\,\, l} = - \bar{\nabla}_i K^l_{\,\, j} + \bar{\nabla}_j K^{l}_{\,\, i}; \\ 
\nonumber & R_{ijl}^{\,\,\,\,\,\,\,\, m} = \bar{R}_{ijl}^{\,\,\,\,\,\,\,\, m} + K_{il} K^{m}_{\,\,\,\, j} - K_{jl} K^{m}_{\,\,\,\, i},
\end{align*}
where $\bar{\nabla}$ is the Levi-Civita connection of $h$ and $\bar{R}_{ijl}^{\,\,\,\,\,\,\,\, m}$ are the components of the Riemann tensor of $h$.
\item
The time derivative of the inverse metric is given by the formula
\[
\frac{\partial h^{ij}}{\partial t} = -2K^{ij}.
\]
\item
The components of the Ricci tensor are
\begin{align*}
& R_{00} = - \frac{\partial}{\partial t} K^{i}_{\,\, i} - K_{ij} K^{ij}; \\
& R_{0i} = - \bar{\nabla}_i K^j_{\,\, j} + \bar{\nabla}_j K^{j}_{\,\, i}; \\
& R_{ij} = \bar{R}_{ij} + \frac{\partial}{\partial t} K_{ij} - 2 K_{il} K^{l}_{\,\, j} + K^{l}_{\,\, l} K_{ij},
\end{align*}
where $\bar{R}_{ij}$ are the components of the Ricci tensor of $h$.
\item
The scalar curvature is
\[
R = \bar{R} + 2 \frac{\partial}{\partial t} K^{i}_{\,\, i} + \left(K^{i}_{\,\, i}\right)^2 + K_{ij} K^{ij},
\]
where $\bar{R}$ is the scalar curvature of $h$.
\item
The component $G_{00}$ of the Einstein tensor is
\[
G_{00} = \frac12 \left( \bar{R} + \left(K^{i}_{\,\, i}\right)^2 - K_{ij} K^{ij} \right).
\]
\end{enumerate}

\item
Let $(M,g)$ be an $(n+1)$-dimensional Lorentzian manifold and $(x^0, \ldots, x^n)$ local coordinates on $M$. Show that:
\begin{enumerate}
\item
The condition for these coordinates to be harmonic is written
\[
\hspace{2cm} \nabla_\alpha \nabla^\alpha x^\mu = 0 \Leftrightarrow H^\mu \equiv \partial_\alpha g^{\alpha\mu} + \frac12 g^{\alpha\mu}g^{\rho\sigma}\partial_\alpha g_{\rho\sigma} = 0
\]
($\mu=0, \ldots, n$).
\item
The reduced Ricci tensor is
\[
\hspace{2cm} R^H_{\mu\nu} \equiv R_{\mu\nu} + g_{\alpha(\mu}\partial_{\nu)}H^\alpha = - \frac12 g^{\alpha\beta} \partial_\alpha \partial_\beta g_{\mu\nu} + F_{\mu\nu}(g,\partial g).
\]
\end{enumerate}

\item \label{development}
Denoting by $h_0$ the standard constant curvature metric on $\bbR^3$, $S^3$ or $H^3$, compute the cosmological constant and determine the maximal globally hyperbolic developments of the following sets of initial data for the vacuum Einstein equations (with cosmological constant):
\begin{enumerate}
\item $(\bbR^3,h_0,0)$;
\item $(\bbR^3,h_0,h_0)$;
\item $(S^3,h_0,0)$;
\item $(H^3,h_0,0)$;
\item $(H^3,h_0,h_0)$;
\item $(\bbR^3,h_0,\diag(p_1,p_2,p_3))$ with $p_1+p_2+p_3={p_1}^2+{p_2}^2+{p_3}^2=1$.
\end{enumerate}

\item
Let $(S,h,K)$ be an initial data set for the vacuum Einstein equations, with $K$ traceless and divergenceless. Given a smooth positive function $u:S \to \bbR$, consider the conformally rescaled metric $\tilde{h}=u^4 h$ and the symmetric tensor $\tilde{K}=u^{-2} K$. By using normal coordinates when convenient, show that:
\begin{enumerate}
\item
The Christoffel symbols $\tilde{\Gamma}^i_{jk}$ of $\tilde{h}$ are related to the Christoffel symbols $\bar{\Gamma}^i_{jk}$ of $h$ by
\[
\hspace{2cm} \tilde{\Gamma}^i_{jk} = \bar{\Gamma}^i_{jk} + 2 \partial_j (\log u) h^{i}_{\,\,\,\,k} + 2 \partial_k (\log u) h^{i}_{\,\,\,\,j} - 2 \partial^i (\log u) h_{jk}.
\]
\item
$\tilde{K}$ is divergenceless for the Levi-Civita connection of $\tilde{h}$.
\item
The Ricci tensor $\tilde{R}_{ij}$ of $\tilde{h}$ is related to the Ricci tensor $\bar{R}_{ij}$ of $h$ by
\begin{align*}
\hspace{2cm}  \tilde{R}_{ij} = & \, \bar{R}_{ij} - 2 \bar{\nabla}_i \partial_j (\log u) - 2 \bar{\Delta} (\log u) h_{ij} \\
& + 4 \partial_i (\log u) \partial_j(\log u) - 4 \left|\grad(\log u)\right|^2 h_{ij}.
\end{align*}
\item
The scalar curvature $\tilde{R}$ of $\tilde{h}$ is related to the scalar curvature $\bar{R}$ of $h$ by
\[
\tilde{R} = u^{-4} \bar{R} - 8 u^{-5} \bar{\Delta} u.
\]
\end{enumerate}

\item
Check that the metric
\[
\hspace{2cm} ds^2 = - \left( \frac{1 - \frac{M}{2r}}{1 + \frac{M}{2r}}\right)^2 dt^2 + \left(1 + \frac{M}{2r}\right)^4 (dr^2 + r^2(d\theta^2 + \sin^2 \theta d\varphi^2))
\]
is indeed the Schwarzschild metric by making the coordinate change
\[
R = r \left(1 + \frac{M}{2r}\right)^2.
\]

\end{enumerate}


\chapter{Positive mass theorem} \label{chapter6}

In this chapter we present the positive mass theorem. Following \cite{W84}, we start by defining the Komar mass for stationary spacetimes. We then discuss field theory and introduce the Einstein-Hilbert action as a means of motivating the definition of the ADM mass. Finally, we prove the (Riemannian) positive mass theorem and the (Riemannian) Penrose inequality for graphs, following \cite{L10}. For more details see \cite{Mars09, B11}.

\section{Komar mass} \label{sec6.1}

Recall that the Newtonian gravitational field satisfies
\[
\dive {\bf G} = - 4 \pi \rho,
\]
where $\rho$ is the mass density of the matter generating the field. Therefore the total mass of a given system is given by
\[
M = \int_{\bbR^3} \rho = - \frac1{4\pi} \int_{\bbR^3} \dive {\bf G} = - \frac1{4\pi} \int_{\Sigma} \left\langle {\bf G}, {\bf n} \right\rangle,
\]
where $\Sigma$ is any surface enclosing all the matter and ${\bf n}$ is the outward unit normal. The fact that $M$ does not depend on $\Sigma$ is equivalent to the statement that $\dive {\bf G}=0$ in the region between any two such surfaces.

For a {\bf static} Lorentzian metric, that is, a metric of the form
\[
ds^2 = - e^{2 \phi} dt^2 + h_{ij} dx^i dx^j
\]
with $\phi$ and $h$ not depending on $t$, the analogue of the gravitational field is minus the acceleration of the observers with constant space coordinates $x^i$, that is, $- \grad \phi$ (see Chapter~\ref{chapter2}). On the other hand, since the metric is static, we expect that the energy computed at a given surface should be multiplied by the redshift factor $e^\phi$ to obtain its reference value at infinity. The relativistic analogue of the formula above is then
\[
M = \frac1{4\pi} \int_{\Sigma} e^\phi (\partial_\mu \phi) n^\mu =  \frac1{4\pi} \int_{\Sigma} (\partial_\mu e^\phi) n^\mu.
\]
We have
\[
\partial_\mu e^\phi = \partial_\mu (- K^\nu K_\nu)^\frac12 = - e^{-\phi} K^\nu \nabla_\mu K_\nu = - N^\nu \nabla_\mu K_\nu,
\]
where $K=\frac{\partial}{\partial t}$ is the timelike Killing vector field and $N=e^{-\phi}\frac{\partial}{\partial t}$ is the unit timelike vector with the same direction, that is, the future-pointing unit normal to the hypersurfaces $S_t$ of constant $t$ (see Figure~\ref{Komar}). Therefore, we can write
\[
M = - \frac1{4\pi} \int_{\Sigma} (\nabla_\mu K_\nu) n^\mu N^\nu = \frac1{4\pi} \int_{\Sigma} (\nabla_\mu K_\nu) N^\mu n^\nu.
\]

\begin{figure}[h!]
\begin{center}
\psfrag{n}{$n$}
\psfrag{N}{$N$}
\psfrag{K}{$K$}
\psfrag{T}{$E_1$}
\psfrag{-grad}{$-\grad \phi$}
\psfrag{S}{$\Sigma$}
\psfrag{St}{$S_t$}
\epsfxsize=1.0\textwidth
\leavevmode
\epsfbox{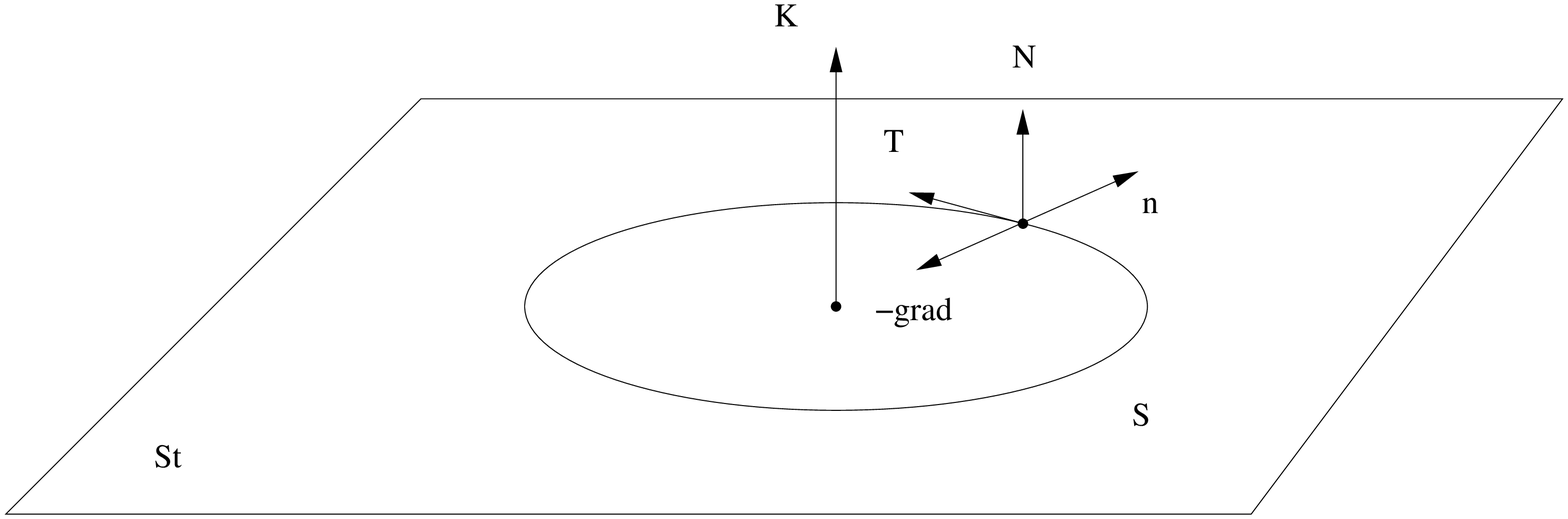}
\end{center}
\caption{Computing the Komar mass on a static spacetime.} \label{Komar}
\end{figure}

Because $K$ is a Killing vector field, $\nabla_\mu K_\nu$ is a $2$-form; more precisely,
\[
(d K^\sharp)_{\mu\nu} = \nabla_\mu K_\nu - \nabla_\nu K_\mu = 2 \nabla_\mu K_\nu.
\]
If $E_1$ and $E_2$ are two unit vector fields tangent to $\Sigma$ such that $\{N,n,E_1,E_2\}$ is a positive orthonormal frame (see Figure~\ref{Komar}), and so $\{-N^\sharp,n^\sharp,E_1^\sharp,E_2^\sharp\}$ is a positive orthonormal coframe, then we can expand
\[
\nabla K^\sharp = - \nabla K^\sharp(N,n) N^\sharp \wedge n^\sharp + \ldots,
\]
and so
\begin{align*}
M & = \frac1{4\pi} \int_{\Sigma} \nabla K^\sharp(N,n) = \frac1{4\pi} \int_{\Sigma} \nabla K^\sharp(N,n) E_1^\sharp \wedge E_2^\sharp \\
&  = - \frac1{4\pi} \int_{\Sigma} - \nabla K^\sharp(N,n) \star(N^\sharp\wedge n^\sharp) = - \frac1{4\pi} \int_{\Sigma} \star \nabla K^\sharp,
\end{align*}
that is
\[
M = - \frac1{8\pi} \int_{\Sigma} \star d K^\sharp.
\]
This expression is the so-called {\bf Komar mass}. Although we arrived at this expression by considering a {\bf static} space, it actually works for any {\bf stationary} spacetime. In other words, the timelike Killing vector field $K$ does not have to be hypersurface-orthogonal.

To show that the Komar mass is well defined, that is, that $\star d K^\sharp$ is a closed $2$-form in vacuum, we start by noticing that if $X$ is any vector field then
\[
(\dive X) \epsilon = d (X \contr \epsilon) = d \star X^\sharp, 
\]
whence
\[
\dive X = - \star d \star X^\sharp.
\]
In particular, for any smooth function $\phi$
\[
\Box \phi = \dive \grad \phi = - \star d \star (\grad \phi)^\sharp = - \star d \star d \phi.
\]
In local coordinates, we have
\begin{align*}
& \star d \phi = \epsilon(\grad \phi, \cdot, \cdot, \cdot) = \sqrt{|\det(g_{\mu\nu})|} \, dx^0 \wedge dx^1 \wedge dx^2 \wedge dx^3(\grad \phi, \cdot, \cdot, \cdot) \\
& = \sqrt{|\det(g_{\mu\nu})|} \, \partial^0 \phi \, dx^1 \wedge dx^2 \wedge dx^3 - \sqrt{|\det(g_{\mu\nu})|} \, \partial^1 \phi \, dx^0 \wedge dx^2 \wedge dx^3 + \ldots
\end{align*}
Therefore
\begin{align*}
d \star d \phi & = \partial_0 \left(\sqrt{|\det(g_{\mu\nu})|} \, \partial^0 \phi \right) dx^0 \wedge dx^1 \wedge dx^2 \wedge dx^3 \\
& + \partial_1 \left(\sqrt{|\det(g_{\mu\nu})|} \, \partial^1 \phi\right) dx^0 \wedge dx^1 \wedge dx^2 \wedge dx^3 + \ldots
\end{align*}
and we obtain the useful formula
\[
\Box \phi = \frac1{\sqrt{|\det(g_{\mu\nu})|}} \partial_\alpha  \left(\sqrt{|\det(g_{\mu\nu})|} \, \partial^\alpha \phi\right).
\]
It is natural to try and generalize this formula for arbitrary $k$-forms.

\begin{Def}
If $\omega$ is a $k$-form then its {\bf Hodge d'Alembertian} is the $k$-form
\[
\Box_H \omega = - \star d \star d \omega - d \star d \star \omega.
\]
\end{Def}

It turns out that in general the Hodge d'Alembertian does not coincide with the usual d'Alembertian
\[
\Box \omega = \nabla_\mu \nabla^\mu \omega
\]
(sometimes called the {\bf rough d'Alembertian}). The relation between these two operators for $1$-forms is given by the following result.

\begin{Thm} ({\bf Weitzenbock formula})
If $\omega$ is a $1$-form then
\[
\Box_H \omega - \Box \omega = - Ric(\omega^\sharp, \cdot).
\]
\end{Thm}

\begin{proof}
We have
\begin{align*}
(\star d\star d\omega)_\delta & = \frac{3 \cdot 2}{3!\,2!} \epsilon_{\gamma\alpha\beta\delta} \nabla^\gamma (\epsilon^{\mu\nu\alpha\beta} \nabla_\mu\,\omega_\nu) = \frac1{2!} \epsilon_{\alpha\beta\gamma\delta} \epsilon^{\mu\nu\alpha\beta} \nabla^\gamma \nabla_\mu\,\omega_\nu \\
& = (- g^{\mu}_{\,\,\,\,\gamma} g^{\nu}_{\,\,\,\,\delta} + g^{\mu}_{\,\,\,\,\delta} g^{\nu}_{\,\,\,\,\gamma}) \nabla^\gamma \nabla_\mu\,\omega_\nu = - \nabla^\mu \nabla_\mu\,\omega_\delta + \nabla^\nu \nabla_\delta \,\omega_\nu  
\end{align*}
and
\begin{align*}
(d\star d\star \omega)_\delta & = \frac{4}{4!} \nabla_\delta (\epsilon_{\gamma\nu\alpha\beta} \nabla^\gamma (\epsilon^{\mu\nu\alpha\beta}\omega_\mu)) = \frac1{3!} \epsilon_{\gamma\nu\alpha\beta}\epsilon^{\mu\nu\alpha\beta}\nabla_\delta\nabla^\gamma\omega_\mu \\
& = - g^{\mu}_{\,\,\,\,\gamma} \nabla_\delta\nabla^\gamma\omega_\mu = - \nabla_\delta\nabla^\mu\omega_\mu.
\end{align*}
Therefore
\[
(\star d\star d\omega + d\star d\star \omega)_\delta = - (\Box\omega)_\delta + \nabla^\mu \nabla_\delta \,\omega_\mu  - \nabla_\delta\nabla^\mu\omega_\mu.
\]
The Weitzenbock formula now follows from
\[
\nabla_\mu \nabla_\delta \,\omega^\mu  - \nabla_\delta\nabla_\mu\,\omega^\mu = R_{\mu\delta\,\,\,\,\nu}^{\,\,\,\,\,\,\,\,\mu} \, \omega^\nu = R_{\delta\nu} \, \omega^\nu.
\]
\end{proof}

Now let $K=\frac{\partial}{\partial t}$ be the timelike Killing vector field in a stationary spacetime. Then
\begin{align*}
d \star K^\sharp & = d \left( \frac{\partial}{\partial t} \contr \left( \sqrt{|\det(g_{\mu\nu})|} \, dt \wedge dx^1 \wedge dx^2 \wedge dx^3 \right)\right) \\
& = d \left( \sqrt{|\det(g_{\mu\nu})|} \, dx^1 \wedge dx^2 \wedge dx^3 \right) \\
& = \frac{\partial}{\partial t} \sqrt{|\det(g_{\mu\nu})|} \, dx^1 \wedge dx^2 \wedge dx^3 = 0,
\end{align*}
since the metric coefficients do not depend on $t$. From the Weitzenbock formula we then have
\[
\star d \star d K^\sharp = - \Box K^\sharp + Ric(K, \cdot).
\]
Now, adding cyclic permutations of the fundamental identity for the Riemann tensor,
\begin{align*}
& \nabla_\mu \nabla_\nu K_\alpha - \nabla_\nu \nabla_\mu K_\alpha = R_{\mu\nu\alpha}^{\,\,\,\,\,\,\,\,\,\,\,\,\beta} K_\beta, \\
& \nabla_\nu \nabla_\alpha K_\mu - \nabla_\alpha \nabla_\nu K_\mu = R_{\nu\alpha\mu}^{\,\,\,\,\,\,\,\,\,\,\,\,\beta} K_\beta,\\
& \nabla_\alpha \nabla_\mu K_\nu - \nabla_\mu \nabla_\alpha K_\nu = R_{\alpha\mu\nu}^{\,\,\,\,\,\,\,\,\,\,\,\,\beta} K_\beta,
\end{align*}
and using the Killing equation and the first Bianchi identity, we have
\[
\nabla_\mu \nabla_\nu K_\alpha = - R_{\nu\alpha\mu}^{\,\,\,\,\,\,\,\,\,\,\,\,\beta} K_\beta,
\]
whence
\[
\Box K_\alpha = - R_{\alpha}^{\,\,\,\,\beta} K_\beta.
\]
We conclude that
\[
\star d \star d K^\sharp = 2 Ric(K, \cdot) \Leftrightarrow d \star d K^\sharp = 2 \star Ric(K, \cdot),
\]
implying that $\star d K^\sharp$ is indeed closed in vacuum, that is, the Komar mass is well defined: any two homologous compact orientable surfaces $\Sigma_1$ and $\Sigma_2$ which enclose the matter content of the stationary spacetime can be used to compute it (Figure~\ref{Komar2}).

\begin{figure}[h!]
\begin{center}
\psfrag{S1}{$\Sigma_1$}
\psfrag{S2}{$\Sigma_2$}
\psfrag{matter}{matter}
\psfrag{vacuum}{vacuum}
\epsfxsize=0.5\textwidth
\leavevmode
\epsfbox{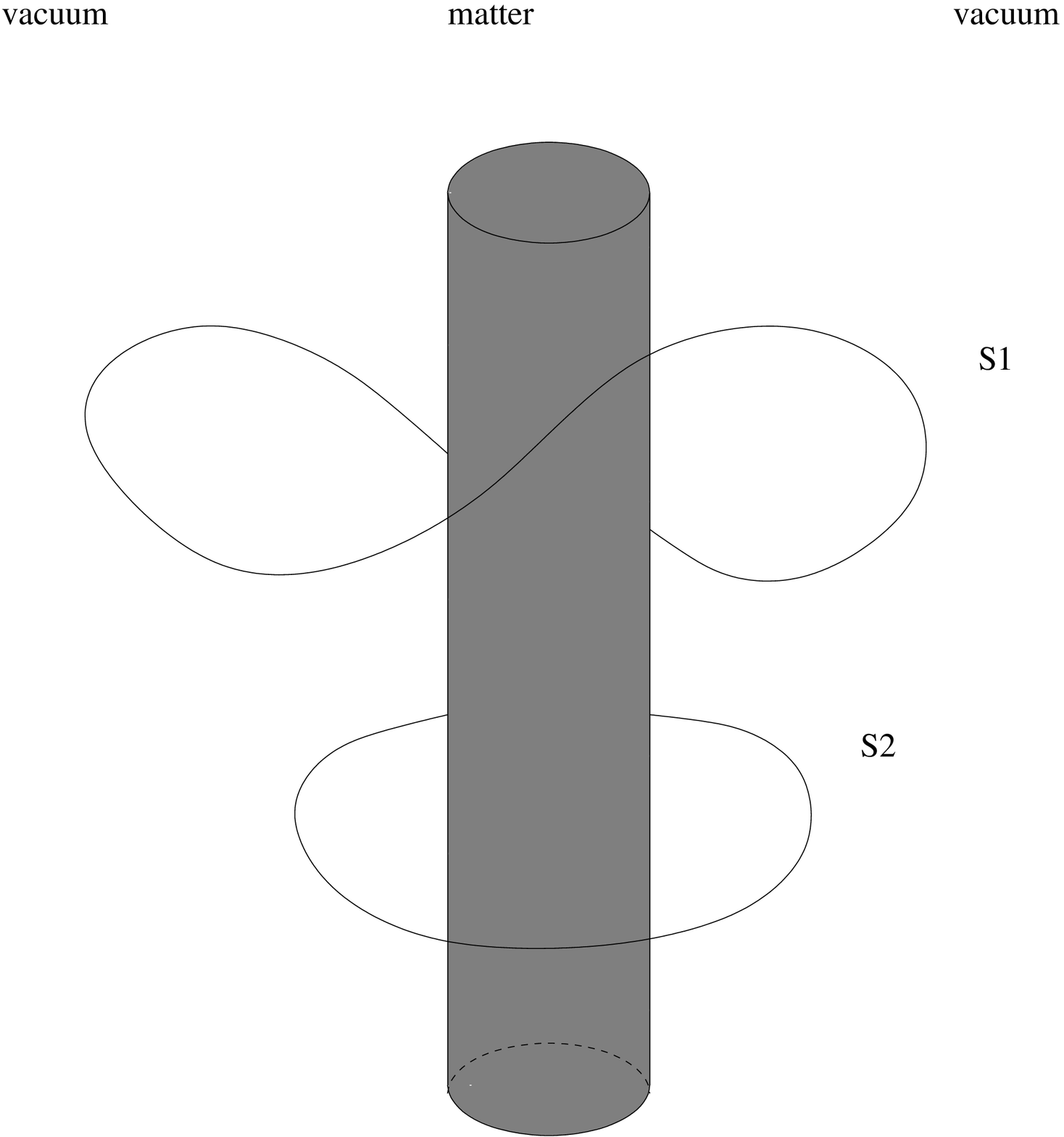}
\end{center}
\caption{Computing the Komar mass with two homologous surfaces.} \label{Komar2}
\end{figure}

If $\Sigma$ is the boundary of a spacelike $3$-dimensional manifold $B$ whose future-pointing unit normal is $N$ then the Komar mass can be written as
\begin{align*}
M & = - \frac1{8\pi} \int_{\Sigma} \star d K^\sharp = - \frac1{8\pi} \int_{B} d \star d K^\sharp = - \frac1{4\pi} \int_{B} \star Ric(K, \cdot) \\
& = - \frac1{4\pi} \int_{B} \epsilon(Ric(K, \cdot)^\sharp, \cdot, \cdot, \cdot) = - \frac1{4\pi} \int_{B} -  \left\langle(Ric(K, \cdot)^\sharp, N \right\rangle \\
& = \frac1{4\pi} \int_{B} R_{\mu\nu} K^\mu N^\nu = 2 \int_{B} \left(T_{\mu\nu} - \frac12 T g_{\mu\nu}\right) K^\mu N^\nu
\end{align*}
(we used $\epsilon = - N^\sharp \wedge \sigma$ in the second line, where $\sigma$ is the volume element for $B$). Note that the Komar mass is not, as one might have guessed, the integral
\[
M' = \int_{B} T_{\mu\nu} K^\mu N^\nu.
\]
This integral is also well defined in a stationary spacetime, since
\[
\nabla^\mu ( T_{\mu\nu} K^\nu) = (\nabla^\mu T_{\mu\nu}) K^\nu + T_{\mu\nu} \nabla^\mu K^\nu = 0,
\]
due to the contracted Bianchi identity and the Killing equation. However, we also have
\[
\nabla^\mu (T K_\mu) = K^\mu \nabla_\mu T + T \nabla^\mu K_\mu = 0,
\]
because $T$ is constant along $K$, which has zero divergence (by contracting the Killing equation). Therefore we also have
\[
\nabla^\mu \left( \left(T_{\mu\nu} - \frac12 T g_{\mu\nu} \right) K^\nu\right)  = 0.
\]
To have an idea of what exactly is measured by the Komar mass, consider a static spacetime, where it is possible to choose $B$ such that $N = e^{-\phi} K$. In this case, for a perfect fluid (whose flow lines are necessarily the integral curves of $K$), we have
\[
M = 2 \int_{B} \left(T(N,N) + \frac12 T\right) e^\phi = 2 \int_{B} \left(\rho + \frac12 (- \rho + 3 p) \right) e^\phi = \int_{B} \left(\rho + 3 p \right) e^\phi.
\]
Thus we see that the Komar mass also includes the pressure; this is reminiscent of the Newtonian formula for the internal energy of a monoatomic gas,
\[
U = \frac32 pV.
\]

\section{Field theory} \label{sec6.2}

Let us consider the problem of how to define the energy of a field $\psi$ in flat Minkowski space. The field equations for the field $\psi$ are usually the equations for the critical points of an {\bf action}
\[
S = \int_{\bbR^4} \mathcal{L}(\psi, \partial\psi) \, dt \, dx^1 dx^2 dx^3,
\]
obtained by integrating a {\bf Lagrangian density} $\mathcal{L}$ (where we assume that the field decays fast enough so that $S$ is well defined). The field equations are then
\[
\partial_\mu \left( \frac{\partial\mathcal{L}}{\partial(\partial_\mu \psi)}\right) - \frac{\partial\mathcal{L}}{\partial\psi} = 0.
\]
Indeed, if $\psi(\lambda)$ is a one-parameter family of fields such that $\psi(0)$ is a critical point of the action, and $\delta \equiv \frac{d}{d \lambda}_{|_{\lambda=0}}$, then
\begin{align*}
\delta S & = \int_{\bbR^4} \left( \frac{\partial\mathcal{L}}{\partial\psi} \delta \psi + \frac{\partial\mathcal{L}}{\partial(\partial_\mu \psi)} \delta (\partial_\mu \psi) \right) = \int_{\bbR^4} \left( \frac{\partial\mathcal{L}}{\partial\psi} \delta \psi + \frac{\partial\mathcal{L}}{\partial(\partial_\mu \psi)} \partial_\mu (\delta\psi) \right) \\
& = \int_{\bbR^4} \partial_\mu \left( \frac{\partial\mathcal{L}}{\partial(\partial_\mu \psi)} \delta\psi \right) + \int_{\bbR^4} \left( \frac{\partial\mathcal{L}}{\partial\psi} - \partial_\mu \left(\frac{\partial\mathcal{L}}{\partial(\partial_\mu \psi)} \right) \right) \delta\psi.
\end{align*}
For field variations $\delta \psi$ with compact support we then have
\[
\delta S = \int_{\bbR^4} \left( \frac{\partial\mathcal{L}}{\partial\psi} - \partial_\mu \left(\frac{\partial\mathcal{L}}{\partial(\partial_\mu \psi)} \right) \right) \delta\psi.
\]
The {\bf canonical energy-momentum tensor} is defined as
\[
T^{\mu}_{\,\,\,\,\nu} = \frac{\partial\mathcal{L}}{\partial(\partial_\mu \psi)} \partial_\nu \psi - \mathcal{L} \delta^{\mu}_{\,\,\,\,\nu}
\]
and satisfies
\begin{align*}
\partial_\mu T^{\mu}_{\,\,\,\,\nu} & = \partial_\mu \left( \frac{\partial\mathcal{L}}{\partial(\partial_\mu \psi)}\right) \partial_\nu \psi + \frac{\partial\mathcal{L}}{\partial(\partial_\mu \psi)} \partial_\mu \partial_\nu \psi \\
& - \frac{\partial\mathcal{L}}{\partial\psi} (\partial_\mu\psi) \delta^{\mu}_{\,\,\,\,\nu} - \frac{\partial\mathcal{L}}{\partial(\partial_\alpha \psi)} (\partial_\mu \partial_\alpha \psi) \delta^{\mu}_{\,\,\,\,\nu} = 0.
\end{align*}
Defining the {\bf Hamiltonian density} to be
\[
\mathcal{H} = -T^0_{\,\,\,\,0} = -\frac{\partial\mathcal{L}}{\partial(\partial_0 \psi)} \partial_0 \psi + \mathcal{L},
\]
it is then clear from the divergence theorem that the {\bf Hamiltonian}
\[
H = \int_{S_t} \mathcal{H} \, dx^1 dx^2 dx^3
\]
is independent of the hypersurface $S_t$ of constant $t$ chosen to compute the integral (assuming that the field decays fast enough so that $H$ is well defined and the boundary integral corresponding to the divergence term vanishes). The Hamiltonian can be identified with the total energy of the field $\psi$, which is therefore constant in time.

If we have $N$ fields $\psi_1, \ldots, \psi_N$ instead of a single field $\psi$, then it is easily seen that the field equations are 
\[
\partial_\mu \left( \frac{\partial\mathcal{L}}{\partial(\partial_\mu \psi_i)}\right) - \frac{\partial\mathcal{L}}{\partial\psi_i} = 0 \qquad (i=1, \ldots, N),
\]
the canonical energy-momentum tensor is
\[
T^{\mu}_{\,\,\,\,\nu} = \frac{\partial\mathcal{L}}{\partial(\partial_\mu \psi_i)} \partial_\nu \psi_i - \mathcal{L} \delta^{\mu}_{\,\,\,\,\nu}
\]
(summed over $i$), and the Hamiltonian density is
\[
\mathcal{H} = - T^0_{\,\,\,\,0} = - \frac{\partial\mathcal{L}}{\partial(\partial_0 \psi_i)} \partial_0 \psi_i + \mathcal{L}.
\]
Note carefully that up until this point the metric was not used to raise or lower indices, or even to apply the divergence theorem. This will be important in Section~\ref{sec6.4}.

If we use Cartesian coordinates then the tensor
\[
T^{\mu\nu} = \frac{\partial\mathcal{L}}{\partial(\partial_\mu \psi_i)} \partial^\nu \psi_i - \mathcal{L} g^{\mu\nu}
\]
satisfies the conservation equation
\[
\nabla_\mu T^{\mu\nu} = 0.
\]
This provides a method for obtaining the energy-momentum tensor that is used in the Einstein field equations for the various matter models. We now list some simple examples.

\subsection{Klein-Gordon field}

The Lagrangian density for the Klein-Gordon field is
\[
\mathcal{L} = \frac12 (g^{\mu\nu}\partial_\mu\phi\,\partial_\nu\phi + m^2\phi^2).
\]
Consequently the canonical energy-momentum tensor for the Klein-Gordon field is
\[
T^{\mu\nu} = \partial^\mu \phi \, \partial^\nu \phi - \frac12 g^{\mu\nu} \left( \partial_\alpha \phi \, \partial^\alpha \phi + m^2 \phi^2 \right),
\]
as claimed in Chapter~\ref{chapter5}.

\subsection{Electromagnetic field}

In this case we can take as our fields the electromagnetic potentials $A_0, A_1, A_2, A_3$. The Lagrangian density (in units where $4\pi\varepsilon_0=1$) is
\[
\mathcal{L} = \frac1{16\pi} g^{\alpha\mu}g^{\beta\nu} F_{\alpha\beta} F_{\mu\nu},
\]
where
\[
F_{\mu\nu} = \partial_\mu A_\nu - \partial_\nu A_\mu.
\]
Using
\[
\frac{\partial F_{\alpha\beta}}{\partial(\partial_\mu A_\nu)} = \delta^{\mu}_{\,\,\,\,\alpha} \delta^{\nu}_{\,\,\,\,\beta} - \delta^{\nu}_{\,\,\,\,\alpha} \delta^{\mu}_{\,\,\,\,\beta},
\]
it is easily seen that
\[
\partial_\mu \left( \frac{\partial\mathcal{L}}{\partial(\partial_\mu A_\nu)}\right) - \frac{\partial\mathcal{L}}{\partial A_\nu} = 0 \Leftrightarrow \partial_\mu F^{\mu\nu} = 0,
\]
which are indeed the Maxwell equations $d\star F = 0$, as
\begin{align*}
& \epsilon_{\mu\gamma\delta\nu} \nabla^\mu \epsilon^{\alpha\beta\gamma\delta} F_{\alpha\beta} =  \epsilon_{\mu\nu\gamma\delta} \epsilon^{\alpha\beta\gamma\delta} \nabla^\mu F_{\alpha\beta} \\
& = -2(\delta^{\mu}_{\,\,\,\,\alpha} \delta^{\nu}_{\,\,\,\,\beta} - \delta^{\nu}_{\,\,\,\,\alpha} \delta^{\mu}_{\,\,\,\,\beta}) \nabla^\mu F_{\alpha\beta} = -2 \nabla^\mu F_{\mu\nu}.
\end{align*}
Note that the equations $dF=0$ follow automatically from the definition $F=dA$. The canonical energy-momentum tensor for the electromagnetic field is
\[
T^{\mu\nu}_{\text{can}} = \frac{\partial\mathcal{L}}{\partial(\partial_\mu A_\alpha)} \partial^\nu A_\alpha - \mathcal{L} g^{\mu\nu} = \frac1{4\pi} \left( F^{\mu\alpha} \partial^\nu A_\alpha - \frac14 F_{\alpha\beta}F^{\alpha\beta} g^{\mu\nu}\right).
\]
This tensor is neither symmetric nor gauge-invariant. However,
\[
F^{\mu\alpha} \partial^\nu A_\alpha = F^{\mu\alpha} F^\nu_{\,\,\,\,\alpha } + F^{\mu\alpha} \partial_\alpha A^\nu,
\]
where the first term is symmetric and gauge-invariant, and the second term is divergenceless:
\[
\partial_\mu \left( F^{\mu\alpha} \partial_\alpha A^\nu \right) = F^{\mu\alpha}\partial_\mu\partial_\alpha A^\nu = 0
\]
(because $F$ is antisymmetric and the partial derivatives commute). Therefore the tensor
\[
T^{\mu\nu} = \frac1{4\pi} \left(F^{\mu\alpha} F^\nu_{\,\,\,\,\alpha } - \frac14 F_{\alpha\beta}F^{\alpha\beta} g^{\mu\nu}\right)
\]
is symmetric, gauge-invariant and divergenceless. This is the true energy-momentum tensor for the electromagnetic field.

\subsection{Relativistic elasticity}

A continuous medium can be described by a Riemannian 3-manifold $(S, k)$ (the {\bf relaxed configuration}) and projection map $\pi: \bbR^4 \to S$ whose level sets are timelike curves (the worldlines of the medium particles), as shown in Figure~\ref{congruence}.

\begin{figure}[h!]
\begin{center}
\psfrag{M}{$\bbR^4$}
\psfrag{p}{$\downarrow\pi$}
\psfrag{S}{$S$}
\psfrag{h}{$h$}
\psfrag{k}{$k$}
\epsfxsize=0.6\textwidth
\leavevmode
\epsfbox{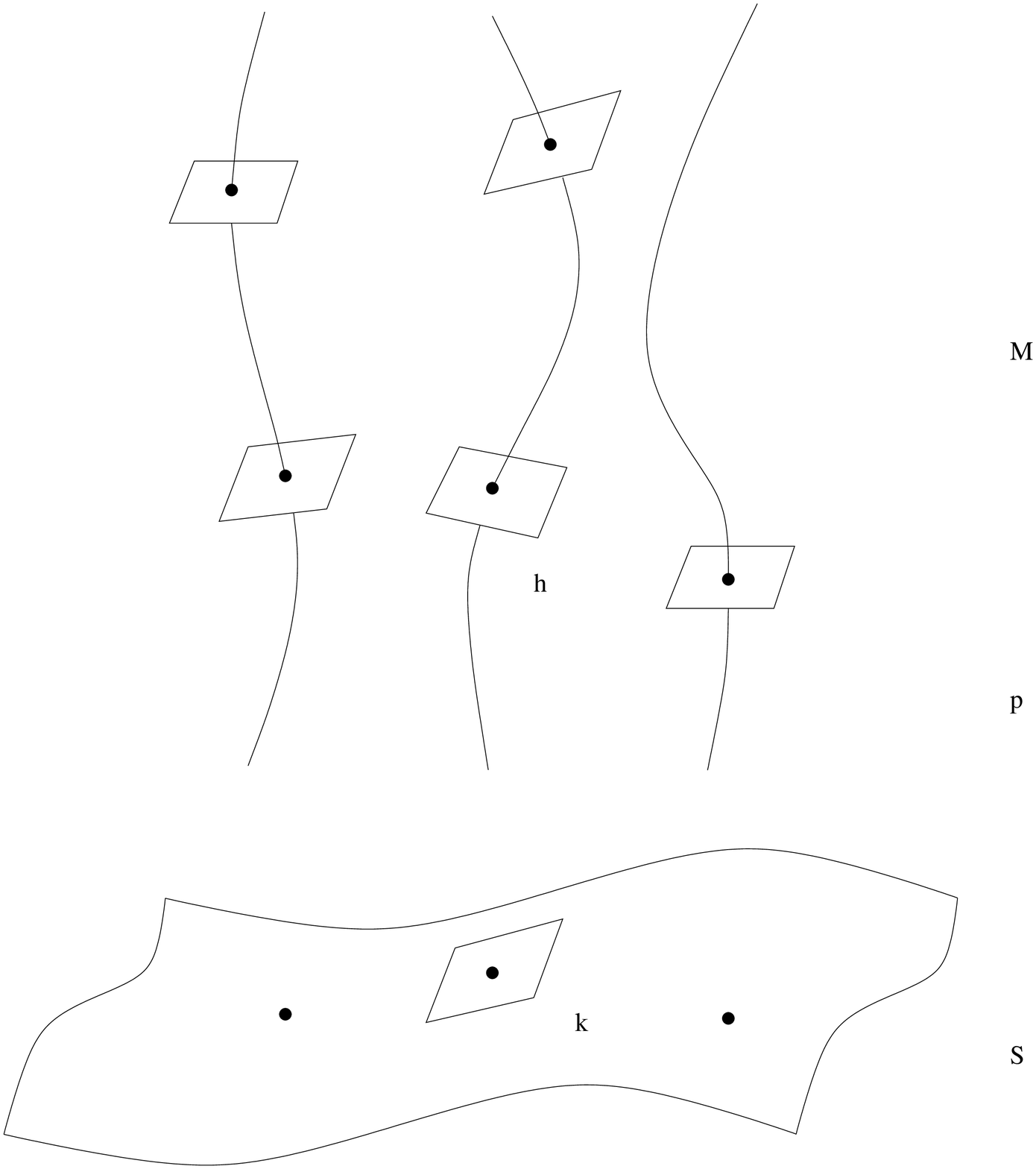}
\end{center}
\caption{A continuous medium in Minkowski's spacetime.} \label{congruence}
\end{figure}

If we choose local coordinates $(\bar{x}^1, \bar{x}^2, \bar{x}^3)$ on $S$ then we can think of $\pi$ as a set of three scalar fields $\bar{x}^1, \bar{x}^2, \bar{x}^3:\bbR^4 \to \bbR$. For a given worldline, we can complete this set of scalar fields into local coordinates $(\bar{t},\bar{x}^1, \bar{x}^2, \bar{x}^3)$ for $\bbR^4$ such that $\bar{t}$ is the proper time along that worldline and its level sets are orthogonal to it:
\[
g = - d\bar{t}^2 + h_{ij} d\bar{x}^i d\bar{x}^j \qquad \text{(on the worldline)}.
\]
Notice that the orthogonal metric
\[
h = h_{ij} d\bar{x}^i d\bar{x}^j
\]
can be thought of as a time-dependent Riemannian metric on $S$, describing the local deformations of the medium along each worldline, that is, the deviations from the natural metric
\[
k = k_{ij} d\bar{x}^i d\bar{x}^j.
\]
We can compute the (inverse) metric $h$ from
\[
h^{ij} = g^{\mu\nu} \frac{\partial \bar{x}^i}{\partial x^\mu} \frac{\partial \bar{x}^j}{\partial x^\nu},
\]
which does not depend on the choice of $\bar{t}$. In other words, the metric $h$ is a quadratic function of the partial derivatives of the fields $\bar{x}^1, \bar{x}^2, \bar{x}^3$. An elastic Lagrangian density $\cL$ for these fields is obtained by assuming that $\cL = \cL(\bar{x}^i, h^{ij})$. The canonical energy-momentum tensor is
\[
T^{\mu\nu} = \frac{\partial\mathcal{L}}{\partial(\partial_\mu \bar{x}^i)} \partial^\nu \bar{x}^i - \mathcal{L} g^{\mu\nu},
\]
and so in the coordinate system $(\bar{t},\bar{x}^1, \bar{x}^2, \bar{x}^3)$ we have
\[
T^{\bar{0}\bar{0}} = \mathcal{L},
\]
that is, the elastic Lagrangian density is just the rest energy density $\rho = T^{\bar{0}\bar{0}}$ measured by each particle in the medium. The choice of $\rho = \rho(\bar{x}^i, h^{ij})$ is called the {\bf elastic law} of the continuous medium.

We define {\bf homogeneous and isotropic materials} to be those for which $\rho$ depends only on the eigenvalues $({s_1}^2,{s_2}^2,{s_3}^2)$ of $h_{ij}$ with respect to $k_{ij}$ (that is, the eigenvalues of the matrix $(h_{ij})$ in a frame where $k_{ij}=\delta_{ij}$). Note that $(s_1,s_2,s_3)$ are the stretch factors along the principal directions given by the eigenvectors of $h_{ij}$, that is, the distance between two nearby points along a principal direction in the current configuration, as measured by the metric $h$, divided by the distance between the same points in the relaxed configuration, as measured by the metric $k$.

Assume that $k_{ij}=\delta_{ij}$, that is, that $(S, k)$ is the Euclidean space. We define the more convenient variables
\begin{align*}
& \lambda_0 = \det (h^{ij}) = \frac1{(s_1 s_2 s_3)^2} ; \\
& \\
& \lambda_1 = \tr (h^{ij}) = \frac1{{s_1}^2} + \frac1{{s_2}^2} + \frac1{{s_3}^2}; \\
& \\
& \lambda_2 = \tr \cof (h^{ij}) = \frac1{(s_1s_2)^2} + \frac1{(s_2 s_3)^2} + \frac1{(s_3 s_1)^2}.
\end{align*}
Note that 
\[
s_1 s_2 s_3 = \left(\frac1{\lambda_0}\right)^\frac12
\]
is the volume occupied in the deformed state by a unit volume of material in the relaxed configuration. Equivalently, 
\[
n = (\lambda_0)^\frac12
\]
is the number density of particles of the medium in the deformed state, if we normalize the number density in the relaxed configuration to be $1$ particle per unit volume. Elastic media whose elastic law depends only on $n$,
\[
\rho=\rho(\lambda_0),
\]
are simply perfect fluids. To check this, we note that from the formula for the inverse of a matrix we have
\[
h_{ij} = \frac1{\lambda_0} A^{ji} = \frac1{\lambda_0} A^{ij},
\]
where $A^{ij}$ is the $(i,j)$-cofactor of $(h^{ij})$. On the other hand, from the Laplace expansion for determinants we have
\[
\lambda_0 = \sum_{j=1}^3 h^{ij} A^{ij}  \qquad (\text{no sum over } i),
\]
and so
\[
\frac{\partial \lambda_0}{\partial h^{ij}} = A^{ij} = \lambda_0 h_{ij}.
\]
Therefore
\begin{align}
T^{\mu\nu} & = \frac{d \rho}{d \lambda_0} \frac{\partial \lambda_0}{\partial h^{ij}} \frac{\partial h^{ij}}{\partial(\partial_\mu \bar{x}^k)} \partial^\nu \bar{x}^k - \rho g^{\mu\nu} \nonumber \\
& = \frac{d \rho}{d \lambda_0} \lambda_0 h_{ij} (g^{\mu\alpha} \delta^{i}_{\,\,\,\,k} \partial_\alpha \bar{x}^j + g^{\mu\alpha} \partial_\alpha \bar{x}^i \delta^{j}_{\,\,\,\,k})\partial^\nu \bar{x}^k - \rho g^{\mu\nu} \nonumber \\
& = 2\lambda_0\frac{d \rho}{d \lambda_0} h_{ij} \partial^\mu \bar{x}^i \partial^\nu \bar{x}^j - \rho g^{\mu\nu}. \nonumber
\end{align}
Since
\[
h_{\mu\nu} = h_{ij} \partial_\mu \bar{x}^i \partial_\nu \bar{x}^j
\]
is simply the metric on the hyperplanes orthogonal to the worldlines, that is,
\[
h_{\mu\nu} = g_{\mu\nu} + U_\mu U_\nu,
\]
where $U$ is the unit tangent vector to the worldlines, we obtain
\begin{align}
T^{\mu\nu} & = 2\lambda_0\frac{d \rho}{d \lambda_0}U^\mu U^\nu + \left(2\lambda_0\frac{d \rho}{d \lambda_0} - \rho\right) g^{\mu\nu} \nonumber \\
& = (\rho + p) U^\mu U^\nu + p g^{\mu\nu} \nonumber
\end{align}
with
\[
p = 2\lambda_0\frac{d \rho}{d \lambda_0} - \rho,
\]
which is indeed the energy-momentum tensor of a perfect fluid. For example, dust corresponds to the elastic law $\rho=\rho_0 \sqrt{\lambda_0}$ (for some positive constant $\rho_0$), yielding $p=0$, and a stiff fluid, with equation of state $p=\rho$, is given by the choice $\rho=\rho_0 \lambda_0$. The ``hard phase" rigid fluid introduced by Christodoulou, with equation of state $p=\rho-\rho_0$, corresponds to $\rho=\frac{\rho_0}2(\lambda_0+1)$. 

To obtain elastic materials that are not fluids we must choose elastic laws that also depend on $\lambda_1$ and $\lambda_2$. For instance, an elastic law is said to be {\bf quasi-Hookean} if it is of the form
\[
\rho = \hat{p}(n) + \hat{\mu}(n) \sigma,
\]
where $\sigma$ is a {\bf shear scalar}, that is, a non-negative function of the stretch factors such that $\sigma=0$ if and only if $s_1 = s_2 = s_3$. The functions $\hat{\rho}$ and $\hat{\mu}$ are called the {\bf unsheared energy density} and the {\bf rigidity modulus} of the elastic material. Examples of these are the {\bf John quasi-Hookean material}, corresponding to the shear scalar
\[
\sigma = \frac{{s_1}^2 + {s_2}^2 + {s_3}^2}{\left({s_1}^2 {s_2}^2 {s_3}^2 \right)^\frac13} - 3,
\]
and the the {\bf Karlovini-Samuelsson quasi-Hookean material}, corresponding to the shear scalar
\[
\sigma = \frac1{12} \left[ \left( \frac{s_1}{s_2} - \frac{s_2}{s_1} \right)^2 + \left( \frac{s_1}{s_3} - \frac{s_3}{s_1} \right)^2 + \left( \frac{s_2}{s_3} - \frac{s_3}{s_2} \right)^2 \right].
\]
It is easily seen that the first elastic law is of the form
\[
\rho=f(\lambda_0)+g(\lambda_0) \lambda_2,
\]
whereas the second is of the form
\[
\rho=f(\lambda_0)+g(\lambda_0) \lambda_1 \lambda_2.
\]
Other examples are the {\bf stiff ultra-rigid material} of Karlovini and Samuelsson, given by
\[
\rho=\frac{\rho_0}4 (\lambda_2 + 1),
\]
and the {\bf Brotas rigid solid}, given by
\[
\rho=\frac{\rho_0}{8}(\lambda_0 + \lambda_1 + \lambda_2 + 1)
\]
(where $\rho_0$ is a positive constant).

\section{Einstein-Hilbert action} \label{sec6.3}

The variational formulation of field theories is coordinate-free, and so it gives a simple method to write the field equations on an arbitrary coordinate system. We must be careful, however, to note that the action written in the new coordinate system must include the Jacobian of the coordinate transformation:
\[
S = \int_{\bbR^4} \mathcal{L}(\psi, \partial\psi)  \sqrt{|\det(g_{\mu\nu})|} \, dx^0 dx^1 dx^2 dx^3.
\]
This suggests that to generalize the field equations to an arbitrary curved spacetime $(M,g)$ one should consider actions of this form, where $\cL$ should be invariant under coordinate changes. In particular, one may wonder if there is a Lagrangian action for the metric itself which yields the Einstein field equations. The answer to this question is affirmative, and the corresponding action is known as the {\bf Einstein-Hilbert action}:
\[
S = \int_{M} R  \sqrt{|\det(g_{\mu\nu})|} \, dx^0 dx^1 dx^2 dx^3 = \int_{M} R \epsilon,
\]
where $R$ and $\epsilon$ are the scalar curvature and the volume element of the metric $g$, and the integral is over an arbitrary (oriented) manifold $M$. Note that $R$ depends on $g$ and its first and second partial derivatives, unlike the Lagrangian densities that we encountered before. Instead of deriving the Euler-Lagrange equations for this case, we will proceed in a more geometric way. We note here, however, that this action is exceptional in that it leads to second-order equations for the metric, as opposed to the fourth order equations that are typical of Lagrangian densities depending on second partial derivatives.

To obtain the Euler-Lagrange equations for the Einstein-Hilbert action we start by considering two affine connections $\nabla$ and $\tilde{\nabla}$ on $M$. Because
\[
(\tilde{\nabla}_X - \nabla_X) (fY) = f (\tilde{\nabla}_X - \nabla_X) Y,
\]
there exists a tensor $C$ such that
\[
(\tilde{\nabla}_X - \nabla_X) Y = C(X,Y) \Leftrightarrow \tilde{\nabla}_\mu Y^\nu = \nabla_\mu Y^\nu + C^\nu_{\mu\alpha} Y^\alpha.
\]
If both $\nabla$ and $\tilde{\nabla}$ are symmetric then
\[
0 = (\tilde{\nabla}_X Y - \tilde{\nabla}_Y X) - (\nabla_X Y - \nabla_Y X) = C(X,Y) - C(Y,X),
\]
that is, $C$ is symmetric:
\[
C^\alpha_{\mu\nu} = C^\alpha_{\nu\mu}. 
\]
Using the Leibnitz rule, it is easy to determine the relation between the covariant derivatives of any tensor using the two connections: for example,
\[
\tilde{\nabla}_\alpha T^\beta_{\mu\nu} = \nabla_\alpha T^\beta_{\mu\nu} + C^\beta_{\alpha\gamma} T^\gamma_{\mu\nu} - C^\gamma_{\alpha\mu} T^\beta_{\gamma\nu} - C^\gamma_{\alpha\nu} T^\beta_{\mu\gamma}.
\]
Assume now that $\tilde{\nabla}$ is the Levi-Civita connection for the metric $g$. Then we have
\[
0 = \tilde{\nabla}_\alpha g_{\mu\nu} = \nabla_\alpha g_{\mu\nu} - C^\beta_{\alpha\mu} g_{\beta\nu} - C^\beta_{\alpha\nu} g_{\mu\beta} = \nabla_\alpha g_{\mu\nu} - C_{\nu\alpha\mu} - C_{\mu\alpha\nu}.
\]
By subtracting this identity from its cyclic permutations,
\begin{align*}
& \nabla_\mu g_{\nu\alpha} = C_{\alpha\mu\nu} + C_{\nu\mu\alpha}, \\
& \nabla_\nu g_{\alpha\mu} = C_{\mu\nu\alpha} + C_{\alpha\nu\mu},
\end{align*}
we readily obtain
\begin{align}
& 2 C_{\alpha\mu\nu} = \nabla_\mu g_{\nu\alpha} + \nabla_\nu g_{\alpha\mu} - \nabla_\alpha g_{\mu\nu} \Leftrightarrow \nonumber \\ 
& C^\alpha_{\mu\nu} = \frac12 g^{\alpha\beta} \left(\nabla_\mu g_{\nu\beta} + \nabla_\nu g_{\mu\beta} - \nabla_\beta g_{\mu\nu}\right). \label{Christoffel}
\end{align} 
Moreover, we have
\begin{align*}
\tilde{\nabla}_\mu  \tilde{\nabla}_\nu X^\alpha = & \, \tilde{\nabla}_\mu (\nabla_\nu X^\alpha + C^\alpha_{\nu\beta} X^\beta ) =  \nabla_\mu \nabla_\nu X^\alpha + C^\alpha_{\mu\beta} \nabla_\nu X^\beta - C^\beta_{\mu\nu} \nabla_\beta X^\alpha \\
& + \nabla_\mu C^\alpha_{\nu\beta} X^\beta + C^\alpha_{\nu\beta} \nabla_\mu X^\beta + C^\alpha_{\mu\gamma} C^\gamma_{\nu\beta} X^\beta - C^\gamma_{\mu\nu} C^\alpha_{\gamma\beta} X^\beta,
\end{align*}
whence
\begin{align*}
\tilde{R}_{\mu\nu\,\,\,\,\beta}^{\,\,\,\,\,\,\,\,\alpha} X^\beta = & \, (\tilde{\nabla}_\mu \tilde{\nabla}_\nu - \tilde{\nabla}_\nu \tilde{\nabla}_\mu) X^\alpha = R_{\mu\nu\,\,\,\,\beta}^{\,\,\,\,\,\,\,\,\alpha} X^\beta \\
& + \left(\nabla_\mu C^\alpha_{\nu\beta} - \nabla_\nu C^\alpha_{\mu\beta} + C^\alpha_{\mu\gamma} C^\gamma_{\nu\beta} - C^\alpha_{\nu\gamma} C^\gamma_{\mu\beta} \right) X^\beta,
\end{align*}
that is,
\begin{equation} \label{curvature}
\tilde{R}_{\mu\nu\,\,\,\,\beta}^{\,\,\,\,\,\,\,\,\alpha} = R_{\mu\nu\,\,\,\,\beta}^{\,\,\,\,\,\,\,\,\alpha} + \nabla_\mu C^\alpha_{\nu\beta} - \nabla_\nu C^\alpha_{\mu\beta} + C^\alpha_{\mu\gamma} C^\gamma_{\nu\beta} - C^\alpha_{\nu\gamma} C^\gamma_{\mu\beta}.
\end{equation}
Note that we retrieve the usual formulae for the Christoffel symbols and the Riemann tensor from equations~\eqref{Christoffel} and \eqref{curvature} in the case when $\nabla_\mu = \partial_\mu$.

Let $g(\lambda)$ be a one-parameter family of Lorentzian metrics on $M$, and choose $\nabla$ and $\tilde{\nabla}$ to be the Levi-Civita connections of $g(0)$ and $g(\lambda)$. The difference between these connections is a tensor $C(\lambda)$ with $C(0)=0$. Again setting $\delta \equiv \frac{d}{d \lambda}_{|_{\lambda=0}}$, we have from~\eqref{Christoffel}
\begin{align*}
\delta C^\alpha_{\mu\nu} & = \frac12 g^{\alpha\beta} \left(\nabla_\mu \delta g_{\nu\beta} + \nabla_\nu \delta g_{\mu\beta} - \nabla_\beta \delta g_{\mu\nu}\right) \\
& = \frac12 \left(\nabla_\mu \delta g_{\nu}^{\,\,\,\,\alpha} + \nabla_\nu \delta g_{\mu}^{\,\,\,\,\alpha} - \nabla^\alpha \delta g_{\mu\nu} \right),
\end{align*}
where $g_{\mu\nu}$ means $g_{\mu\nu}(0)$ and all indices are raised or lowered with this metric. Note carefully that $\delta g$ means the tensor $\delta g_{\mu\nu}$, possible with some indices raised. Thus for instance
\[
\delta(g^{\mu\nu}) = - g^{\mu\alpha}g^{\nu\beta} \delta g_{\alpha\beta} = - \delta g^{\mu\nu}.
\]
From \eqref{curvature} we have
\[
\delta R_{\mu\nu\,\,\,\,\beta}^{\,\,\,\,\,\,\,\,\alpha} =  \nabla_\mu \delta C^\alpha_{\nu\beta} - \nabla_\nu \delta C^\alpha_{\mu\beta},
\]
whence
\begin{align*}
\delta R_{\mu\nu} & =  \nabla_\alpha \delta C^\alpha_{\mu\nu} - \nabla_\mu \delta C^\alpha_{\alpha\nu} \\
& = \frac12 \nabla_\alpha \left( \nabla_\mu \delta g_{\nu}^{\,\,\,\,\alpha} + \nabla_\nu \delta g_{\mu}^{\,\,\,\,\alpha} - \nabla^\alpha \delta g_{\mu\nu} \right) - \frac12 \nabla_\mu \nabla_\nu \delta g_{\alpha}^{\,\,\,\,\alpha}.
\end{align*}
Note that
\[
g^{\mu\nu} \delta R_{\mu\nu} = - \nabla_\mu \nabla^\mu \delta g_{\nu}^{\,\,\,\,\nu} +  \nabla^\mu \nabla^\nu \delta g_{\mu\nu} = \nabla^\mu \left( -\nabla_\mu \delta g_{\nu}^{\,\,\,\,\nu} + \nabla^\nu \delta g_{\mu\nu}\right)
\]
is a divergence with respect to the metric $g(0)$, and will vanish when integrated for variations of the metric with compact support.

The variation of the Einstein-Hilbert action is
\[
\delta S = \delta \int_M R \epsilon = \int_M \delta(R \epsilon) = \int_M( \delta R \epsilon + R \delta \epsilon).
\]
We have
\[
\delta R = \delta(g^{\mu\nu} R_ {\mu\nu}) = - \delta g^{\mu\nu} R_ {\mu\nu} + \nabla^\mu \left( -\nabla_\mu \delta g_{\nu}^{\,\,\,\,\nu} + \nabla^\nu \delta g_{\mu\nu}\right)
\]
and, using the identity
\[
\delta \det A = (\det A) \tr(A^{-1} \delta A),
\]
for any matrix-valued function $A$,
\begin{align*}
\delta \epsilon & = \delta \sqrt{-\det(g_{\mu\nu})} \, dx^0 \wedge dx^1 \wedge dx^2 \wedge dx^3 \\
& = - \frac12 \left(\det(g_{\mu\nu})\right) g^{\mu\nu} \delta g_{\mu\nu} \left(-\det(g_{\mu\nu})\right)^{-\frac12} \, dx^0 \wedge dx^1 \wedge dx^2 \wedge dx^3 \\
& = \frac12 g_{\mu\nu} \delta g^{\mu\nu} \epsilon.
\end{align*}
We conclude that
\begin{align*}
\delta S & = - \int_M \left(R_ {\mu\nu} - \frac12 R g_{\mu\nu} \right) \delta g^{\mu\nu} \epsilon + \int_M \nabla^\mu \left( -\nabla_\mu \delta g_{\nu}^{\,\,\,\,\nu} + \nabla^\nu \delta g_{\mu\nu}\right) \epsilon \\
& = - \int_M G_{\mu\nu} \delta g^{\mu\nu} \epsilon.
\end{align*}
for variations of the metric with compact support. We conclude that the Euler-Lagrange equations for the Einstein-Hilbert action are the Einstein equations $G_{\mu\nu} = 0$.

\begin{Remark}
It is easy to see that the Einstein tensor of a compact surface $(M,g)$ vanishes identically. Therefore we have
\[
\delta \int_M R \epsilon = 0
\]
automatically. If $g_1$ and $g_2$ are any two Riemannian metrics on $M$ then $g(\lambda)=(1-\lambda)g_1 + \lambda g_2$ is a Riemannian metric which interpolates between $g_1$ and $g_2$. We then have
\[
\frac{d}{d\lambda} \int_M R(\lambda) \epsilon(\lambda) = 0  \Rightarrow \int_M R_1 \epsilon_1 = \int_M R_2 \epsilon_2,
\]
that is, the integral of the scalar curvature does not depend on the metric. This statement is known as the {\bf Gauss-Bonnet theorem}.
\end{Remark}

To include matter fields $\psi$ and a cosmological constant $\Lambda$ in the Einstein equations we consider the action
\[
S = \int_M \left(\cL(g,\psi) - \frac1{16 \pi}(R - 2\Lambda)\right) \epsilon.
\]
It is clear that
\[
\delta S = \int_M E(\psi) \delta \psi \, \epsilon + \frac1{16 \pi} \int_M \left(G_{\mu\nu} + \Lambda g_{\mu\nu} - 8 \pi T_{\mu\nu} \right) \delta g^{\mu\nu} \epsilon,
\]
where we have defined
\[
\delta \int_M \cL(g,\psi) \epsilon = \int_M E(\psi) \delta \psi \, \epsilon - \frac12 \int_M T_{\mu\nu} \delta g^{\mu\nu} \epsilon.
\]
The Euler-Lagrange equations are then the Einstein equations with sources plus the field equations for $\psi$:
\[
\begin{cases}
G_{\mu\nu} + \Lambda g_{\mu\nu} = 8 \pi T_{\mu\nu} \\
E(\psi) = 0
\end{cases} .
\]
The energy-momentum tensor is then
\[
T_{\mu\nu} = - 2 \frac{\delta \cL}{\delta g^{\mu\nu}} - \cL g_{\mu\nu},
\]
where the second term comes from the variation of the volume element and we have set
\[
\delta \cL = \frac{\delta \cL}{\delta g^{\mu\nu}} \delta g^{\mu\nu} + E(\psi) \delta \psi.
\]
It is interesting to note that the energy-momentum tensor as defined here often agrees with what one would expect from the canonical energy-momentum tensor associated to the Lagrangian density $\cL$ in Minkowski's spacetime.

\section{Gravitational waves} \label{sec6.35}

Since we have computed the variation of the Ricci tensor, we make a short digression to discuss the linearized Einstein vacuum equations, which describe the propagation of gravitational waves on a fixed solution $(M,g)$ of the full nonlinear vacuum equations $R_{\mu\nu} = \Lambda g_{\mu\nu}$. The linearized equations are simply
\[
\delta R_{\mu\nu} = \Lambda \delta g_{\mu\nu},
\]
that is,
\begin{equation} \label{linearEinstein}
\nabla_\alpha \left( \nabla_\mu \delta g_{\nu}^{\,\,\,\,\alpha} + \nabla_\nu \delta g_{\mu}^{\,\,\,\,\alpha} - \nabla^\alpha \delta g_{\mu\nu} \right) - \nabla_\mu \nabla_\nu \delta g_{\alpha}^{\,\,\,\,\alpha} = 2 \Lambda \delta g_{\mu\nu}.
\end{equation}
Note that some variations of the metric are trivial, in that they arise from the diffeomorphism invariance of the Einstein equations: if $\psi_\lambda$ is a one-parameter family of diffeomorphisms, then the variation
\[
g(\lambda) = {\psi_\lambda}^* g
\]
yields metrics which are isometric to $g=g(0)$, but expressed in different coordinates. In this case we have
\[
\delta g = \mathcal{L}_V g,
\]
that is
\[
\delta g_{\mu\nu} = \nabla_\mu V_\nu + \nabla_\nu V_\mu,
\]
where $V$ is the vector field defined at each point $p \in M$ by
\[
V_p = \frac{d}{d \lambda}_{|_{\lambda=0}} \psi_\lambda(p).
\]
Therefore, the linearized Einstein equations have gauge freedom: we can always add a Lie derivative of $g$ to any given variation of the metric without altering its physical meaning. This corresponds to gauge transformations of the form
\[
\delta g_{\mu\nu} \to \delta g_{\mu\nu} + \nabla_\mu V_\nu + \nabla_\nu V_\mu.
\]
We will now construct a gauge where equations~\eqref{linearEinstein} look particularly simple. To do that, we consider the trace-reversed metric perturbation
\[
\overline{\delta g}_{\mu\nu} = \delta g_{\mu\nu} - \frac12 (\delta g_{\alpha}^{\,\,\,\,\alpha}) g_{\mu\nu},
\]
which transforms under a gauge transformation as
\[
\overline{\delta g}_{\mu\nu} \to \overline{\delta g}_{\mu\nu} + \nabla_\mu V_\nu + \nabla_\nu V_\mu - (\nabla_\alpha V^\alpha) g_{\mu\nu},
\]
so that its divergence transforms as
\begin{align*}
\nabla^\mu \overline{\delta g}_{\mu\nu} \to & \nabla^\mu \overline{\delta g}_{\mu\nu} + \Box V_\nu + \nabla_\mu \nabla_\nu V^\mu - \nabla_\nu \nabla_\alpha V^\alpha \\
& = \nabla^\mu \overline{\delta g}_{\mu\nu} + \Box V_\nu + R_{\mu\nu\,\,\,\,\alpha}^{\,\,\,\,\,\,\,\,\mu} V^\alpha \\
& = \nabla^\mu \overline{\delta g}_{\mu\nu} + \Box V_\nu + R_{\nu\alpha} V^\alpha.
\end{align*}
Assume that $(M,g)$ is globally hyperbolic. By solving the wave equation
\[
\Box V_\nu + R_{\nu\alpha} V^\alpha + \nabla^\mu \overline{\delta g}_{\mu\nu} = 0,
\]
we can then change to a gauge where
\[
\nabla^\mu \overline{\delta g}_{\mu\nu} = 0,
\]
that is,
\[
\nabla^\mu \delta g_{\mu\nu} = \frac12 \nabla_\nu \delta g_{\alpha}^{\,\,\,\,\alpha}.
\]
Taking the trace of equation~\eqref{linearEinstein} we then obtain
\begin{equation} \label{waveperturbation}
\Box \, \delta g_{\alpha}^{\,\,\,\,\alpha} + 2 \Lambda \delta g_{\alpha}^{\,\,\,\,\alpha} = 0.
\end{equation}
Note that we still have a residual gauge freedom, corresponding to vector fields $V$ such that
\[
\Box V_\nu + R_{\nu\alpha} V^\alpha = 0.
\]
Choosing $V$ and its derivatives on some Cauchy hypersurface $S$ such that
\[
\delta g_{\alpha}^{\,\,\,\,\alpha} + 2 \nabla_\alpha V^\alpha = N \cdot \left( \delta g_{\alpha}^{\,\,\,\,\alpha} + 2 \nabla_\alpha V^\alpha \right) = 0,
\]
where $N$ is the future-pointing unit normal to $S$, and solving the wave equation for $V$ above, we guarantee the existence of a gauge where 
\[
\delta g_{\alpha}^{\,\,\,\,\alpha} = N \cdot \left( \delta g_{\alpha}^{\,\,\,\,\alpha}\right) = 0
\]
on $S$. The wave equation~\eqref{waveperturbation} then guarantees that
\[
\delta g_{\alpha}^{\,\,\,\,\alpha} = 0
\]
on $M$, and so
\[
\nabla^\mu \delta g_{\mu\nu} = 0.
\]
This is the so-called {\bf transverse traceless gauge}. In this gauge, the linearized Einstein equation~\eqref{linearEinstein} can be written as
\[
\nabla_\alpha \nabla_\mu \delta g_{\nu}^{\,\,\,\,\alpha} + \nabla_\alpha \nabla_\nu \delta g_{\mu}^{\,\,\,\,\alpha} - \Box \delta g_{\mu\nu} = 2 \Lambda \delta g_{\mu\nu}.
\]
Using
\begin{align*}
\nabla_\alpha \nabla_\mu \delta g_{\nu}^{\,\,\,\,\alpha} & = \nabla_\mu \nabla_\alpha \delta g_{\nu}^{\,\,\,\,\alpha} + R_{\alpha\mu\nu\beta} \delta g^{\beta\alpha} + R_{\alpha\mu\,\,\,\,\beta}^{\,\,\,\,\,\,\,\,\alpha} \delta g_{\nu}^{\,\,\,\,\beta} \\
& = R_{\alpha\mu\nu\beta} \delta g^{\alpha\beta} + R_{\mu\beta} \delta g_{\nu}^{\,\,\,\,\beta} \\
& = R_{\alpha\mu\nu\beta} \delta g^{\alpha\beta} + \Lambda \delta g_{\mu\nu},
\end{align*}
we finally obtain
\[
\Box \delta g_{\mu\nu} - 2R_{\alpha\mu\nu\beta} \delta g^{\alpha\beta} = 0.
\]

\section{ADM mass} \label{sec6.4}

If we write the metric in the Gauss Lemma form,
\[
g = - dt^2 + h_{ij}(t,x) dx^i dx^j 
\]
then from the exercises in Chapter~\ref{chapter5} we have
\[
R = \bar{R} + 2 \frac{\partial}{\partial t} \left(K^i_{\,\,\,\,i}\right) + \left(K^i_{\,\,\,\,i}\right)^2 + K_{ij} K^{ij}.
\]
Setting
\[
X = K^i_{\,\,\,\,i} \frac{\partial}{\partial t}\, ,
\]
we have
\[
\dive X = \nabla_0 X^0 + \nabla_i X^i = \partial_0 X^0 + \Gamma_{i0}^i X^0 = \partial_0 X^0 + K^i_{\,\,\,\,i} X^0,
\]
that is,
\[
\frac{\partial}{\partial t} \left(K^i_{\,\,\,\,i}\right) = \dive X - \left(K^i_{\,\,\,\,i}\right)^2.
\]
We conclude that
\[
R = \bar{R} - \left(K^i_{\,\,\,\,i}\right)^2 + K_{ij} K^{ij} + 2 \dive X,
\]
and so the Einstein-Hilbert action corresponds to the Lagrangian density
\[
\cL = \sqrt{\det(h_{ij})} \left(\bar{R} - \left(K^i_{\,\,\,\,i}\right)^2 + K_{ij} K^{ij}\right).
\]
Therefore the Hamiltonian density is
\begin{align*}
\mathcal{H} & = - \frac{\partial \cL}{\partial(\partial_0 h_{ij})} \partial_0 h_{ij} + \cL = - \frac{\partial \cL}{\partial(K_{ij})} K_{ij} + \cL \\
& = - \sqrt{\det(h_{ij})} \left( 2 K^{ij} - 2 K^l_{\,\,\,\,l} h^{ij} \right) K_{ij} + \cL \\
& =\sqrt{\det(h_{ij})} \left( \bar{R} + \left(K^i_{\,\,\,\,i}\right)^2 - K_{ij} K^{ij}\right) = 2\sqrt{\det(h_{ij})} \, G_{00} = 0
\end{align*}
for any solution of the vacuum Einstein field equations. That is, the total energy associated to the fields $h_{ij}$ is simply zero.

\begin{figure}[h!]
\begin{center}
\psfrag{n}{$n$}
\psfrag{M}{$M$}
\psfrag{dM}{$\partial M$}
\psfrag{d/dt}{$\frac{\partial}{\partial t}$}
\psfrag{S}{$\Sigma$}
\epsfxsize=0.4\textwidth
\leavevmode
\epsfbox{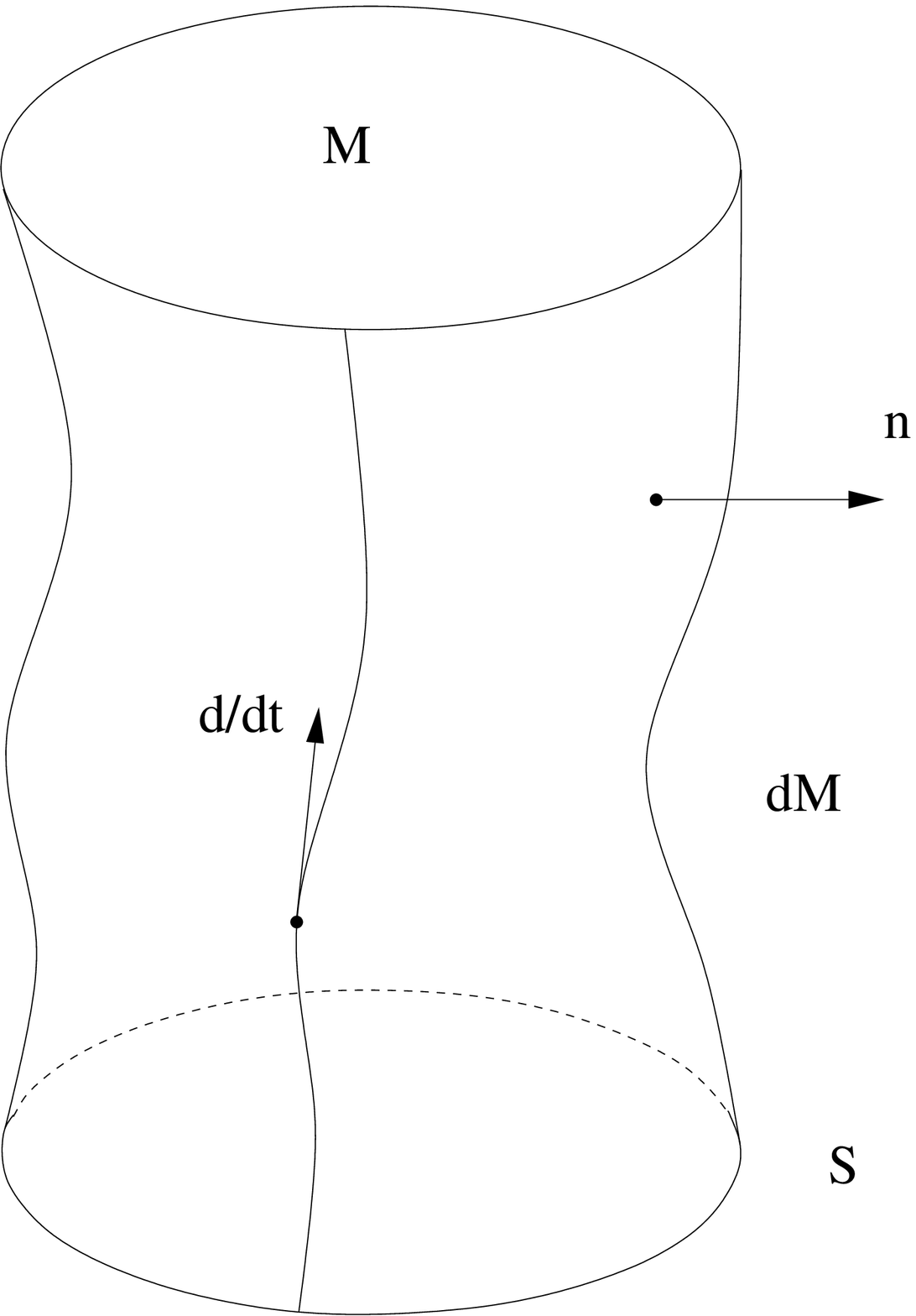}
\end{center}
\caption{Computing the boundary terms for the Einstein-Hilbert action.} \label{boundary}
\end{figure}

To try to obtain a nonzero quantity we reconsider the boundary terms that were discarded when varying the Einstein-Hilbert action:
\[
\int_M \nabla^\mu \left( -\nabla_\mu \delta g_{\nu}^{\,\,\,\,\nu} + \nabla^\nu \delta g_{\mu\nu}\right) \epsilon = \int_{\partial M} \left( -\nabla_\mu \delta g_{\nu}^{\,\,\,\,\nu} + \nabla^\nu \delta g_{\mu\nu}\right) n^\mu \omega,
\]
where we take $M$ to be a manifold with a timelike boundary $\partial M$ at infinity, tangent to $\frac{\partial}{\partial t}$, with unit normal $n$ and volume element $\omega$ (Figure~\ref{boundary}). Note that it is not necessary to consider the flux of the vector field $X$ which we discarded above as it is orthogonal to $n$. The boundary integral can be written as
\[
\int_\bbR \int_{\Sigma} \left( - \partial_i (h^{jk}\delta h_{jk}) + h^{jk}\bar\nabla_k \delta h_{ij}\right) n^i \sigma dt,
\]
where we take $\partial M$ to be the flow by $\frac{\partial}{\partial t}$ of a spacelike surface $\Sigma$ and $\omega = dt \wedge \sigma$. Suppose that $h$ approaches the Euclidean metric at infinity, so that  in an appropriate coordinate system $(x^1,x^2,x^3)$ we have
\[
h_{ij} = \delta_{ij} + O\left(r^{-p}\right), \quad \delta h_{ij} = O\left(r^{-p}\right), \quad \partial_k h_{ij}, \partial_k \delta h_{ij}, \Gamma^k_{ij} = O\left(r^{-p-1}\right) 
\]
as $r \to +\infty$, with $r^2 = {(x^1)}^2 + {(x^2)}^2 + {(x^3)}^2$ and $p>\frac12$. Then the boundary integral can be written as
\begin{align*}
& \lim_{r \to \infty} \int_\bbR \int_{S_r} \left( - \partial_i \delta h_{jj} + \partial_j \delta h_{ij} \right) \frac{x^i}{r} dt \\
& = \delta \lim_{r \to \infty} \int_\bbR \int_{S_r} \left( - \partial_i h_{jj} + \partial_j h_{ij} \right) \frac{x^i}{r} dt = \delta I,
\end{align*}
where $S_r$ is a coordinate sphere of radius $r$. Note that $\delta$ can be brought outside the integral because we have replaced the metric-dependent terms $\bar\nabla$, $n$ and $\sigma$. 

If $S$ is the Einstein-Hilbert action, we then have
\[
\delta S = - \int_M G_{\mu\nu} \delta g^{\mu\nu} + \delta I,
\]
and so the Einstein equations hold if and only if
\[
\delta(S-I) = 0.
\]
One can think of $I$ as the integral of a singular Lagrangian density $\mathcal{I}$. The Einstein-Hilbert Lagrangian density $\cL$ should therefore be replaced by $\cL - \mathcal{I}$, and so, since $\mathcal{I}$ does not depend on $K_{ij}$, the corresponding Hamiltonian density $\mathcal{H}$ should be replaced by $\mathcal{H} - \mathcal{I} = - \mathcal{I}$. The Hamiltonian should then be
\[
H = - \lim_{r \to \infty} \int_{S_r} \left( \partial_j h_{ij} - \partial_i h_{jj} \right) \frac{x^i}{r}.
\]
This Hamiltonian suggests a definition of the total energy of the gravitational field at a given time slice, provided that $h$ approaches the Euclidean metric at infinity.

\begin{Def}
A $3$-dimensional Riemannian manifold $(S,h)$ is said to be {\bf asymptotically flat} if there exist:
\begin{enumerate}[(i)]
\item
A compact set $K \subset S$ such that $S \setminus K$ is diffeomorphic to $\bbR^3 \setminus \overline{B_1}(0)$;
\item
A chart $(x^1,x^2,x^3)$ on $S \setminus K$ (called a {\bf chart at infinity}) such that 
\[
|h_{ij} - \delta_{ij}| + r |\partial_k h_{ij}| + r^2 |\partial_k \partial_l h_{ij}|= O(r^{-p}) \text{ and } \bar{R}=O(r^{-q})
\]
for some $p > \frac12$ and $q>3$, where $r^2 = {(x^1)}^2 + {(x^2)}^2 + {(x^3)}^2$ and $\bar{R}$ is the scalar curvature of $h$.
\end{enumerate}

\end{Def}

\begin{figure}[h!]
\begin{center}
\psfrag{S}{$S$}
\psfrag{K}{$K$}
\epsfxsize=1.0\textwidth
\leavevmode
\epsfbox{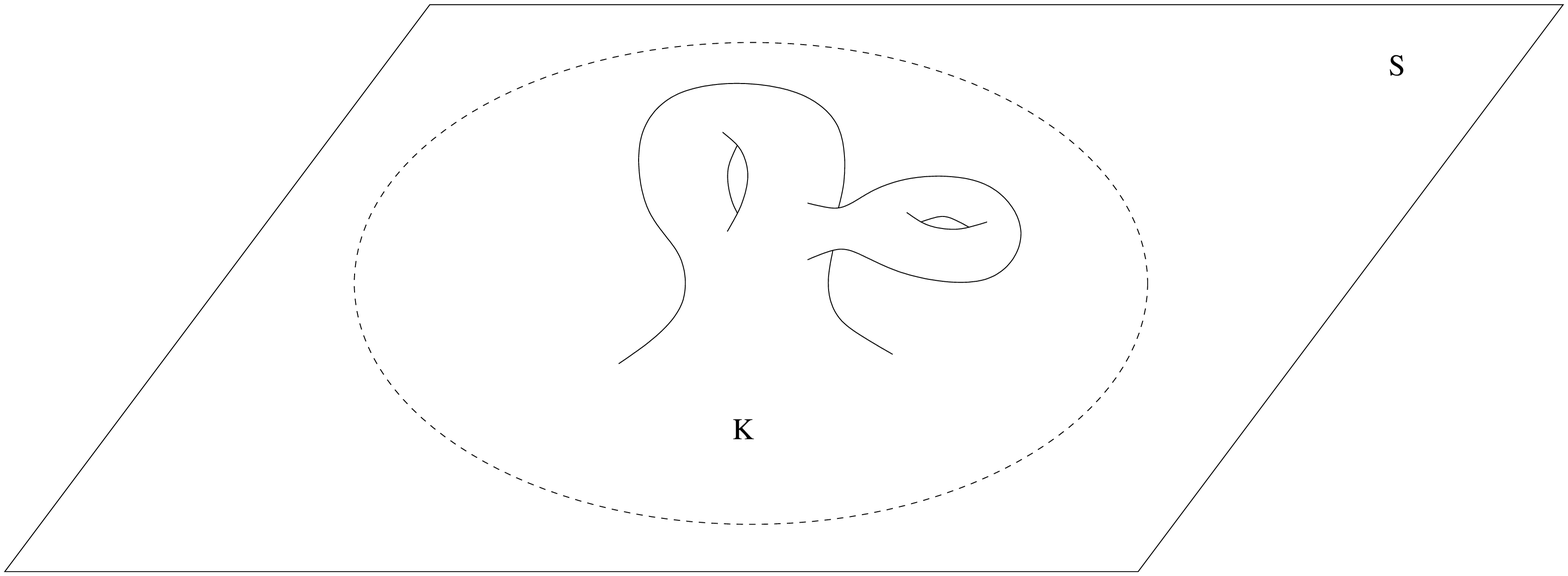}
\end{center}
\caption{Asymptotically flat manifold.} \label{asymp_flat}
\end{figure}

\begin{Def} ({\bf Arnowitt-Deser-Misner \cite{ADM61}})
The {\bf ADM mass} of an asymptotically flat Riemannian manifold $(S,h)$ is
\[
M = \lim_{r \to + \infty} \frac{1}{16 \pi} \int_{S_r}  \left( \partial_j h_{ij} - \partial_i h_{jj} \right) \frac{x^i}r,
\]
where $S_r$ is a sphere of radius $r$ in the chart at infinity $(x^1,x^2,x^3)$.
\end{Def}

Note that in this definition the Hamiltonian has been multiplied by the factor $-\frac1{16\pi}$ that is used when coupling to matter fields.

\begin{Thm} ({\bf Ashtekar \cite{AA79}}) 
If $(S,h)$ is asymptotically flat and the maximal Cauchy development of $(S,h,K)$ is stationary then the Komar mass of the maximal Cauchy development coincides with the ADM mass of $(S,h)$.
\end{Thm}

\begin{Thm} ({\bf Bartnik \cite{B86}}) 
The ADM mass is well defined, that is, it does not depend on the choice of the chart at infinity.
\end{Thm}

\section{Positive mass theorem} \label{sec6.5}

Since the gravitational field is attractive, its energy is presumably negative, at least for bound states. On the other hand, we expect gravitational waves to carry positive energy. It is therefore an important question to decide whether there is a lower bound for the ADM mass (the inexistence of which would signal an instability). In the simplest case of time-symmetric initial data ($K=0$) the restriction equations reduce to
\[
\bar{R} = 16 \pi \rho,
\]
and we expect $\rho \geq 0 \Leftrightarrow \bar{R} \geq 0$ for reasonable matter fields. In this case, we have the following famous result.

\begin{Thm} ({\bf Schoen and Yau \cite{SY81}}) 
Let $(S,h)$ be a complete asymptotically flat Riemannian $3$-manifold with scalar curvature $\bar{R} \geq 0$. Then:
\begin{enumerate}[(i)]
\item
Its ADM mass is nonnegative, $M \geq 0$.
\item
If $M=0$ then $S=\bbR^3$ and $h$ is the Euclidean metric.
\end{enumerate}
\end{Thm}

\begin{proof}
Following \cite{L10}, we give a proof of $(i)$ only for asymptotically flat Riemannian $3$-manifolds $S$ that are graphs of smooth functions $f:\bbR^3 \to \bbR$ with the metric $h$ induced by the Euclidean metric of $\bbR^4$. Using the Cartesian coordinates $(x^1,x^2,x^3)$ of $\bbR^3$ as global coordinates on the graph we have
\[
h_{ij} = \delta_{ij} + \partial_i f \partial_j f.
\]
From that expression one can show that the scalar curvature of the graph is the divergence of a vector field on $\bbR^3$:
\[
\bar{R} = \partial_i \left( \frac{1}{1 + |\grad f|^2} (\partial_i f \partial_j\partial_j f - \partial_i \partial_j f \partial_j f)\right).
\]
On the other hand,
\begin{align*}
M & = \lim_{r \to + \infty} \frac{1}{16 \pi} \int_{S_r}  \left( \partial_j h_{ij} - \partial_i h_{jj} \right) \frac{x^i}r \\
& = \lim_{r \to + \infty} \frac{1}{16 \pi} \int_{S_r}  \left( \partial_i \partial_j f \partial_j f + \partial_i f \partial_j \partial_j f - 2 \partial_i \partial_j f \partial_j f \right) \frac{x^i}r \\
& = \lim_{r \to + \infty} \frac{1}{16 \pi} \int_{S_r}  \left( \partial_i f \partial_j \partial_j f - \partial_i \partial_j f \partial_j f \right) \frac{x^i}r .
\end{align*}
Since $(S,h)$ is asymptotically flat, the derivatives of $f$ approach zero with certain decays as $r \to + \infty$. One can then easily show that
\begin{align*}
M & = \lim_{r \to + \infty} \frac{1}{16 \pi} \int_{S_r}  \left( \frac{1}{1 + |\grad f|^2} (\partial_i f \partial_j \partial_j f - \partial_i \partial_j f \partial_j f) \right) \frac{x^i}r \\
& = \frac{1}{16 \pi} \int_{\bbR^3}  \partial_i \left( \frac{1}{1 + |\grad f|^2} (\partial_i f \partial_j \partial_j f - \partial_i \partial_j f \partial_j f) \right) \\
& = \frac{1}{16 \pi} \int_{\bbR^3} \bar{R} \geq 0.
\end{align*}
We do not prove the rigidity statement $(ii)$. It is interesting to note that this statement, together with the formula above for the ADM mass, implies that any graph with zero scalar curvature is flat.
\end{proof}

The volume element of the graph is
\[
\epsilon = \sqrt{1 + |\grad f|^2} \, dx^1 \wedge dx^2 \wedge dx^3,
\]
since the eigenvalues of the matrix $(h_{ij})$ are $1 + |\grad f|^2$ for the eigenvector $\grad f$ and $1$ for the eigenvectors orthogonal to $\grad f$. Consequently the ADM mass is
\[
M = \frac{1}{16 \pi} \int_{S} \frac{\bar{R}}{\sqrt{1 + |\grad f|^2}} \, \epsilon =  \int_{S} \frac{\rho}{\sqrt{1 + |\grad f|^2}} \, \epsilon < \int_{S} \rho \, \epsilon. 
\]
The difference
\[
M - \int_{S} \rho \, \epsilon < 0
\]
can be thought of as the (negative) gravitational binding energy.

\section{Penrose inequality} \label{sec6.6}

The positive mass theorem admits a refinement in the case when black holes are present, known as the Penrose inequality. The idea is that black hole horizons correspond to minimal surfaces $\Sigma$ on the Riemannian manifold $(S,h)$, each contributing with a mass $M$ at least as big as the mass of a Schwarzschild black hole with the same event horizon area $A$:
\[
M \geq \sqrt{\frac{A}{16\pi}}.
\]
To understand the motivation for this inequality, recall that in the proof of the Penrose singularity theorem we defined the outward null expansion of a $2$-surface $\Sigma$ on a Cauchy hypersurface $S$ as
\[
\theta = \frac12 \gamma^{AB} \frac{\partial \gamma_{AB}}{\partial r} = \frac12 \tr \left(\cL_{\frac{\partial}{\partial r}} g\right)_{|_{T\Sigma}}
\]
If $N$ is the future-pointing unit normal to $S$ and $n$ is the outward unit normal to $\Sigma$ on $S$ we then have
\[
\theta = \frac12 \tr \left(\cL_{N + n} g\right)_{|_{T\Sigma}} = \tr K_{|_{T\Sigma}} + \frac12 \tr \left(\cL_{n} g\right)_{|_{T\Sigma}},
\]
where we used the fact that $\frac{\partial}{\partial r} = N + n$ on $\Sigma$, and also that if $X,Y$ are tangent to $\Sigma$ then
\[
\left(\cL_Z g\right) (X,Y) = \left\langle X, \nabla_Y Z \right\rangle + \left\langle Y, \nabla_X Z \right\rangle
\]
depends only on $Z$ along $\Sigma$. For time-symmetric initial data ($K=0$) this becomes
\[
\theta = \frac12 \tr \left(\cL_{n} g\right)_{|_{T\Sigma}} = \frac12 \tr \left(\cL_{n} h\right)_{|_{T\Sigma}} = \tr \kappa,
\]
where $h$ is the metric induced by $g$ on $S$ and $\kappa$ is the second fundamental form of $\Sigma$ on $S$. We conclude that $\Sigma$ is {\bf marginally trapped}, that is, has zero outward null expansion, if and only $\tr \kappa = 0$, which is precisely the condition for $\Sigma$ to be a minimal surface. Now a marginally trapped surface anticipates the formation of trapped surfaces, which will lead to geodesic incompleteness of the resulting spacetime, presumably due to singularities. If one believes the {\bf weak cosmic censorship conjecture}, these singularities should only occur inside black holes, and so there should be a black hole horizon envelloping any marginally trapped surface.

\begin{Thm} ({\bf Huisken and Ilmanen \cite{HI01}, Bray \cite{B01}}) 
Let $(S,h)$ be a complete asymptotically flat Riemannian manifold with scalar curvature $\bar{R} \geq 0$. Then: 
\begin{enumerate}[(i)]
\item
Its ADM mass satisfies $M \geq \sqrt{\frac{A}{16\pi}}$, where $A$ is the sum of the areas of the outer minimal surfaces.
\item
If $M=\sqrt{\frac{A}{16\pi}}$ then the restriction of $(S,h,0)$ to the exterior of the outer minimal surfaces coincides with the initial data for the Schwarzschild solution of mass $M$ outside the event horizon.
\end{enumerate}
\end{Thm}

\begin{proof}
Following \cite{L10}, we give a proof of only for asymptotically flat Riemannian $3$-manifolds $S$ that are graphs of smooth functions $f:U\subset \bbR^3 \to \bbR$, with the metric $h$ induced by the Euclidean metric of $\bbR^4$. Here
\[
U = \bbR^3 \setminus \bigcup_{a=1}^N K_a,
\]
where $K_1, \ldots, K_N$ are disjoint convex compact sets with smooth boundaries $\partial K_1, \ldots, \partial K_N$ to which $f$ extends as a constant and where $|\grad f|\to+\infty$ (so that they are minimal surfaces of the graph). Applying the divergence theorem as in the proof of the positive mass theorem for graphs we now obtain
\[
M = \frac{1}{16 \pi} \int_{U} \bar{R} + \frac{1}{16 \pi} \sum_{a=1}^N \int_{\partial K_a}  \left( \frac{1}{1 + |\grad f|^2} (\partial_i f \partial_j \partial_j f - \partial_i \partial_j f \partial_j f) \right) n^i,
\]
where $n$ is the outward unit normal to $\partial K_a$. Now as is well known the Laplacian $\Delta f$ of $f$ in $U$ is related to the Laplacian $\bar{\Delta} f$ of $f$ in $\partial K_a$ by
\[
\Delta f = \bar{\Delta} f + Hf(n,n) +  (n \cdot f) \tr \kappa,
\]
where $Hf$ is the Hessian of $f$ and $\kappa$ is the second fundamental form of $\partial K_a$ in $\bbR^3$. Since $f$ is constant on $\partial K_a$ we have $\bar{\Delta} f = 0$ and so
\begin{align*}
& (\partial_i f \partial_j \partial_j f - \partial_i \partial_j f \partial_j f) n^i = (n \cdot f) \Delta f - Hf(n, \grad f) \\
& = (n \cdot f) Hf(n,n) + (n \cdot f)^2 \tr \kappa - \langle \grad f, n \rangle Hf(n,n) \\
& = \langle \grad f, n \rangle^2 \tr \kappa = |\grad f|^2 \tr \kappa,
\end{align*}
where we used $\grad f = \langle \grad f, n \rangle n$. Therefore
\begin{align*}
M & = \frac{1}{16 \pi} \int_{U} \bar{R} + \frac{1}{16 \pi} \sum_{a=1}^N \int_{\partial K_a}  \left( \frac{|\grad f|^2}{1 + |\grad f|^2} \right) \tr \kappa \\
& = \frac{1}{16 \pi} \int_{U} \bar{R} + \frac{1}{16 \pi} \sum_{a=1}^N \int_{\partial K_a}  \tr \kappa.
\end{align*}
Now {\bf Minkowski's inequality} for the smooth boundaries of compact convex sets states that
\[
\int_{\partial K_a}  \tr \kappa \geq \sqrt{16 \pi A_a},
\]
where $A_a$ is the area of $\partial K_a$. Since $\bar{R} \geq 0$ we conclude that
\begin{align*}
M & \geq \sum_{a=1}^N \sqrt{\frac{A_a}{16 \pi}} \geq \sqrt{\frac{1}{16 \pi} \sum_{a=1}^N A_a}, 
\end{align*}
where we used $\sqrt{a+b} \leq \sqrt{a} + \sqrt{b}$ for $a,b>0$.

We do not prove the rigidity statement $(ii)$. It is interesting to note that this statement, together with the formula above for the ADM mass, implies that any graph of the type considered above with zero scalar curvature is a Flamm paraboloid.
\end{proof}

\begin{figure}[h!]
\begin{center}
\psfrag{S}{$S$}
\psfrag{A1}{$A_1$}
\psfrag{A2}{$A_2$}
\epsfxsize=1.0\textwidth
\leavevmode
\epsfbox{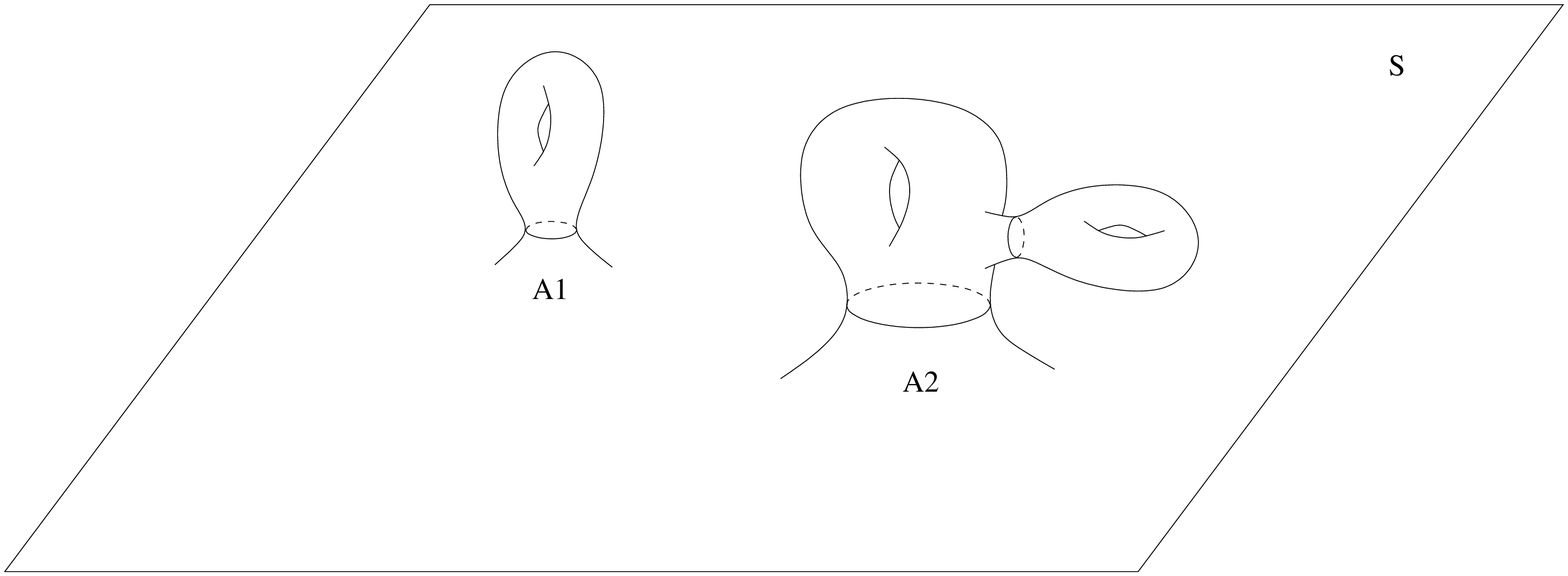}
\end{center}
\caption{Penrose inequality; the sum of the areas of the outer minimal surfaces is $A=A_1 + A_2$.} \label{Penrose_ineq}
\end{figure}

\section{Exercises} \label{sec6.7}

\begin{enumerate}

\item
Let $g$ be a static spherically symmetric Lorentzian metric on $\bbR^4$ whose matter fields have spatially compact support and satisfy the dominant energy condition. There exist smooth functions $\phi=\phi(r)$ and $m=m(r)$ such that
\[
g = - e^{2\phi(r)} dt^2 + \frac{dr^2}{1-\frac{2m(r)}{r}} + r^2 \left( d\theta^2 + \sin^2\theta d\varphi^2 \right).
\]
Show that:
\begin{enumerate}
\item
The Einstein equations imply
\begin{align*}
& \frac{dm}{dr}=4\pi r^2\rho; \\
& \frac{d\phi}{dr}=\frac{m+4\pi r^3p}{r(r-2m)},
\end{align*}
where $\rho$ and $p$ are the energy density and the radial pressure as measured by the static observers. 
\item
There exist constants $M\geq 0$ and $\Phi \in \bbR$ such that
\[
\lim_{r \to +\infty} m(r)=M \quad \text{ and } \quad \lim_{r \to +\infty} \phi(r)=\Phi.
\]
\item
If we choose the coordinate $t$ such that $\Phi=0$ then $M$ is the Komar mass of $g$ with respect to the timelike Killing vector field $\frac{\partial}{\partial t}$.
\item
The constant $M$ satisfies $M \leq E$, where
\[
E = \int_{\{t=0\}} \rho,
\]
with equality exactly when $g$ is the Minkowski metric. 
\end{enumerate}

\item
Starting with the appropriate Lagrangian density, show that the Klein-Gordon equation can be written as
\[
\hspace{2cm} \frac{1}{\sqrt{|\det(g_{\mu\nu})|}} \partial_\alpha \left( \sqrt{|\det(g_{\mu\nu})|} \, \partial^\alpha \phi \right) - m^2 \phi = 0.
\]

\item
Starting with the Einstein-Hilbert-Klein-Gordon action
\[
\hspace{2cm} S = \int_M \left[ \frac12\left(g^{\mu\nu}\partial_\mu \phi \, \partial_\nu \phi + m^2 \phi^2\right) - \frac1{16\pi} R \right] \epsilon
\]
obtain the energy-momentum tensor for $\phi$:
\[
\hspace{2cm} T_{\mu\nu} = \partial_\mu\phi\,\partial_\nu\phi - \frac12 g_{\mu\nu} \left(\partial_\alpha\phi\,\partial^\alpha\phi+m^2 \phi^2\right).
\]
Note that this agrees with what one would expect from the canonical energy-momentum tensor.

\item
Show that the Einstein-Hilbert action is equivalent to the (first order, non-geometric) {\bf Einstein action}
\[
\hspace{2cm} S = \int_{M} g^{\alpha\beta} \left( \Gamma_{\alpha\gamma}^\delta \Gamma_{\beta\delta}^\gamma - \Gamma_{\alpha\beta}^\gamma \Gamma_{\gamma\delta}^\delta \right) \sqrt{|\det(g_{\mu\nu})|} \, dx^0 dx^1 dx^2 dx^3,
\]
by showing that they differ by a term of the form $\partial_\alpha Q^\alpha$, where
\[
\hspace{2cm} Q^\alpha = \left( g^{\beta \gamma} \Gamma^{\alpha}_{\beta\gamma} - g^{\alpha \beta} \Gamma^{\gamma}_{\beta\gamma} \right) \sqrt{|\det(g_{\mu\nu})|}.
\]
The following formulae will be useful:
\[
\hspace{2cm} \partial_\alpha \sqrt{|\det(g_{\mu\nu})|} = - \frac12 g_{\beta\gamma} \partial_\alpha g^{\beta\gamma} \sqrt{|\det(g_{\mu\nu})|}
\]
and
\[
\hspace{2cm} \nabla_\alpha g^{\beta\gamma} = 0 \Leftrightarrow \partial_\alpha g^{\beta\gamma} = - \Gamma^{\beta}_{\alpha\delta} g^{\delta\gamma} - \Gamma^{\gamma}_{\alpha\delta} g^{\beta\delta}.
\] 

\item
Let $f:\bbR^3 \to \bbR$ be a smooth function and consider the metric $h$ induced on its graph $S$ by the Euclidean metric in $\bbR^4$,
\[
h_{ij} = \delta_{ij} + \partial_i f \partial_j f.
\]
Show that:
\begin{enumerate}
\item
The inverse metric is
\[
h^{ij} = \delta_{ij} - \frac{\partial_i f \partial_j f}{1 + |\grad f|^2}.
\]
\item
The Christoffel symbols are
\[
\bar{\Gamma}^k_{ij} = \frac{\partial_k f \partial_i\partial_j f}{1 + |\grad f|^2}.
\]
\item
The Ricci tensor is
\begin{align*}
\hspace{2cm} \bar{R}_{ij} = & \frac{\partial_i\partial_j f \partial_k \partial_k f - \partial_i\partial_k f \partial_j \partial_k f }{1 + |\grad f|^2} \\
& + \frac{\partial_k f \partial_l f \partial_i\partial_k f \partial_j \partial_l f - \partial_k f \partial_l f \partial_i\partial_j f \partial_k \partial_l f}{\left(1 + |\grad f|^2\right)^2}.
\end{align*}
\item
The scalar curvature is
\begin{align*}
\hspace{2cm} \bar{R} = & \frac{\partial_i\partial_i f \partial_j \partial_j f - \partial_i\partial_j f \partial_i \partial_j f }{1 + |\grad f|^2} \\
& - \frac{2 \partial_j f \partial_k f (\partial_i\partial_i f \partial_j \partial_k f - \partial_i\partial_j f \partial_i \partial_k f )}{\left(1 + |\grad f|^2\right)^2}.
\end{align*}
\item
The scalar curvature can be written as
\[
\hspace{2cm} \bar{R} = \partial_i \left( \frac{1}{1 + |\grad f|^2} (\partial_i f \partial_j\partial_j f - \partial_i \partial_j f \partial_j f)\right).
\]
\end{enumerate}

\item
Let $h$ be the spherically symmetric Riemannian metric defined in $\bbR^3$ by
\[
h = \frac{dr^2}{1-\frac{2m(r)}{r}} + r^2 \left( d\theta^2 + \sin^2\theta d\varphi^2 \right),
\]
where $m$ is a smooth function whose derivative has compact support.
\begin{enumerate}
\item
Check that in Cartesian coordinates we have
\[
h_{ij} = \delta_{ij} + \frac{\frac{2m(r)}{r^3}}{1-\frac{2m(r)}{r}} x^i x^j.
\]
\item
Show that if the limit
\[
M=\lim_{r \to +\infty} m(r)
\]
exists then $h$ is asymptotically flat with ADM mass $M$ (which in particular coincides with the Komar mass when appropriate).
\item
Check that $h$ has scalar curvature
\[
\bar{R} = \frac{4}{r^2}\frac{dm}{dr},
\]
and use this to prove the Riemannian positive mass theorem for $h$.
\item
Show that $r=r_0$ is a minimal surface if and only if $m(r_0)=\frac{r_0}2$ (in which case $r$ is a well-defined coordinate only for $r>r_0$), and use this to prove the Riemannian Penrose inequality for $h$. 
\end{enumerate}

\item
Consider a Riemannian metric given in the Gauss Lemma form
\[
g = dt^2 + h_{ij}(t,x) dx^i dx^j,
\]
so that the hypersurface $t=0$ is a Riemannian manifold with induced metric $h=h_{ij}dx^i dx^j$ and second fundamental form 
\[
K=\frac12\frac{\partial h_{ij}}{\partial t}dx^idx^j.
\]
Show that:
\begin{enumerate}
\item
The Laplacian operators $\Delta$ and $\bar{\Delta}$ of $g$ and $h$ are related by
\[
\Delta f = \bar{\Delta} f + Hf(\partial_0,\partial_0) + (\partial_0f) \tr K,
\]
where $Hf$ is the Hessian of $f$.
\item
The metric induced on the hypersurface $t=\lambda f(x)$ is
\[
h(\lambda) = \left[h_{ij}(\lambda f(x),x) + \lambda^2 \partial_i f \partial_j f\right] dx^idx^j.
\]
\item
The first variation of this metric is
\[
\delta h \equiv \frac{d}{d\lambda}_{|_{\lambda=0}} h(\lambda)= 2fK_{ij}dx^i dx^j.
\]
\item
The first variation of the volume element is
\[
\delta \sigma \equiv \frac{d}{d\lambda}_{|_{\lambda=0}} \sigma(\lambda) = f \tr K \sigma.
\]
\item
The second variation of the volume element is
\begin{align*}
\hspace{2cm} \delta^2 \sigma & \equiv \frac{d^2}{d\lambda^2}_{|_{\lambda=0}} \sigma(\lambda) \\
& = \left[ \frac12 f^2 \left(\bar{R}-R+\left(K^{i}_{\,\, i}\right)^2-K_{ij}K^{ij}\right) + |\grad f|^2 \right] \sigma,
\end{align*}
where $R$ and $\bar{R}$ are the scalar curvatures of $g$ and $h$.
\item
There is no metric on the $3$-torus with positive scalar curvature. (You will need to use the fact that any metric in the $3$-torus admits a minimizing $2$-torus; this type of idea is used in the proof of the rigidity statement in the positive mass theorem.)
\end{enumerate}

\item
Consider the surfaces $\Sigma_t$ obtained from the boundary $\partial K \equiv \Sigma_0$ of a compact convex set $K \subset \bbR^3$ by flowing a distance $t$ along the unit normal. Let $A(t)$ and $V(t)$ be the area of $\Sigma_t$ and the volume bounded by $\Sigma_t$.
\begin{enumerate}
\item
Show that
\[
\dot{A}(0) = \int_{\partial K} \tr \kappa,
\]
where $\kappa$ is the second fundamental form of $\Sigma$.
\item
Prove that $\ddot{A}(t)=8\pi$, implying $A(t)=4\pi t^2 + \dot{A}(0)t + A(0)$.
\item
Conclude that $V(t)=\frac{4\pi}{3}t^3 + \frac{\dot{A}(0)}{2}t^2 + A(0)t + V(0)$.
\item
Use the isoperimentric inequality $A(t)^3 \geq 36 \pi V(t)^2$ to prove Minkowski's inequality:
\[
\dot{A}(0) \geq \sqrt{16 \pi A(0)}.
\]
\end{enumerate}

\end{enumerate}


\chapter{Black holes} \label{chapter7}

In this chapter we study black holes and the laws of black hole thermodynamics, following \cite{Townsend97} (see also \cite{BCH73, Poisson07}). An elementary discussion of quantum field theory in curved spacetime and the Hawking radiation can be found in \cite{Carroll03}.

\section{The Kerr solution} \label{sec7.1}

General rotating black holes are described by the {\bf Kerr metric}, given in the so-called {\bf Boyer-Lindquist coordinates} by
\begin{align*}
ds^2 = & - \left( 1 - \frac{2Mr}{\rho^2} \right) dt^2 - \frac{4Mar\sin^2\theta}{\rho^2} dt d\varphi +  \frac{\rho^2}{\Delta} dr^2 \\
& + \rho^2 d \theta^2 + \left( r^2 + a^2 + \frac{2Ma^2r\sin^2\theta}{\rho^2} \right) \sin^2\theta d\varphi^2,
\end{align*}
where 
\begin{align*}
& \rho^2 = r^2 + a^2 \cos^2 \theta, \\
& \Delta = r^2 - 2Mr + a^2,
\end{align*}
and $M,a \in \bbR$ are constants. Note that the Schwarzschild metric is a particular case, corresponding to $a=0$. It is possible to prove that the Kerr metric solves the vacuum Einstein field equations (see for instance \cite{ONeill95}). 

The Kerr metric is not spherically symmetric, but admits a two-dimensional group of isometries, generated by the Killing vector fields $X = \frac{\partial}{\partial t}$ and $Y=\frac{\partial}{\partial \varphi}$. The Komar mass associated to $X$ is
\[
M_{\text{Komar}} = - \frac1{8\pi} \int_{\Sigma} \star d X^\sharp,
\]
where $\Sigma$ is a $2$-surface of constant $(t,r)$, and can be computed to be
\[
M_{\text{Komar}} = M.
\]

The expression for the Komar mass in terms of the energy-momentum tensor, given in Chapter~\ref{chapter6}, suggests the definition of the {\bf Komar angular momentum} as
\[
J_{\text{Komar}} = \frac1{16\pi} \int_{\Sigma} \star d Y^\sharp
\]
(note the change in sign and absolute value of the constant, due to the fact that $Y$ is now spacelike and essentially orthogonal to the timelike unit normal $N$). The same exact argument as was done for the Komar mass shows that $J_{\text{Komar}}$ does not depend on the choice of $\Sigma$. Performing the calculation for the Kerr metric yields
\[
J_{\text{Komar}} = Ma,
\]
and so the parameter $a$ can be interpreted as the angular momentum per unit mass.

The Killing vector $X$ becomes null on the hypersurface given by the equation
\[
r = M + \sqrt{M^2-a^2\cos^2\theta},
\]
known as the {\bf ergosphere}. However, it is easy to show that the metric induced on this hypersurface is Lorentzian, and so it cannot be the black hole event horizon (since it can be crossed both ways by timelike curves). The event horizon corresponds to the hypersurface $r=r_+$, where
\[
r_+ = M + \sqrt{M^2-a^2}.
\]
Indeed, the function $\Delta$ changes sign on this hypersurface, and so $\grad r$ becomes timelike (meaning that $r$ must decrease along causal curves). Note that the ergosphere encloses the event horizon, touching it only at the poles, as shown in Figure~\ref{ergosphere}. The region in between, where $X$ is spacelike, is called the {\bf ergoregion}, because matter fields satisfying the dominant energy condition can have negative energy there, and so extract energy from the black hole when absorbed (a mechanism known as the {\bf Penrose process} in the case of particles or {\bf superradiance} in the case of fields). Note also that the existence of an event horizon requires that $|a|\leq M$. If $|a|< M$ the black hole is said to be {\bf subextremal}, and if $|a|=M$ it is called {\bf extremal}.

\begin{figure}[h!]
\begin{center}
\psfrag{event horizon}{event horizon}
\psfrag{ergosphere}{ergosphere}
\psfrag{ergoregion}{ergoregion}
\psfrag{black hole}{black hole}
\epsfxsize=0.8\textwidth
\leavevmode
\epsfbox{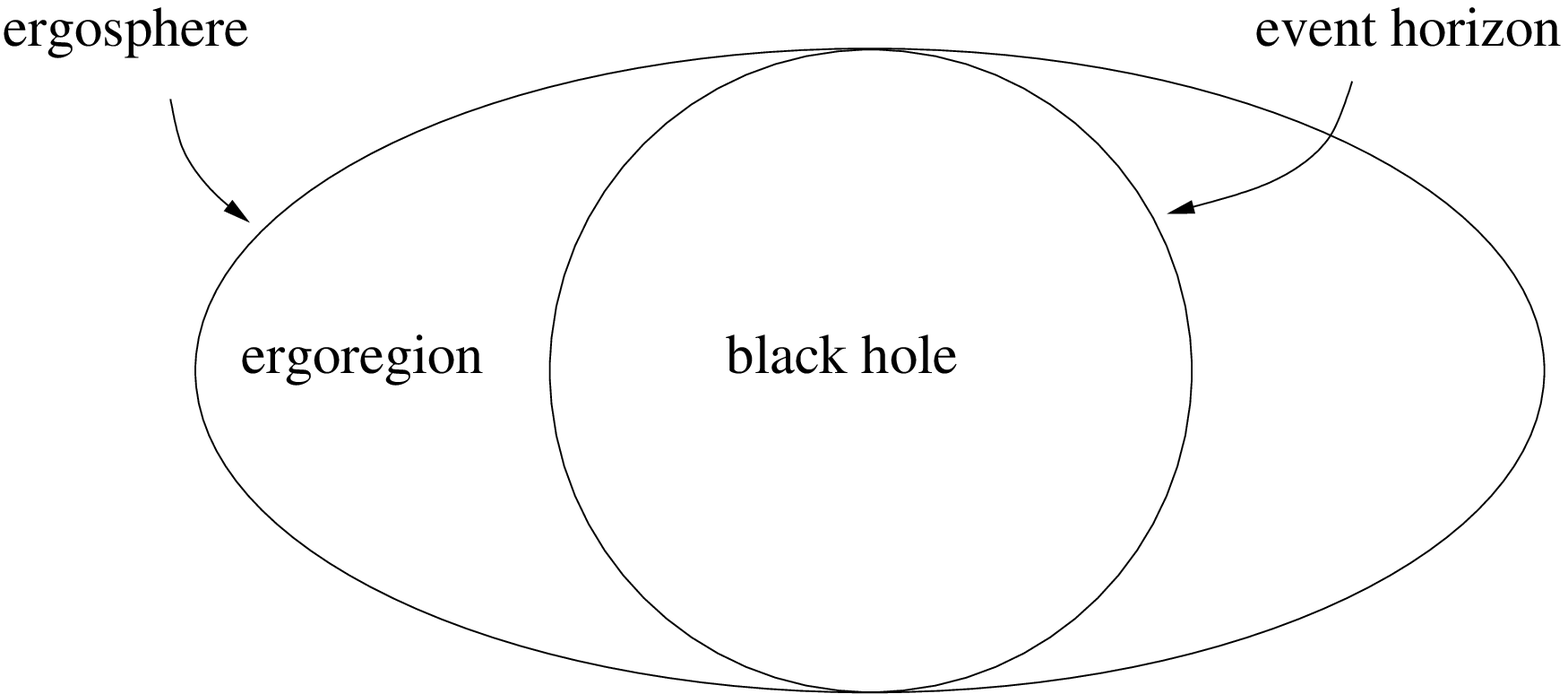}
\end{center}
\caption{Spacelike cross-section of the Kerr solution.} \label{ergosphere}
\end{figure}

A time-oriented spacetime $(M,g)$ is said to be {\bf asymptotically flat} if it contains an open set $\mathcal{I}$ where the metric is well approximated (in a certain sense which we will not make precise) by the Minkowski metric in the region $\{x^2 + y^2 + z^2 > R^2\}$, for sufficiently large $R>0$. For such spacetimes we can define the {\bf black hole region} as $\mathcal{B}=M\setminus J^-(\mathcal{I})$, so that $\mathcal{B}$ consists of the events which cannot send signals to $\mathcal{I}$. If $\mathcal{B}\neq \varnothing$, the spacetime is called a {\bf black hole spacetime}, and $\mathscr{H^+}=\partial\mathcal{B}$ is called the {\bf event horizon}. Finally, an asymptotically flat spacetime is called {\bf stationary} if it admits a Killing vector field which is timelike on $\mathcal{I}$.

The importance of the Kerr solution stems from the following result.

\begin{Thm} ({\bf Israel \cite{Israel67}, Carter \cite{Carter71}, Hawking \cite{H72}, Robinson \cite{Robinson75}, Chru\'sciel and Costa \cite{CC08}}) 
The Kerr solution is the only real-analytic, stationary black hole spacetime satisfying the vacuum Einstein equations.
\end{Thm}

\section{Killing horizons and the zeroth law} \label{sec7.2}

\begin{Def}
A {\bf Killing horizon} is a null surface which is orthogonal to a nonvanishing Killing vector field.
\end{Def}

Note that the Killing vector field is therefore null on the corresponding Killing horizon, and tangent to it. Killing horizons are important due to the following result.

\begin{Thm} ({\bf Hawking \cite{H72}})
The event horizon of a stationary black hole spacetime whose matter fields satisfy hyperbolic equations and the weak energy condition is a Killing horizon.
\end{Thm}

\begin{Prop}
The integral curves of a nonvanishing normal to a null hypersurface (e.g.\ the Killing vector field orthogonal to a Killing horizon) are reparameterized null geodesics.
\end{Prop}

\begin{proof}
Let $Z$ be nonvanishing and normal to a null hypersurface $S$, and let $p \in S$. For any tangent vector $v \in T_pS$ it is easy to construct a local vector field $V$ which is tangent to $S$, commutes with $Z$ and satisfies $V_p=v$. We have
\[
\left\langle V, \nabla_Z Z \right\rangle = - \left\langle \nabla_Z V, Z \right\rangle = - \left\langle \nabla_V Z, Z \right\rangle = -\frac12 \, V \cdot \left\langle Z, Z \right\rangle = 0.
\]
Since $p$ and $v$ are arbitrary, we conclude that $\nabla_Z Z$ is orthogonal to $S$, i.e.
\[
\nabla_Z Z = k Z,
\]
for some function $k:S \to \bbR$.
\end{proof}

\begin{Def}
If $\mathscr{H}$ is a Killing horizon associated to the Killing vector field $Z$ then the function $k:\mathscr{H} \to \bbR$ such that
\[
\nabla_Z Z = k Z,
\]
on $\mathscr{H}$ is called the {\bf surface gravity} of $\mathscr{H}$ relative to $Z$.
\end{Def}

\begin{Thm} ({\bf Zeroth law of black hole thermodynamics})
The surface gravity of a Killing horizon on a spacetime satisfying the dominant energy condition is a constant function. 
\end{Thm}

\begin{proof}
Let $Z$ be a Killing vector field associated to a Killing horizon $\mathscr{H}$. Since $Z$ is orthogonal to $\mathscr{H}$, we have from the Frobenius theorem that
\[
Z^\sharp \wedge dZ^\sharp = 0
\]
on $\mathscr{H}$. Because $Z^\sharp$ is nonvanishing, and can therefore be completed to a coframe, we necessarily have
\[
dZ^\sharp = 2 Z^\sharp \wedge U^\sharp \Leftrightarrow \nabla_\alpha Z_\beta = Z_\alpha U_\beta - U_\alpha Z_\beta
\]
on $\mathscr{H}$, for some vector field $U$. If the vector fields $X$ and $Y$ are tangent to $\mathscr{H}$ then
\begin{equation} \label{FrobeniusKilling}
X^\alpha Y^\beta \nabla_\alpha Z_\beta = 0,
\end{equation}
and consequently
\begin{equation} \label{FrobeniusKilling2}
Y^\alpha Z^\beta \nabla_\alpha X_\beta = - Y^\alpha X^\beta \nabla_\alpha Z_\beta = 0.
\end{equation}
on $\mathscr{H}$. Taking the derivative of \eqref{FrobeniusKilling} along any vector field $V$ tangent to $\mathscr{H}$ yields
\begin{align*}
V^\mu X^\alpha Y^\beta \nabla_\mu \nabla_\alpha Z_\beta & = - V^\mu (\nabla_\mu X^\alpha) Y^\beta \nabla_\alpha Z_\beta - V^\mu X^\alpha (\nabla_\mu Y^\beta) \nabla_\alpha Z_\beta\\
& = - V^\mu ( Y^\beta \nabla_\mu X^\alpha + X^\alpha \nabla_\mu Y^\beta) (Z_\alpha U_\beta - U_\alpha Z_\beta) = 0,
\end{align*}
in view of \eqref{FrobeniusKilling2}. On the other hand, we saw in Chapter~\ref{chapter6} that any Killing vector field $Z$ satisfies
\begin{equation} \label{RiemannKilling}
\nabla_\mu \nabla_\alpha Z_\beta = - R_{\alpha\beta\mu}^{\,\,\,\,\,\,\,\,\,\,\,\,\nu} Z_\nu,
\end{equation}
whence
\begin{equation} \label{RiemannzeroonH}
R_{\alpha\beta\mu\nu} X^\alpha Y^\beta V^\mu Z^\nu = 0
\end{equation}
for any three vector fields $X,Y,V$ tangent to $\mathscr{H}$. If $X$ and $Y$ are orthonormal then they can always be completed to a local frame $\{X,Y,Z,W\}$ such  $W$ is null, orthogonal to $X$ and $Y$, and normalized against $Z$,
\[
Z_\mu W^\mu = -1.
\]
In this frame the metric is written
\[
g_{\mu\nu} = X_\mu X_\nu + Y_\mu Y_\nu - Z_\mu W_\nu - W_\mu Z_\nu,
\]
and so we obtain, for any vector field $V$ tangent to $\mathscr{H}$,
\begin{align} \label{Ralphabeta}
R_{\alpha\beta} V^\alpha Z^\beta & = R_{\alpha\mu\beta\nu} V^\alpha Z^\beta g^{\mu\nu} = - R_{\alpha\mu\beta\nu} V^\alpha Z^\beta Z^\mu W^\nu \\
& = V^\alpha Z^\beta W^\nu \nabla_\alpha \nabla_\beta Z_\nu, \nonumber
\end{align}
where we used \eqref{RiemannzeroonH} and \eqref{RiemannKilling}.

We have
\begin{equation} \label{expressionfork}
k = - \left\langle W, \nabla_Z Z \right\rangle = - W^\nu Z^\mu \nabla_\mu Z_\nu.
\end{equation}
Differentiating \eqref{expressionfork} along a vector field $V$ tangent $\mathscr{H}$ yields
\[
V \cdot k = - (V^\alpha \nabla_\alpha W^\nu) Z^\mu \nabla_\mu Z_\nu - W^\nu (V^\alpha \nabla_\alpha Z^\mu) \nabla_\mu Z_\nu - W^\nu Z^\mu V^\alpha \nabla_\alpha \nabla_\mu Z_\nu.
\]
The first term in the right-hand side of this equation is
\[
- (V^\alpha \nabla_\alpha W^\nu) k Z_\nu = k W^\nu V^\alpha \nabla_\alpha Z_\nu = k W^\nu V^\alpha (Z_\alpha U_\nu - U_\alpha Z_\nu) = k V^\alpha U_\alpha, 
\]
whereas the second term is
\[
- W^\nu V^\alpha (Z_\alpha U^\mu - U_\alpha Z^\mu) \nabla_\mu Z_\nu = V^\alpha U_\alpha W^\nu Z^\mu \nabla_\mu Z_\nu = - k V^\alpha U_\alpha.
\]
Therefore these terms cancel out, and, using \eqref{Ralphabeta}, we obtain
\[
V \cdot k = - R_{\alpha\beta} V^\alpha Z^\beta.
\]
From \eqref{Ralphabeta} it is clear that
\[
R_{\alpha\beta} Z^\alpha Z^\beta = 0,
\]
implying that the vector field
\[
I_\alpha = R_{\alpha\beta} Z^\beta
\]
is tangent to $\mathscr{H}$. By Einstein's equation, we have
\[
I_\alpha = \left(\frac12 R g_{\alpha\beta} - \Lambda g_{\alpha\beta} + 8 \pi T_{\alpha\beta} \right) Z^\beta = \frac12 R Z_\alpha - \Lambda Z_\alpha + 8 \pi T_{\alpha\beta} Z^\beta,
\]
where $R$ is the scalar curvature, $\Lambda$ is the cosmological constant and $T_{\alpha\beta}$ is the energy-momentum tensor. Therefore the vector field
\[
J_\alpha = T_{\alpha\beta} Z^\beta
\]
is also tangent to $\mathscr{H}$. Since $T_{\alpha\beta}$ satisfies the dominant energy condition, the vector field $J$ must be causal, and since it is tangent to $\mathscr{H}$ it can only be null and parallel to $Z$. We conclude that the vector field $I$ is also proportional to $Z$, and so
\[
V \cdot k = - I_{\alpha} V^\alpha = 0,
\] 
that is, $k$ is constant along $\mathscr{H}$.
\end{proof}

Let $\lambda$ be a local coordinate that is an affine parameter for the null geodesics along $\mathscr{H}$, that is,
\[
\left \langle \frac{\partial}{\partial \lambda}, \frac{\partial}{\partial \lambda} \right \rangle = 0 \qquad \text{ and } \qquad  \nabla_{\frac{\partial}{\partial \lambda}} \frac{\partial}{\partial \lambda} = 0
\] 
on $\mathscr{H}$. Then it is easy to see that
\[
Z = k (\lambda-\lambda_0) \frac{\partial}{\partial \lambda}
\]
on $\mathscr{H}$, where $\lambda_0$ may depend on the null geodesic. In other words, $Z$ vanishes on some cross-section of $\mathscr{H}$. If $Z$ and $\frac{\partial}{\partial \lambda}$ are future-pointing then $Z$ vanishes to the past when $k$ is positive (that is, $\lambda>\lambda_0$), and to the future if $k$ is negative (that is, $\lambda<\lambda_0$).

If $t$ is a local coordinate in a neighborhood of $\mathscr{H}$ such that
\[
Z = \frac{\partial}{\partial t}
\]
then we have on $\mathscr{H}$
\[
Z = \frac{d\lambda}{dt} \frac{\partial}{\partial \lambda} = k (\lambda-\lambda_0) \frac{\partial}{\partial \lambda},
\]
implying that
\[
\lambda-\lambda_0 = C e^{kt},
\]
where $C$ may depend on the null geodesic. By rescaling $\lambda$ conveniently we may assume that
\[
Z = e^{kt} \frac{\partial}{\partial \lambda}
\]
on $\mathscr{H}$.

Now consider a timelike congruence crossing $\mathscr{H}$, with tangent unit timelike vector field $U$ satisfying
\[
[Z,U] = 0.
\]
This can be accomplished by taking a timelike hypersurface transverse to $Z$, ruled by timelike curves, and moving each curve by the flow of $Z$. The quantity
\[
E = - \left \langle \frac{\partial}{\partial \lambda}, U \right \rangle = - \left \langle e^{-kt} Z, U \right \rangle
\]
represents the energy of a given null geodesic as measured by an observer of the congruence when crossing the Killing horizon, and is related to the frequency of the associated wave. Since $Z$ is a Killing field, we have
\begin{align*}
Z \cdot E & = - \mathcal{L}_Z  \left \langle e^{-kt} Z, U \right \rangle =  - \left \langle \mathcal{L}_Z (e^{-kt} Z), U \right \rangle -  \left \langle e^{-kt} Z, \mathcal{L}_Z U \right \rangle \\
& = - \left \langle [Z, e^{-kt} Z], U \right \rangle -  \left \langle e^{-kt} Z, [Z, U] \right \rangle = \left \langle k e^{-kt} Z, U \right \rangle = - k E,
\end{align*}
implying that
\[
E = E_0 e^{-kt}.
\]
In other words, the energy of the null geodesic as measured by the observers of the congruence decreases exponentially if $k>0$ ({\bf redshift effect}), and increases exponentially if $k<0$ ({\bf blueshift effect}).

\section{Smarr's formula and the first law} \label{sec7.3}

Unlike the case of the Schwarzschild black hole, where the event horizon is the Killing horizon corresponding to $X=\frac{\partial}{\partial t}$, the event horizon of the Kerr black hole is a Killing horizon for the Killing vector field
\[
Z = X + \Omega Y,
\]
where $Y=\frac{\partial}{\partial \varphi}$ and $\Omega \in \bbR$ is an appropriate constant.

\begin{Def}
$\Omega$ is called the {\bf angular velocity of the event horizon}.
\end{Def}

To find $\Omega$, we write the quadratic equation in $\omega$ for the vector $X + \omega Y$ to be null at some point with $r>r_+$ (so that we can use Boyer-Lindquist coordinates):
\[
- \left( 1 - \frac{2Mr}{\rho^2} \right) - \frac{4Mar\sin^2\theta}{\rho^2} \, \omega + \left( r^2 + a^2 + \frac{2Ma^2r\sin^2\theta}{\rho^2} \right) \sin^2\theta \, \omega^2 = 0.
\]
The discriminant of this equation has the simple form $\Delta \sin^2 \theta$, and consequently vanishes when we take $r=r_+$, in which case the quadratic equation has the single solution
\[
\Omega = \frac{a}{r_+^2 + a^2} = \frac{a}{2Mr_+}.
\]
This solution is the limit as $r \to r_+$ of the values of $\omega=\omega(r,\theta)$ such that $X + \omega Y$ is null, and so it must coincide with the angular velocity of the event horizon.

From the expressions of the Komar mass and angular momentum, it is clear that
\[
M - 2 \Omega J = -\frac1{8\pi} \int_{\Sigma} \star d Z^\sharp
\]
for any compact orientable $2$-surface $\Sigma$ enclosing the event horizon $\mathscr{H}$. Let us consider the case when $\Sigma$ is a spacelike cross-section of $\mathscr{H}$ (Figure~\ref{Smarr}). We can uniquely define a future-pointing unit timelike vector field $N$ and an unit spacelike vector field $n$, both orthogonal to $\Sigma$, such that $Z = N+n$. Because $Z$ is a Killing vector field, $\nabla_\mu Z_\nu$ is a $2$-form; more precisely,
\[
(d Z^\sharp)_{\mu\nu} = \nabla_\mu Z_\nu - \nabla_\nu Z_\mu = 2 \nabla_\mu Z_\nu.
\]
If $E_1$ and $E_2$ are two unit vector fields tangent to $\Sigma$ such that $\{N,n,E_1,E_2\}$ is a positive orthonormal frame, and so $\{-N^\sharp,n^\sharp,E_1^\sharp,E_2^\sharp\}$ is a positive orthonormal coframe, then we can expand
\[
\nabla Z^\sharp = - \nabla Z^\sharp(N,n) N^\sharp \wedge n^\sharp + \ldots.
\]
Therefore,
\[
M - 2 \Omega J = - \frac1{4\pi} \int_{\Sigma} \star \nabla Z^\sharp = \frac1{4\pi} \int_{\Sigma} \nabla Z^\sharp(N,n) E_1^\sharp \wedge E_2^\sharp.
\]
Since
\[
\nabla Z^\sharp(N,n) = \nabla Z^\sharp(Z,n) = \left\langle \nabla_Z Z, n \right\rangle = \left\langle k Z, n \right\rangle = k,
\]
we finally obtain the {\bf Smarr formula}:
\[
M = \frac{kA}{4\pi} + 2 \Omega J,
\]
where $A$ is the area of the cross-section $\Sigma$ (which in particular is the same for all cross-sections of $\mathscr{H}$).

\begin{figure}[h!]
\begin{center}
\psfrag{n}{$n$}
\psfrag{N}{$N$}
\psfrag{H}{$\mathscr{H}$}
\psfrag{E}{$E_1$}
\psfrag{S}{$\Sigma$}
\psfrag{Z}{$Z$}
\epsfxsize=0.8\textwidth
\leavevmode
\epsfbox{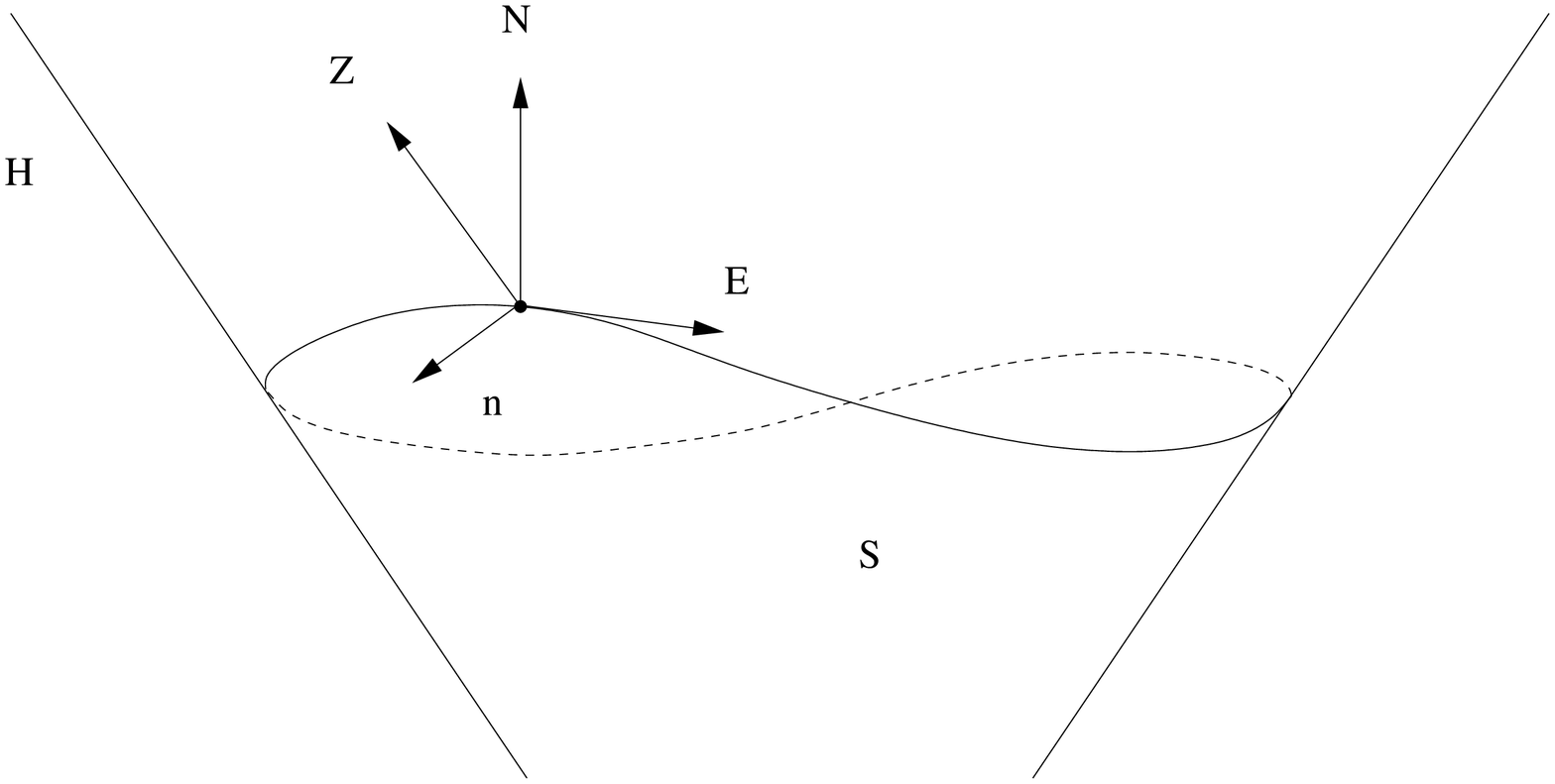}
\end{center}
\caption{Spacelike cross-section of the event horizon.} \label{Smarr}
\end{figure}

A cross section of the event horizon can be obtained by taking the limit as $r \to r_+$ of a surface of constant $(t,r)$. The induced metric is then
\[
ds^2 = (r_+^2 + a^2 \cos^2 \theta) d \theta^2 + \left( r_+^2 + a^2 + \frac{2Ma^2r_+\sin^2\theta}{r_+^2 + a^2 \cos^2 \theta} \right) \sin^2\theta d\varphi^2,
\]
and its area is
\[
A = 4 \pi (r_+^2 + a^2) = 8 \pi (M^2 + M \sqrt{M^2 - a^2}),
\]
whence
\[
M = \sqrt{\frac{A}{16\pi} + \frac{4\pi J^2}{A}}.
\]
Noting that $M=M(A,J)$ is homogeneous of degree $1/2$, we know from Euler's homogeneous function theorem that
\[
\frac12 M = \frac{\partial M}{\partial A} A + \frac{\partial M}{\partial J} J.
\]
On the other hand, it is easy to check that
\[
\frac{\partial M}{\partial J} = \Omega.
\]
Smarr's formula then implies the following result:

\begin{Thm} ({\bf First law of black hole thermodynamics})
The function $M=M(A,J)$ giving the mass of a Kerr black hole as a function of the area of (spacelike cross-sections of) its horizon and its angular momentum satisfies
\[
dM = \frac{k}{8\pi} dA + \Omega dJ.
\]
\end{Thm}

This formula provides an easy way to compute the surface gravity of a Kerr black hole horizon:
\[
k = \frac1{4M} - M \Omega^2.
\]
In particular, for an extremal black hole ($r_+ = a = M$) we have
\[
\Omega = \frac{1}{2M} \Rightarrow k = 0.
\]

\section{Second law} \label{sec7.4}

We consider arbitrary test fields propagating on a Kerr background. Apart from ignoring their gravitational backreaction, we make no further hypotheses on the fields: they could be any combination of scalar or electromagnetic fields, fluids, elastic media, or other types of matter. By the Einstein equation, their combined energy-momentum tensor $T$ must satisfy
\[
\nabla_\mu T^{\mu\nu} = 0.
\]
Using the symmetry of $T$ and the Killing equation, 
\[
\nabla_\mu X_\nu + \nabla_\nu X_\mu = 0,
\]
we have
\[
\nabla_\mu (T^{\mu\nu} X_\nu) = 0.
\]
This conservation law suggests that the total field energy on a given spacelike hypersurface $S$ extending from the black hole event horizon $\mathscr{H^+}$ to infinity (Figure~\ref{Penrose_energy}) should be
\[
E = \int_S T^{\mu\nu} X_\nu N_\mu ,
\]
where $N$ is the future-pointing unit normal to $S$.

\begin{figure}[h!]
\begin{center}
\psfrag{i+}{$i^+$}
\psfrag{i0}{$i^0$}
\psfrag{i-}{$i^-$}
\psfrag{H}{$H$}
\psfrag{H+}{$\mathscr{H^+}$}
\psfrag{H-}{$\mathscr{H^-}$}
\psfrag{I+}{$\mathscr{I^+}$}
\psfrag{I-}{$\mathscr{I^-}$}
\psfrag{I}{$\mathscr{I}$}
\psfrag{S0}{$S_0$}
\psfrag{S1}{$S_1$}
\epsfxsize=0.5\textwidth
\leavevmode
\epsfbox{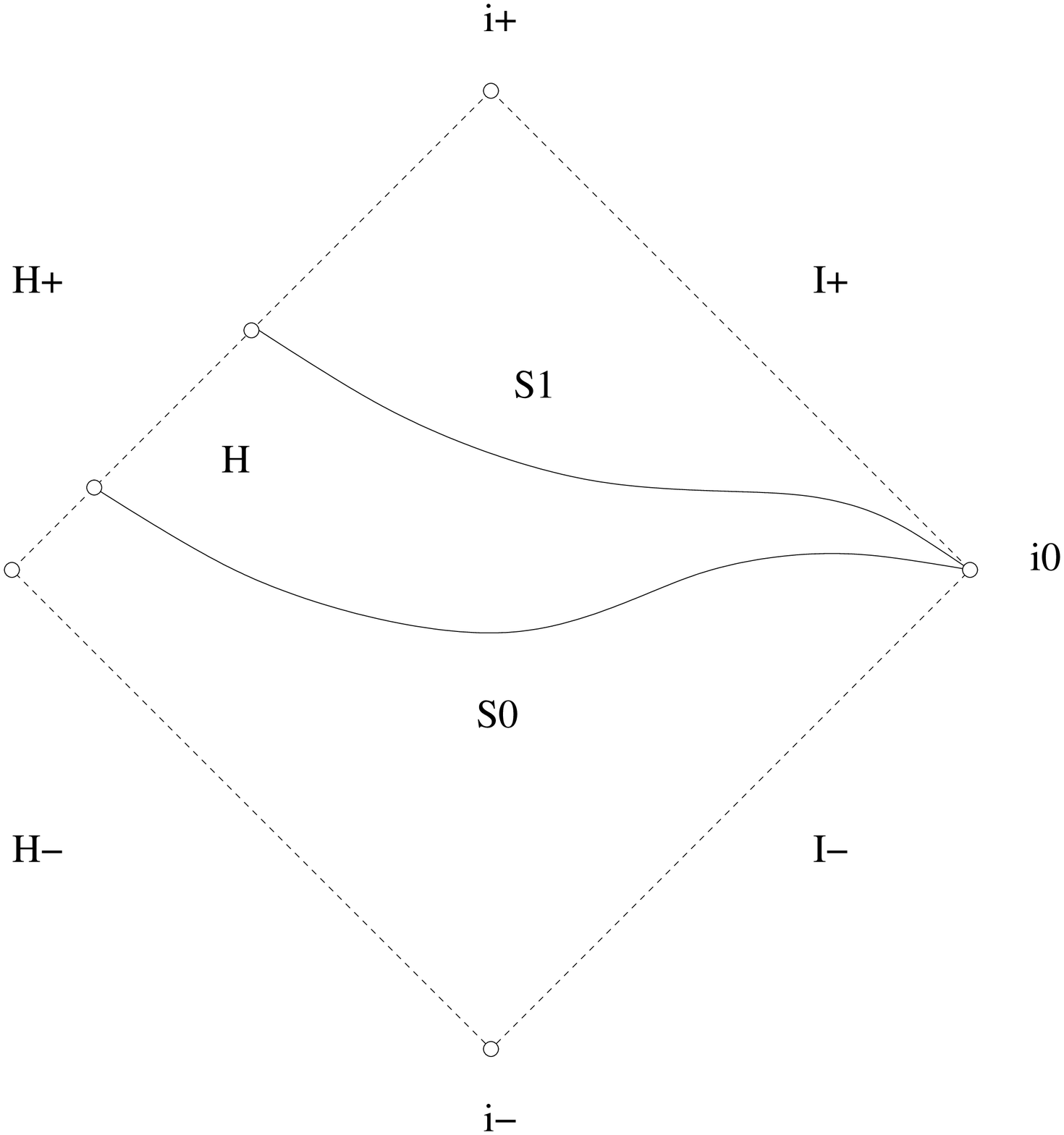}
\end{center}
\caption{Penrose diagrams for the region of outer communication of the Kerr spacetime.} \label{Penrose_energy}
\end{figure}

Analogously, the total field angular momentum on a spacelike hypersurface $S$ extending from the event horizon to infinity is
\begin{equation}
L = - \int_S T^{\mu\nu} Y_\nu N_\mu ,
\end{equation}
where the minus sign accounts for the timelike unit normal. 

Consider now two such spacelike hypersurfaces, $S_0$ and $S_1$, with $S_1$ to the future of $S_0$ (Figure~\ref{Penrose_energy}). The energy absorbed by the black hole across the subset $H$ of $\mathscr{H^+}$ between $S_0$ and $S_1$ is then
\[
\Delta M = \int_{S_0} T^{\mu\nu} X_\nu N_\mu  - \int_{S_1} T^{\mu\nu} X_\nu N_\mu ,
\]
whereas the angular momentum absorbed by the black hole across $H$ is
\[
\Delta J = - \int_{S_0} T^{\mu\nu} Y_\nu N_\mu  + \int_{S_1} T^{\mu\nu} Y_\nu N_\mu .
\]
Therefore, we have
\[
\Delta M  - \Omega \Delta J = \int_{S_0} T^{\mu\nu} Z_\nu N_\mu  - \int_{S_1} T^{\mu\nu} Z_\nu N_\mu .
\]
Because $Z$ is also a Killing vector field,
\[
\nabla_\mu (T^{\mu\nu} Z_\nu) = 0,
\]
and so the divergence theorem, applied to the region bounded by $S_0$, $S_1$ and $H$, yields
\[
\Delta M  - \Omega \Delta J = \int_{H} T^{\mu\nu} Z_\nu Z_\mu 
\]
(we use $-Z$ as the null normal on $H$). Therefore we have the following result.

\begin{Thm} ({\bf Second law of black hole thermodynamics, test field version})
If the energy-momentum tensor $T$ corresponding to any collection of test fields propagating on a Kerr background satisfies the null energy condition at the event horizon then the energy $\Delta M$ and the angular momentum $\Delta J$ absorbed by the black hole satisfy
\[
\Delta M  \geq \Omega \Delta J.
\]
\end{Thm}

If we think of Kerr black holes as stationary states and imagine that the interaction of a Kerr black hole with test fields results in a new Kerr black hole, then, in view of the first law of black hole thermodynamics, we can rewrite the result above as
\[
dA = \frac{8\pi}{k} (dM - \Omega dJ) \geq 0,
\]
that is, the area of the event horizon of a Kerr black hole can only increase as a result of its interaction with test fields. 

In fact, it is possible to prove a general result about the area of the event horizon of a black hole spacetime.

\begin{Prop}\label{Prop_horizon}
The event horizon of a black hole spacetime is ruled by null geodesics.
\end{Prop}

\begin{proof}
Given the causal structure of Minkowski's spacetime, it is clear that $J^-(\mathcal{I})$ coincides with $I^-(\mathcal{I})$, an open set. Let $p \in \mathscr{H}$ be any point in the event horizon and let $U \ni p$ be a simple neighborhood (see Proposition~\ref{compact} in Chapter~\ref{chapter4}). Given a sequence $\{p_n\} \subset U \cap J^-(\mathcal{I})$  converging to $p$, let $c_n$ be a future-pointing causal curve connecting $p_n$ to $\mathcal{I}$, let $q_n$ be the first intersection of $c_n$ with $\partial U$, and let $\gamma_n$ be the future-pointing causal geodesic connecting $p_n$ to $q_n$ in $U$ (see Figure~\ref{horizon}). Since the exponential map centered on any point in $U$ is a diffeomorphism, these geodesics converge to a future-pointing causal geodesic $\gamma$ in $U$ with initial point $p$. Note that $\gamma$ cannot enter $J^-(\mathcal{I})$, because then we would have $p \in J^-(\mathcal{I})\equiv\inte J^-(\mathcal{I})$, in contradiction with $p \in \mathscr{H}\equiv\partial J^-(\mathcal{I})$. Moreover, every point in $\gamma$ is the limit of points in $\gamma_n$, hence points in $J^-(\mathcal{I})$. We conclude that $\gamma$ is a curve on $\overline{J^-(\mathcal{I})} \setminus \inte J^-(\mathcal{I}) = \mathscr{H}$. Finally, $\gamma$ cannot be a timelike geodesic, because then the sequence $q_n$ would enter the open set $I^+(p)$, and we would again have $p \in J^-(\mathcal{I})$. We conclude that $\gamma$ is a null geodesic. If we extend $\gamma$ maximally towards the future we obtain a future-inextendible null geodesic; covering this curve with simple neighborhoods and applying similar arguments to the above, one can easily show that it never leaves $\mathscr{H}$.
\end{proof}

\begin{figure}[h!]
\begin{center}
\psfrag{U}{$U$}
\psfrag{p}{$p$}
\psfrag{pn}{$p_n$}
\psfrag{qn}{$q_n$}
\psfrag{H}{$\mathscr{H}$}
\epsfxsize=0.4\textwidth
\leavevmode
\epsfbox{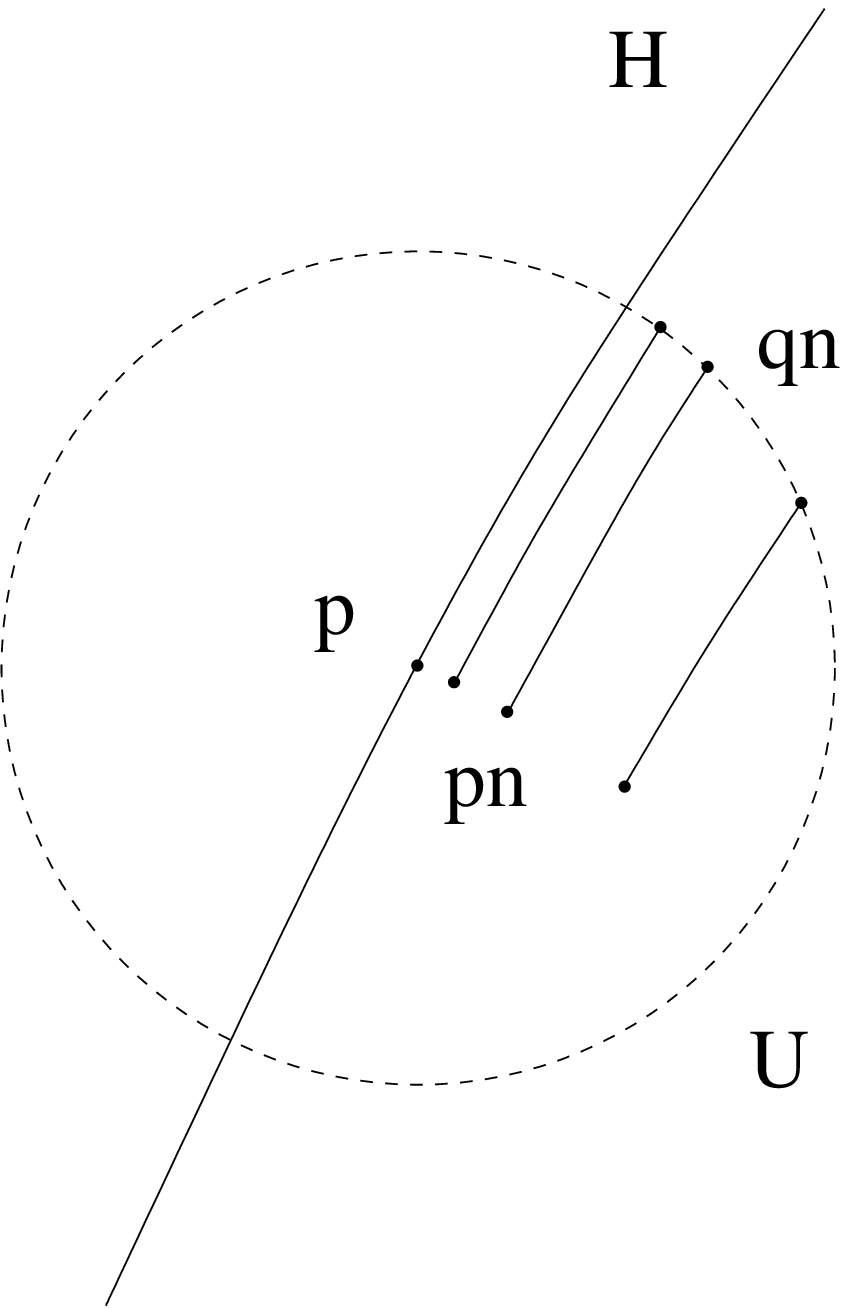}
\end{center}
\caption{Proof of Proposition~\ref{Prop_horizon}.} \label{horizon}
\end{figure}

\pagebreak

\begin{Thm} ({\bf Second law of black hole thermodynamics, Hawking's version \cite{H72}})
If the energy-momentum tensor of a black hole spacetime satisfies the null energy condition at the event horizon and the null geodesics ruling the event horizon are complete towards the future then their expansion is never negative. In particular, the area of any spacelike cross-section of the event horizon cannot decrease towards the future.
\end{Thm}

\begin{proof}
Note that the null geodesics ruling the event horizon are orthogonal to any spacelike cross-section $\Sigma$, so that the discussion preceding the proof of Penrose's singularity theorem in Chapter~\ref{chapter4} applies. Suppose that the expansion of some null geodesic were negative at some point $p \in \mathscr{H}$. Then, by the analogue of Proposition~\ref{conjugate_null} in Chapter~\ref{chapter4}, that null geodesic would have a conjugate point to the future of $p$, after which, by the analogue of Proposition~\ref{conjugate_null_I+} in Chapter~\ref{chapter4}, it would leave $\mathscr{H}$. Since this would contradict Proposition~\ref{Prop_horizon}, the expansion can never be negative.
\end{proof}


\section{Hawking radiation and black hole thermodynamics} \label{sec7.6}

Recall the three laws of thermodynamics for, say, a gas:

\begin{itemize}
\item
{\bf Zeroth law:} The temperature is constant throughout the gas when thermal equilibrium has been reached.
\item
{\bf First law:} The internal energy $U$ as a function of the gas entropy $S$ and volume $V$ satisfies
\[
dU = T dS - p dV,
\]
where $T$ is the temperature and $p$ is the pressure.
\item
{\bf Second law:} The entropy of the gas cannot decrease towards the future.
\end{itemize}

These are remarkably similar to the three laws of black hole thermodynamics, if we identify the black hole mass $M$ with the internal energy, (some multiple of) the horizon's surface gravity with the black hole's temperature and (some multiple of) the horizon's area with the black hole's entropy. Inspired by this analogy, Bekenstein \cite{Bek72} proposed in 1972 that black holes are indeed thermodynamic systems. Hawking initially resisted this suggestion, since objects in thermal equilibrium at a given temperature must emit black body radiation, which black holes cannot (classically) do. In 1974, however, Hawking \cite{H74} applied methods of quantum field theory in curved spacetime to show that black holes do indeed emit particles with a thermal spectrum corresponding to the temperature
\[
T = \frac{k}{2\pi}.
\]
In particular, this fixed the black hole entropy as
\[
S = \frac{A}4.
\]

\section{Exercises} \label{sec7.7}

\begin{enumerate}

\item
To compute the Komar mass and angular momentum of the Kerr solution we consider the region $r \gg M,a$.
\begin{enumerate}
\item
Show that a positive orthonormal coframe is approximately given in this region by
\begin{align*}
& \omega^0 \sim dt, \qquad \omega^r \sim dr, \qquad \omega^\theta \sim r d\theta, \\
& \omega^\varphi \sim r \sin\theta d\varphi - \frac{2Ma\sin\theta}{r^2} dt.
\end{align*}
\item
Establish the following asymptotic formulas:
\begin{align*}
& X^\sharp \sim - \left( 1 - \frac{2M}{r} \right) dt - \frac{2Ma\sin^2\theta}{r} d\varphi; \\
& Y^\sharp \sim - \frac{2Ma\sin^2\theta}{r} dt + r^2 \sin^2\theta d\varphi; \\
& dX^\sharp \sim \frac{2M}{r^2} \omega^0 \wedge \omega^r + \ldots ; \\
& dY^\sharp \sim - \frac{6Ma\sin^2\theta}{r^2} \omega^0 \wedge \omega^r + \ldots .
\end{align*}
\item
Prove that $M_{\text{Komar}} = M$ and $J_{\text{Komar}} = Ma$.
\end{enumerate}

\item
Show that the metric induced on the ergoshpere is
\begin{align*}
ds^2 = & - 2a\sin^2\theta dt d\varphi + \frac{2M^3r}{(r-M)^2} d \theta^2 \\
& + \left( r^2 + a^2 + a^2\sin^2\theta \right) \sin^2\theta d\varphi^2,
\end{align*}
where $r = M + \sqrt{M^2-a^2\cos^2\theta}$. Prove that this metric is Lorentzian.

\item
The symmetry semiaxis $\theta=0$ is a totally geodesic submanifold of the Kerr solution, with metric
\[
ds^2 = - \frac{\Delta}{r^2 + a^2} dt^2 + \frac{r^2 + a^2}{\Delta} dr^2.
\]
Obtain the maximal analytical extension of this submanifold. Note that $r=0$ is not a singularity, and so the metric can be continued for negative values of $r$. 

\item
Consider the static and spherically symmetric metric given in local coordinates by
\[
\hspace{2cm} ds^2 = - V(r) dt^2 + V(r) dr^2 + r^2 \left( d\theta^2 + \sin^2 \theta d\varphi^2 \right).
\]
\begin{enumerate}
\item
Show that this metric can be written in the form
\[
\hspace{2cm} ds^2 = - V(r) dv^2 + 2 dv dr + r^2 \left( d\theta^2 + \sin^2 \theta d\varphi^2 \right),
\]
where
\[
v = t + \int \frac{dr}{V(r)},
\]
\item
Assume that there $V$ has an isolated zero at $r=r_H$, so that this hypersurface is a Killing horizon. Show that the corresponding surface gravity relative to $\frac{\partial}{\partial v}$ is
\[
k = \frac12 V'(r_H).
\]
\end{enumerate}

\item
Consider a one-parameter family of null geodesics $\gamma(\tau,\lambda)$ connecting two timelike curves $c_0(\tau)=\gamma(\tau,0)$ and $c_1(\tau)=\gamma(\tau,1)$. Show that
\[
\frac{\partial}{\partial \lambda} \left\langle \frac{\partial \gamma}{\partial \tau}, \frac{\partial \gamma}{\partial \lambda} \right\rangle = 0,
\]
and use this to prove that if $\tau$ is the proper time for the curve $c_0$ and
\[
\tau'(\tau) = \int_{\tau_0}^\tau \left|\dot{c}_1(t)\right|dt
\]
is the proper time for the curve $c_1$ then
\[
\frac{d\tau'}{d\tau} = \frac{E_0}{E_1},
\]
where
\[
E_0 = - \left\langle \dot{c}_0, \frac{\partial \gamma}{\partial \lambda} \right\rangle \quad \text{ and } \quad E_1 = - \left\langle \frac{\dot{c}_1}{|\dot{c}_1|}, \frac{\partial \gamma}{\partial \lambda} \right\rangle 
\]
are the energies of the null geodesic $\gamma(\tau,\lambda)$ as measured by the observers corresponding to $c_0$ and $c_1$.

\item
Compute the area, the angular velocity and the surface gravity of a Kerr black hole horizon (you may find the relation $r_+^2 + a^2 = 2Mr_+$ to be useful here).

\item
Prove that test fields satisfying the null energy condition at the event horizon cannot destroy an extremal Kerr black hole. More precisely, prove that if an extremal black hole is characterized by the physical quantities $(M,J)$, and absorbs energy and angular momentum $(\Delta M,\Delta J)$ by interacting with the test fields, then the metric corresponding to the physical quantities $(M+\Delta M, J + \Delta J)$ represents, to first order in $\Delta M$ and $\Delta J$, either a subextremal or an extremal Kerr black hole.

\item
Show that if a test field satisfying the null energy condition at the event horizon extracts energy from a Kerr black hole then $a=J/M$ always decreases. What fraction of the black hole's mass can be extracted?

\item
Check that the second law holds for the black hole resulting from an Oppenheimer-Snyder collapse.

\item
Use the second law of black hole thermodynamics to:
\begin{enumerate}
\item
Prove that a Schwarzschild black hole cannot split into two Kerr black holes;
\item
Give an upper bound to the energy released in the form of gravitational waves when two Kerr black holes coalesce to form a Schwarzschild black hole. Can the efficiency of this process ever exceed $50\%$?
\end{enumerate}

\end{enumerate}


\chapter*{Appendix: Mathematical concepts for physicists} \label{appendix}

In this appendix we list some mathematical concepts which will be used in the main text, for the benefit of readers whose background is in Physics.

\section*{Topology} \label{secA.1}

\begin{Def}
A {\bf topological space} is a set $M$ with a {\bf topology}, that is, a list of the {\bf open subsets} of $M$, satisfying:
\begin{enumerate}
\item Both $\varnothing$ and $M$ are open;
\item Any union of open sets is open;
\item Any finite intersection of open sets is open.
\end{enumerate}
\end{Def}

All the usual topological notions can now be defined. For instance, a {\bf closed set} is a set whose complement is open. The {\bf interior} $\inte A$ of a subset $A \subset M$ is the largest open set contained in $A$, its {\bf closure} $\overline{A}$ is the smallest closed set containing $A$, and its {\bf boundary} is $\partial A = \overline{A} \setminus \inte A$.

The main object of topology is the study of limits and continuity.

\begin{Def}
A sequence $\{ p_n \}$ is said to {\bf converge} to $p \in M$ if for any open set $U \ni p$ there exists $k \in \bbN$ such that $p_n \in U$ for all $n \geq k$.
\end{Def}

\begin{Def}
A map $f:M \to N$ between two topological spaces is said to be {\bf continuous} if for each open set $U\subset N$ the preimage $f^{-1}(U)$ is an open subset of $M$. A bijection $f$ is called a {\bf homeomorphism} if both $f$ and its inverse $f^{-1}$ are continuous.
\end{Def}

A system of local coordinates on a manifold is an example of a homeomorphism between the coordinate neighborhood and an open set of $\bbR^n$.

Two fundamental concepts in topology are compactness and connectedness.

\begin{Def}
A subset $A \subset M$ is said to be {\bf compact} if every cover of $A$ by open sets admits a finite subcover. It is said to be  {\bf connected} it is impossible to write $A = (A \cap U) \cup (A \cap V)$ with $U,V$ disjoint open sets and $A \cap U, A \cap V \neq \varnothing$.
\end{Def}

The following result generalizes the theorems of Weierstrass and Bolzano.

\begin{Thm}
Continuous maps carry compact sets to compact sets, and connected sets to connected sets.
\end{Thm}

\section*{Metric spaces} \label{secA.2}

\begin{Def}
A {\bf metric space} is a set $M$ and a {\bf distance function} $d:M\times M \to [0,+\infty)$ satisfying:
\begin{enumerate}
\item
{\bf Positivity:} $d(p,q) \geq 0$ and $d(p,q)=0$ if and only if $p=q$;
\item
{\bf Symmetry:} $d(p,q)=d(q,p)$;
\item
{\bf Triangle inequality:} $d(p,r) \leq d(p,q) + d(q,r)$,
\end{enumerate}
for all $p,q,r \in M$.
\end{Def}

The {\bf open ball} with center $p$ and radius $\varepsilon$ is the set
\[
B_{\varepsilon}(p)=\{ q \in M \mid d(p,q) < \varepsilon \}.
\]
Any metric space has a natural topology, whose open sets are unions of open balls. In this topology $p_n \to p$ if and only if $d(p_n,p) \to 0$, $F \subset M$ is closed if and only if every convergent sequence in $F$ has limit in $F$, and $K \subset M$ is compact if and only if every sequence in $K$ has a sublimit in $K$.

A fundamental notion for metric spaces is completeness.

\begin{Def}
A sequence $\{ p_n \}$ in $M$ is said to be a {\bf Cauchy sequence} if for all $\varepsilon > 0$ there exists $N \in \bbN$ such that $d(p_n,p_m)<\varepsilon$ for all $m,n>N$. A metric space is said to be {\bf complete} if all its Cauchy sequences converge.
\end{Def}

In particular any compact metric space is complete.

\section*{Hopf-Rinow theorem} \label{secA.3}

\begin{Def}
A Riemannian manifold is said to be {\bf geodesically complete} if any geodesic is defined for every value of its parameter.
\end{Def}

\begin{Def}
Let $(M,g)$ be a connected Riemannian manifold and $p,q \in M$. The {\bf distance} between $p$ and $q$ is defined as
\[
d(p,q)=\inf \{ l(\gamma) \mid \gamma \text{ is a smooth curve connecting } p \text{ to } q \}.
\]
\end{Def}

It is easily seen that $(M,d)$ is a metric space. Remarkably, the completeness of this metric space is equivalent to geodesic completeness.

\begin{Thm}
{\em (Hopf-Rinow)} A connected Riemannian manifold $(M,g)$ is geodesically complete if and only if $(M,d)$ is a complete metric space.
\end{Thm}

\section*{Differential forms} \label{secA.4}

\begin{Def}
A {\bf differential-form} $\omega$ {\bf of degree} $k$ is simply a completely anti-symmetric $k$-tensor: $\omega_{\alpha_1 \cdots \alpha_k} = \omega_{[\alpha_1 \cdots \alpha_k]}$.
\end{Def}

For instance, covector fields are differential forms of degree $1$. Differential forms are useful because of their rich algebraic and differential structure.

\begin{Def}
If $\omega$ is a $k$-form and $\eta$ is an $l$-form then their {\bf exterior product} is the $(k+l)$-form
\[
(\omega \wedge \eta)_{\alpha_1 \cdots \alpha_k \beta_1 \cdots \beta_l} = \frac{(k+l)!}{k!\,l!} \omega_{[\alpha_1 \cdots \alpha_k} \eta_{\beta_1 \cdots \beta_l]},
\]
and the {\bf exterior derivative} of $\omega$ is the $(k+1)$-form
\[
(d\omega)_{\alpha \alpha_1 \cdots \alpha_k} = (k+1) \nabla_{[\alpha} \omega_{\alpha_1 \cdots \alpha_k]},
\]
where $\nabla$ is any symmetric connection.
\end{Def}

It is easy to see that any $k$-form $\omega$ is given in local coordinates by
\begin{equation} \label{kform}
\omega = \sum_{\alpha_1 < \cdots < \alpha_k} \omega_{\alpha_1 \cdots \alpha_k}(x) dx^{\alpha_1} \wedge \cdots \wedge dx^{\alpha_k},
\end{equation}
and therefore has $n \choose k$ independent components on an $n$-dimensional manifold.

\begin{Prop} \label{propertiesforms}
If $\omega$, $\eta$ and $\theta$ are differential forms then:
\begin{enumerate}
\item
$\omega \wedge (\eta \wedge \theta) = (\omega \wedge \eta) \wedge \theta$;
\item
$\omega \wedge \eta = (-1)^{(\deg \omega)(\deg \eta)} \eta \wedge \omega$;
\item
$\omega \wedge (\eta + \theta) = \omega \wedge \eta + \omega \wedge \theta$;
\item
$d(\omega + \eta) = d \omega + d \eta$;
\item
$d(\omega \wedge \eta) = d \omega \wedge \eta + (-1)^{\deg \omega} \omega \wedge d \eta$;
\item
$d^2 \omega = 0$.
\end{enumerate}
\end{Prop}

It is clear from these properties that if the $k$-form $\omega$ is given in local coordinates by \eqref{kform} above then
\[
d\omega = \sum_{\alpha_1 < \cdots < \alpha_k} \sum_{\alpha} \partial_\alpha \omega_{\alpha_1 \cdots \alpha_k} dx^{\alpha} \wedge dx^{\alpha_1} \wedge \cdots \wedge dx^{\alpha_k}.
\]

The last property in Proposition~\ref{propertiesforms} has a converse, known as the {\bf Poincar\'e Lemma}.

\begin{Lemma}{\em (Poincar\'e)}
If $d \omega = 0$ then locally $\omega = d \eta$.
\end{Lemma}

A related result is the {\bf Frobenius Theorem}. Here we present a particular case of this result.

\begin{Thm}{\em (Frobenius)}
The nonvanishing $1$-form $\omega$ is locally orthogonal to a family of hypersurfaces if and only if $\omega \wedge d \omega = 0$.
\end{Thm}

To prove the easy direction in this equivalence it suffices to note that $\omega$ is locally orthogonal to a family of hypersurfaces $\{ f = \text{constant} \}$ if and only if $\omega = g df$ for some nonvanishing function $g$. Note that in particular this is always true for $1$-forms in $2$-dimensional manifolds.

We will now assume that our $n$-dimensional manifold is {\bf oriented}, that is, that an orientation can be, and has been, consistently chosen on every tangent space. Any $n$-form $\omega$ is written in local coordinates as
\[
\omega = a(x) dx^1 \wedge \cdots \wedge dx^n.
\]
If the coordinate system is positive, that is, if the coordinate basis $\{ \partial_1, \ldots, \partial_n \}$ has positive orientation at all points, we define
\[
\int_U \omega = \int_{x(U)} a(x) dx^1 \ldots dx^n,
\]
where $U$ is the coordinate neighborhood. This formula does not depend on the choice of local coordinates because $dx^1 \wedge \cdots \wedge dx^n$ transforms by the determinant of the change of variables. 

\begin{Thm}{\em (Stokes)}
If $M$ is an oriented $n$-dimensional manifold with boundary $\partial M$ and $\omega$ is an $(n-1)$-form then
\[
\int_M d\omega = \int_{\partial M} \omega.
\]
\end{Thm}

In this theorem the orientations of $M$ and $\partial M$ are related as follows: if $\partial M$ is a level set of $x^1$ in a positive coordinate system and $\partial_1$ points outwards then the coordinate system $(x^2, \ldots, x^n)$ on $\partial M$ is positive.

If $M$ has a metric $g$ then its {\bf volume element} is the $n$-form $\epsilon$ which is $1$ when contracted with a positive orthonormal frame. In positive local coordinates we have
\[
\epsilon = \sqrt{|\det(g_{\mu\nu})|} \, dx^1 \wedge \cdots \wedge dx^n.
\]
It is easily seen that $\nabla \epsilon = 0$, where $\nabla$ is the Levi-Civita connection. If $\omega$ is a $k$-form then its {\bf Hodge dual} is the $(n-k)$-form $\star \omega$ given by
\[
(\star \omega)_{\beta_1 \cdots \beta_{n-k}} = \frac1{k!} \omega^{\alpha_1 \cdots \alpha_k} \epsilon_{\alpha_1 \cdots \alpha_k \beta_1 \cdots \beta_{n-k}}.
\]
The operator $\star$, called the Hodge star, can alternatively be defined as follows: if $\{\omega^1, \ldots, \omega^n \}$ is any positively oriented orthonormal coframe (so that the volume element is $\epsilon = \omega^1 \wedge \ldots \wedge \omega^n$) then
\begin{align*}
\star (\omega^1 \wedge \ldots \wedge \omega^k) & = g(\omega^1, \omega^1) \cdots g(\omega^k, \omega^k) \, \omega^{k+1} \wedge \ldots \wedge \omega^n \\
& = \epsilon((\omega^1)^\sharp, \ldots, (\omega^k)^\sharp, \ldots).
\end{align*}

\section*{Lie derivative} \label{secA.5}

A vector field $X$ can be identified with the differential operator that corresponds to taking derivatives along $X$. In local coordinates, this operator is given by
\[
X \cdot f = X^\mu \partial_\mu f.
\]
It turns out that the commutator of two vector fields $X$ and $Y$, regarded as differential operators, is also a vector field:
\begin{align*}
[X,Y] \cdot f  & = X \cdot (Y^\mu \partial_\mu f) - Y \cdot (X^\mu \partial_\mu f) \\
& = (X \cdot Y^\mu) \partial_\mu f + Y^\mu X^\nu \partial_\nu \partial_\mu f - (Y \cdot X^\mu) \partial_\mu f - X^\mu Y^\nu \partial_\nu \partial_\mu f \\
& = (X \cdot Y^\mu - Y \cdot X^\mu) \partial_\mu f.
\end{align*}

\begin{Def}
The {\bf Lie bracket} of two vector fields $X$ and $Y$ is the vector field
\[
[X,Y] = (X \cdot Y^\mu - Y \cdot X^\mu) \partial_\mu.
\]
\end{Def}

This operation is intimately related with the exterior derivative.

\begin{Prop}
If $\omega$ is a $1$-form then
\[
d \omega(X,Y) = X \cdot \omega(Y) - Y \cdot \omega(X) - \omega([X,Y])
\]
for all vector fields $X$ and $Y$.
\end{Prop}

\begin{proof}
In local coordinates we have
\begin{align*}
d \omega(X,Y) & = ( \partial_\mu \omega_\nu - \partial_\nu \omega_\mu ) X^\mu Y^\nu = Y^\nu X \cdot \omega_\nu - X^\nu Y \cdot \omega_\nu \\
& = X \cdot (\omega_\nu Y^\nu) - \omega_\nu X \cdot Y^\nu - Y \cdot (\omega_\nu X^\nu) + \omega_\nu Y \cdot X^\nu \\
& = X \cdot \omega(Y) - Y \cdot \omega(X) - \omega_\nu (X \cdot Y^\nu - Y \cdot X^\nu).
\end{align*}
\end{proof}

If the vector field $X$ is nonzero at some point $p$ then there exists a coordinate system defined in a neighborhood of $p$ such that $X = \partial_1$. In fact, we just have to fix local coordinates $(x^2, \ldots, x^n)$ on a hypersurface $\Sigma$ transverse to $X$ at $p$ and let $x^1$ be the parameter for the flow of $X$ starting at $\Sigma$. If $T$ is any tensor, we define its {\bf Lie derivative} along $X$ as the tensor with components
\[
(\cL_X T)^{\alpha_1 \cdots \alpha_k}_{\beta_1 \ldots \beta_l} = \partial_1 T^{\alpha_1 \cdots \alpha_k}_{\beta_1 \ldots \beta_l}.
\]
This can be extended to points where $X$ vanishes by continuity. Although this definition seems to depend on the coordinate system, it is actually invariant. To check this, we just have to find an invariant expression for the Lie derivative of functions, vector fields and $1$-forms and then notice that the Leibnitz rule applies.

\begin{Prop}
If $X$ is a vector field then:
\begin{enumerate}
\item
$\cL_X f = X \cdot f$ for functions $f$;
\item
$\cL_X Y = [X,Y]$ for vector fields $Y$;
\item
$\cL_X \omega = X \contr d \omega + d (X \contr \omega)$ for $1$-forms $\omega$
\end{enumerate}
(where $\contr$ means contraction in the first index).
\end{Prop}

\begin{proof}
The formula for functions is immediate. In the coordinate system where $X = \partial_1$,
\[
\cL_X Y = \partial_1 Y^\mu \partial_\mu = (X \cdot Y^\mu - Y \cdot X^\mu) \partial_\mu = [X,Y].
\]
Finally, we have
\begin{align*}
(X \contr d \omega + d (X \contr \omega))(Y) & = d \omega(X,Y) + Y \cdot \omega(X) = X \cdot \omega(Y) - \omega([X,Y]) \\
& = \cL_X (\omega(Y)) - \omega(\cL_XY) = (\cL_X\omega) (Y),
\end{align*}
where we used the Leibnitz rule. This formula is sometimes called {\bf Cartan's magic formula}.
\end{proof}

\section*{Cartan structure equations} \label{secA.6}

Let $\{ E_\mu \}$ be an orthonormal frame, and $\{ \omega^\mu \}$ the corresponding orthonormal coframe, so that
\[
\omega^\mu(E_\nu)=\delta^\mu_{\,\,\,\,\nu}.
\]
Note that the metric can be written as
\[
ds^2 = \eta_{\mu\nu} \, \omega^\mu \otimes \omega^\nu,
\]
where $(\eta_{\mu\nu})=\diag(-1,1,1,1)$ is the flat space metric (which we will use to raise and lower indices).

\begin{Def}
The {\bf connection forms} associated to the orthonormal frame $\{ E_\mu \}$ are the $1$-forms $\omega^\mu_{\,\,\,\,\nu}$ such that
\[
\nabla_X E_\nu = \omega^\mu_{\,\,\,\,\nu}(X) E_\mu
\]
for all vector fields $X$. The {\bf curvature forms} associated this frame are the $2$-forms $\Omega^\mu_{\,\,\,\,\nu}$ such that
\[
R(X,Y) E_\nu =  \Omega^\mu_{\,\,\,\,\nu}(X,Y) E_\mu
\]
for all vector fields $X, Y$.
\end{Def}

Note that the components of the Riemann tensor in the orthonormal frame can be retrieved from the curvature forms by noticing that
\[
\Omega^\mu_{\,\,\,\,\nu} = R_{\alpha\beta\,\,\,\,\nu}^{\,\,\,\,\,\,\,\,\mu} \, \omega^\alpha \otimes \omega^\beta = \sum_{\alpha < \beta} R_{\alpha\beta\,\,\,\,\nu}^{\,\,\,\,\,\,\,\,\mu} \, \omega^\alpha \wedge \omega^\beta.
\]
These forms can be computed by using the so-called {\bf Cartan structure equations}. This is by far the most efficient way to compute the curvature. 

\begin{Thm}
The connection forms are the unique solution of {\bf Cartan's first structure equations}
\[
\begin{cases}
\omega_{\mu\nu}=-\omega_{\nu\mu}\\
d\omega^\mu + \omega^\mu_{\,\,\,\,\nu} \wedge \omega^\nu = 0
\end{cases},
\]
and the curvature forms are given by {\bf Cartan's second structure equations}
\[
\Omega^\mu_{\,\,\,\,\nu} = d\omega^\mu_{\,\,\,\,\nu} + \omega^\mu_{\,\,\,\,\alpha} \wedge \omega^\alpha_{\,\,\,\,\nu} \,.
\]
\end{Thm}

\begin{proof}
The first condition is equivalent to
\[
X \cdot \langle E_\mu, E_\nu \rangle = 0
\]
for all vector fields $X$, which in turn is equivalent to the compatibility of the connection with the metric. Using
\[
d \omega^\mu (X,Y) = X \cdot \omega^\mu(Y) - Y \cdot \omega^\mu(X) - \omega^\mu([X,Y])
\]
for all vector fields $X$ and $Y$, it is easy to see that the second condition is equivalent to
\[
[E_\alpha, E_\beta] = \nabla_{E_\alpha} E_\beta - \nabla_{E_\beta} E_\alpha,
\]
which in turn is equivalent to the symmetry of the connection. Since the Levi-Civita connection is the only connection which is symmetric and compatible with the metric, we conclude that Cartan's first structure equations have a unique solution.

Finally, the third condition can be derived by writing 
\begin{align*}
\Omega^\mu_{\,\,\,\,\nu}(X,Y) E_\mu = R(X,Y) E_\nu = \nabla_X \nabla_Y E_\nu - \nabla_Y \nabla_X E_\nu - \nabla_{[X,Y]} E_\nu
\end{align*}
in terms of the connection forms.
\end{proof}


\bibliographystyle{amsalpha}
\bibliography{math_rel}
\printindex
\end{document}